\documentclass[12pt]{article}
\usepackage{amsmath,amssymb,amsthm,bm}

\usepackage{graphicx}  
\usepackage{graphicx,subfig,psfrag,epsf,bm}
\usepackage{multirow}
\usepackage{enumerate,array,booktabs}
\usepackage{natbib}
\usepackage[hidelinks]{hyperref}
\usepackage{enumitem}
\usepackage{caption}
\usepackage{rotating}
\usepackage {mathabx}
\usepackage{amsmath,color}
\usepackage{amsfonts}
\usepackage{dsfont}
\usepackage{multirow}
\usepackage{pdflscape}
\usepackage{alltt}
\usepackage[T1]{fontenc}
\usepackage{bm}
\usepackage{pifont}
\usepackage{booktabs}
\IfFileExists{upquote.sty}{\usepackage{upquote}}{}
\usepackage{algorithm}
\usepackage{algpseudocode} 

\usepackage{setspace}
\usepackage{etoolbox}
\AtBeginEnvironment{algorithmic}{\setstretch{1.7}}

\usepackage{setspace}
\makeatletter
\def\singlespace{\def\baselinestretch{1}\@normalsize}

\newcommand{\kV}{{\cal V}}
\newcommand{\kU}{{\cal U}}

\newcommand{\kG}{{\cal G}} 
\newcommand{\kA}{{\cal A}}
\newcommand{\kS}{{\cal S}}

\newcommand{\vd}{{\bm d}}
\newcommand{\ve}{{\bm e}}

\newcommand{\vS}{{\bm S}}
\newcommand{\vE}{{\bm E}}

\newcommand{\vr}{{\bm r}}

\newcommand{\vu}{{\bm u}}
\newcommand{\vv}{{\bm v}}
\newcommand{\vw}{{\bm w}}
\newcommand{\vx}{{\bm x}}
\newcommand{\vy}{{\bm y}}

\newcommand{\vA}{{\bm A}}

\newcommand{\vI}{{\bm I}}
\newcommand{\vJ}{{\bm J}}

\newcommand{\vnull}{{\bm 0}}

\newcommand{\lam}{{\lambda}}

\newcommand{\vDelta}{\bm{\Delta}}

\newcommand{\vSigma}{\bm{\Sigma}}
\newcommand{\vOmega}{\bm{\Omega}}
\newcommand{\vTheta}{\bm{\Theta}}
\newcommand{\valpha}{\bm{\alpha}}
\newcommand{\vbeta}{\bm{\beta}}
\newcommand{\vgamma}{\bm{\gamma}}

\newcommand{\veta}{\bm{\eta}}

\newcommand{\vLam}{\bm{\Lambda}}

\newcommand{\vnu}{\bm{\nu}}

\newcommand{\vtheta}{\bm{\theta}}
\newcommand{\vxi}{\bm{\xi}}
\newcommand{\vphi}{\bm{\phi}}

\newcommand{\vkappa}{\bm{\kappa}}
\newcommand{\ep}{\epsilon}

\newcommand{\bay}{\begin{array}}
	\newcommand{\eay}{\end{array}}

\newcommand{\vxh}{\widehat{\vx}}
\newcommand{\xh}{\widehat{x}}

\newcommand{\vdh}{\widehat{\vd}}

\newcommand{\vSigmah}{\widehat{\vSigma}}
\newcommand{\vOmegah}{\widehat{\vOmega}} 
\newcommand{\vbetah}{\widehat{\vbeta}}
\newcommand{\vthetah}{\widehat{\vtheta}}
\newcommand{\vThetah}{\widehat{\vTheta}}

\newcommand{\vkappah}{\widehat{\vkappa}}
\newcommand{\vetah}{\widehat{\veta}}

\newcommand{\vxw}{\widetilde{\vx}}
\newcommand{\xw}{\widetilde{x}}

\newcommand{\vbetaw}{\widetilde{\vbeta}}

\newcommand{\E}{\mbox{E}}
\newcommand{\Var}{\mbox{Var}}
\newcommand{\Cov}{\mbox{Cov}}

\newcommand{\ra}{\rightarrow}

\newcommand{\rad}{\xrightarrow{d}}
\newcommand{\rap}{\xrightarrow{p}}

\newtheorem{theorem}{Theorem}
\newtheorem{lemma}{Lemma} 
\newtheorem{corollary}{Corollary}
\newtheorem{definition}{Definition}

\newcommand{\bqa}{\begin{eqnarray*}}
	\newcommand{\eqa}{\end{eqnarray*}}
\newcommand{\bqan}{\begin{eqnarray}}
\newcommand{\eqan}{\end{eqnarray}}
\newcommand{\bqt}{\begin{quote}}
	\newcommand{\eqt}{\end{quote}}
\newcommand{\bt}{\begin{tabbing}}
	\newcommand{\et}{\end{tabbing}}
\newcommand{\bit}{\begin{itemize}}
	\newcommand{\eit}{\end{itemize}}
\newcommand{\ben}{\begin{enumerate}}
	\newcommand{\een}{\end{enumerate}}
\newcommand{\beq}{\begin{equation}}
\newcommand{\eeq}{\end{equation}}
\newcommand{\bdefi}{\begin{definition}}
	\newcommand{\edefi}{\end{definition}}
\newcommand{\bpro}{\begin{proposition}}
	\newcommand{\epro}{\end{proposition}}
\newcommand{\blem}{\begin{lemma}}
	\newcommand{\elem}{\end{lemma}}
\newcommand{\bco}{\begin{corollary}}
	\newcommand{\eco}{\end{corollary}}
\newcommand{\bdes}{\begin{description}}
	\newcommand{\edes}{\end{description}}

\newcommand{\bbR}{\mathbb{R}}
\newcommand{\bbK}{\mathbb{K}}
\newcommand{\bbB}{\mathbb{B}}
\newcommand{\bbE}{\mathbb{E}}

\newcommand{\bbM}{\mathbb{M}}

\newcommand{\bbU}{\mathbb{U}} 
\newcommand{\bbP}{\mathbb{P}}

\def\boxit#1{\vbox{\hrule\hbox{\vrule\kern6pt
			\vbox{\kern6pt#1\kern6pt}\kern6pt\vrule}\hrule}}

\newcommand{\sgn}{\mbox{sgn}}
\DeclareMathOperator*{\argmin}{arg\,min}

\usepackage{xcite}
\usepackage{xr}
\makeatletter
\newcommand*{\addFileDependency}[1]{% argument=file name and extension
	\typeout{(#1)}
	\@addtofilelist{#1} 
	\IfFileExists{#1}{}{\typeout{No file #1.}} 
}
\makeatother

\begin{document}

	\title{\large
		\bf Model-Assisted Uniformly Honest Inference for Optimal Treatment Regimes in High Dimension}
	\author{Yunan Wu, Lan Wang and Haoda Fu}
	\date{}
	
	\maketitle

	\begin{singlespace}
		\begin{footnotetext}[1]
			{ Yunan Wu is Assistant Professor, Department of Mathematical Sciences, University of Texas at Dallas. Emails: yunan.wu@utdallas.edu. 
				Lan Wang is Professor, Department of Management Science, University of
				Miami. Emails: lanwang@mbs.miami.edu.
				Wang and Wu's research was partly supported by
				NSF DMS-1952373 and
				NSF OAC-1940160.
				Dr. Haoda Fu is Research Fellow, Enterprise Lead for Machine Learning and AI,
                Eli Lilly and Company.  Email: fu\_haoda@lilly.com. Wang and Wu's research was partly supported by
				NSF DMS-1952373 and
				NSF OAC-1940160.
				The authors are grateful to the referees, the associate editor and the Co-editor for their
				valuable comments, which have significantly improved the paper.}	
		\end{footnotetext}
	\end{singlespace}

	\bigskip
	\begin{abstract}
		This paper develops new tools to quantify uncertainty in optimal decision making and to gain insight into 
		which variables one should collect information about given the potential cost of measuring a large number of variables.
		We investigate simultaneous inference to determine if a group of variables is relevant for estimating an optimal decision rule in
		a high-dimensional semiparametric framework.
		The unknown link function permits flexible modeling of the interactions between the treatment and the covariates, but 
		leads to nonconvex estimation in high dimension and
		imposes significant challenges for inference.  We first establish that a local restricted strong convexity 
		condition holds with high probability and that
		any feasible local sparse solution of the estimation problem can achieve the near-oracle
		estimation error bound. We further rigorously verify that a wild bootstrap procedure based on a debiased version of the local solution 
		can provide asymptotically honest uniform inference for the effect of a group of variables on optimal decision making. 
		The advantage of honest inference is that it does not require the initial estimator to achieve perfect model selection 
		and does not require the zero 
		and nonzero effects to be well-separated.	We also propose an efficient algorithm for estimation. 
		Our simulations suggest satisfactory performance. An example from a diabetes study illustrates the real
		application. 
	\end{abstract}
	
	\noindent%
	
	%\newpage
	%\spacingset{1.5} % DON'T change the spacing!

	\section{Introduction}	\label{sec:intro}
	
	Precision medicine is an innovative practice for disease treatment that takes into account individual variability in genes, environment, and lifestyle for each patient. 
	Substantial efforts have recently been devoted to studying how to estimate the optimal personalized treatment regime 
	given the individual-level information, which aims to
	yield the best expected outcome if the treatment regime is followed by each individual in the population.
	Several successful approaches have been developed for this estimation problem, 
	including Q-learning and A-learning based methods \citep{watkins1992q, robins2000, Murphy2003, moodie2010, Qian2011}, 
	and classification-based methods
	\citep{zhang2012robust, zhao2012estimating, zhao2015,  wang2018, qi2018d}, among others.
	We refer to \citet{CM2013} and \citet{KM2016} for a general introduction to this area and other relevant references.
	
	Inference or uncertainty quantification is important in practice. This paper studies the following inference problem for optimal 
	personalized decision making: 
	suppose we have a large number of covariates 
	(e.g., hundreds of genes),
	how will we determine if a given subset of covariates (e.g., genes associated with a given biological pathway)
	is relevant for making the optimal treatment recommendation?   
	Scientifically, this knowledge would enable the doctors and researchers to identify critical characteristics (e.g., gender, age, gene pathways) that
	are influential for the optimal decision. It also helps gain insight into what information
	is worth collecting to be more cost effective, given the possibility of measuring a large number of variables (genetic, clinic, etc).

	In the last few years, important progress has been made in inference with optimal decision rules. 
	\citet{laber2010} developed a novel locally consistent adaptive confidence
	interval for the Q-learning approach.  \citet{chakraborty2013inference} proposed a practically convenient adaptive $m$-out-of-$n$ bootstrap
	method for inference for Q-learning. \citet{Song2015} studied penalized Q-learning.
	\citet{Jeng2018} developed Lasso-based debiased procedure 
	for A-learning.
	Different but related, \citet{Chakraborty2014} 
	and \citet{Luedtke2016}, \citet{zhu2018proper} developed confidence intervals for another quantity of interest: the value function.
	However, existing work mostly deals with the classical
	asymptotic setting of fixed $p$ and large $n$, where $p$ is the number of covariates and $n$ is the 
	sample size, and have not addressed
	the challenge of inference with high-dimensional variables.
	Moreover, the aforementioned work often assumes that the interaction between the covariates and the treatment has a known functional form.

	Motivated by the overarching goal of precision medicine to incorporate genetic information (e.g, measurements on thousands
	of genes) in the decision making process, this paper investigates inference about the effect of a group of variables on the optimal decision rule
	in the high-dimensional setting. 
	The existing frameworks are known to face challenges for the purpose of inference in high dimension.
	The Q-learning approach is prone to model-misspecification. 
	Robust model-free procedures 
	that directly estimate the Bayes rule (e.g., \citet{zhang2012robust})
	have a nonstandard convergence rate, see for example, the recent analysis in 
	\citet{wang2018} on the cubic-root convergence rate. 
	On the other hand, the theory of Hinge-loss based O-learning (\cite{ zhao2012estimating})
	has been focused on the generalization error bound. Inference 
	for the Bayes rule based on the nonsmooth surrogate loss
	is very challenging in high dimension.
	We alleviate the above difficulty by
	adopting a flexible semiparametric model-assisted approach for 
	optimal decision estimation and inference. The semiparametric structure permits nonparametric main effects
	and nonlinear interaction effect between the covariates and treatment
	via an unknown smooth link function. This semiparametric framework incorporates many existing models as special cases.
	
	When the interaction effects are nonlinear, the 
	parameter indexing the optimal decision rule does not necessarily correspond
	to the solution of a convex problem.
	For inference, we first propose and study a preliminary
	estimator based on a high-dimensional penalized
	profile estimation equation.
	This estimator is motivated by earlier work on classical single-index models 
	(e.g., \citet{Powell1989}, \citet{duan1991slicing},
	\citet{ichimura1993}, \citet{ZhuXue}, \citet{carroll1997},
	\citet{xia1999extended}, \citet{yu2002penalized},
	\citet{wang2010estimation},
	\citet{ma2013}, \citet{ma2016}, among others).
	Several paper recently studied estimation 
	for  high-dimensional single-index models (e.g.,
	\citet{radchenko2015high}, \citet{neykov2016l1},
	\citet{yang2017high}, \citet{lin2019sparse}, among others) but 
	focused on statistical properties of the global solution which may not be numerically achieved
	due to the nonconvex nature of the problem. 
	Adopting tools from modern
	empirical process and random matrix theory,
	we establish that a local restricted strong convexity 
	condition holds with high probability in high dimension and that
	any local sparse solution of the penalized estimation equation can achieve desirable
	estimation accuracy. Moreover, we propose a new algorithm for efficient computation in high dimension. 
	%we derive near-optimal bound for
	%feasible local solutions, under relatively weaker regularity conditions.
	
	Our research also makes new contributions to statistical inference in high-dimensional semiparametric
	models.
	Recent work on inference has been mostly limited to linear regression or generalized linear
	regression, see
	\citet{zhang2014confidence,van2014asymptotically,javanmard2014confidence, 
		Belloni2015uniform, cai2017confidence, 
		Ning2017, zhang2017, zhu2018linear, Shi2020}, among others.
	High-dimensional inference in the semiparametric setting with 
	estimated nonparametric components
	is a substantially harder problem and has been little studied.
	We have a particularly challenging setting where the parameter of interest and nonparametric
	component are bundled together, that is, the nuisance functions
	depend on the parameter of interest (\citet{ding2011sieve}).
	So far, statistical inference
	for single-index model has mostly been limited to
	the lower-dimensional setting (e.g.,
	\citet{liang2010estimation}), \citet{gueuning2016confidence}).
	
	Our approach is inspired by the de-biasing (or de-sparsifying) idea 
	proposed in 
	\citet{zhang2014confidence} and \citet{van2014asymptotically}, which intuitively can be thought of
	inverting the
	Karush-Kuhn-Tucker conditions \citep{van2014asymptotically}.
	We generalize this idea to the semiparametric setting and prove 
	that valid honest uniform inference can be obtained based on a debiased version of a local solution.
	Specifically, we derive simultaneous confidence intervals for inference on a group of variables 
	while allowing the number of covariates to exceed the sample size. The confidence intervals enjoy the {\it honest}
	property in the following sense
	\bqa
	\sup_{\vbeta_0: ||\vbeta_0||_0\leq s}\sup_{\alpha\in(0,1)}\Big|P\Big( \sqrt{n}\max_{j\in\kG}|\widetilde{\beta}_j-\beta_{0j}|\leq c^*_{1-\alpha}  \Big)-(1-\alpha)\Big| = o(1),
	\eqa
	where $\vbeta_0=(\beta_{01},\ldots, \beta_{0p})^T $ is the population parameter indexing the optimal treatment regime,
	$\widetilde{\beta}_j$'s denote debiased estimators that will be introduced later, $\kG$ denotes the group of variables of interest,
	$||\cdot||_{0}$ denotes the $l_0$ norm of a vector, and $s$ is a positive integer denoting the sparsity size.
	The significance of the honest property is that the coverage probability is asymptotically 
	valid uniformly over a class of $s$-sparse models.
	An immediate implication is that it relaxes the assumption on signal strength and does not require the  
	zero and nonzero effects to be well-separated (so-called $\beta_{\mbox{min}}$ condition). 
	In particular, this procedure does not require the initial estimator to achieve perfect model selection.
	It avoids the problems associated with the nonuniformity of the limiting theory for penalized
	estimators, see discussions in \citet{li1989honest, potscher2009confidence, van2014asymptotically, mckeague2015}, among others.
	It is also worth noting that
	the number of variables in $\kG$ can be 
	either small or large. For example, one may be interested in assessing how a group of genes
	corresponding to a particular biological pathway, the size of which can be comparable with or even larger than the sample size, 
	affect optimal decision making. The critical value $c^*_{1-\alpha}$ is obtained using a wild bootstrap 
	procedure, which automatically accounts for the dependence of the coordinates for testing component-wise hypotheses and
	leads to more accurate finite-sample performance.

	The remainder of the paper is organized as follows. Section 2 introduces the new methodology.
	Section 3 studies the statistical properties. Section 4 provides the details on computation and reports
	numerical results from Monte Carlo studies. Section 5 illustrates the new methods on a real data example from a diabetes study.
	Section 6 discusses some extensions. %The appendix summaries the regularity conditions and presents several useful technical lemmas.
	The regularity conditions,  all the proofs and additional numerical examples are given in the online supplementary material.

	\section{Methodology}\label{sec:method}	
	\subsection{A Semiparametric Framework}\label{sec:notation}
	
	For notational simplicity, we will focus on the binary decision setting. 
	Let $A\in \mathcal{A}=\{0,1\}$ denote a binary treatment
	and $\vx\in \mathcal{X}$ denote a $p$-dimensional vector of baseline covariates.
	Let $Y$ denote the outcome of interest.  Without loss of generality, we assume a larger value
	of the outcome is preferred. 
	The observed data consist of $\{(\vx_i,A_i,Y_i): i=1,\cdots,n\}$.
	We are interested in the setting where $p\gg n$.
	
	A treatment regime is an individualized decision rule that can be represented as a function $d(\vx): \mathcal{X}\rightarrow \mathcal{A}$.
	The optimal treatment regime is defined as the decision rule which, if followed by the whole population, will 
	achieve the largest average outcome. Formally, it is defined using the potential outcome framework in causal inference 
	\citep{Neyman1990,Rubin78}.
	Let $Y^*(a)$ be the potential outcome had the subject been assigned to treatment $a\in \{0,1\}$.
	Given a treatment regime $d(\vx)$,  the corresponding potential
	outcome is $Y^*(d)=Y^*(1)d(\vx)+Y^*(0)(1-d(\vx))$. The optimal treatment regime is defined as
	$d^{\scalebox{.8}{\mbox{opt}}}(\vx)=\arg \max_{d}\E\{Y^*(d)\}$.
	It is now well known that $d^{\scalebox{.8}{\mbox{opt}}}(\vx)=\arg \max_{a\in \mathcal{A}}\E(Y|\vx, A=a)$ \citep{Qian2011}.

	This paper considers a 
	flexible semiparametric framework for optimal treatment regime estimation and inference
	in the high-dimensional setting.
	Specifically, we assume
	\begin{align}
	Y_i = g(\vx_i) +(A_i-1/2) f_0(\vx_i^T\vbeta_0)+\epsilon_i,\quad i=1, \ldots, n, \label{model}
	\end{align}
	where $\vbeta_0=(\beta_{01},\beta_{02},\cdots,\beta_{0p})^T$, $g(\vx_i)$ is the unknown main effect, 
	and $f_0(\cdot)$ is an unknown  function that describes the interaction between the
	treatment and covariates, and the random error $\ep_i$ satisfies
	$\E(\epsilon_i|\vx_i)=0$, $i=1, \ldots, n$.
	For identification purpose, we assume that there exists a relevant covariate which has a continuous density given the other covariates \citep{ichimura1993}. Such an identification condition is required even in the lower-dimensional setting when the true model is known. Without loss of generality,  we assume that the first covariate
	$x_1$ satisfies this condition and normalize its coefficient $\beta_{01}$ such that $\beta_{01}=1$, see Remark (c) in Section~S2 of the online supplementary material for more discussions on the identifiability condition.   We denote $\bbB_0 = \{\vbeta=(\beta_1,\cdots,\beta_p)^T:   \beta_1=1\}$ as the candidate set for $\vbeta_0$.
	Under model (\ref{model}), the optimal treatment regime is $d^{\scalebox{.8}{\mbox{opt}}}(\vx)=\mbox{I}\big(f_0(\vx_i^T\vbeta_0)>0\big)$, where $\mbox{I}(\cdot)$
	denotes the indicator function. Note that the class of index rules are popular in practice due to its interpretability.
	%and the fact that the optimal treatment regime corresponding to a class of commonly used parametric models has this form.

	Existing work on inference for optimal treatment regime is mostly based on a parametric
	generative model, which is prone to model misspecification. 
	The semiparametric structure alleviates this difficulty. In particular, it  
	allows for possible nonlinear interaction effects between the covariates and treatment.
	It also circumvents the curse of dimensionality associated with a fully nonparametric model.

	Our goal is to estimate $\vbeta_0$ and make inference on its components in the high-dimensional setting.
	In the special case $f_0(u)=u$, which is popularly used in practice, the problem can be formulated
	as a high-dimensional convex estimation problem. However, when $f_0$ is nonlinear,
	it generally leads to a high-dimensional nonconvex problem. Both estimation and inference 
	need to overcome new challenges.

	\subsection{Profiled Semiparametric Estimation}\label{sec:est_method}
	We start with introducing a 
	penalized profiled semiparametric estimation equation for estimating the parameter indexing the optimal treatment regime.
	We consider data from a random
	experiment, that is, $P(A_i=0)=P(A_i=1)=1/2$, $i=1, \ldots, n$. Extension to data from observational studies is discussed in Section~\ref{sec:discuss}. 
	Inspired by an observation made for the linear model \citep{Tian}, we observe
	\begin{align}
	2(2A_i-1)Y_i =  f_0(\vx_i^T\vbeta_0)+2(2A_i-1)\big[\epsilon_i+g(\vx_i)\big].\label{modified_model1}
	\end{align}
	Let $\widetilde{Y}_i=2(2A_i-1)Y_i$ be the modified response, and let $\widetilde{\epsilon}_i=2(2A_i-1)\big[\epsilon_i+g(\vx_i)\big]$ be the modified error. 
	We have
	\begin{align}
	\E\{\widetilde{Y}_i| \vx_i\}=f_0(\vx_i^T\vbeta_0).\label{modified_model2}
	\end{align} 
	In the ideal situation where the link function $f_0$ is known, we have $\vbeta_0=\arg \min_{\vbeta}\E\big[\widetilde{Y}_i-f_0(\vx_i^T\vbeta_0)\big]^2$.
	It is noteworthy that for a nonlinear function $f_0$, the objective function is usually nonconvex in $\vbeta$.
	\citet{ichimura1993} carefully studied the properties of the global minimizer for a semiparametric nonlinear least-squares approach
	in the classical finite-dimensional setting.
	
	To estimate $\vbeta_0$ in the high-dimensional setting with an known $f_0$, we consider a penalized profiled semiparametric estimation equation. 
	In the ideal situation where $f_0$ is known a prior,  
	$\vbeta_0$ satisfies the following unbiased estimating equation
	\bqan\label{score0}
	\E\big\{\big[\widetilde{Y}_i-f_0(\vx_i^T\vbeta_0)\big]f'_0(\vx_i^T\vbeta_0)\vx_i\big\}=\vnull,
	\eqan
	where $f_0'(\cdot)$ denotes the derivative of $f_0(\cdot)$. We will replace the unknown
	$f_0$ and $f_0'$ by their respective profiled nonparametric estimator,
	and consider an appropriately penalized version of the estimated score function to handle the high-dimensional covariates.
	
	We summarize the main steps of estimation as follows. 
	%First, to estimate $f_0(u)$ for a given parameter $\vbeta$, we employ the Nadaraya-Watson estimator estimator. 
	Define  $G(t|\vbeta)=\E\{\widetilde{Y} | \vx^T\vbeta=t\}$. Note that $G(t|\vbeta_0) = f_0(t)$.
	However, when $\vbeta\neq \vbeta_0$, $G(t|\vbeta)$ usually has a functional form different from $f_0$.
	\citet{ichimura1993} showed that $\frac{\partial G(\vx_i^T\vbeta|\vbeta)}{\partial \vbeta}\approx f_0'(\vx_i^T\vbeta)\big[ \vx_i-\E(\vx_i|\vx_i^T\vbeta)\big]^T $ for $\vbeta$
	close to $\vbeta_0$.
	Consider the Nadaraya-Watson kernel estimator for $G(t|\vbeta) $:
	\bqan
	\widehat{G}(t|\vbeta) = \sum_{i=1}^n  W_{ni}(t,\vbeta)\widetilde{Y}_i,\label{Ghat_def}
	\eqan
	where $K_h(z) = h^{-1}K(z/h)$, and $ W_{ni}(t,\vbeta) = \frac{K_h(t-\vx_i^T\vbeta)}{\sum_{j=1}^nK_h(t-\vx_j^T\vbeta)}.$
	Write $G^{(1)}(t|\vbeta) = \frac{d }{d t} G(t|\vbeta) $ and $W_{ni}^{(1)}(t,\vbeta) =  \frac{d}{d t}W_{ni}(t,\vbeta) $. Then
	the kernel estimator for the derivative $G^{(1)}(t|\vbeta)$ is 
	\bqan
	\widehat{G}^{(1)}(t|\vbeta) = \sum_{i=1}^n  W_{ni}^{(1)}(t,\vbeta)\widetilde{Y}_i.\label{G1hat_def}
	\eqan
	Write $G(\vx^T\vbeta|\vbeta) = \E\{\widetilde{Y} | \vx^T\vbeta\}$.  
	%Denote $G(\vx^T \vbeta) = G(\vx^T \vbeta|\vbeta) $, $G^{(1)}(\vx^T \vbeta) = G^{(1)}(\vx^T \vbeta|\vbeta) $, $\widehat{G}(\vx^T \vbeta) = \widehat{G}(\vx^T \vbeta| \vbeta)$ and $\widehat{G}^{(1)}(\vx^T \vbeta) = \widehat{G}^{(1)}(\vx^T \vbeta| \vbeta)$ for notational simplicity. 
	To estimate $\widehat{G}(\vx_j^T \vbeta|\vbeta) $ and $\widehat{G}^{(1)}(\vx_j^T \vbeta|\vbeta) $, we employ the following leave-one-out estimators
	\bqan
	\widehat{G}(\vx_j^T \vbeta|\vbeta) = \sum_{i=1, i\neq j}^n W_{nij}(\vx_j^T \vbeta,\vbeta)\widetilde{Y}_i,\label{Ghat_def_j}\quad 
	\widehat{G}^{(1)}(\vx_j^T \vbeta|\vbeta) = \sum_{i=1,i\neq j}^n W_{nij}^{(1)}(\vx_j^T \vbeta,\vbeta)\widetilde{Y}_i,\label{G1hat_def_j}
	\eqan
	where $W_{nij}(\vx_j^T \vbeta,\vbeta)  = \frac{K_h(\vx_j^T \vbeta-\vx_i^T\vbeta)}{\sum_{k\neq j}K_h(\vx_j^T \vbeta-\vx_k^T\vbeta)},$ and $W_{nij}^{(1)}(\vx_j^T \vbeta,\vbeta) = \frac{d}{d t}W_{nij}(t,\vbeta)\Big|_{t=\vx_j^T \vbeta} $.
	Similarly, we estimate $\E(\vx|\vx^T\vbeta_0)$ by
	$
	\widehat{\E}(\vx_j|\vx_j^T \vbeta) = \sum_{i=1, i\neq j}^n W_{ni}(\vx_j^T \vbeta,\vbeta)\vx_i.
	%\label{Ehat_def_j}
	$
	Denote  $\vx_i = (x_{i,1},\vx_{i,-1}^T)^T$.
	Motivated by the semiparaemtric efficient score derived in \citet{liang2010estimation}, we consider the following profiled semiparametric estimating function 
	\bqan\label{Sn}
	\vS_n(\vbeta,\widehat{G},\widehat{\E})=-n^{-1}\sum_{i=1}^n\big[\widetilde{Y}_i-\widehat{G}(\vx_i^T\vbeta|\vbeta)\big]\widehat{G}^{(1)}(\vx_i^T\vbeta|\vbeta)[\vx_{i,-1}-\widehat{\E}(\vx_{i,-1}|\vx_i^T\vbeta)].
	\eqan

	In the high-dimensional setting, the estimating equation $\vS_n(\vbeta)=\vnull$ is ill-posed when $p\gg n$. 
	Let $\vbetah=(\widehat{\beta}_1, \ldots, \widehat{\beta}_p)^T=(1,  \vbetah_{-1}^T)^T$ be a solution in  $\bbB_0$ that solves the following penalized semiparametric profiled estimating equation
	\bqan\label{local}
	\vS_n(\vbeta,\widehat{G},\widehat{\E}) + \lambda\vkappa = \vnull, 
	\eqan
	where $\lambda>0$ is a tuning parameter,
	$\vkappa=(\kappa_2, \ldots, \kappa_p)^T\in\partial ||\vbeta_{-1}||_1$ with $||\vbeta_{-1}||_1$ denoting the 
	$l_1$ norm of $\vbeta_{-1}=(\beta_2, \ldots, \beta_p)^T$ and
	$\partial ||\vbeta_{-1}||_1$ denoting the subdifferential of $ ||\vbeta_{-1}||_1$, that is
	$\kappa_j=\mbox{sign}(\beta_j)$ if $\beta_j\neq 0$, and $\kappa_j\in [-1, 1]$ otherwise, $j=2, \ldots, p$.
	In (\ref{local}), $\widehat{G}$ and $\widehat{\E}$
	are evaluated at the corresponding $\vbeta$ in the estimating equations, hence here they
	stand for $\widehat{G}(\vx_i^T\vbeta|\vbeta)$
	and $\widehat{\E}(\vx_i |\vx_i^T\vbeta)$,
	respectively.
	Note that (\ref{local})
	may have multiple solutions. The theory we develop
	in Section~\ref{sec:est_theorem} provides a near-optimal error bound for any sparse local solution of the estimating equation.
	The satisfactory performance of the proposed profiled estimator is demonstrated in the numerical simulations in Section~\ref{sec:sim_res}.

	%Then to estimate $G(\vx_{i}^T \vbeta)$ and $G^{(1)}(\vx_{i}^T \vbeta)$ for $\vx_i$, we consider  $$\widehat{G}(\vx_{i}^T \vbeta) \triangleq \widehat{G}(\vx_{i}^T \vbeta, \vbeta) =  \frac{\sum_{j\neq i} K_h(\vx_i^T\vbeta-\vx_j^T\vbeta) \widetilde{Y}_j}{\sum_{j\neq i} K_h(\vx_{i}^T \vbeta-\vx_j^T\vbeta)},$$
	%Many exsisting work can be applied to estimate $f_0(\cdot)$ and $\vbeta_{0}$. In this paper, we consider the profile least-squares procedure proposed by %\citet{Liang2010}.

	%\subsection{Inference}
	\subsection{Inference on the Optimal Decision Rule} \label{sec:inf_method}
	To quantify the importance of the covariates on optimal decision making, we 
	will construct confidence intervals for the individual components of 
	$\vbeta_0=(1, \vbeta_{0,-1}^T)^T$ via
	debiasing a local solution to the semiparametric estimating equation (\ref{local}).
	This generalizes the work of debiased confidence intervals for high-dimensional linear regression in
	\citet{zhang2014confidence} and \citet{van2014asymptotically} to the semiparametric setting where
	the initial estimator is an estimating equation solution and 
	an estimated infinite-dimensional functional is present. The theory for semiparametric
	inference  in high dimension is highly nontrivial and is carefully studied in Section~\ref{sec:theorem}.
	We further investigate a wild bootstrap procedure for testing a general group hypothesis, which aims to achieve accurate
	finite-sample performance.
	
	Let $\vbetah=(1,  \vbetah_{-1}^T)^T$ denote a solution satisfying (\ref{local}). 
	In the high-dimensional linear regression setting, 
	the main idea of debiased estimator is to invert the Karush–Kuhn–Tucker (KKT) condition of the lasso.
	Inspired by this idea, we consider the following debiased estimator of $\vbeta_{0,-1}$:
	\begin{align}
	\vbetaw_{-1} = \vbetah_{-1} - \vThetah^T 	\vS_n(\vbetah,\widehat{G},\widehat{\E}), \label{desp}
	\end{align}
	where the $(p-1)\times (p-1)$ matrix $\vThetah$ is an approximation to the inverse of $\nabla  	\vS_n(\vbetah,\widehat{G},\widehat{\E})$, the derivative matrix of 
	$	\vS_n(\vbeta,\widehat{G},\widehat{\E})$ 
	with respect to $\vbeta_{-1}$ evaluated at $\vbeta=\vbetah$. %Denote $\widehat{G}^{(2)}(t) = \frac{d^2}{d t^2} \widehat{G}(t, \vbeta)$, we can derive that 
	%$$\nabla^2 L(\vbetah) = \frac{1}{n}\sum_{i=1}^n[\widehat{G}^{(1)}(\vx_i^T\vbetah)]^2\vx_i\vx_i^T -  \frac{1}{n}\sum_{i=1}^n[\widetilde{Y}_i - \widehat{G}(\vx_i^T\vbetah)]\widehat{G}^{(2)}(\vx_i^T\vbetah)\vx_i\vx_i^T.$$
	To construct the approximate inverse $\vThetah$,  we propose a nodewise Dantzig estimator.
	Specifically, given the initial estimator $\vbetah$ and a positive number $\eta$, 
	for $j=2,\cdots,p$, define
	\begin{align}
	\vd_j(\vbetah,\eta) =  \argmin\limits_{\bm{v}\in\bbR^{p-2}}||\vv||_1 \mbox{ s.t. }   \Big|\Big|n^{-1}\sum_{i=1}^n \big[\widehat{G}^{(1)}(\vx_i^T\vbetah|\vbetah)\big]^2 (\xh_{i,j}-\vxh_{i,-j*}^T\vv)\vxh_{i,-j*}\Big|\Big|_\infty\leq \eta,\label{dj_def}
	\end{align} 
	where $||\cdot||_\infty$ denotes the infinity norm of a vector, $\vxh _i= \vx_i-\widehat{\E}(\vx_i|\vx_i^T\vbetah)$,
	$\xh_{i,j}$ denotes the $j^{th}$ entry of the vector $\vxh_i$, $\vxh_{i,-1}$ denotes the $(p-1)$-subvector of $\vxh_i$ that excludes the $1^{st}$ entry, and the $\vxh_{i,-j*}$ denotes the $(p-2)$-subvector of $\vxh_i$ that excludes the $1^{st}$ and $j^{th}$ entries. Furthermore, for $j=2,\cdots,p$, we define 
	\begin{align}
	\vphi_j(\vbetah,\eta)& = \Big(-\big(\vd_j(\vbetah,\eta) \big)^T_{1:(j-2)}, 1, -\big(\vd_j(\vbetah,\eta) \big)^T_{(j-1):(p-2)}\Big)^T,\label{phij_def}\\
	\tau_j^{2}(\vbetah,\eta)&=n^{-1}\sum_{i=1}^n\big[ \widehat{G}^{(1)}(\vx_i^T\vbetah|\vbetah)\big]^2\xh_{i,j} \vxh_{i,-1}^T\vphi_j(\vbetah,\eta),\label{tauj_def} \\
	\vtheta_j(\vbetah,\eta) &= \tau_j^{-2}(\vbetah,\eta)\vphi_j(\vbetah,\eta),\label{thetaj_def}
	\end{align} 
	where for a vector $\vu= (u_1,\cdots,u_p)^T$,  
	given $1\leq i\leq j\leq p$,
	$(\vu)_{i:j}$ returns the subvector $(u_i,\cdots,u_j)^T$, and for any $ i> j$, $(\vu)_{i:j}$ returns the empty vector.  
	For notational simplicity, denote $\vdh_j = \vd_j(\vbetah,\eta)$, $\widehat{\tau}_j^{2} = \tau_j^{2}(\vbetah,\eta)$, and $\vthetah_j =\vtheta_j(\vbetah,\eta) $. The approximate inverse of  $\nabla	\vS_n(\vbetah,\widehat{G},\widehat{\E})$ is then constructed as  
	\bqa
	\vThetah = (\vthetah_2,...,\vthetah_p).
	\eqa 
	
	%and consequently, the desparsified estimator is defined as
	%\begin{align}
	%\vbetaw \triangleq \vbetah +\vThetah^T\Big\{ \frac{1}{n}\sum_{i=1}^n\{\widetilde{Y}_i- \widehat{G}(\vx_i^T\vbetah)\} %\widehat{G}^{(1)}(\vx_i^T\vbetah)\vx_i\Big\}.\label{desp}
	%\end{align}

	%Let the desparsified estimator 
	%$\vbetaw = $ be computed as in (\ref{desp}) with
	%the above approximate inverse $\vThetah$.
	The validity of $\vThetah$ as an approximation to the inverse of $\nabla 	\vS_n(\vbetah,\widehat{G},\widehat{\E})$ is given in Lemma~\ref{dbound} of Section~\ref{sec:inf_theorem}.
	Section~\ref{sec:theorem} will also present the statistical properties of the debiased estimator $\vbetaw_{-1}=
	(\widetilde{\beta}_2,\cdots,\widetilde{\beta}_p)^T$. This then leads to
	the following asymptotic $100(1-\alpha)\%$ confidence interval for $\beta_{0j}$, 
	\begin{align}
	\Big\{\ \widetilde{\beta}_j - \Phi^{-1}(1-\alpha/2)\big( \widehat{\Sigma}_{jj}/n \big)^{1/2}, \widetilde{\beta}_j +  \Phi^{-1}(1-\alpha/2)\big( \widehat{\Sigma}_{jj}/n \big)^{1/2}\ \Big\},\label{conf_j}
	\end{align}
	where $j=2,\cdots,p$,  $ \Phi^{-1}(\cdot)$ is the quantile function of the standard normal distribution, and $\widehat{\Sigma}_{jj}$ denotes the $(j-1)^{th}$ diagonal entry of $\vSigmah(\vbetah)$, with
	\begin{align} 
	\vSigmah(\vbetah) \triangleq\vThetah^T\Big\{\frac{1}{n}\sum_{i=1}^n \big[\widetilde{Y}_i-\widehat{G}(\vx_i^T\vbetah|\vbetah)\big]^2 [\widehat{G}^{(1)}(\vx_i^T\vbetah|\vbetah)]^2 \vxh_{i,-1}\vxh_{i,-1}^T\Big\} \vThetah.\label{desp_var}
	\end{align}
	%Its consistency is  shown in %Lemma~\ref{cor_Sig} of the %supplementary material. 
	Corollary~\ref{cor:marginal_conf} in Section~\ref{sec:theorem} justifies the asymptotic uniform validity of this marginal confidence interval.

	Next, we consider the following more general simultaneous testing problem
	\bqan\label{group}
	H_{0,\kG}: \beta_{0j} = 0\mbox{ for all }j\in \kG \mbox{\quad versus \quad}H_{1,\kG}: \beta_{0j} \neq 0\mbox{ for some }j\in \kG, \label{sim_test}
	\eqan
	where $\kG$ is a prespecified subset of $\{2,\ldots,p\}$.  The size of $\kG$ may depend on the sample size $n$.  
	Such a hypothesis naturally arises in the high-dimensional setting. 
	For example, researchers may want to test whether a gene pathway, consisting of multiple genes 
	%(the number of genes in the pathway %can be large relative to the sample %size) 
	for the same biological functions, is important for optimal treatment regime recommendation.
	For this purpose, we propose an effective bootstrap procedure. Although the asymptotic normal distribution of the debiased estimator 
	(see Theorem~\ref{desparse}) allows for construction of confidence intervals for individual coefficients (or  fixed-dimensional subvector of coefficients), applying it to make inference for groups of variables when the group size diverges (potentially larger than $n$) is not straightforward. Moreover, 
	confidence intervals
	based on the asymptotic distribution have been observed to sometimes lead to undercoverage for nonzero coefficients in finite samples. The bootstrap procedure we study  automatically accounts for the dependence structure of the variables in the group and provides more accurate critical value.

	When deriving the asymptotic property of the debiased estimator (in the proof of Theorem~\ref{desparse}), it is observed that the asymptotic property
	of $\sqrt{n}(\widetilde{\vbeta}_{-1}-\vbeta_{0,-1})$ is determined by 
	the leading term $\sqrt{n}\vThetah^T \vS_n(\vbeta_0, G, \E)$. This suggests that
	we approximate the distribution of $\sqrt{n}(\widetilde{\beta}_j-\beta_{0j})$,
	$j=2,\ldots,p$,
	by the distribution of the following multiplier bootstrap statistic
	\bqan\label{btest}
	\delta^*_j&\triangleq  \frac{1}{n}\sum_{i=1}^nr_i \big[\widetilde{Y}_i- \widehat{G}(\vx_i^T\vbetah|\vbetah)\big] \widehat{G}^{(1)}(\vx_i^T\vbetah|\vbetah)\vxh_{i,-1}^T\vthetah_j,
	%\vdelta^*_\kG&\triangleq \vThetah_\kG^T\Big\{\frac{1}{n}\sum_{i=1}^nr_i \{\widetilde{Y}_i- \widehat{G}(\vx_i^T\vbetah)\} \widehat{G}^{(1)}(\vx_i^T\vbetah)\vx_i\Big\},
	\eqan
	where %$\vThetah_\kG = (\vthetah_j)_{j\in \kG}$, and 
	$r_1,\cdots,r_n$  are i.i.d. standard normal random variables, independent of the data. 
	%Let $\vdelta_\kG = \vThetah_\kG^T\Big\{\frac{1}{n}\sum_{i=1}^n\{\widetilde{Y}_i- \widehat{G}(\vx_i^T\vbetah)\} \widehat{G}^{(1)}(\vx_i^T\vbetah)\vx_i\Big\}$ be the observed test statistic.
	Let $c_{1-\alpha}^*$ be the upper $\alpha$-quantile of the distribution of $\max_{j\in \kG}|\delta^*_j|$ conditional on the data, 
	which can be easily simulated by generating multiple independent copies of the random weights. 
	We reject the null hypothesis at level $\alpha$ if $\max_{j\in \kG}|\widetilde{\beta_j}|>c_{1-\alpha}^*$.
	The asymptotic validity of the bootstrap procedure is formally established in Section~\ref{sec:theorem}.
	Its performance is demonstrated in the numerical simulations in Section~\ref{sec:sim_res}.

	%The distribution of $\vdelta^*_\kG$ can be estimated from a large number of bootstrap samples. Let $q_\kG^{1-\alpha}$ be the %$(1-\alpha)^{th}$ quantile of the bootstrap distribution of $\sqrt{n}||\vdelta^*_\kG||_\infty$. 
	%An asymptotic level-$\alpha$ test for testing (\ref{sim_test}) would reject $H_{0,G}$ if 
	%$ \sqrt{n}||\vdelta_\kG||_\infty \geq q_\kG^{1-\alpha},$
	%where $\vdelta_\kG = \vThetah_\kG^T\Big\{\frac{1}{n}\sum_{i=1}^n \{\widetilde{Y}_i- \widehat{G}(\vx_i^T\vbetah)\} %\widehat{G}^{(1)}(\vx_i^T\vbetah)\vx_i\Big\}$,

	\section{Statistical Properties}\label{sec:theorem}
	
	%A significant challenge for inference in the high-dimensional semparametric framework is that the corresponding estimation problem is not convex.
	%A key result we show in Section~\ref{sec:est_theorem} is that a local restricted strong convexity condition holds 
	%with high probability and that all local solutions within a small ball of $\vbeta_0$ enjoy a near-optimal error rate
	%under mild conditions. Section~\ref{sec:inf_theorem} further derives the asymptotic normality 
	%of the debiased version of any such local solution and proves the validity of inference procedure in Section~{sec:inf_method}.
	%In this section, we carefully study %the statistical theory for estimation %and inference.
	%The proof of the theory can be found %in the supplementary material. 

	\subsection{Theory for Estimation}\label{sec:est_theorem}
	Making inference about the optimal treatment regime requires an adequate
	initial estimator for $\vbeta_0$.
	To obtain such an initial estimator
	in the high-dimensional semiparametric framework,
	a significant challenge 
	is that the corresponding estimation problem is not necessarily convex.
	To tackle this, we first establish in Lemma~\ref{lem:local_LRC} below
	that the  estimated $(p-1)$-dimensional gradient $\vS_n(\cdot,\widehat{G},\widehat{\E})$ in (\ref{Sn}) possesses an important
	local restricted strong convexity property with high probability. 
	Theorem~\ref{Lasso_error} then shows that all local sparse solutions within a small neighborhood of $\vbeta_0$ enjoy a near-optimal error rate
	under mild conditions.
	In the sequel, we use $a\vee b$ to denote $\max(a,b)$, and $a\wedge b$ to denote $\min(a,b)$. Let 
	$s=||\vbeta_0||_0$ be the sparsity size of $\vbeta_0$, the population parameter indexing the optimal treatment regime.

	\blem \label{lem:local_LRC} (local restricted strong convexity property)
	Assume conditions (A1)--(A5) in Section~S2 of the online supplementary material are satisfied. 
	If $ d_0\big[\frac{s\log (p\vee n)}{n}\big]^{1/5}\leq h<1$  for some constant $d_0>0$, 
	then there exist universal positive constants $c_0$, $c_1$, $c_2$ and $r\leq 1$, which do not depend on $n$, $p$ and $\vbeta_0$, such that 
	\bqa
	&&P\Big(\big\langle	\vS_n(\vbeta,\widehat{G},\widehat{\emph{E}})-	\vS_n(\vbeta_0,\widehat{G},\widehat{\emph{E}}),  \vbeta_{-1}-\vbeta_{0,-1}\big\rangle \geq c_0 ||\vbeta-\vbeta_0||_2^2 - c_1 h^2||\vbeta-\vbeta_0||_2,\ \forall \ \vbeta\in\bbB\Big) \\
	&\geq &1-\exp( -c_2  \log p  ),%\label{local_LRC} 
	\eqa
	for all $n$ sufficiently large, where $\bbB=\{\vbeta\in\bbB_0: ||\vbeta-\vbeta_0||_2\leq r,  ||\vbeta||_0\leq ks\}$
	and $k>1$ is a positive constant.
	\elem
	
	\noindent{\it Remark 1.}  Lemma~\ref{lem:local_LRC} characterizes the local geometry of the profiled score function.
	For high-dimensional regression with convex loss function such as 
	$L_1$ penalized linear regression, restricted strong convexity plays an important role
	on the theory of the regularized estimator \citet{Negahban2012}.
	Local restricted strong convexity condition were investigated in \citet{Loh2015} and \citet{Mei2018}
	for some specific nonconvex loss functions. Those results, however, do not apply to our setting due to the estimated infinite-dimensional parameter.

	Theorem 1 below presents non-asymptotic high-probability error bounds for any local sparse solution  $\vbetah$ that satisfies the penalized
	profiled estimation equation (\ref{local}).
	\begin{theorem}
		\label{Lasso_error}
		Assume conditions (A1)--(A5) in Section~S2 of the online supplementary material are satisfied. 
		Suppose $\lambda = d_1 h^2 $ %$\lambda = d_1 \max\{h^2,\sqrt{\frac{\log p}{n}}\}$ %$\lambda = a_0 \sqrt{\frac{\log p}{n}}$ 
		for some constant $d_1>0$, and $ d_0\big[\frac{s\log (p\vee n)}{n}\big]^{1/5}\leq h\leq d_0n^{-1/6}$ for some constant $d_0>0$.  Then there
		exist universal positive constants $c_0$ and $c_1$ such that
		for any solution $\vbetah$ in $\bbB$, we have
		% then there exist positive constants $c_0$  such that 
		\begin{align*}%\label{errorbound}
		||\vbetah-\vbeta_0||_2\leq \frac{6}{c_0}\lambda\sqrt{s}, \quad	||\vbetah-\vbeta_0||_1\leq \frac{24}{c_0}\lambda s,
		\end{align*}
		with probability at least  $1-\exp(-c_1\log p)$,  for all $n$ sufficiently large.
		%where $s = ||\vbeta_0||_0$. 
	\end{theorem}

	\noindent{\it Remark 2.} Theorem~\ref{Lasso_error} shows that under some mild regularity conditions, local solutions of the profiled estimation equation (\ref{local}) enjoy desirable estimation error rates, same as Lasso does for high-dimensional linear regression.	 For the purpose of inference, the initial estimator is not require to achieve perfect variable selection.
	The debiased estimator, however, can achieve the $n^{-1/2}$ rate for each individual coefficient,
	as we will show in Section~\ref{sec:inf_theorem}.
	Carefully going through the proof of the theorem also reveals that the above error bounds hold uniformly for all
	$\vbeta_0$ such that $||\vbeta_0||_0\leq s$. \\
	
 	\noindent{\it Remark 3.} Based on Theorem~\ref{Lasso_error}, Lemmas A5--A6 in the 
		online supplement establish the uniform convergence rates for the nonparametric estimator $\widehat{G}(\vx_i^T\vbeta|\vbeta)$ and $\widehat{G}^{(1)}(\vx_i^T\vbeta|\vbeta)$, which are of independent interest. %Denote $\bbT = \{t \in\bbR : |t|\leq\sqrt{ \log (p\vee n)}\}$. 
		Under the assumptions of Theorem~\ref{Lasso_error},  we show that there exist universal positive constants $c_0$ and $c_1$  such that 
		\begin{align*}
		&P\Big(\max_{ 1\leq i \leq n}\sup\limits_{\vbeta \in \bbB}\big|\widehat{G}(\vx_i^T\vbeta|\vbeta)-G(\vx_i^T\vbeta|\vbeta)  \big| \geq c_0h^2\Big)\leq  \exp[-c_1\log(p\vee n)], \\
		&P\Big(\max_{ 1\leq i \leq n}\sup\limits_{\vbeta \in \bbB}\big|\widehat{G}^{(1)}(\vx_i^T\vbeta|\vbeta)-G^{(1)}(\vx_i^T\vbeta|\vbeta)  \big| \geq c_0h\Big)\leq \exp[-c_1\log(p\vee n)]. % \label{G1bound}.
		\end{align*}

	\subsection{Theory for Inference}\label{sec:inf_theorem}
	We first introduce some additional notation. 
	Let $\vxw_i = \vx_i-\E(\vx_i|\vx_i^T\vbeta_0)$, and let $\vxw_{i,-1}$ denote the $(p-1)$-subvector of $\vxw_i$ that excludes its $1^{st}$ entry.
	Let $\vOmega =\E \big\{[G^{(1)}(\vx_i^T\vbeta_0|\vbeta_0)]^2 \vxw_{i,-1}\vxw_{i,-1}^T\big\}$.
	Assume the $(p-1)\times (p-1)$ matrix $\vOmega$ is positive definite and write
	its inverse $\vOmega^{-1}\triangleq\vTheta = (\vtheta_2,...,\vtheta_p) $. For $j=2, \ldots, p$, let $\vOmega_{-(j-1),-(j-1)}\in\bbR^{(p-2)\times (p-2)}$ 
	be the submatrix of $\vOmega$ with its $(j-1)^{th}$ row and $(j-1)^{th}$ column removed;
	similarly $\vOmega_{-(j-1),(j-1)}\in\bbR^{p-1}$ denotes the $(j-1)^{th}$ column of $\vOmega$ with its $(j-1)^{th}$ entry removed.
	Note that $\vOmega_{-(j-1),-(j-1)}$ is positive definite. Define
	$\vd_{0j} = (\vOmega_{-(j-1),-(j-1)})^{-1}\vOmega_{-(j-1),(j-1)}$, $s_j = ||\vd_{0j}||_0$, $\widetilde{s}= \max_{2\leq j\leq p}s_j $ and
	$\tau^2_{0j} =\vOmega_{(j-1),(j-1)} - \vd_{0j}^T\vOmega_{-(j-1),(j-1)}=(\Theta_{(j-1),(j-1)})^{-1}$,  $j=2, \ldots, p$.
	
	Lemma \ref{dbound} below establishes useful properties of the approximate inverse of $\nabla 	\vS_n(\vbetah,\widehat{G},\widehat{\E})$,
	defined in Section~\ref{sec:inf_method}.
	
	\begin{lemma}
		\label{dbound}
		%Assume conditions \ref{A1}--\ref{K4} in the appendix hold. Let  $\lambda = d_1 \max\{h^2,\sqrt{\frac{\log p}{n}}\}$, and 
		Assume the conditions of Theorem~\ref{Lasso_error} are satisfied.
		Let $\eta = d_2h$ for some positive constant $d_2>0$. If $\eta \widetilde{s}\leq d_0$ and $d_0\big[\frac{s \log (p\vee n)}{n}\big]^{1/5}\leq h\leq d_0n^{-1/6}$  for some constant $d_0>0$,  then there exist some universal positive constants $d_2$, $c_0$ and $c_1$ such that results (1)-(3) below hold uniformly in $j=2,\ldots,p$, with probability at least $1-\exp(-c_1\log p)$ for all $n$ sufficiently large:\\
		(1) $  ||\vdh_j -\vd_{0j}||_2 \leq \frac{8\sqrt{s_j}\eta}{\xi_2}$, and $ ||\vdh_j -\vd_{0j}||_1 \leq \frac{16s_j\eta}{\xi_2}$;\\
		(2) $ |\tau^2_{0j}-\widehat{\tau}_j^{2}| \leq c_0\sqrt{s_j}\eta$, and $ |\tau^{-2}_{0j}-\widehat{\tau}_j^{-2}| \leq c_0\sqrt{s_j}\eta$;\\
		(3)  $ ||\vthetah_j-\vtheta_{j} ||_2 \leq c_0\sqrt{s_j}\eta$, and $ ||\vthetah_j-\vtheta_{j} ||_1 \leq c_0s_j\eta$;\\
		%\blue{(4) $\max_{1\leq j\leq p}\vthetah_j^T \Big(\frac{1}{n}\sum_{i=1}\vxh_i\vxh_i^T\Big)\vthetah_j  \leq 4\xi_1\xi_2^{-2}$;\\
		where %$\xi_1$ is the largest eigenvalue of $\E\big[\Cov (\vx|\vx^T\vbeta_0)\big]$, and }
		$\xi_2>0$ is the smallest eigenvalue of $\vOmega$.
		%that  \bfred{$\min\limits_{2\leq j\leq p}\big\{\min\limits_{j: %\vphi_{0j}\neq\ve_{j-1}}\inf\limits_{\vv\in \mathcal{V}_{2j}}\vv^T\vOmega\vv, %\min\limits_{j:\vphi_{0j}=\ve_{j-1}} \Omega_{(j-1),(j-1)}\}\geq \xi_2$, } where 
		%\bfred{$\mathcal{S}_{\vphi_j}$ is the support set of $\vphi_{0j} = \tau_{0j}^2\vtheta_j$, %and $\mathcal{V}_{2j}=
		%	\{\vv=(v_1,\cdots,v_{p-1})^T: ||\vv_{\mathcal{S}_{\vphi_j}^C}||_1 \leq %||\vv_{\mathcal{S}_{\vphi_j}}||_1, ||\vv||_2=1, v_{j-1}=0\}$}, for any $j\in\{2,\cdots,p\}$ such that $\vphi_{0j}\neq\ve_{j-1}$.  
		%is the smallest eigenvalue of $\vOmega$.
		%$\min_{1\leq i\leq n}|f'(\vx_i^T\vbeta_0)|\geq a>0$ . 
	\end{lemma}

	Lemma \ref{dbound} requires $\tilde{s}=\max_{2\leq j\leq p}s_j$ to be of order $O(h^{-1})$. 
	For high-dimensional generalized linear models (Theorem 3.1,  Van de Geer et al. [2014]), the corresponding sparsity constraint is
	$\tilde{s}=o(\sqrt{n/\log p})$. Our constrain is somewhat stricter due to the need to estimate the infinite-dimensional
	nuisance parameter.
	Building on Lemma \ref{dbound}, we prove the statistical property of the debiased estimator $\vbetaw_{-1}$ defined in (\ref{desp}).  
	%in Theorem~\ref{desparse}. %\bfblue{Define $\widetilde{s}$ here. }
	\begin{theorem}
		\label{desparse}
		%Assume \ref{A1}--\ref{K4} hold, with $\lambda = d_1 \max\{h^2,\sqrt{\frac{\log p}{n}}\}$, and $\eta = d_2h$ for positive constants $d_1$, $d_2$.
		Assume the conditions of Lemma~\ref{dbound} are satisfied.
		Let  $\Delta_{n,p} =s h^3 \sqrt{n} +  \widetilde{s}h \sqrt{\log  p}$.
		%$\Delta_{n,p} = s h^3 \sqrt{n} +  \sqrt{h\log (p\vee n)}+ h [\log (p\vee n) ]^{3/2}+  \widetilde{s}h \sqrt{\log  p}$.
		Assume $\Delta_{n,p} =o (1)$ and $ s \log (p\vee n)\leq d_0nh^5$  for some constant $d_0>0$. Then for all $n$ sufficiently large, 
		\begin{align*}
		\sqrt{n}\big(\widetilde{\beta}_j  - \beta_{0j} \big) =W_j+\Delta_j,  \quad j= 2, \ldots, p,
		\end{align*}
		with
		% %N(0,\ve_j^T\vTheta^T\vLam\vTheta\ve_j)$ %with $\vLam =\E %\big\{[\widetilde{\epsilon}_iG^{(1)}(\vx_%i^T\vbeta_0|\vbeta_0)]^2 %\vxw_i\vxw_i^T\big\}$, 
		\bqa
		&& W_j=n^{-1/2}\ve_{j-1}^T\sum_{i=1}^{n}\widetilde{\epsilon}_{i}  G^{(1)}(\vx_i^T\vbeta_0|\vbeta_0)\vxw_{i,-1},\\
		&&P\big(\max_{2\leq j\leq p}|\Delta_j|\geq c_0\Delta_{n,p}\big)\leq  \exp(-c_1\log p),
		\eqa
		where $c_0$, $c_1$ are universal positive constants, and $\ve_{j-1}$ denotes the $(p-1)$-dimensional vector with the $(j-1)^{th}$ 
		entry being one and all the other entries equal to zero.
	\end{theorem}
	
	\noindent{\it Remark 4.} 
	Theorem~\ref{desparse} suggests that
	%the desparsified individual coefficient estimator %$\widetilde{\beta}_j$
	%has an asymptotic normal distribution. 
	if we consider a lower-dimensional linear combination of coefficients
	$\valpha^T\vbeta_{0,-1}$, where $\valpha$ is a $(p-1)$-dimensional nonzero vector of constants, then 
	$\valpha^T(\vbetaw_{-1}-\vbeta_{0,-1})$ has the asymptotic distribution  $N(0,\valpha^T\vTheta^T\vLam\vTheta\valpha)$ with $\vLam =\E \big\{[\widetilde{\epsilon}_iG^{(1)}(\vx_i^T\vbeta_0|\vbeta_0)]^2\vxw_{i,-1}\vxw_{i,-1}^T\big\}$. The asymptotic covariance matrix resembles that in the literature on profiled estimation for index models in lower dimension, see \citet{liang2010estimation}, \citet{ma2016}, among other. The assumption $\Delta_{n,p} =o (1)$ is a sufficient condition for the remaining term of the linear approximation of $\sqrt{n}\big(\widetilde{\beta}_j  - \beta_{0j} \big)$ to be uniformly negligible. It still allows $p$ to grow at an exponential rate of $n$.\\
	
	\noindent{\it Remark 5.}
	The proof of Theorem~\ref{desparse} is given in the online supplement.
	To build the theory, we show that
	\begin{align*}
	\sqrt{n}(\vbetaw_{-1}-\vbeta_{0,-1}) 
	=&\sqrt{n}\vThetah^T\vS_n(\vbeta_0,G, \E) + \sqrt{n}(\vI_{p-1}- \vThetah^T\vJ_1) (\vbetah_{-1}-\vbeta_{0,-1}) \\
	&-\sqrt{n}\vThetah^T[\vS_n(\vbetah,\widehat{G},\widehat{\E}) -\vS_n(\vbeta_0,G, \E)-  \vJ_1(\vbetah_{-1}-\vbeta_{0,-1})  ]\\ 
	\triangleq& \vA_{n1} +\vA_{n2} +\vA_{n3},
	\end{align*}
	where  $\vJ_1 =  n^{-1}\sum_{i=1}^{n}  [\widehat{G}^{(1)}(\vx_i^T\vbetah|\vbetah)]^2\vxh_{i,-1}\vxh_{i,-1}^T$ is the leading term in the approximation to $\nabla \vS_n(\vbeta_0,G, \E) $. In the proof, we carefully justify that:
	(1) The $(j-1)^{th}$ component of $\vA_{n1}$ can be approximated by $W_j$ in the theorem, for $2\leq j \leq p$,
	(2) $P(|| \vA_{n2} ||_\infty\geq c_0\Delta_{n,p})\leq \exp(-c_1\log p)$,
	and
	(3) $P(|| \vA_{n3} ||_\infty\geq c_0\Delta_{n,p})\leq \exp(-c_1\log p)$,
	for some positive constants $c_0$ and $c_1$.
	Furthermore, to provide a deeper insight into the extension into the semiparametric setting, we consider the Gateaux functional derivative
	of the estimating function with respect to the
	infinite-dimensional nuisance parameters.
	Consider the functional
	$
	M(z; \vbeta,G,E)=[\widetilde{Y}-G(\vx^T\vbeta|\vbeta)\big]G^{(1)}(\vx^T\vbeta|\vbeta)[\vx_{-1}-\E(\vx_{-1}|\vx^T\vbeta)],
	$
	where $z=(A, X, \widetilde{Y})$ denotes a vector of random observations
	of the data. The Gateaux derivative of $M(z; \vbeta,G,E)$ at $G$ in the direction $[\overline{G}-G]$ is  defined as
	\bqa 
	\lim_{\tau\ra 0}
	\frac{\mathbb{E}\big\{M(z;\vbeta,
		G+\tau(\overline{G}-G),
		E) -M(z;\vbeta,G,E)\big\}}
	{\tau}. \eqa
	It is easy to see that this Gateaux derivative at $G$ is zero when evaluated at $\vbeta=\vbeta_0$. Similarly, the Gateaux derivative with respect to $E$
	vanishes at the true value $\vbeta_0$. This orthogonality behavior
	suggests the insensitivity of the estimating function to the infinite-dimensional nuisance parameters.

	%	\blue{Note that $\Delta_{n,p} =o (1)$ for all %sufficiently large $n$ ensures $\eta \widetilde{s}= O(1)$ %and $nh^6\leq d_0$ in Lemma~\ref{dbound}. }
	
	The following corollary establishes uniform validity of the marginal confidence intervals  (\ref{conf_j}) introduced in Section~\ref{sec:inf_method}.
	\begin{corollary}
		\label{cor:marginal_conf}
		Under the conditions of Theorem~\ref{desparse}, 
		%let $\bbB_0 = \{\vbeta\in\bbR^p: ||\vbeta||_0\leq s,|\beta_1|=1\}$, then
		\begin{align*}%\label{marginal_conf}
		\sup_{\vbeta_0\in\bbB_0:||\vbeta_0||_0\leq s}\max_{2\leq j \leq p}\sup_{\alpha\in(0,1)}\Big|P\Big(\big| \sqrt{n}(\widetilde{\beta}_j-\beta_{0j})\widehat{\Sigma}_{jj}^{-1/2} \big|\leq \Phi^{-1}(1-\alpha/2) \Big)-(1-\alpha)\Big| = o(1),
		\end{align*}
		where $\widehat{\Sigma}_{jj}$ denotes the $(j-1)^{th}$ diagonal entry of $\vSigmah(\vbetah)$ defined in Section~\ref{sec:inf_method}, and
		$\Phi^{-1}(\cdot)$ is the quantile function of $N(0,1)$.
	\end{corollary}
	
	Finally, Theorem~\ref{multi_boots}  below establishes the validity of the bootstrap procedure introduced
	in Section~\ref{sec:inf_method} for testing the group hypothesis (\ref{group}). 
	Given a group of variables $\kG\subseteq \{2,\ldots,p\}$, 
	the wild bootstrap test statistic is defined as $\sqrt{n}\max_{j\in\kG}| \delta^*_j|$, where
	$\delta^*_j\triangleq  n^{-1}\sum_{i=1}^nr_i \{\widetilde{Y}_i- \widehat{G}(\vx_i^T\vbetah|\vbetah)\} \widehat{G}^{(1)}(\vx_i^T\vbetah|\vbetah)\vxh_{i,-1}^T\vthetah_j$, 
	and $r_1,\cdots,r_n$  are standard normal random variables that are independent of the data. 
	Denote $\vr =\{r_1,\cdots,r_n\}$, and let
	$\vw = \{w_1,\cdots, w_n\}$ denote the random sample $w_i = (A_i, \vx_i, \widetilde{Y}_i)$.  
	Given $0<\alpha<1$, recall that the bootstrap critical value for a level-$\alpha$ test is defined as
	\begin{align}\label{cstar_def}
	c^*_{1-\alpha} = \inf\Big\{t\in\bbR: P\Big( \sqrt{n}\max_{j\in\kG}| \delta^*_j| \leq t \big| \{w_i\}_{i=1}^n\Big)\geq 1-\alpha\Big\}. 
	\end{align}
	
	\begin{theorem}
		\label{multi_boots} 
		%Assume \ref{A1}--\ref{K4} hold, with $\lambda = d_1 h^2$, 	and $\eta = d_2h$ for positive constants $d_1$, $d_2$.  
		Assume the conditions of Theorem~\ref{desparse} are satisfied.
		If $\Delta_{n,p} \sqrt{\log p}=o (1)$, $   h\geq d_0\big[\frac{s\log (p\vee n)}{n}\big]^{1/5}$ for some constant $d_0>0$, and  $\sqrt{\widetilde{s}} h\log^2 p =o(1)$, then 
		% $ d_0\big[\frac{s\log (p\vee n)}{n}\big]^{1/5}\leq h< d_0( \sqrt{\widetilde{s}}\log^2 p )^{-1}$ for some constant $d_0>0$,  then 
		\begin{align*}%\label{sup_alpha}
		\sup_{\vbeta_0\in\bbB_0:||\vbeta_0||_0\leq s}\sup_{\alpha\in(0,1)}\Big|P\Big( \sqrt{n}\max_{j\in\kG}|\widetilde{\beta}_j-\beta_{0j}|\leq c^*_{1-\alpha}(\kG)  \Big)-(1-\alpha)\Big| = o(1).
		\end{align*}  
	\end{theorem}  
	Theorem \ref{multi_boots}  ensures that the multiplier bootstrap procedure is valid for the simultaneous testing problem (\ref{sim_test}).
	It is also {\it honest} in the sense of being valid uniformly over $s$-sparse models of the form (\ref{model}).
	It does not require the nonzero components of $\vbeta_0$ to be well-separated from zero. In particular, the multiple bootstrap procedure
	does not require the local solution of the profiled estimation to achieve perfect variable selection, which is usually unrealistic
	in practice. 
	%Note that the asymptotic distribution of any subvector of $\vdelta^*$ can be easily obtained according to Theorem~\ref{multi_boots}. It implies the  multiplier bootstrap procedures is valid in the simultaneous testing (\ref{sim_test}), for any group $\kG$. %Similar to Corollary~\ref{cor1}, corollary for the clean case is concluded as below.

	\section{Monte Carlo Studies}\label{sec:simulate}
	\subsection{Algorithm for Estimation}\label{sec:algo_est}
	To solve  the  penalized
	high-dimensional profiled estimating equation
	for the initial estimator $\vbetah$, 
	%Recall that $	%\vS_n(\vbeta,\widehat{G},\widehat{\E})=-n^{-1%}\sum_{i=1}^n\big[\widetilde{Y}_i-\widehat{G}%(\vx_i^T\vbeta|\vbeta)\big]\widehat{G}^{(1)}(%\vx_i^T\vbeta|\vbeta)[\vx_i-\widehat{\E}(\vx_%i|\vx_i^T\vbeta)].$ 
	%To increase the computational efficiency, we 
	we extend 
	%the proximal algorithm in \citet{Nesterov} %for linear regression.. 
	%Similarly as 
	the composite gradient algorithm  (\citet{Nesterov, agarwal2012}) for high-dimensional M-estimator without nuisance parameters.
	A summary of the proposed algorithm is given in Algorithm~1 in Section~S10.1 of the supplementary material. 
	
	Specifically, given a current estimator $\vbeta^t=(1, (\vbeta_{-1}^t)^T)^T$ at step $t$, we update the estimate by
	\begin{align}\label{update}
	\vbeta_{-1}^{t+1} = \argmin\limits_{\substack{\vbeta_{-1}\in\bbR^{p-1}\\ ||\vbeta_{-1}||_1\leq \rho} }
	\Big\{
	\frac{\gamma_u}{2}||\vbeta_{-1}-\vbeta_{-1}^t||_2^2 + [	\vS_n(\vbeta^t,\widehat{G},\widehat{\E})]^T(\vbeta_{-1}-\vbeta_{-1}^t)  + \lambda||\vbeta_{-1}||_1\Big\},
	\end{align} 
	where $\gamma_u$ is the step size, $\rho$ is a positive constant such that $||\vbeta_{0,-1}||_1\leq \rho$. 
	%The constraint $||\vbeta||_1\leq \rho$ will  %guarantee  good performance in the first few %iterations.
	An appealing practical property of the algorithm is that the update in step (\ref{update}) can be done efficiently by the following formula:
	\begin{align}
	\vbeta_{-1}^{t+1} =T_s\Big(\vbeta_{-1}^t-\frac{1}{\gamma_u}	\vS_n(\vbeta^t,\widehat{G},\widehat{\E}), \lambda\Big), \label{update1}
	\end{align} 
	where the function $T_s(\vbeta_{-1}, \lambda) = \Big(\sgn(\beta_j)*\max(|\beta_j|-\lambda,0)\Big)_{j=2,\cdots,p}$ is the soft-threshold operator. Then to ensure the constraint $||\vbeta_{-1}||_1\leq \rho$, we employ the projection method introduced in \citet{Duchi2008}, which is described in Algorithm~2 in Section~S10.1 of the online supplement.
	%This update procedure can be easily implemented.   

	In implementation, we choose the kernel function $K(\cdot)$ as the distribution function of the standard normal distribution.
	The bandwidth is set to be  $h=0.9n^{-1/6}\min\{\mbox{std}(\vx_i^T\vbeta),\\ \mbox{ IQR}(\vx_i^T\vbeta)/1.34\}$, as motivated by \citet{silverman1986density}, where “std” denotes the standard deviation, and “IQR” denotes the interquartile range.  For the step-size parameter,  inspired by \citet{agarwal2012}, we employ an expanding series for $\gamma_u$, which ensures that the stepsize diminishes during the update process.
	Given a set of candidate tuning parameters $\{\lambda_k\}$ and the corresponding estimators $\vbetah_{\lambda_k}$, we employ $5$-fold cross-validation to select the optimal tuning parameter $\lambda$ by minimizing 
	MSE$(\lambda) = n^{-1}\sum_{i=1}^n \{\widetilde{Y}_i - \widehat{G}(\vx_i^T\vbetah_\lambda|\vbetah_\lambda)\}^2$.
	
	To obtain the debiased estimator $\vbetaw$, the nodewise Dantzig estimator $\vd_j(\vbetah,\eta)$ in (\ref{dj_def}) is computed via linear programming, see details in Section~S10.2 of the supplementary material.

	\subsection{Monte Carlo Results}\label{sec:sim_res}
	We generate random data from the model $Y=(\vx^T\veta)^2 + (A-\frac{1}{2})f_0(\vx^T\vbeta_0) + \ep$, where 
	$\ep\sim N(0, 1)$, $A\sim \mbox{Bernoulli}(0.5)$, and $\vx$  follows a $p$-dimensional multivariate normal distribution with mean zero and identity covariance matrix,  $\veta = (0.5, 0.5,-0.5,-0.5, 0,\cdots,0)^T$, $\vbeta_0=(1,-1,-0.5,0.4,-0.3, 0,\cdots,0)^T$,  and $f_0(u)=20*\{[1+\exp(-u)]^{-1}-0.5\}$.  We consider $n=300, 500$ and  $p=200, 800, 2000$ in the Monte Carlo experiment.

	We first investigate the finite-sample performance of the penalized profiled semiparametric estimator in Section~\ref{sec:est_method}. Table~\ref{tab:est} reports the average $l_1$- and $l_2$-estimation errors, the average number of false negatives (nonzero components incorrectly identified as zero) and 
	false positives (zero components incorrectly identified as nonzero), with their standard errors in the parentheses,
	based on 500 simulation runs. Results in Table~\ref{tab:est} demonstrate satisfactory performance of the profiled estimator for both the scenarios $p<n$ and $p>n$.

	\begin{table}[!h]
		\centering
		\caption{Performance of the penalized profile least-squares estimator}{
			\begin{tabular}{cccccc}
				\hline 
				$n$&$p$& $l_1$ error & $l_2$ error & False Negative & False Positive \\ 
				\hline 
				\multirow{3}{*}{300}
				&200& 0.85 (0.02)& 0.31 (0.00) & 0.01 (0.01) & 10.95 (0.32) \\  
				&800& 1.10 (0.03)& 0.37 (0.00) & 0.07 (0.01) & 19.05 (1.13) \\  
				&2000& 1.32 (0.03)& 0.40 (0.00) & 0.09 (0.01) & 31.25 (1.57)\\ 
				\hline 
				\multirow{3}{*}{500}
				&200& 0.58 (0.01)& 0.22 (0.00) & 0.00 (0.00) & 9.30 (0.30) \\  
				&800& 0.79 (0.02)& 0.27 (0.00) & 0.00 (0.00) & 17.39 (0.66) \\  
				&2000&0.94 (0.02) & 0.31(0.00) & 0.01 (0.00) &25.60 (1.18) \\ 
				\hline
		\end{tabular}}
		\label{tab:est}
	\end{table}

	%Given the estimation for $\vbeta_0$ and the optimal treatment, we can also compute the estimator for the optimal value function $V(\vbeta_0) = \E\big[ g(\vx)+\frac{1}{2}|f_0(\vx^T\vbeta_0)|\big]$, given model (\ref{model}). In Table~\ref{tab:value}, we compare the optimal value function estimators  and standard errors by our methods and O-learning, which is implemented by ``DTRlearn2'' package \citep{DTRlearn2}. Given our model described as above, the optimal value turns out to be 3.423. Table~\ref{tab:value} shows that our method is uniformly better than O-learning. It is because O-learning is nonparametric, without considering any outcome model specification. However, our method is based on a semi-parametric model, as in (\ref{model}).
	
	%\begin{table}[!h]
	%	\centering
	%	\caption{Performance of estimating value function.}{
	%		\begin{tabular}{c|c|cc}
	%			\hline 
	%			$n$&$p$&New & O-learning\\
	%			\cline{3-4} 
	%			\hline 
	%			\multirow{3}{*}{300}
	%			&200& 3.322 (0.005) & 2.944 (0.012) \\  
	%			&800& 3.296 (0.003) & 2.679 (0.012) \\  
	%			&2000&3.280 (0.003) & 2.582 (0.012) \\ 
	%			\hline 
	%			\multirow{3}{*}{500}
	%			&200& 3.376 (0.002) & 3.069 (0.009) \\  
	%			&800& 3.350 (0.003) & 2.771 (0.009) \\  
	%			&2000&3.346 (0.003) & 2.643 (0.009) \\ 
	%			\hline 
	%	\end{tabular}}
	%	\label{tab:value}
	%\end{table}
	
	Next we investigate the wild bootstrap procedure introduced
	in Section~\ref{sec:inf_method} for testing the group hypothesis (\ref{group}). We consider the following six different choices for the groups: $\kG_1 = \{6,7,8,9\}$, $\kG_2 = \{5,6,7,8,9\}$, $\kG_3= \{4,6,7,8,9\}$, $\kG_4 = \{4,5,6,7,8,9\}$, $\kG_5 = \{3,6,7,8,9\}$ and $\kG_6 = \{2,6,7,8,9\}$. Note that $\kG_1$ consists of only zero entries in $\vbeta_0$, while all the other groups include at least one non-zero elements.  Table~\ref{tab:inference} summarizes the average Type I errors and powers for each scenario, based on 1000 Bootstrap samples and 500 simulation runs.
	
	\begin{table}[!h]
		\centering
		\caption{Performance of the bootstrap procedure in Section~\ref{sec:inf_method} for simultaneous testing.}{
			\begin{tabular}{cc|c|ccccc}
				\hline 
				\multirow{2}{*}{$n$}&\multirow{2}{*}{$p$}& Type I error & \multicolumn{5}{c}{Power} \\
				\cline{3-8}
				&& $\kG_1$ &$\kG_2$ &$\kG_3$ &$\kG_4$  &$\kG_5$ &$\kG_6$  \\ 
				\hline 
				\multirow{3}{*}{300}
				&200& 5.6\%& 96.4\% & 96.2\% & 97.8\% & 98.6\% & 100\%\\  	
				&800& 5.4\%& 94.6\% & 97.6\% & 99.0\%& 99.6\% &100\% \\  
				&2000&3.2\% & 92.4\% & 96.8\% &98.4\% & 99.0\% & 100\%\\ 
				\hline 
				\multirow{3}{*}{500}
				&200& 4.4\%& 100\% &100\% & 100\% & 100\% & 100\%\\  
				&800& 5.0\%& 99.6\% & 99.6\% & 100\%& 99.2\% & 100\% \\  
				&2000&4.6\% & 98.8\% & 98.6\% &99.0\% & 99.2\% & 100\%\\ 
				\hline 
		\end{tabular}}
		\label{tab:inference}
	\end{table}
	Table~\ref{tab:inference} indicates that type I errors are reasonable controlled for all scenarios. 
	Power performance generally depends on the number and magnitudes of the nonzero components.
	The hypothesis corresponding to $\kG_2$ represents a more challenging situation where the only non-zero element is -$0.3$, close to $0$. The average powers for
	this case for different values of $p$ are still over 90\%. %  $\kG_4$ contains two non-zero elements, one is $-0.3$ and the other is $0.4$, while $\kG_3$ has only one non-zero element $0.4$. \bfred{We observe that the magnitude of the smallest nonzero entry has more influence on the power than the number of nonzero coefficients.}
	
	Note that for inference, we need to estimate the approximate inverse of $\nabla S_n(\vbetah)$ 
	which involves an additional tuning parameter $\eta$. We observe that
	the inference procedure is not overly sensitive to its choice and fix
	it at the value $\eta = 25h$ to save computational time. Alternatively,
	it can also be selected via cross-validation similarly as what has been done for $\lambda$ selection. We provide additional simulation results in Section~S10.3 of the online supplement, including
	investigation on the choice of $\eta$ and comparing with alternative procedures for
	estimating the optimal value function.

	\section{A Real Data Example}\label{sec:realdata}
	We illustrate the application on a clinical data set introduced by \citet{realdata}. This is a randomized, double-blind, parallel treatment arm,
	phase III clinical trial to compare the efficacy and safety of pioglitazone versus gliclazide on metabolic control in naive patients with Type 2 diabetes mellitus. This data set we consider contains information on clinical characteristics for 813 individuals with Type 2 diabetes. The patients were randomized into two treatment arms: pioglitazone (treatment 0) and gliclazide (treatment 1). Their glycosylated haemoglobin A$_{1c}$ (HbA$_{1c}$) and fasting plasma glucose (FPG) levels were recorded every four weeks, up to week 52. 
	
	The primary efficacy endpoint is the change of HbA$_{1c}$ from baseline to the last available post-treatment value. 
	We consider the main effects of
	22 baseline covariates and their two-way interactions in the model. The dimension of the model is over 250. In the analysis, we standardize the covariates to have mean zero and sample variance one.

	We consider testing the significance of six different groups of variables.
	Table~\ref{tab:realdata} summarizes these six different groups and their respective $p$-values, based on the bootstrap procedure 
	in Section~\ref{sec:inf_method}. The estimated coefficients are reported in Section~S10.3 of the supplementary.

	\begin{table}[!h]
		\centering
		\caption{Real data analysis: evaluation of the significance of different groups of variables}{
			\begin{tabular}{c|c|c }
				\hline 
				Group& Variables & $p-$value \\
				\hline 
				1&HbA$_{1c}$, creatinine,  BMI, waist circumference, HomaS& $0.003$\\ 
				\multirow{2}{*}{2}& all variables in Group 1, all their two-way interactions,&\multirow{2}{*}{$0.011$}\\
				& and  their  interactions with fasting insulin&   \\  
				3&HbA$_{1c}$,   HomaS& $<0.001$\\ 
				4&BMI, creatinine,  waist circumference,& $0.242$  \\  
				5& LDL-C, total cholesterol, age, weight& $ 0.494$ \\ 
				\hline 
		\end{tabular}}
		\label{tab:realdata}
	\end{table}
	
	Based on the scientific literature and suggestions from our clinical collaborators, 
		fasting insulin is important for estimating the optimal treatment regime. We normalize its coefficient as 1 in our model.
	The first group includes the main effects of five characteristics, which are the baseline average levels for HbA$_{1c}$, creatinine,   BMI, waist circumference and homeostatic model assessment insulin sensativity (HomaS).  
	The variables in this group are those identified by diabetes experts to be potentially important for optimal treatment regime estimation.
	The bootstrap procedures suggests a significant $p$-value
	(0.003) for this group, which indicates that at one variable in this group is influential
	for making an optimal personalized decision in the choice of the two treatments. Group 2 augments Group 1 by including all the two-way interaction of these six characteristics (including fasting insulin), hence includes 20 variables in total. The estimated $p-$value is  0.011. Group 3 and Group 4 are subgroups of Group 1. The third group only includes two main effects: baseline HbA$_{1c}$  and HomaS, while the fourth group includes the remaining three main effects. The estimated $p-$values suggest that the significant characteristics are among those in Group 3 rather than Group 4.
	Group 5 consists of four variables: the baseline average levels for the   low-density lipoprotein cholesterol (LDL-C), total cholesterol, age and weight. This group of variables is of interest because 
	Glucose and lipid metabolism are linked to each other in many ways \citep{parhofer2015}. Age and weight are also always taken into account for optimal treatment regime estimation.
	Our test suggests that Group 5 does not appear to be influential in optimal treatment recommendation.

	\section{Discussions}\label{sec:discuss}
	We propose a flexible semi-parametric approach for making honest simultaneous inference about the importance of
	a group of variables on optimal treatment regime estimation. We develop new statistical theory to overcome the challenges of
	nonconvexity, high dimensionality and infinite-dimensional nonparametric components.
	
	In this paper, we focus on a randomized trial.
	For observation studies,
	let $\pi(\vx)=P(A=1|\vx)$ be the propensity score. Observing that $\E\{[A-\pi(\vx)]g(\vx)\} = 0$, we have
	$$4[ A_i-\pi(\vx_i)] Y_i =4[ A_i-\pi(\vx_i)]  g(\vx_i) +4[ A_i-\pi(\vx_i)] (A_i-1/2) f_0(\vx_i^T\vbeta_0)+4[ A_i-\pi(\vx_i)] \epsilon_i.$$
	Let $\widetilde{Y}_i = 4[ A_i-\pi(\vx_i)] Y_i $, $\widetilde{\epsilon}_i = 4[ A_i-\pi(\vx_i)][\epsilon_i+g(\vx_i)]$, then we have $$\E\widetilde{Y}_i =4[ A_i-\pi(\vx_i)] (A_i-1/2) f_0(\vx_i^T\vbeta_0). $$
	Denote $G(t|\vbeta) = \E(\widetilde{Y}|\vx^T\vbeta=t) = 2\E\{[ A-\pi(\vx)] (2A-1) f_0(\vx^T\vbeta_0)|\vx^T\vbeta=t \}$, $G^{(1)}(t|\vbeta) =\frac{d}{d t}G(t|\vbeta) $, and define $\widehat{G}(t|\vbeta)$, $\widehat{G}^{(1)}(t|\vbeta)$ similarly as in Section~\ref{sec:est_method}.
	Assume $\pi(\vx)=P(A=1|\vx)$ can be modeled as
	$\pi(\vx, \vxi)$, where $\vxi$ is a finite-dimensional parameter. 	Let $\widehat{\vxi}$ be an estimate of $\vxi$, such as the one based on the regularized logistic regression.
	Define the profiled semiparametric estimating function $\vS_n(\vbeta,\widehat{G},\widehat{\E},\widehat{\vxi}) = -n^{-1}\sum_{i=1}^n\{ 4[ A_i-\pi(\vx_i,\widehat{\vxi})] Y_i- \widehat{G}(\vx_i^T\vbeta|\vbeta) \}  \widehat{G}^{(1)}(\vx_i^T\vbeta|\vbeta)[\vx_{i,-1}-\widehat{\E}(\vx_{i,-1}|\vx_i^T\vbeta)].$
	We then estimate $\vbeta_0$ through the following penalized semiparametric profiled estimating equation
	$\vS_n(\vbeta,\widehat{G},\widehat{\E},\widehat{\vxi}) + \lambda\vkappah =\vnull.$
	%	For randomized trials, we have $\pi(\vx)\equiv %1/2$. Then the estimator mentioned above is %equivalent to the estimator we proposed in %Section~\ref{sec:est_method}.
	%Let $\widehat{\xi}$ be an estimate of $\vxi$, such as the one based on the popular logistic regression.
	%Under the popular assumption of nounmeasured confounders, a penalized profile semiparametric estimator  
	%for $\vbeta_0$ can be constructed as
	%\begin{align*}
	%\argmin_{\vbeta\in\bbR^p: |\beta_1|=1}\frac{1}{2n} \sum_{i=1}^n\Big\{\frac{ \widetilde{Y}_i -\widehat{G}(\vx_i^T\vbeta|\vbeta)}{A_i\pi(\vx, \widehat{\vxi}) + (1-A_i)(1-\pi(\vx, \widehat{\vxi}))}\Big\}^2+\sum_{j=1}^pp_\lambda(|\beta_j|).
	%\end{align*}
	%The bootstrap inference procedure can be implemented %similarly as described in Section~\ref{sec:inf_method}.
	Promising numerical performance of this estimator is reported in Section~S10.3 of the supplementary. Our approach can still be applied to investigate the theory but is it more complex due to the additional nuisance
	parameter.
	We will explore 
	the complete theory
	for the above estimator in the future work. Alternative approaches that can potentially be extended to our setting
	include \citet{nie2017}, \citet{knzel2018}, among ohers.

	Our approach for high-dimensional inference generalizes the ``inverting KKT condition" technique 
	in \citet{van2014asymptotically}. An alternative approach, which is more suitable if one is interested in some targeted
	lower-dimensional parameter is based on the idea of orthogonalization, see
	for example
	%\citet{BelloniChernozhukovHansen2014},  
	\citet{Belloni2015uniform}, \citet{Ning2017}, \citet{chernozhukov2018double}.
	In contrast, our approach is able to achieve debiasing for the $p$-dimensional coefficient vector simultaneously.
	The main idea of the orthogonalization approach is to construct a lower-dimensional estimating equation
	which is locally insensitive to the nuisance parameters. The construction
	of such a lower-dimensional moment condition is nontrivial for high-dimensional
	semiparametric setting, particularly for index model, where the challenge of bundled parameter arises.

	\bibliography{Hreference6}

	\section*{Supplementary material}
	
	Section~\ref{sec:proofnotation} of the supplementary material summarizes all the commonly used notation in the proof. Section~\ref{sec:appendix} summaries the regularity conditions and presents some discussions on these conditions.
	Section~\ref{sec:lemmas} presents the technical lemmas used in the proof.
	Section~\ref{sec:proof3.1} and Section~\ref{sec:proof3.2} of the supplementary material provide the proofs of 
	the theoretical results in Section 3.1 and 3.2 of the main paper, respectively. Section~\ref{sec:proof_append} presents the proofs of the technical lemmas in Section~\ref{sec:lemmas}.  Section~\ref{sec:proof_auxil} provides additional technical details.
	Section~\ref{sec:id_cond} introduces the identification conditions for the single index models mentioned in Assumption~\ref{A1}-(c).
	Section~\ref{sec:normal_verify} provides examples for verifying the regularity conditions. Section~\ref{sec:algo_numeric} presents the algorithms introduced in Section~\ref{sec:algo_est} of the main paper and some additional numerical results.
	
	\setcounter{section}{0}
	\setcounter{equation}{0}
	\renewcommand{\thesection}{S\arabic{section}}  
	\renewcommand{\theequation}{S\arabic{equation}} 
	\renewcommand{\thetable}{S\arabic{table}} 
	%\section{Preliminary knowledge about sub-Gaussian} \label{sec:subG}
	\section{Review of some useful notation}\label{sec:proofnotation}
	%\subsection{Notations for model and estimation}
	We collect below some notation 
	introduced in the main paper for easy reference in the proof.
	First, recall model (1) in the main paper is $$Y_i = g(\vx_i) +(A_i-1/2) f_0(\vx_i^T\vbeta_0)+\epsilon_i,\quad i=1, \ldots, n.$$ 
	Recall $G(t|\vbeta)=\E\{ f_0(\vx^T\vbeta_0) | \vx^T\vbeta=t\}$ and $G^{(1)}(t|\vbeta)=\frac{d}{d t}G(t|\vbeta)$. Note that $G(t|\vbeta_0) = f_0(t)$ and $G^{(1)}(t|\vbeta_0) = f_0'(t)$. We assume that $\vbeta_0\in \bbB_0 = \{\vbeta=(\beta_1,\cdots,\beta_p)^T:   \beta_1=1\}$.
	
	Denote $\vx_i =(x_{i,1},\vx_{i,-1}^T)^T$ and $\vbeta =(\beta_{1},\vbeta_{-1}^T)^T$. 
	The estimated profiled score function is
	\begin{align*}
	\vS_n(\vbeta,\widehat{G},\widehat{\E})=&-n^{-1}\sum_{i=1}^n\big[\widetilde{Y}_i-\widehat{G}(\vx_i^T\vbeta|\vbeta)\big]\widehat{G}^{(1)}(\vx_i^T\vbeta|\vbeta)\big[\vx_{i,-1}-\widehat{\E}(\vx_{i,-1}|\vx_i^T\vbeta)\big],
	\end{align*}
	where $\widetilde{Y}_i = 2(2A_i-1)Y_i$, $\widetilde{\epsilon}_i = 2(2A_i-1)[\epsilon_i+g(\vx_i)]$, and 
	\begin{align*}
	\widehat{G}(t|\vbeta) = \sum_{i=1}^n  W_{ni}(t,\vbeta)\widetilde{Y}_i,\quad	\widehat{G}^{(1)}(t|\vbeta) = \sum_{i=1}^n  W_{ni}^{(1)}(t,\vbeta)\widetilde{Y}_i,\quad 	\widehat{\E} (\vx_{-1}|\vx^T \vbeta=t) = \sum_{i=1}^n W_{ni}(t,\vbeta)\vx_{i,-1},
	\end{align*}
	with $K_h(z) = h^{-1}K(z/h)$, $ W_{ni}(t,\vbeta) = \frac{K_h(t-\vx_i^T\vbeta)}{\sum_{j=1}^nK_h(t-\vx_j^T\vbeta)}$, and $W_{ni}^{(1)}(t,\vbeta) =  \frac{d}{d t}W_{ni}(t,\vbeta) $. Note that to estimate $\widehat{G}(\vx_j^T \vbeta|\vbeta) $, $\widehat{G}^{(1)}(\vx_j^T \vbeta|\vbeta) $ and $\widehat{\E}(\vx_{j,-1}|\vx_j^T \vbeta) $, we  employ the  leave-one-out estimators.

	%\subsection{Notations for inference}
	The debiased estimator is defined as
	$$\vbetaw_{-1} = \vbetah_{-1} - \vThetah^T 	\vS_n(\vbetah,\widehat{G},\widehat{\E}),$$
	where the $(p-1)\times(p-1)$ matrix $\vThetah= (\vthetah_2,...,\vthetah_p)$ is an approximation to the inverse of $\nabla \vS_n(\vbetah,\widehat{G},\widehat{\E})$,
	where $\nabla$ denotes the gradient with respect to the components of $\vbeta_{-1}$. 
	Given an initial estimator $\vbetah$, to compute $\vthetah_{j}$, let
	\begin{align*}
	\vdh_j\triangleq \vd_j(\vbetah,\eta) = \argmin\limits_{\bm{v}\in\bbR^{p-2} }||\vv||_1 \mbox{  subject to }   \Big|\Big|n^{-1}\sum_{i=1}^n \big[\widehat{G}^{(1)}(\vx_i^T\vbetah|\vbetah)\big]^2 (\xh_{i,j}-\vxh_{i,-j*}^T\vv)\vxh_{i,-j*}\Big|\Big|_\infty\leq \eta,
	\end{align*} 
	for some $\eta>0$, $j=2,\cdots,p$,
	where $||\cdot||_\infty$ denotes the infinity norm of a vector, $\vxh _i= \vx_i-\widehat{\E}(\vx_i|\vx_i^T\vbetah)$,
	$\xh_{i,j}$ denotes the $j^{th}$ entry of the vector $\vxh_i$, $\vxh_{i,-1}$ denotes the $(p-1)$-subvector of $\vxh_i$ that excludes its $1^{st}$ entry, and $\vxh_{i,-j*}$ denotes the $(p-2)$-subvector of $\vxh_i$ that excludes both its $1^{st}$ and $j^{th}$ entries. Furthermore, for $j=2,\cdots,p$, let 
	\begin{align*}
	\vphi_j(\vbetah,\eta)& = \Big(-\big(\vd_j(\vbetah,\eta) \big)^T_{1:(j-2)}, 1, -\big(\vd_j(\vbetah,\eta) \big)^T_{(j-1):(p-2)}\Big)^T,\\
	\widehat{\tau}_j^{2} \triangleq\tau_j^{2}(\vbetah,\eta)&=n^{-1}\sum_{i=1}^n\big[ \widehat{G}^{(1)}(\vx_i^T\vbetah|\vbetah)\big]^2\xh_{i,j} \vxh_{i,-1}^T\vphi_j(\vbetah,\eta) ,  \\
	\vthetah_j \triangleq\vtheta_j(\vbetah,\eta) &= \tau_j^{-2}(\vbetah,\eta)\vphi_j(\vbetah,\eta).
	\end{align*}   
	Then we can define the matrix $\vSigmah(\vbetah)$ as below:
	$$\vSigmah(\vbetah) \triangleq\vThetah^T\Big\{\frac{1}{n}\sum_{i=1}^n \big[\widetilde{Y}_i-\widehat{G}(\vx_i^T\vbetah|\vbetah)\big]^2 [\widehat{G}^{(1)}(\vx_i^T\vbetah|\vbetah)]^2 \vxh_{i,-1}\vxh_{i,-1}^T\Big\} \vThetah.$$
	
	Define the matrix $\vOmega =\E \big\{[G^{(1)}(\vx_i^T\vbeta_0|\vbeta_0)]^2 \vxw_{i,-1}\vxw_{i,-1}^T\big\}$, with $\vxw_{i,-1} = \vx_{i,-1}-\E(\vx_{i,-1}|\vx_i^T\vbeta_0)$,
	and its inverse $\vOmega^{-1}\triangleq\vTheta = (\vtheta_2,...,\vtheta_p) $. For $j=2, \ldots, p$, let $\vOmega_{-(j-1),-(j-1)}\in\bbR^{(p-2)\times (p-2)}$ 
	be the submatrix of $\vOmega$ with its $(j-1)^{th}$ row and $(j-1)^{th}$ column removed;
	similarly $\vOmega_{-(j-1),(j-1)}\in\bbR^{p-2}$ denotes the $(j-1)^{th}$ column of $\vOmega$ with its $(j-1)^{th}$ entry removed.
	%Assume $\vOmega_{-j,-j}$ is positive definite. 
	Define
	$\vd_{0j} = (\vOmega_{-(j-1),-(j-1)})^{-1}\vOmega_{-(j-1),(j-1)}$, $s_j = ||\vd_{0j}||_0$, $\widetilde{s}= \max_{2\leq j\leq p}s_j $ and
	$\tau^2_{0j} =\Omega_{(j-1),(j-1)} - \vd_{0j}^T\vOmega_{-(j-1),(j-1)}=[\Theta_{(j-1),(j-1)}]^{-1}$, $\vphi_{0j} = \tau^{2}_{0j}\vtheta_j$. Denote $\bbK(p,s_0) = \{\vv\in\bbR^p:||\vv||_2\leq 1,||\vv||_0\leq s_0\}$ for any integers $p$ and $s_0$.
	Finally, recall $s=||\vbeta_0||_0$, $\bbB_1= \big\{\vbeta\in\bbB:||\vbeta-\vbeta_0||_2\leq c_0 \sqrt{s}h^2 , ||\vbeta||_0\leq ks\big\},$ where $\bbB=\big\{\vbeta\in\bbB_0: ||\vbeta-\vbeta_0||_2\leq r,  ||\vbeta||_0\leq ks\big\}$.
	
	For any $p-$dimensional vector $\vv=(v_1,\cdots,v_p)$, we denote $||\vv||_1=\sum_{j=1}^p|v_j|$, $||\vv||_2=\sqrt{\sum_{j=1}^p|v_j|^2}$, and $||\vv||_\infty=\max_{1\leq j\leq p}|v_j|$. For  any matrix $\vA=(a_{ij})\in\bbR^{p_1\times p_2}$, where $p_1$, $p_2$ are two arbitrary integers, we denote $||\vA||_\infty=\max_{ 1\leq i \leq p_1,1\leq j\leq p_2}|a_{ij}|$.

	\section{Regularity Conditions}\label{sec:appendix}
	
	We %recall and 
	define some notation first.  %Recall $s=||\vbeta_0||_0$, $\vxw_{i,-1} = \vx_{i,-1}-\E(\vx_{i,-1}|\vx_i^T\vbeta_0)$,
	%and $\ve_{j-1}$ denotes the $(p-1)$-dimensional vector with the $(j-1)^{th}$  entry being one and all the other entries equal to zero.  Recall $\vOmega=\E \big\{[G^{(1)}(\vx_i^T\vbeta_0|\vbeta_0)]^2 \vxw_{i,-1}\vxw_{i,-1}^T\big\}$, 
	%with $\vxw_{i,-1} = \vx_{i,-1}-\E(\vx_{i,-1}|\vx_i^T\vbeta_0)$,
	%and $x_j$ denotes the $j^{th}$ component of $\vx$. 
	Given any square matrix $A$, $\lambda_{\max}(A)$ and $\lambda_{\min}(A)$ denote the largest and the smallest eigenvalues of $A$, respectively.
	%Let $\mathcal{S}_{\vphi_j}$  be the support set of $\vphi_{0j} = \tau_{0j}^2\vtheta_j$, $j=2, \ldots, p$.
	Let $\mathcal{V}_1=
	\{\vv\in\bbR^{p-1}:  ||\vv||_2=1,  	||\vv||_0\leq 2ks\}$, 
	where $k>$ is a positive integer.
	%Let \bfred{$\mathcal{V}_{2j}= \{\vv=(v_1,\cdots,v_{p-1})^T: ||\vv_{\mathcal{S}_{\vphi_j}^C}||_1 \leq ||\vv_{\mathcal{S}_{\vphi_j}}||_1, ||\vv||_2=1, v_{j-1}=0\}$},  for any $j=2,\ldots,p$ such that $\vphi_{0j}\neq\ve_{j-1}$.
	%Recall $\bbB=\{\vbeta\in\bbB_0: ||\vbeta-\vbeta_0||_2\leq r,  ||\vbeta||_0\leq ks\}$ and $\bbB_1=\big\{\vbeta\in\bbB:||\vbeta-\vbeta_0||_2\leq c_0 \sqrt{s}h^2, ||\vbeta||_0\leq ks\big\}$. 
	Let
	$\E^{(j)}(\vx_{-1}|\vx^T\vbeta=t)$ denote the derivatives of
	$\E(\vx_{-1}|\vx^T\vbeta=t)$ with respect to $t$, $j=1,2$.
	Let 
	$\E^{(1)}(\vx_{-1}\vx_{-1}^T|\vx^T\vbeta=t)$ denote the derivative
	of $\E(\vx_{-1}\vx_{-1}^T|\vx^T\vbeta=t)$ with respect to $t$.
	
	We state below a set of regularity conditions, followed by Remarks (a)--(c)
	to discuss these conditions. 
	
	\label{sec:assump}
	\begin{enumerate} [label=(A\arabic*)]
		\item \label{A1}
		\begin{enumerate}
			\item  The distributions of $\vx\in\bbR^p$ and $\ep$ are sub-Gaussian with variance proxy $\sigma_x^2$ and $\sigma_\epsilon^2$, respectively, where $p\geq 2$. 
			%Suppose that $\lambda_{max}(\E(\vx\vx^T))\leq \xi_1$	 for some constant $\xi_1>0$, where $\lambda_{max}(\cdot)$ denotes the largest eigenvalue.
			\item The function $f_0(\cdot)$ satisfies 	%\bfred{$\\n^{-1}\sum_{i=1}^n[f'_0(\vx_i^T\vbeta_0)]^2\geq a^2>0$},  
			$\E\{[f'_0(\vx^T\vbeta_0)]^2\}= a^2>0$ 
			and $\max_{ 1\leq i \leq n}| f_0'(\vx_i^T\vbeta_0)|\leq b$ for some positive 
			constants $a$ and $b$, where 
			$f'_0$ denotes its first derivative, and $||\vbeta_0||_2$ is bounded.
			Its second derivative $ f_0''(z) $ and third derivative $ f_0'''(z) $ are bounded for $z\in\bbR$.
			%The function $f_0(\cdot)$ is twice differentiable. The first derivative satisfies %$\min_{1\leq i\leq n}\\|f'_0(\vx_i^T\vbeta_0)|\geq a$ 
			%\bfred{$\\n^{-1}\sum_{i=1}^n[f'_0(\vx_i^T\vbeta_0)]^2\geq a^2>0$}, 
			%and $| f_0'(z)|\leq b$ for any $z\in\bbR$;  and the second derivative $ f_0''(z) $ is bounded and Lipschitz for $z\in\bbR$.
			\item The lower-dimensional true model $f_0(\vx^T\vbeta_0)$ satisfies the 
			identifiability conditions for the classical single index models (e.g., \citet{ichimura1993,horowitz2012semiparametric}, see 
			Section~\ref{sec:id_cond} of the online supplement  for details).
			%and $||\vbeta_0||_0\geq 2$. 
			\item The main effect $g(\cdot)$ satisfies
			%$||\E(\vx\vx^T)||_\infty \leq M$, 
			$P\big(\max_{ 1\leq i \leq n}|g(\vx_i)|\leq M \big)=1$ for some positive constant $M$.
			%$P(\sup_{\vbeta\in\bbB,1\leq i\leq n}|\vx_i^T\vbeta|\leq M )=1$,
			%where $\bbB=\{\vbeta\in\bbR^p: ||\vbeta-\vbeta_0||_2\leq r, ||\vbeta||_0\leq ks\}$,
			%	with $r>0$ and $k>1$ being positive constants and $s= ||\vbeta_0||_0$.
		\end{enumerate}  
		\item \label{A2}
		\begin{enumerate}
			\item  
			There exist some positive constants $M$, $\xi_0$, $\xi_1$, $\xi_2$, $\xi_3$ and $\xi_4$   such that we have %$\E\{[x_j-\E(x_j|\vx^T\vbeta_0)]^2\} \leq M$,  for any $j\in\{1,,\cdots,p\}$.  It is assumed that
			$\inf_{\vv\in \mathcal{V}_1}\vv^T\E\big[\Cov (\vx_{-1}|\vx^T\vbeta_0)\big]\vv \geq \xi_0$,
			$\lambda_{\max}(\E\big[\Cov (\vx_{-1}|\vx^T\vbeta_0)\big]) \leq \xi_1$,
			%\bfred{$\min\limits_{2\leq j\leq p}\big\{\min\limits_{j: \vphi_{0j}\neq\ve_{j-1}}\inf\limits_{\vv\in \mathcal{V}_{2j}}\vv^T\vOmega\vv, \min\limits_{j:\vphi_{0j}=\ve_{j-1}} \Omega_{(j-1),(j-1)}\}\geq \xi_2$, }
			$\lambda_{\min}(\vOmega)\geq \xi_2$, 
			and $\lambda_{\max}(\E(\vx\vx^T))\leq \xi_3$.
			Also,  $\sup_{\vbeta\in\bbB}n^{-1}\sum_{i=1}^n\left[\lambda_{\max}\big(\E(\vx_i\vx_i^T|\vx_i^T\vbeta)\big)\right]^2\leq \xi_4$
			and $\max_{ 1\leq i \leq n}\sup_{\vbeta\in\bbB_1}   \lambda_{\max}\big(\E(\vx_i\vx_i^T|\vx_i^T\vbeta)\big) \leq M  \log(p\vee n)$, for all $n$ sufficiently large.
			%with probability at least $1-\exp[-c\log(p\vee n)]$ 
			%all $n$ sufficiently large. 
			%It satisfies that  and some constant $M>0$.
			%The smallest eigenvalue of  $\E\big[\Cov (\vx|\vx^T\vbeta_0)\big]$, denoted as $\xi_0$, is strictly positive, and its largest eigenvalue $\xi_p$ is bounded.   

			\item $\E(\vx_{-1}|\vx^T\vbeta=t)$ is twice-differentiable with respect to $t$,
			and $\E(\vx_{-1}\vx_{-1}^T|\vx^T\vbeta=t)$ is differentiable with respect to $t$.
			There exists some positive constant $M$ such that for any $\veta\in\bbR^{p-1}$,
			$\max_{ 1\leq i \leq n}\sup_{\vbeta\in\bbB} |\E^{(1)}(\vx_{i,-1}^T\veta|\vx_i^T\vbeta)| \leq M ||\veta||_2$,  %$\max_{ 1\leq i \leq n}\sup_{\vbeta\in\bbB} ||\E^{(2)}(\vx_i|\vx_i^T\vbeta)||_2\leq M$, 
			and $\sup_{|t|\leq 2||\vbeta_0||_2\sigma_x\sqrt{\log(p\vee n)}}\sup_{\vbeta\in\bbB} |\E^{(2)}(\vx_{-1}^T\veta|\vx^T\vbeta=t)|\leq M ||\veta||_2$. \\
			%	and $\max_{ 1\leq i \leq n}\sup_{\vbeta\in\bbB} \big\{||\E^{(1)}[(\vx_i^T\veta)^2|\vx_i^T\vbeta]||_2/|\vx_i^T\vbeta|\big\}\leq M ||\veta||_2^2$, for any $\veta\in\bbR^p$.
			Furthermore, $\max_{ 1\leq i \leq n}\sup_{\vbeta\in\bbB_1}   \big\{\big|\E^{(1)}[(\vx_{i,-1}^T\veta)^2] |\vx_i^T\vbeta\big| \leq M ||\veta||_2^2\sqrt{\log(p\vee n)}$,
			%with probability at least $1-\exp[-c\log(p\vee n)]$, 
			for all $n$ sufficiently large.
			%The derivatives $\E^{(1)}(\vx|\vx^T\vbeta=t)$ and $\E^{(2)}(\vx|\vx^T\vbeta=t)$ have bounded $l_2$ norms,\bfred{ and $\Big|\E^{(1)}[(\vx^T\veta)^2|\vx^T\vbeta=t]\Big|\leq C|t|*||\veta||_2^2$},  for any $\vbeta\in\bbB$, $\veta\in\bbR^p$, $ |t|\leq \sqrt{\log(p\vee n)}$, and some constant $C>0$.
			%Suppose that $\E(\vx|\vx^T\vbeta=t)$ is twice-differentiable with respect to $t$.  The derivatives $\E^{(1)}(\vx|\vx^T\vbeta=t)$ and $\E^{(2)}(\vx|\vx^T\vbeta=t)$ have bounded $l_2$ norms for any $\vbeta\in\bbB$ and $t\in\bbT  = \{t \in\bbR : |t|\leq\sqrt{ \log (p\vee n)}\}$.
			\item For some positive constant $C$,
			\begin{align*}
			&\sup_{\vv\in\bbK(p,2ks+\widetilde{s})}\big|[\E(\vx|\vx^T\vbeta_1)-\E(\vx|\vx^T\vbeta_2)]^T\vv\big|\\
			\leq & C \Big(|\vx^T\vbeta_1-\vx^T\vbeta_2|  +\max(|\vx^T\vbeta_1|,|\vx^T\vbeta_2|)  ||\vbeta_1-\vbeta_2||_2\Big),
			%\\	\big|\E [(\vx^T\veta)^2|\vx^T\vbeta_1] &-\E [(\vx^T\veta)^2|\vx^T\vbeta_2]\big| \\
			%\leq  C||\veta||_2^2\Big[\big|(\vx^T\vbeta_1)^2-(\vx^T\vbeta_2)^2\big|  & + \max(1,|\vx^T\vbeta_1|^2,|\vx^T\vbeta_2|^2) *||\vbeta_1-\vbeta_2||_2\Big],
			%\big|[\E(\vx|\vx^T\vbeta_1)-\E(\vx|\vx^T\vbeta_2)]^T\veta\big|\leq& C||\veta||_2\Big(|\vx^T\vbeta_1-\vx^T\vbeta_2|  +  ||\vbeta_1-\vbeta_2||_2\Big),\\
			%\big|\E [(\vx^T\veta)^2|\vx^T\vbeta_1] -\E [(\vx^T\veta)^2|\vx^T\vbeta_2]\big| \leq & C||\veta||_2^2\Big(|\vx^T\vbeta_1-\vx^T\vbeta_2|  +||\vbeta_1-\vbeta_2||_2\Big),
			\end{align*}
			for any $\vbeta_1,\vbeta_2\in\bbB$, $\vx\in\bbR^p$, where $\bbK(p,2ks+\widetilde{s}) = \{\vv\in\bbR^p:||\vv||_2\leq 1,||\vv||_0\leq 2ks+\widetilde{s}\}$, and  $\widetilde{s}=\max_{ 1\leq j\leq p}||\vd_{0j}||_0$. 
		\end{enumerate} 
		\item \label{K1} The kernel function $K(\cdot)$ is nonnegative,  symmetric about $0$, and 
		twice differentiable and bounded on the real line. 
		The function $K(\cdot)$ and its derivatives $K'(\cdot)$, $K''(\cdot)$ are all Lipschitz on the real line.
		Furthermore,
		$\lim\limits_{|\nu| \rightarrow \infty} K(\nu)= 0$, $\int_{-\infty}^{\infty}  K(\nu) d\nu=1$,  $\int_{-\infty}^{\infty} \nu K'(\nu) d\nu=-1$, and $\int_{-\infty}^{\infty} \nu^2 K''(\nu) d\nu=2$. 
		%\item \label{K2}  
		For any integer $0\leq i\leq 4$
		$\int |\nu^i K(\nu)| d\nu<\infty$; for integer $0\leq i\leq 3$, $\int |\nu^i K'(\nu)| d\nu<\infty$; for integer $0\leq i\leq 2$, $\int |\nu^i K''(\nu) |d\nu<\infty$.
		%$\int \nu^2[K(\nu)]^2 d\nu$, $\int \nu^2[K'(\nu)]^2 d\nu$, $\int \nu^4K(\nu) d\nu $, and $\int \nu^4|K'(\nu)| d\nu$ are all finite. 
		\item \label{K3} Let  $f_{\bm{\beta}}(\cdot)$ denote the density function of $\vx^T\vbeta$.
		Suppose that $f_{\bm{\beta}}(\cdot)$ is twice differentiable, and $f_{\bm{\beta}}(\cdot)$, $f_{\bm{\beta}}'(\cdot)$, 
		$f_{\bm{\beta}}''(\cdot)$ are all bounded on the real line. Furthermore, for some positive constant $M$, $P(\max_{ 1\leq i \leq n}\sup_{\vbeta\in\bbB}  f^{-1}_{\bm{\beta}}(\vx_i^T\vbeta) \leq M)=1$.
		%$f^{-1}_{\bm{\beta}}(t) \leq M$, $\forall \ |t|\leq \sqrt{\log(p\vee n)}$ and $\vbeta\in\bbB$,
		.
		\item \label{K4}  \begin{enumerate}
			\item For any $\vbeta\in\bbB$ and $t\in \bbR$, 
			$G(t|\vbeta)=\E\{f_0(\vx^T\vbeta_0) | \vx^T\vbeta=t\}$ is twice differentiable with respect to $t$.  Its first derivative satisfies 
			$P(\max_{ 1\leq i \leq n}\sup_{\vbeta\in\bbB}  |G^{(1)}(\vx_i^T\vbeta|\vbeta)|\leq b)=1$, for 
			some positive constant  $b$.
			Its second derivative $G^{(2)}(t|\vbeta)$ is bounded.
			\item %$G(\vx^T\vbeta|\vbeta)$ and
			$G^{(1)}(t|\vbeta)$ satisfies 
			\begin{align*}
			%\big|G(\vx^T\vbeta|\vbeta) - G(\vx^T\vbeta_0|\vbeta_0)\big|\leq C\big( |\vx^T\vbeta-\vx^T\vbeta_0| +||\vbeta-\vbeta_0||_2 \big),\\ 
			n^{-1} \sum_{i=1}^{n}\big[G^{(1)}(\vx_i^T\vbeta|\vbeta) - G^{(1)}(\vx_i^T\vbeta_0|\vbeta_0)\big]^2 \leq c_1 ||\vbeta-\vbeta_0||_2^{2},
			\end{align*}
			%with probability at least $1-\exp(-c_2\log p)$, 
			for any $\vbeta\in\bbB$, some positive constant $c_1$ and all $n$ sufficiently large. 
			\item $G(t|\vbeta)$ and $G^{(1)}(t|\vbeta)$ satisfy the local Lipschitz conditions:
			%for any $\vbeta_1,\vbeta_2\in\bbB$ and $\vx\in\bbR^p$,
			\begin{align*}
			%\big|G(\vx^T\vbeta_1|\vbeta_1) &- G(\vx^T\vbeta_2|\vbeta_2)\big|\\
			%\leq C\Big\{|\vx^T\vbeta_1-\vx^T\vbeta_2| +\big[1+&\max(|\vx^T\vbeta_1|,|\vx^T\vbeta_2|) %\big]*\sqrt{||\vbeta_1-\vbeta_2||_2} \Big\},\\
			%\big|G^{(1)}(\vx^T\vbeta_1|\vbeta_1) &- G^{(1)}(\vx^T\vbeta_2|\vbeta_2)\big|\\
			%\leq C\Big\{|\vx^T\vbeta_1-\vx^T\vbeta_2|+\big[1+&\max(|\vx^T\vbeta_1|,|\vx^T\vbeta_2|) \big]*\sqrt{||\vbeta_1-\vbeta_2||_2} \Big\},
			%n^{-1}\sum_{i=1}^{n}\big|G(\vx_i^T\vbeta_1|\vbeta_1) - G(\vx_i^T\vbeta_2|\vbeta_2)\big|\leq c_1||\vbeta_1-\vbeta_2||_2^{1/2} \log (p\vee n),\\
			%n^{-1}\sum_{i=1}^{n}\big|G^{(1)}(\vx_i^T\vbeta_1|\vbeta_1) - G^{(1)}(\vx_i^T\vbeta_2|\vbeta_2)\big|\leq c_1||\vbeta_1-\vbeta_2||_2^{1/2} \log (p\vee n),
			\sup_{|t|\leq c_0\sqrt{s\log(p\vee n)}}  \big[G(t|\vbeta_1) - G(t|\vbeta_2)\big]^2 \leq c_1s\log (p\vee n)||\vbeta_1-\vbeta_2||_2 ,\\
			\sup_{ |t|\leq c_0\sqrt{s\log(p\vee n)}} \big[G^{(1)}(t|\vbeta_1) - G^{(1)}(t|\vbeta_2)\big]^2 \leq c_1s\log (p\vee n)||\vbeta_1-\vbeta_2||_2 ,
			\end{align*}
			%with probability at least $1-\exp[-c_2\log (p\vee n)]$, 
			for any $\vbeta_1,\vbeta_2\in \bbB$, %and $||\vbeta_1-\vbeta_2||_2\leq \frac{c_2sh^4}{\log(p\vee n)}$,
			for some positive constants $c_0$,  $c_1$, and all $n$ sufficiently large.
			%			\item  \bfred{$G(\vx_i^T\vbeta_0|\vbeta)$ satisfies the local   conditions:
			%			$$\max_{ 1\leq i \leq n}\sup_{\vbeta\in\bbB_1}\left|\left|\frac{\partial}{\partial\vbeta}G(\vx_i^T\vbeta_0|\vbeta)-\left[\frac{\partial}{\partial\vbeta}G(\vx_i^T\vbeta_0|\vbeta)\Big|_{\vbeta=\vbeta_0}\right]\right|\right|_\infty\leq c_0\sqrt{\log (p\vee n)},$$ 
			%			%with probability $1-\exp[-c_1\log(p\vee n)]$, 
			%			for some positive constant $c_0$ and all $n$ sufficiently large.}
		\end{enumerate}
		%\item \label{K5}  %$\inf_{\vbeta\in\bbB}\lambda_{\min}\big[\Cov (\vx|\vx^T\vbeta)\big]\geq d$ 
		%$\lambda_{\min}\big[\Cov (\vx|\vx^T\vbeta_0)\big]\geq d$ with probability one for constant $d>0$, where $\lambda_{\min}(\cdot)$ means the smallest eigenvalue of a matrix.
		%\item \label{B1} The smallest eigenvalue $\xi_1$  of  $\E(\vx\vx^T)$ is strictly positive and $\xi_1^{-1}\leq M_2$ for constant $M_2>0$.  
	\end{enumerate}

	\noindent{\it Remark (a).}
	%Assumption  \ref{A1} includes the assumptions about data and model (\ref{model}).
	Our assumptions on the covariates in \ref{A1}  are similar to those in the literature on high-dimensional inference
	with random designs (e.g., \citet{van2014asymptotically,
		Belloni2015uniform}, among others). 
	We assume $p\geq 2$. If $p=1$, then the index model degenerates to a nonparametric model.
	The conditions in \ref{K1} are standard assumptions on 
	the kernel function for nonparametric smoothing. 
	The assumptions in \ref{K3} on the distributions of $\vx^T\vbeta$
	are common for index models. 
	Assumption \ref{A2} involves restricted eigenvalue types assumptions and conditions on $\E(\vx|\vx^T\vbeta)$ 
	and \ref{K4} imposes conditions on the function $G(t|\vbeta)$ and $G^{(1)}(t|\vbeta)$.
	In Section~\ref{sec:normal_verify} of the supplementary material, we verify that 
	these key assumptions
	are satisfied when $\vx$ follows the multivariate normal distribution
	in the high-dimensional setting.\\

	\noindent{\it Remark (b).}
	Comparing with low-dimensional single-index models, conditions (A2) and (A5) are worthy of some discussions. The conditional expectations
	in both conditions depend on $\vx^T\vbeta$  and $\vbeta$, possibly in
	a nonlinear fashion. As an example, in the multivariate normal distribution setting,
	the linearity condition $\E(\vx^T\veta|\vx^T\vbeta)=c_{\veta, \vbeta}\vx^T\vbeta$, where
	$c_{\veta, \vbeta}$  is a non-stochastic constant,
	plays an important role in the 
	low-dimensional theory. In the high-dimensional setting,
	$c_{\veta, \vbeta}$ (depending on $\vbeta$, $\veta$ nonlinearly) 
	requires more careful analysis.\\
	
	\noindent{\it Remark (c).}
	For identifiability, we assume that there exists a covariate with a nonzero coefficient. In practice, domain experts may help suggest such a candidate continuous covariate and the statisticians can run confirmatory analysis (e.g., comparing the conditional treatment effect conditional on this covariate) to verify if this is a viable choice.
	In the literature, another popular condition for identifiability is $||\vbeta_0||_2=1$. However, it was also recognized (\citet{yu2002penalized}, \citet{ZhuXue}, \citet{wang2010estimation}, among others) that technical derivation under this identifiability condition is more involved 
	due to the fact $\vbeta_0$ is a boundary point of a unit sphere and the derivative 
	does not exist at $\vbeta_0$. To handle this, the aforementioned literature suggested
	a delete-one-component approach. It was assumed that the true $\vbeta$ has a positive component $\beta_{r}$.
	Let $\vbeta=(\beta_{1},\cdots,\beta_{p})^T$ and $\vbeta^{(r)}$ be a $(p-1)$-subvector of $\vbeta$ that excludes the $r^{th}$ entry.  Then we can write $\vbeta=(\beta_{1},\cdots,\beta_{r-1}, \sqrt{1-||\vbeta^{(r)}||_2^2},\beta_{r+1},\cdots,\beta_{p})^T.$ Thus the model can be reparametrized using the $(p-1)$-dimensional parameter $\vbeta^{(r)}$.
	Under the assumption $||\vbeta^{(r)}||_2<1$ (reasonable under the assumption that the 
	underlying model has dimension at least two, otherwise it degenerates to a fully nonparametric model), the Jacobian matrix of the transformation can be computed as $\vJ_{\vbeta^{(r)}} = (\vgamma_1,\cdots,\vgamma_p)^T$,
	where $\vgamma_s = \ve_s$ ($s^{th}$ column of the identity matrix $\vI_p$), for $1\leq s\leq p,s\neq r$, and $\vgamma_r = (1-||\vbeta^{(r)}||_2^2)^{-1/2}(\beta_{1},\cdots,\beta_{r-1}, -\sqrt{1-||\vbeta^{(r)}||_2^2},\beta_{r+1},\cdots,\beta_{p})^T$. 
	Note that this transformation analysis also relies on knowing a covariate with a positive coefficient.
	%\\
	%If the true parameter $\vbeta^{(r)}$ turns out $||\vbeta^{(r)}||_2=1$, which indicates that %$\beta_r=0$, then the Jacobian matrix is similar as above, but with $\vgamma_r=\vnull$. %Therefore, we have $\mbox{det}(\vJ_{\vbeta^{(r)}})  = 0$. Then this transformation makes no %sense. 
	%\\

	\section{Some useful definitions and lemmas} \label{sec:lemmas}
	In this section, we introduce several useful definitions and lemmas which will be used in the proof of the theory. The proofs of these lemmas are given in Section~\ref{sec:proof_append}.
	
	\begin{definition} 
		A random vector $\vx\in\bbR^p$ is said to be sub-Gaussian with variance proxy $\sigma_x^2$ if $\emph{E}\vx=\vnull$, and for each (fixed) unit vector $\vv\in\bbR^p$, 
		$$\emph{E}[\exp(s\vx^T\vv)]\leq \exp\Big(\frac{\sigma_x^2s^2}{2}\Big), \qquad \forall\ s\in\bbR.$$
		%We write $\vx\sim \mbox{subG}(\sigma^2)$. 
		An equivalent definition is that for each (fixed) unit vector $\vv\in\bbR^p$ and any $t>0$, $P(|\vx^T\vv|\geq t)\leq 2\exp\big(-\frac{t^2}{2\sigma_x^2}\big)$. 
		
		Property: %$\vx_1,\cdots,\vx_n\stackrel{i.i.d}{\sim}\mbox{subG}(\sigma^2)$, 
		Let $\vx_1,\cdots,\vx_n\in\bbR^p$ be independent sub-Gaussian random vectors in 
		$\bbR^p$ with variance proxy $\sigma_x^2$.
		Then $\forall\ t>0$, 
		$P\big(\max_{1\leq i\leq n}||\vx_i||_\infty>t\big)\leq 2np \exp\big(-\frac{t^2}{2\sigma_x^2}\big)$. As a result,
		\begin{align}
		P\Big(\max_{1\leq i\leq n}||\vx_i||_\infty>2\sigma_x\sqrt{\log (np)}\Big)\leq 2\exp\big[-\log (np)\big].\label{subG_prop}
		\end{align} 
	\end{definition}
	
	\begin{definition} 
		Let $L_k(\bbP_n) = \left|n^{-1}\sum_{i=1}^{n}\vgamma^k(Z_i)\right|^{1/k}$, $k=1,2$.
		For $\delta>0$, the $\delta-$covering number $N(\delta,\Gamma,L_k(\bbP_n))$ of 
		the class of functions $\Gamma$ is the minimum number of $L_k(\bbP_n)-$balls with radius $\delta$ to cover $\Gamma$. The entropy is $H(\cdot,\Gamma,L_k(\bbP_n))\triangleq \log\left[N(\delta,\Gamma,L_k(\bbP_n))\right]$.
	\end{definition}
	
	\begin{definition} 
		A Rademacher sequence is a sequence $\epsilon_1,\cdots,\epsilon_ n$ of i.i.d copies of a random variable $\epsilon$ taking values in $\{1,-1\}$, with $P(\ep=1) = P(\ep=-1)=1/2$.
	\end{definition}

	\setcounter{lemma}{0}
	\renewcommand{\thelemma}{A\arabic{lemma}} 
	
	\blem
	\label{Lgrad}
	%Suppose \ref{A1}--\ref{K4} hold. 
	Under the assumptions of Theorem~\ref{Lasso_error},
	there exist universal positive constants $d_0$ and $d_1$  such that 
	$\big|\big|\vS_n(\vbeta_0,\widehat{G},\widehat{\emph{E}})\big|\big|_\infty\leq d_0  \sqrt{\frac{\log (p\vee n)}{n}} $ with probability at least $1- \exp[-d_1\log (p\vee n)]$.
	\elem

	\blem \label{lem:subg} 
	If $\vx\in\bbR^p$ is sub-Gaussian with variance proxy $\sigma_x^2$, then for any $\vbeta\in\bbB$, $\emph{E}(\vx|\vx^T\vbeta)$ and $\vx-\emph{E}(\vx|\vx^T\vbeta)$ are both sub-Gaussian, with variance proxy $\sigma_x^2$ and $2\sigma_x^2$, respectively. Furthermore, under assumption~\ref{A1}, $\widetilde{\epsilon} = 2(2A-1)\big[\epsilon+g(\vx)\big]$ is  sub-Gaussian with variance proxy $4(\sigma_\ep^2 + M^2)$. 
	In addition, if  $x$ is a random variable such that $|x|\leq \sigma_x$, for some positive constant $\sigma_x$, and $y$ is a sub-Gaussian random variable with variance proxy $\sigma_y^2$, then $xy-\emph{E}(xy)$ is  sub-Gaussian  with variance proxy no larger than $4\sigma_x\sigma_y$.
	\elem
	
	\blem \label{lem:events} 
	Define the following four events:
	\begin{align*} 
	%\quad\mathcal{G}_n &= \Big\{\max_{1\leq i \leq n,2\leq j\leq p}\big|G^{(1)}(\vx_i^T\vbeta_0|\vbeta_0) \big[\vx_i-\emph{E}(\vx_i|\vx_i^T\vbeta_0)\big]^T\vtheta_j\big|\leq\sigma_x\sqrt{\log (p\vee n)}\Big\},\\
	%\mathcal{H}_n &= \Big\{\max_{2\leq j \leq p, 1\leq i \leq n}\big|\widetilde{\epsilon}_i\big[\vx_i-\emph{E}(\vx_i|\vx_i^T\vbeta_0)\big]^T\vtheta_{j}\big|\leq\sigma_x(\sigma_\ep+M_g)  \log (p\vee n) \Big\},\\
	%\mathcal{K}_n &= \big\{\max_{ 1\leq i \leq n}||\vx_i||_\infty\leq \sigma_x\sqrt{\log (p\vee n)}\big\}.
	\mathcal{G}_n &= \left\{\max_{ 2\leq j\leq p}\sup_{\vbeta\in\bbB }n^{-1}\sum_{i=1}^n\left|  \left[\vx_{i,-1}-\emph{E}(\vx_{i,-1}|\vx_i^T\vbeta)\right]^T\vtheta_j\right|^2\leq d_0\xi_2^{-2}\sigma_x^2\right\},\\ 
	\mathcal{H}_n &= \left\{\max_{ 2\leq j\leq p}n^{-1}\sup_{\vbeta\in\bbB }\sum_{i=1}^n\left| %\widetilde{\epsilon}_i
	2[\ep_i+g(\vx_i)]*	 \left[\vx_{i,-1}-\emph{E}(\vx_{i,-1}|\vx_i^T\vbeta)\right]^T\vtheta_j\right|^2\leq d_1\xi_2^{-2}\sigma_x^2(\sigma_\ep^2+M^2)\right\},\\
	\mathcal{J}_n &= \left\{\max_{1\leq i \leq n}\sup_{\vbeta\in\bbB_1}\left| \vx_i^T\vbeta \right|\leq 2||\vbeta_0||_2\sigma_x\sqrt{\log(p\vee n)}\right\},\\
	\mathcal{K}_n &= \left\{\sup_{\vv\in\bbK(p,s_0)}n^{-1}\sum_{i=1}^n\left| \vx_i^T\vv \right|^2\leq 2\sigma_x^2\right\},
	\end{align*} 
	for some positive constants $d_0>4$ and $d_1>256\sqrt{2}$.
	Under the assumptions of Theorem~\ref{Lasso_error}, %if $\widetilde{s}\log p\leq c_0n$ for some positive constant $c_0$,
	there exists some universal positive constant $c$ such that $P\left(\mathcal{G}_n\cap\mathcal{H}_n\cap\mathcal{J}_n\cap\mathcal{K}_n\right)\geq 1-\exp[-c\log(p\vee n)]$, for all $n$ sufficiently large. 
	%all four events hold simultaneously with probability at least $1-\exp[-c\log(p\vee n)]$, for all $n$ sufficiently large. 
	\elem
	
	\blem \label{lem:Gbound}  
	Under the assumptions of Theorem 1, %there exists universal positive constant $C$ such that 
	for any $\vbeta\in\bbB$,
	\begin{align}
	G(\vx^T\vbeta|\vbeta) - G(\vx^T\vbeta_0|\vbeta_0)= &f_0'(\vx^T\vbeta)\big[\vx_{-1}^T\vgamma  -\emph{E}(\vx_{-1}^T\vgamma  |\vx^T\vbeta)\big] \nonumber \\
	&-  \Big\{ h(\vx_{-1}^T\vgamma) - \emph{E}\big[h(\vx_{-1}^T\vgamma) |\vx^T\vbeta\big] \Big\}, \label{G-1}\\ 
	G^{(1)}(\vx^T\vbeta|\vbeta) - G^{(1)}(\vx^T\vbeta_0|\vbeta_0)=&f''_0(\vx^T\vbeta) \big[ \vx_{-1}^T\vgamma-\emph{E}( \vx_{-1}^T\vgamma|\vx^T\vbeta)\big] -h_1(\vx_{-1}^T\vgamma)\nonumber\\
	&-   f_0'(\vx^T\vbeta)\emph{E}^{(1)}(\vx_{-1}^T\vgamma|\vx^T\vbeta) +\emph{E}^{(1)}\big[h(\vx_{-1}^T\vgamma) |\vx^T\vbeta\big],  \label{G1-1}
	%\\ \frac{\partial}{\partial\vbeta}G(t|\vbeta) \Big|_{\vbeta=\vbeta_0}=& -f_0'(t) \emph{E}(\vx  |\vx^T\vbeta_0=t) ,  \label{Gt}
	\end{align}
	where $\vgamma=\vbeta_{-1}-\vbeta_{0,-1}$,  $h(u) = \int_{0}^uaf_0''(a+\vx^T\vbeta_0)da$, $ h_1(u)=\int^{u}_{0} af'''_0(a+\vx^T\vbeta_0)da$, and $\emph{E}^{(1)}(\cdot|\vx^T\vbeta=t)$ is the first derivative of $\emph{E}(\cdot|\vx^T\vbeta=t)$ with respect to $t$. Moreover,  there exist universal positive constants $c_1$ and $c_2$ such that for   all $n$ sufficiently large,
	\begin{align}
	\max_{1\leq i \leq n}\sup_{\vbeta_1,\vbeta_2\in\bbB}\left[ G(\vx_i^T\vbeta_1|\vbeta_1)-G(\vx_i^T\vbeta_2|\vbeta_2) \right]^2\leq& c_1||\vbeta_1-\vbeta_2||_2s\log( p\vee n),\label{Gt}\\
	\max_{1\leq i \leq n}\sup_{\vbeta_1,\vbeta_2\in\bbB}\left[  G^{(1)}(\vx_i^T\vbeta_1|\vbeta_1)-G^{(1)}(\vx_i^T\vbeta_2|\vbeta_2) \right]^2\leq& c_1||\vbeta_1-\vbeta_2||_2s\log( p\vee n),\label{Gtt}\\
	%&\max_{1\leq i \leq n}\sup_{\vbeta\in\bbB_1}\left|\left|\frac{\partial}{\partial\vbeta}G(\vx_i^T\vbeta_0|\vbeta) \right|\right|_\infty\leq c_1\sqrt{\log( p\vee n)},\label{Gtt}\\
	n^{-1}\sum_{i=1}^{n}\big[G(\vx_i^T\vbeta|\vbeta) - G(\vx_i^T\vbeta_0|\vbeta_0)\big]^2\leq &c_1||\vbeta-\vbeta_0||_2^2,\ \forall\ \vbeta\in \bbB\label{G-2} 
	\end{align}
	with probability at least $1-\exp[-c_2\log(p\vee n)]$.  
	\elem

	\blem \label{Gfunc}
	%Under assumptions \ref{A1}--\ref{K4},  
	Under the assumptions of Theorem~\ref{Lasso_error}, %denote $\bbT = \{t \in\bbR : |t|\leq\sqrt{ \log (p\vee n)}\}$,  
	there exist universal positive constants $c_0$ and $c_1$  such that for all $n$ sufficiently large,
	\begin{align*}
	%P\Big(\sup \limits_{t \in \bbT, \vbeta \in \bbB}\big|\widehat{G}(t|\vbeta)-G(t|\vbeta)  \big| \geq c_0h^2\Big)\leq  \exp(-c_1nh^5). %\label{Gbound}  
	P\Big(\max_{ 1\leq i \leq n}\sup\limits_{\vbeta \in \bbB}\big|\widehat{G}(\vx_i^T\vbeta|\vbeta)-G(\vx_i^T\vbeta|\vbeta)  \big| \geq c_0h^2\Big)\leq  \exp[-c_1\log(p\vee n)].
	\end{align*} 
	\elem
	
	\blem \label{G1func}
	%Under assumptions \ref{A1}--\ref{K4},
	Under the assumptions of Theorem~\ref{Lasso_error}, there exist universal positive constants $c_0$ and $c_1$  such that for all $n$ sufficiently large,
	\begin{align*}
	%P\Big(\sup \limits_{t \in \bbT, \vbeta \in \bbB}\big|\widehat{G}^{(1)}(t|\vbeta)-G^{(1)}(t|\vbeta)  \big| \geq c_0h\Big)\leq \exp(-c_1nh^5). % \label{G1bound}.
	P\Big(\max_{ 1\leq i \leq n}\sup\limits_{\vbeta \in \bbB}\big|\widehat{G}^{(1)}(\vx_i^T\vbeta|\vbeta)-G^{(1)}(\vx_i^T\vbeta|\vbeta)  \big| \geq c_0h\Big)\leq \exp[-c_1\log(p\vee n)].
	\end{align*}  
	\elem
	
	\blem \label{lem:Ebound} 
	%Under assumptions \ref{A1} -- \ref{A4} and \ref{K1} -- \ref{K4},
	Under the assumptions of Theorem 1, there exist universal positive constants $c_0$, $c_1$ such that 
	for all $n$ sufficiently large,
	\begin{align*}
	%P\Big(\sup \limits_{t \in \bbT,  \vbeta \in \bbB} \Big|\Big[\widehat{\E}(\vx|\vx^T\vbeta=t)-\E(\vx|\vx^T\vbeta=t)\Big]^T\veta\Big| \geq c_0h^2||\veta||_2\Big)\leq \exp(-c_1nh^5),\\
	%P\Big(\sup \limits_{t \in \bbT,  \vbeta \in \bbB} \Big|\Big|\widehat{\E}(\vx|\vx^T\vbeta=t)-\E(\vx|\vx^T\vbeta=t)\Big|\Big|_\infty \geq c_0h^2\Big)\leq \exp(-c_1nh^5),
	P\Big(\max_{ 1\leq i \leq n}\sup\limits_{\substack{\vbeta \in \bbB\\\vv\in\bbK(p,2ks)}} \Big|\Big[\widehat{\emph{E}}(\vx_{i}|\vx_i^T\vbeta)-\emph{E}(\vx_{i}|\vx_i^T\vbeta)\Big]^T\vv\Big| \geq c_0h^2\Big)\leq \exp[-c_1\log(p\vee n)],\\
	P\Big(\max_{ 1\leq i \leq n}\sup\limits_{\vbeta \in \bbB} \Big|\Big|\widehat{\emph{E}}(\vx_{i}|\vx_i^T\vbeta)-\emph{E}(\vx_{i}|\vx_i^T\vbeta)\Big|\Big|_\infty \geq c_0h^2\Big)\leq \exp[-c_1\log(p\vee n)].
	\end{align*}  
	Furthermore, if $\widetilde{s}\log p\leq d_0n$ for some positive constant $d_0$, where $s_j=||\vd_{0j}||_0$, $\widetilde{s}=\max_{ 2\leq j\leq p}s_j$, then  there exist universal positive constants $d_1$, $d_2$ such that for all $n$ sufficiently large,
	$$	P\Big(\max_{ 1\leq i \leq n}\sup\limits_{\substack{\vbeta \in \bbB\\\vv\in\bbK(p,2ks+ \widetilde{s})}} \Big|\Big[\widehat{\emph{E}}(\vx_{i}|\vx_i^T\vbeta)-\emph{E}(\vx_{i}|\vx_i^T\vbeta)\Big]^T\vv\Big| \geq d_1h^2\Big)\leq \exp[-d_2\log(p\vee n)].$$ 
	\elem
	
	%\blem \label{lem:xbeta_bound}  
	%Let $\bbB_1 = \big\{\vbeta\in\bbB:||\vbeta-\vbeta_0||_2\leq c_0 \sqrt{s}h^2 \big\}, ||\vbeta||_0\leq ks\big\}$ for positive constants $c_0$ and $k$, with $ks=O(n)$.  
	%\elem

	\blem \label{lem:ghat1_g1} 
	%Under assumptions \ref{A1}--\ref{K4},
	%Let $\gamma_1 = \sqrt{h\log (p\vee n)} + h \big[\log (p\vee n) \big]^{3/2}$.
	Under the assumptions of Theorem~\ref{desparse}, there exist universal positive constants $c_0$, $c_1$, such that for all $n$ sufficiently large,
	\begin{align}
	P\left(\max_{2\leq j \leq p}   \left|n^{-1/2}\sum_{i=1}^n\vtheta_j^T\vgamma(Z_i,\vbetah,\widehat{G}^{(1)})  \right| \geq c_0\left[h^2\log(p\vee n)\right]^{1/4}\right)\leq\exp(-c_1\log p),\label{ghat1_g1} \\
	P\left(  \left|\left|n^{-1/2}\sum_{i=1}^n \vgamma(Z_i,\vbetah,\widehat{G}^{(1)}) \right|\right|_\infty\geq c_0\left[h^2\log(p\vee n)\right]^{1/4}\right)\leq\exp(-c_1\log p),\label{ghat1_g1_inf} 
	\end{align} 
	where $Z_i=(\vx_i,\epsilon_i,A_i)$, $\vgamma(Z_i,\vbetah,\widehat{G}^{(1)}) = \big[\widehat{G}^{(1)}(\vx_i^T\vbeta|\vbeta) - G^{(1)}(\vx_i^T\vbeta_0|\vbeta_0)\big] \widetilde{\epsilon}_i \big[\vx_{i,-1} -\emph{E}(\vx_{i,-1}|\vx_i^T\vbeta_0)\big]$, with $\widetilde{\epsilon}_i=2(2A_i-1)[\ep_i+g(\vx_i)]$.
	\elem
	
	%\begin{align}
	%P\Big(\max_{2\leq j \leq p} \sup_{\vbeta\in\bbB_1,m\in\bbM} \big|n^{-1/2}\sum_{i=1}^n\vtheta_j^T\vgamma(Z_i,\vbeta,m) \big| >\gamma\Big)\leq\exp-[c\log (p\vee n)],\label{ghat1_g1} \\
	%P\Big( \sup_{\vbeta\in\bbB_1,m\in\bbM}\big|\big|n^{-1/2}\sum_{i=1}^n \vgamma(Z_i,\vbeta,m) %\big|\big|_\infty>\gamma\Big)\leq\exp-[c\log (p\vee n)],\label{ghat1_g1_inf} 
	%\end{align} 
	%where $Z_i=(\vx_i,\epsilon_i,A_i)$, $\vgamma(Z_i,\vbeta,m) = \big[m(\vx_i^T\vbeta|\vbeta) - G^{(1)}(\vx_i^T\vbeta_0|\vbeta_0)\big] \widetilde{\epsilon}_i \big[\vx_i -\emph{E}(\vx_i|\vx_i^T\vbeta_0)\big] $,  %$||m||_\infty = \sup_{\vbeta\in\bbB_1}||m(\cdot|\vbeta)||_\infty$, 
	%$m(\vx^T\vbeta|\vbeta)$ depends on $\vx$ only through $\vx^T\vbeta$, and 
	%\begin{align*}
	%\bbM = &\Big\{m(\cdot|\vbeta): \vbeta\in\bbB_1,\ m(\cdot|\vbeta) \mbox{ is Lipschitz } \forall\ \vbeta,\\
	%&\max_{  1\leq i\leq n}\sup_{\vbeta\in\bbB_1}||m(\vx_i^T\vbeta|\vbeta)-G^{(1)}(\vx_i^T\vbeta_0|\vbeta_0 )||_\infty\leq c_1h\Big\}.
	%\end{align*}
	% \sup_{\vbeta\in\bbB_1}||m(\cdot|\vbeta)-G^{(1)}(\cdot|\vbeta_0 )||_\infty\leq c_1h\Big\}.$$ 

	\blem \label{lem:ghat_g}  
	Under the assumptions of Theorem~\ref{desparse}, there exist  universal positive constants $c_0$, $c_1$ such that for all $n$ sufficiently large,
	\begin{align}
	P\left(\max_{2\leq j \leq p}  \left|n^{-1/2}\sum_{i=1}^n \vtheta_j^T\vnu_1(Z_i,\vbetah,\widehat{G}^{(1)}) \right|\geq c_0sh^3\sqrt{n} \right)\leq\exp(-c_1\log p),\label{ghat_g1} \\
	P\left( \left|\left|n^{-1/2}\sum_{i=1}^n \vnu_1(Z_i,\vbetah,\widehat{G}^{(1)}) \right|\right|_\infty\geq c_0sh^3\sqrt{n} \right)\leq\exp(-c_1\log p),\label{ghat_g1_inf}\\
	P\left(\max_{2\leq j \leq p}  \left|n^{-1/2}\sum_{i=1}^n \vtheta_j^T\vnu_2(Z_i,\vbetah,\widehat{G},\widehat{G}^{(1)}) \right|\geq c_0sh^3\sqrt{n} \right)\leq\exp(-c_1\log p),\label{ghat_g2} \\
	P\left( \left|\left|n^{-1/2}\sum_{i=1}^n \vnu_2(Z_i,\vbetah,\widehat{G},\widehat{G}^{(1)}) \right|\right|_\infty\geq c_0sh^3\sqrt{n} \right)\leq\exp(-c_1\log p),\label{ghat_g2_inf}
	\end{align} 
	where $Z_i=(\vx_i,\epsilon_i,A_i)$, 
	$\vnu_1(Z_i,\vbetah,\widehat{G}^{(1)}) = \big[G(\vx_{i}^{T} \vbeta_0|\vbeta_0)-G(\vx_{i}^{T} \vbetah|\vbetah) - G^{(1)}(\vx_i^T\vbetah|\vbetah) \vxh_{i,-1}^T(\vbeta_{0,-1}-\vbetah_{-1})\big] \widehat{G}^{(1)}(\vx_i^T\vbetah|\vbetah)\vxh_{i,-1} $, $\vnu_2(Z_i,\vbetah,\widehat{G},\widehat{G}^{(1)}) = \big[G(\vx_{i}^{T} \vbetah|\vbetah) -\widehat{G}(\vx_{i}^{T} \vbetah|\vbetah) \big]\widehat{G}^{(1)}(\vx_i^T\vbetah|\vbetah)\vxh_{i,-1} $, with $\vxh_{i,-1} = \vx_{i,-1} -\widehat{\emph{E}}(\vx_{i,-1}|\vx_i^T\vbetah)$.
	\elem

	\blem \label{lem:Ehat_E}  
	Under the assumptions of Theorem~\ref{Lasso_error}, there exist universal positive constants $c_0$ and $c_1$ such that for all $n$ sufficiently large,
	\begin{align}
	P\left(\max_{2\leq j \leq p} \left|n^{-1/2}\sum_{i=1}^n \vtheta_j^T\vxi(Z_i,\vbetah,\widehat{\emph{E}}) \right| \geq  c_0h\sqrt{s\log(p\vee n)}\right)\leq\exp(-c_1\log p)\label{Ehat_E},\\
	P\left( \left|\left|n^{-1/2}\sum_{i=1}^n\vxi(Z_i,\vbeta,\widehat{\emph{E}}) \right|\right|_\infty\geq c_0h\sqrt{s\log(p\vee n)}\right)\leq\exp(-c_1\log p)\label{Ehat_E_inf},
	\end{align} 
	where $Z_i=(\vx_i,\epsilon_i,A_i)$, 
	$\vxi(Z_i,\vbetah,\widehat{\emph{E}}) =  \widetilde{\epsilon}_i G^{(1)}(\vx_i^T\vbeta_0|\vbeta_0) \left[\widehat{\emph{E}}(\vx_{i,-1}|\vx_i^T\vbetah)-\emph{E}(\vx_{i,-1}|\vx_i^T\vbeta_0)\right] $, with $\widetilde{\epsilon}_i=2(2A_i-1)[\ep_i+g(\vx_i)]$.
	\elem
	%	Under the assumptions of Theorem~\ref{Lasso_error}, there exist universal positive constants $c_0$ and $c_1$ such that for all $n$ sufficiently large,
	%	\begin{align}
	%	P\Big(\max_{2\leq j \leq p}\sup_{\vbeta\in\bbB_1,\vE\in\bbE} \big|n^{-1/2}\sum_{i=1}^n \vtheta_j^T\vxi(Z_i,\vbeta,\vE) \big| > c_0\sqrt{h}\Big)\leq\exp[-c_1\log (p\vee n)]\label{Ehat_E},\\
	%	P\Big(\sup_{\vbeta\in\bbB_1,\vE\in\bbE}\big|\big|n^{-1/2}\sum_{i=1}^n\vxi(Z_i,\vbeta,\vE) \big|\big|_\infty> c_0\sqrt{h}\Big)\leq\exp\big[-c_1\log (p\vee n)]\label{Ehat_E_inf},
	%	\end{align} 
	%	where 
	%	$\vxi(Z_i,\vbeta,\vE) =  \widetilde{\epsilon}_i G^{(1)}(\vx_i^T\vbeta_0|\vbeta_0) [\vE(\vx_i^T\vbeta|\vbeta) -\emph{E}(\vx_i|\vx_i^T\vbeta_0)\big] $,    $\vE(\vx^T\vbeta|\vbeta)$ depends on $\vx$ only through $\vx^T\vbeta$,  and
	%	\begin{align*}\bbE  = \Big\{\vE(\cdot|\vbeta):\vbeta\in\bbB_1;  \mbox{ for any fixed }\veta\in\bbR^p, \vE(\cdot|\vbeta)^T\veta \mbox{ is Lipschitz }  \forall\ \vbeta, \\ \max_{1\leq i \leq n}\sup_{\vbeta\in\bbB_1}\big|\big|[\vE(\vx_i^T\vbeta|\vbeta)-\emph{E}(\vx_i |\vx_i^T\vbeta_0)]^T\veta\big|\big|_\infty\leq c_1 |\veta||_2sh^2\sqrt{\log(p\vee n)},
	%	%\sup_{\vbeta\in\bbB_1}\big|\big|[\vE(\cdot|\vbeta)-\E(\vx |\vx^T\vbeta_0=\cdot)]^T\veta\big|\big|_\infty\leq c_1 |\veta||_2sh^2\sqrt{\log(p\vee n)},
	%	\Big\}.
	%	\end{align*} 
	%	

	\blem \label{Gbetafunc}
	%Under the assumptions \ref{A1}--\ref{K4},
	Under the assumptions of Lemma~\ref{dbound}, %for any sub-Gaussian random vector $\vx$, 
	there exist universal positive constants $c_0$ and $c_1$ such that for all $n$ sufficiently large,%, and $\mathcal{X} = \{\vx\in\bbR^p:||\vx||_\infty\leq \sqrt{\log (p\vee n)}\}$, 
	\begin{align}
	&P\Big(\max_{ 1\leq i \leq n}\sup \limits_{ \vbeta\in\bbB_1}\big|\widehat{G} (\vx_i^T\vbeta|\vbeta)-G (\vx_i^T\vbeta_0|\vbeta_0)  \big| \geq c_0sh^2\sqrt{\log (p\vee n)} \Big)\leq  \exp(-c_1\log p) \label{Gbetabound},\\
	&P\Big(\max_{ 1\leq i \leq n}\sup \limits_{\vbeta\in\bbB_1}\big|\widehat{G}^{(1)} (\vx_i^T\vbeta|\vbeta)-G^{(1)} (\vx_i^T\vbeta_0|\vbeta_0)  \big| \geq c_0h \Big) \leq \exp(-c_1\log p) \label{G1betabound},\\
	&P\Big(\max_{ 1\leq i \leq n}\sup \limits_{\substack{\vbeta \in \bbB\\\vv\in\bbK(p,2ks+\widetilde{s})}}\big| [\widehat{\emph{E}} (\vx_{i}|\vx_i^T\vbeta)-\emph{E}(\vx_{i}|\vx_i^T\vbeta_0)]^T\vv \big| \geq c_0 sh^2\sqrt{\log (p\vee n)} \Big)\leq  \exp(-c_1\log p) \label{Ebetabound}.
	\end{align}  
	%where $s_j=||\vd_{0j}||_0$, $\widetilde{s}=\max_{ 2\leq j\leq p}s_j$.
	\elem
	
	\blem \label{Efunc} 
	Under the assumptions of Lemma~\ref{dbound}, there exist universal positive constants $c_0$, $c_1$, %and $\mathcal{X} = \{\vx\in\bbR^p:||\vx||_\infty\leq \sqrt{\log (p\vee n)}\}$, 
	such that  all $n$ sufficiently large,
	\begin{align}
	P\Big(\sup \limits_{\substack{\vbeta \in \bbB\\\vv\in\bbK(p,2ks+\widetilde{s})}}\Big|\frac{1}{n}\sum_{i=1}^n\vv^T\big[\widehat{\emph{E}} (\vx_{i}|\vx_i^T\vbeta)-\emph{E}(\vx_{i}|\vx_i^T\vbeta_0)\big]& \big[\widehat{\emph{E}} (\vx_{i}|\vx_i^T\vbeta)-\emph{E}(\vx_{i}|\vx_i^T\vbeta_0)\big]^T\vv \Big|\nonumber\\
	\geq& c_0 sh^4 \Big)\leq  \exp(-c_1\log p) \label{Eboundphi},\\
	P\Big(\sup \limits_{\vbeta\in\bbB_1}\Big|\Big|\frac{1}{n}\sum_{i=1}^n\big[\widehat{\emph{E}} (\vx_{i}|\vx_i^T\vbeta)-\emph{E}(\vx_{i}|\vx_i^T\vbeta_0)\big]& \big[\widehat{\emph{E}} (\vx_{i}|\vx_i^T\vbeta)-\emph{E}(\vx_{i}|\vx_i^T\vbeta_0)\big] ^T\Big|\Big|_\infty \nonumber\\
	\geq& c_0 sh^4 \Big)\leq  \exp(-c_1\log p) \label{Ebound}.
	\end{align}   
	%where $s_j=||\vd_{0j}||_0$, $\widetilde{s}=\max_{ 2\leq j\leq p}s_j$.
	\elem

	Lemma~\ref{lem:thetaj} below gives an alternative expression for $\vphi_{0j}$.
	\blem \label{lem:thetaj}  
	%Denote  $\vphi_{0j} = \tau^{2}_{0j}\vtheta_j$. 
	Under assumption \ref{A2}, we have
	\begin{align*}
	\vphi_{0j} = \Big(-(\vd_{0j})_{1:(j-2)}^T,1,-(\vd_{0j})_{(j-1):(p-1)}^T\Big)^T,
	\end{align*}
	and $ ||\vtheta_j||_0\leq \widetilde{s}+1$ uniformly in $j=2, \ldots, p$. 
	%where $\vd_{0j} = (\vOmega_{-j,-j})^{-1}\vOmega_{-j,j}$, and $\tau^2_{0j} =\vOmega_{j,j} - \vd_{0j}^T\vOmega_{-j,j}$. 
	Under Assumption \ref{A1} and Assumption \ref{A2},   we have  
	%(1) $\inf_{||\vv||_2=1}\vv^T\vOmega\vv\geq a^2\xi_0$;\qquad (2) 
	$\tau^{-2}_{0j}\leq ||\vtheta_j||_2 \leq  \xi_2^{-1} $ and $\tau^{2}_{0j}\leq b^2\xi_1$ uniformly in $j=2, \ldots, p$. 
	
	\elem
	
	\blem 
	\label{cor_Sig}
	Under the assumptions of Theorem~\ref{desparse}, $||\vSigmah(\vbetah) - \vTheta^T\vLam\vTheta||_\infty  = o_p(1)$.
	\elem

	\section{Proofs of results in Section~\ref{sec:est_theorem} of the main paper} \label{sec:proof3.1}
	\begin{proof}[Proof of Lemma~\ref{lem:local_LRC}] 
		Note that
		\begin{align*} 
		\vS_n(\vbeta,\widehat{G},\widehat{\E}) =&- \frac{1}{n}\sum_{i=1}^n \big[\widetilde{Y}_i  - \widehat{G}(\vx_i^T\vbeta|\vbeta)\big]\widehat{G}^{(1)}(\vx_i^T\vbeta|\vbeta)\big[\vx_{i,-1} -\widehat{\E}(\vx_{i,-1}|\vx_i^T\vbeta)\big]\\
		=&- \frac{1}{n}\sum_{i=1}^n \widetilde{\epsilon}_i\widehat{G}^{(1)}(\vx_i^T\vbeta|\vbeta)\big[\vx_{i,-1} -\widehat{\E}(\vx_{i,-1}|\vx_i^T\vbeta)\big]\\
		&- \frac{1}{n}\sum_{i=1}^n \big[G(\vx_i^T\vbeta_0|\vbeta_0)  - \widehat{G}(\vx_i^T\vbeta|\vbeta)\big]\widehat{G}^{(1)}(\vx_i^T\vbeta|\vbeta)\big[\vx_{i,-1} -\widehat{\E}(\vx_{i,-1}|\vx_i^T\vbeta)\big],
		\end{align*} 
		where  $\widetilde{\epsilon}_i  =  2(2A_i-1)[\epsilon_i+g(\vx_i)]$.
		Denote $\vgamma=\vbeta_{-1}-\vbeta_{0,-1}$. We have
		\begin{align*}
		&\langle \vS_n(\vbeta,\widehat{G},\widehat{\E}) - \vS_n(\vbeta_0,\widehat{G},\widehat{\E}),\vgamma \rangle \\ 
		=&- \frac{1}{n}\sum_{i=1}^n \widetilde{\epsilon}_i\widehat{G}^{(1)}(\vx_i^T\vbeta|\vbeta)\big[\widehat{\E}(\vx_{i,-1}|\vx_i^T\vbeta_0) -\widehat{\E}(\vx_{i,-1}|\vx_i^T\vbeta)\big]^T\vgamma\\
		&- \frac{1}{n}\sum_{i=1}^n \widetilde{\epsilon}_i\big[\widehat{G}^{(1)}(\vx_i^T\vbeta|\vbeta)-\widehat{G}^{(1)}(\vx_i^T\vbeta_0|\vbeta_0)\big]\big[\vx_{i,-1} -\widehat{\E}(\vx_{i,-1}|\vx_i^T\vbeta_0)\big]^T\vgamma\\
		&- \frac{1}{n}\sum_{i=1}^n \big[G(\vx_i^T\vbeta_0|\vbeta_0)  - \widehat{G}(\vx_i^T\vbeta|\vbeta)\big]\widehat{G}^{(1)}(\vx_i^T\vbeta|\vbeta)\big[\widehat{\E}(\vx_{i,-1}|\vx_i^T\vbeta_0) -\widehat{\E}(\vx_{i,-1}|\vx_i^T\vbeta)\big]^T\vgamma\\
		& - \frac{1}{n}\sum_{i=1}^n \big[G(\vx_i^T\vbeta_0|\vbeta_0)  - \widehat{G}(\vx_i^T\vbeta|\vbeta)\big]\big[\widehat{G}^{(1)}(\vx_i^T\vbeta|\vbeta)-\widehat{G}^{(1)}(\vx_i^T\vbeta_0|\vbeta_0)\big]\big[\vx_{i,-1} -\widehat{\E}(\vx_{i,-1}|\vx_i^T\vbeta_0)\big]^T\vgamma\\
		& - \frac{1}{n}\sum_{i=1}^n \big[\widehat{G}(\vx_i^T\vbeta_0|\vbeta_0)  - \widehat{G}(\vx_i^T\vbeta|\vbeta)\big]\widehat{G}^{(1)}(\vx_i^T\vbeta_0|\vbeta_0)\big[\vx_{i,-1} -\widehat{\E}(\vx_{i,-1}|\vx_i^T\vbeta_0)\big]^T\vgamma\\
		\triangleq&-\sum_{k=1}^5 A_k(\vbeta),
		\end{align*} 
		where the definition of $A_k$'s, $k=1,\cdots,5$, is clear from the context. Each $A_k$ can be further decomposed.
		\bqa
		&&A_1 (\vbeta)\\
		&=&  \frac{1}{n}\sum_{i=1}^n \widetilde{\epsilon}_iG^{(1)}(\vx_i^T\vbeta|\vbeta)\big[\E(\vx_{i,-1}|\vx_i^T\vbeta_0) - \E(\vx_{i,-1}|\vx_i^T\vbeta)\big]^T\vgamma  \\
		& &+ \frac{1}{n}\sum_{i=1}^n \widetilde{\epsilon}_iG^{(1)}(\vx_i^T\vbeta|\vbeta)\big[\widehat{\E}(\vx_{i,-1}|\vx_i^T\vbeta_0) -\E(\vx_{i,-1}|\vx_i^T\vbeta_0)-\widehat{\E}(\vx_{i,-1}|\vx_i^T\vbeta_0)+\E(\vx_{i,-1}|\vx_i^T\vbeta)\big]^T\vgamma\\
		&& + \frac{1}{n}\sum_{i=1}^n \widetilde{\epsilon}_i\big[\widehat{G}^{(1)}(\vx_i^T\vbeta|\vbeta)-G^{(1)}(\vx_i^T\vbeta|\vbeta)\big]\big[\widehat{\E}(\vx_{i,-1}|\vx_i^T\vbeta_0) -\widehat{\E}(\vx_{i,-1}|\vx_i^T\vbeta)\big]^T\vgamma\\
		&\triangleq&\sum_{l=1}^3 A_{1l}(\vbeta).
		\eqa 
		\bqa
		&& A_2(\vbeta)\\
		& =&  \frac{1}{n}\sum_{i=1}^n \widetilde{\epsilon}_i\big[G^{(1)}(\vx_i^T\vbeta|\vbeta)-G^{(1)}(\vx_i^T\vbeta_0|\vbeta_0)\big]\big[\vx_{i,-1} -\E(\vx_{i,-1}|\vx_i^T\vbeta_0)\big]^T\vgamma\\
		&& + \frac{1}{n}\sum_{i=1}^n \widetilde{\epsilon}_i\big[G^{(1)}(\vx_i^T\vbeta|\vbeta)-G^{(1)}(\vx_i^T\vbeta_0|\vbeta_0)\big]\big[\E(\vx_{i,-1}|\vx_i^T\vbeta_0) -\widehat{\E}(\vx_{i,-1}|\vx_i^T\vbeta_0)\big]^T\vgamma\\
		& &+ \frac{1}{n}\sum_{i=1}^n \widetilde{\epsilon}_i\big[\widehat{G}^{(1)}(\vx_i^T\vbeta|\vbeta)-G^{(1)}(\vx_i^T\vbeta|\vbeta)-\widehat{G}^{(1)}(\vx_i^T\vbeta_0|\vbeta_0)+G^{(1)}(\vx_i^T\vbeta_0|\vbeta_0)\big]\\
		&&\qquad\qquad*\big[\vx_{i,-1} -\widehat{\E}(\vx_{i,-1}|\vx_i^T\vbeta_0)\big]^T\vgamma\\
		&\triangleq&\sum_{l=1}^3 A_{2l}(\vbeta),\\\ \\ \  \\
		&&A_3(\vbeta)\\
		&=&   \frac{1}{n}\sum_{i=1}^n \big[G(\vx_i^T\vbeta_0|\vbeta_0)  - G(\vx_i^T\vbeta|\vbeta)\big]G^{(1)}(\vx_i^T\vbeta|\vbeta)\big[\E(\vx_{i,-1}|\vx_i^T\vbeta_0) - \E(\vx_{i,-1}|\vx_i^T\vbeta)\big]^T\vgamma\\
		&&+\frac{1}{n}\sum_{i=1}^n\big[G(\vx_i^T\vbeta_0|\vbeta_0)  - G(\vx_i^T\vbeta|\vbeta)\big]G^{(1)}(\vx_i^T\vbeta|\vbeta)\\
		&&\qquad\quad*\big[\widehat{\E}(\vx_{i,-1}|\vx_i^T\vbeta_0)-\E(\vx_{i,-1}|\vx_i^T\vbeta_0)-\widehat{\E}(\vx_{i,-1}|\vx_i^T\vbeta_0)+\E(\vx_{i,-1}|\vx_i^T\vbeta)\big]^T\vgamma\\
		&&+ \frac{1}{n} \sum_{i=1}^n \big[G(\vx_i^T\vbeta_0|\vbeta_0)  -G(\vx_i^T\vbeta|\vbeta)\big] \big[\widehat{G}^{(1)}(\vx_i^T\vbeta|\vbeta) -G^{(1)}(\vx_i^T\vbeta|\vbeta)\big]\\
		&&\qquad\qquad*\big[\widehat{\E}(\vx_{i,-1}|\vx_i^T\vbeta_0) -\widehat{\E}(\vx_{i,-1}|\vx_i^T\vbeta)\big]^T\vgamma\\
		&& +  \frac{1}{n}\sum_{i=1}^n \big[G(\vx_i^T\vbeta|\vbeta)  - \widehat{G}(\vx_i^T\vbeta|\vbeta)\big]\widehat{G}^{(1)}(\vx_i^T\vbeta|\vbeta)\big[\widehat{\E}(\vx_{i,-1}|\vx_i^T\vbeta_0) -\widehat{\E}(\vx_{i,-1}|\vx_i^T\vbeta)\big]^T\vgamma\\
		&\triangleq&\sum_{l=1}^4 A_{3l}(\vbeta).
		\eqa
		\bqa
		&&A_4(\vbeta)+ A_5(\vbeta)\\
		&=& \frac{1}{n}\sum_{i=1}^n \big[G(\vx_i^T\vbeta_0|\vbeta_0)  - G(\vx_i^T\vbeta|\vbeta)\big]G^{(1)}(\vx_i^T\vbeta|\vbeta)\big[\vx_{i,-1} - \E(\vx_{i,-1}|\vx_i^T\vbeta_0)\big]^T\vgamma\\ 
		&&+\frac{1}{n}\sum_{i=1}^n \big[G(\vx_i^T\vbeta_0|\vbeta_0)  - G(\vx_i^T\vbeta|\vbeta)\big]G^{(1)}(\vx_i^T\vbeta|\vbeta)\big[\E(\vx_{i,-1}|\vx_i^T\vbeta_0) -\widehat{\E}(\vx_{i,-1}|\vx_i^T\vbeta_0)\big]^T\vgamma\\
		&& + \frac{1}{n}\sum_{i=1}^n \big[G(\vx_i^T\vbeta_0|\vbeta_0)  - G(\vx_i^T\vbeta|\vbeta)\big]\big[\widehat{G}^{(1)}(\vx_i^T\vbeta|\vbeta)-G^{(1)}(\vx_i^T\vbeta|\vbeta)][\vx_{i,-1} -\widehat{\E}(\vx_{i,-1}|\vx_i^T\vbeta_0)\big]^T\vgamma\\	
		&& +  \frac{1}{n}\sum_{i=1}^n \big[G(\vx_i^T\vbeta|\vbeta)  - \widehat{G}(\vx_i^T\vbeta|\vbeta)\big]\big[\widehat{G}^{(1)}(\vx_i^T\vbeta|\vbeta)-\widehat{G}^{(1)}(\vx_i^T\vbeta_0|\vbeta_0)\big]\big[\vx_{i,-1} -\widehat{\E}(\vx_{i,-1}|\vx_i^T\vbeta_0)\big]^T\vgamma\\
		&& +  \frac{1}{n}\sum_{i=1}^n \big[\widehat{G}(\vx_i^T\vbeta_0|\vbeta_0)  -G(\vx_i^T\vbeta_0|\vbeta_0) - \widehat{G}(\vx_i^T\vbeta|\vbeta) +G(\vx_i^T\vbeta|\vbeta)\big]\\
		&&\qquad\qquad*\widehat{G}^{(1)}(\vx_i^T\vbeta_0|\vbeta_0)\big[\vx_{i,-1} -\widehat{\E}(\vx_{i,-1}|\vx_i^T\vbeta_0)\big]^T\vgamma\\
		&\triangleq&\sum_{l=1}^5 B_{l}(\vbeta).
		\eqa 
		
		The proof involves evaluating the order of $A_k$ ($k=1,2,3$) and $A_4+A_5$.
		We provide the details of analyzing $B_1(\vbeta)$ and $B_2(\vbeta)$ which are two of the most challenging terms to study. All the other terms can be handled similarly. First, for the function $h(\cdot)$ defined in Lemma~\ref{lem:Gbound},
		we have $h(\vx_{-1}^T\vgamma) = \int_0^{\vx_{-1}^T\vgamma}af''_0(a+\vx^T\vbeta_0)da  = \frac{1}{2}f''_0(\vx^T\vbeta_1)(\vx_{-1}^T\vgamma)^2$ for some $\vbeta_1$ between $\vx^T\vbeta_0$ and $\vx^T\vbeta=\vx^T\vbeta_0+\vx_{-1}^T\vgamma $.
		By (\ref{G-1}) in Lemma~\ref{lem:Gbound}, we can write $-B_1=\sum_{q=1}^6B_{1q}$, where 
		%let us bound $G(\vx_i^T\vbeta_0|\vbeta_0)  - G(\vx_i^T\vbeta|\vbeta)$.
		%	By Taylor expansion,
		%	\begin{align*}
		%	&G(\vx_{i}^T \vbeta|\vbeta) -G(\vx_{i}^T \vbeta_0|\vbeta_0) = \E\big[f_0(\vx_i^T\vbeta_0)|\vx_i^T\vbeta\big] - f_0(\vx_i^T\vbeta_0)\\
		%	=& \E\big[f_0(\vx_i^T\vbeta) - f_0'(\vx_i^T\vbeta)\vx_i^T\vgamma+ \frac{1}{2} f''_0(\vx_i^T\vbeta_1) (\vx_i^T\vgamma ) ^2 |\vx_i^T\vbeta\big] - f_0(\vx_i^T\vbeta_0)\\
		%	=&f_0(\vx_i^T\vbeta) - f_0(\vx_i^T\vbeta_0)- f_0'(\vx_i^T\vbeta)\E(\vx_i|\vx_i^T\vbeta)^T\vgamma  + \frac{1}{2}\E\big[  f''_0(\vx_i^T\vbeta_1) (\vx_i^T\vgamma)^2 |\vx_i^T\vbeta\big]   \\
		%	=& f_0'(\vx_i^T\vbeta) \vx_i^T\vgamma-  \frac{1}{2}  f''_0(\vx_i^T\vbeta_2) (\vx_i^T\vgamma)^2 - f_0'(\vx_i^T\vbeta)\E(\vx_i|\vx_i^T\vbeta)^T\vgamma  + \frac{1}{2}\E\big[  f''_0(\vx_i^T\vbeta_1) (\vx_i^T\vgamma)^2 |\vx_i^T\vbeta\big]  \\
		%	=&f_0'(\vx_i^T\vbeta)\big[ \vx_i-\E(\vx_i|\vx_i^T\vbeta)\big]^T\vgamma- \frac{1}{2}  f''_0(\vx_i^T\vbeta_2) (\vx_i^T\vgamma)^2+ \frac{1}{2}\E\big[  f''_0(\vx_i^T\vbeta_1) (\vx_i^T\vgamma)^2 |\vx_i^T\vbeta\big],
		%	\end{align*}
		%	for some $\vbeta_1$ and $\vbeta_2$  between $\vbeta$ and $\vbeta_0$. Note that $G(\vx^T\vbeta_0|\vbeta_0) = f(\vx^T\vbeta_0)$, and $G^{(1)}(\vx^T\vbeta_0|\vbeta_0) = f'(\vx^T\vbeta_0)$.
		%We therefore write $-B_1=  \sum_{q=1}^6B_{1q}$, where 
		\begin{align*}
		B_{11}&=n^{-1}\sum_{i=1}^n \big[f_0'(\vx_i^T\vbeta_0)\big]^2\vgamma^T\big[ \vx_i-\E(\vx_i|\vx_i^T\vbeta_0)\big]\big[ \vx_{i,-1}-\E(\vx_{i,-1}|\vx_i^T\vbeta_0)\big]^T\vgamma ,\\
		B_{12}&=n^{-1}\sum_{i=1}^nf_0'(\vx_i^T\vbeta_0) [f_0'(\vx_i^T\vbeta)-f_0'(\vx_i^T\vbeta_0)]\vgamma^T\big[ \vx_{i,-1}-\E(\vx_{i,-1}|\vx_i^T\vbeta_0)\big]\big[ \vx_{i,-1}-\E(\vx_{i,-1}|\vx_i^T\vbeta_0)\big]^T\vgamma ,\\
		B_{13}&=n^{-1}\sum_{i=1}^n[G^{(1)}(\vx_{i}^T \vbeta|\vbeta)-G^{(1)}(\vx_{i}^T \vbeta_0|\vbeta_0)] f_0'(\vx_i^T\vbeta) \Big\{\big[ \vx_{i,-1}-\E(\vx_{i,-1}|\vx_i^T\vbeta_0)\big]^T\vgamma\Big\}^2 ,\\
		B_{14}&=n^{-1}\sum_{i=1}^nG^{(1)}(\vx_{i}^T \vbeta|\vbeta) f_0'(\vx_i^T\vbeta)\vgamma^T\big[\E(\vx_{i,-1}|\vx_i^T\vbeta_0)-\E(\vx_{i,-1}|\vx_i^T\vbeta)\big] \big[ \vx_{i,-1}-\E(\vx_i|\vx_{i,-1}^T\vbeta_0)\big]^T\vgamma ,\\ 
		B_{15}&=(2n)^{-1}\sum_{i=1}^nG^{(1)}(\vx_{i}^T \vbeta|\vbeta) f''_0(\vx_i^T\vbeta_2) (\vx_{i,-1}^T\vgamma)^2 \big[ \vx_{i,-1}-\E(\vx_{i,-1}|\vx_i^T\vbeta_0)\big]^T\vgamma ,\\
		B_{16}&=(2n)^{-1}\sum_{i=1}^nG^{(1)}(\vx_{i}^T \vbeta|\vbeta) \E\big[  f''_0(\vx_i^T\vbeta_1) (\vx_{i,-1}^T\vgamma)^2 |\vx_i^T\vbeta\big]  \big[ \vx_{i,-1}-\E(\vx_{i,-1}|\vx_i^T\vbeta_0)\big]^T\vgamma,
		\end{align*}
		for some $\vbeta_1$ and $\vbeta_2$  between $\vbeta$ and $\vbeta_0$. Note that $G(\vx^T\vbeta_0|\vbeta_0) = f(\vx^T\vbeta_0)$, and $G^{(1)}(\vx^T\vbeta_0|\vbeta_0) = f'(\vx^T\vbeta_0)$.
		
		Observe that 
		\begin{align*}
		&B_{11} \\
		=&n^{-1}\sum_{i=1}^n \big[f_0'(\vx_i^T\vbeta_0)\big]^2\vgamma^T\E\big[\Cov (\vx_{-1}|\vx^T\vbeta_0)\big]\vgamma \\
		&+n^{-1}\sum_{i=1}^n \big[f_0'(\vx_i^T\vbeta_0)\big]^2\vgamma^T \big\{\big[ \vx_{i,-1}-\E(\vx_{i,-1}|\vx_i^T\vbeta_0)\big]\big[ \vx_{i,-1}-\E(\vx_{i,-1}|\vx_i^T\vbeta_0)\big]^T - \E\big[\Cov (\vx_{-1}|\vx^T\vbeta_0)\big]\big\}\vgamma\\
		=&B_{111} + B_{112},
		\end{align*}
		where the definitions of $B_{111}$ and $B_{112}$ are clear from the context.
		Assumption~\ref{A1}-(b) implies that $P\big(n^{-1}\sum_{i=1}^n [f_0'(\vx_i^T\vbeta_0)]^2\leq a^2/2\big)\leq \exp\left(\frac{-na^2}{4b^2}\right)$, according to Hoeffding's inequality. Combined with Assumption~\ref{A2}, we have $B_{111}\geq a^2\xi_0||\vgamma||_2^2/2$,  with probability at least $1-\exp\left(\frac{-na^2}{4b^2}\right)$. 
		Note that $\vx_{i,-1}-\E(\vx_{i,-1}|\vx_i^T\vbeta_0)$ is sub-Gaussian by Lemma~\ref{lem:subg}.
		Lemma~\ref{lem15NCL} implies  $|B_{112}|\leq c_0||\vgamma||_2^2\sqrt{\frac{s\log p}{n}}$, with probability at least $1-\exp(-c_1s\log p)$, for some positive constants $c_0$, $c_1$, and all $n$ sufficiently large. Therefore, we obtain that	$B_{11} \geq   ||\vgamma||_2^2 \left(a^2\xi_0/2- \sqrt{\frac{s\log p}{n}} \right) $  with probability at least $1-\exp(-c_1s\log p)$, for some positive constant $c_1$, and all $n$ sufficiently large.
		
		To evaluate $B_{12}$, we observe that there exists a point $\vbeta^r$ between $\vbeta_0$ and $\vbeta$ such that 
		\begin{align*}
		|B_{12}| =&  \Big|n^{-1}\sum_{i=1}^nf_0'(\vx_i^T\vbeta_0) f_0''(\vx_i^T\vbeta^r)(\vx_{i,-1}^T\vgamma)\vgamma^T\big[ \vx_{i,-1}-\E(\vx_{i,-1}|\vx_i^T\vbeta_0)\big]\big[ \vx_{i,-1}-\E(\vx_{i,-1}|\vx_i^T\vbeta_0)\big]^T\vgamma\Big|\\
		\leq &C  n^{-1}\sum_{i=1}^n|\vx_{i,-1}^T\vgamma|*\vgamma^T\big[ \vx_{i,-1}-\E(\vx_{i,-1}|\vx_i^T\vbeta_0)\big]\big[ \vx_{i,-1}-\E(\vx_{i,-1}|\vx_i^T\vbeta_0)\big]^T\vgamma\\
		\leq &C  \sqrt{n^{-1}\sum_{i=1}^n(\vx_{i,-1}^T\vgamma)^2}*\sqrt{n^{-1}\sum_{i=1}^n\left\{\big[ \vx_{i,-1}-\E(\vx_{i,-1}|\vx_i^T\vbeta_0)\big]^T\vgamma\right\}^4},
		\end{align*}
		for some positive constant $C$, given Assumption~\ref{A1}-(b). Recall that $\vxw_{i,-1}=  \vx_{i,-1}-\E(\vx_{i,-1}|\vx_i^T\vbeta_0) $. Lemma~\ref{lem:cube_rate} indicates that
		$$P\left( \Big|n^{-1}\sum_{i=1}^n\big[|\vxw_{i,-1}^T\vgamma|^4 -\E(|\vxw_{i,-1}^T\vgamma|^4)\big]\Big|\geq c_1||\vgamma||_2^4\left[\sqrt{\frac{s\log p}{n}}\vee \frac{s^2\log^2 p}{n}\right]\right)\leq \exp(-c_2  s\log p ),$$
		for some positive constants $c_1$, $c_2$, and all $n$ sufficiently large. 
		Hence $|B_{12}| \leq c_1 ||\vgamma||_2 *||\vgamma||_2^2= c_1||\vgamma||_2^3$,
		with probability at least $1-\exp(-c_2s\log p)$, for some positive  constants $c_1$, $c_2$, and all $n$ sufficiently large.

		H\"older's Inequality and Assumption~\ref{K4}-(b) imply that
		\begin{align*}
		|B_{13}| \leq&  b \Big( n^{-1}\sum_{i=1}^n|G^{(1)}(\vx_{i}^T \vbeta|\vbeta)-G^{(1)}(\vx_{i}^T \vbeta_0|\vbeta_0)|^2\Big)^{1/2}* \Big( n^{-1}\sum_{i=1}^n\big|\big[ \vx_{i,-1}-\E(\vx_{i,-1}|\vx_i^T\vbeta_0)\big]^T\vgamma\big|^4\Big)^{1/2}\\
		\leq &c_1||\vgamma||_2 * \Big( n^{-1}\sum_{i=1}^n\left|\big[ \vx_{i,-1}-\E(\vx_{i,-1}|\vx_i^T\vbeta_0)\big]^T\vgamma\right|^4\Big)^{1/2},
		\end{align*}
		with probability at least $1-\exp(-c_2\log p)$, for some positive constants $c_1$, $c_2$, and all $n$ sufficiently large.  
		Similarly as the derivation for $B_{12}$,  $|B_{13}| \leq c_1 ||\vgamma||_2 *||\vgamma||_2^2= c_1||\vgamma||_2^3$,
		with probability at least $1-\exp(-c_2\log p)$, for some positive  constants $c_1$, $c_2$, and all $n$ sufficiently large. 
		
		To evaluate $B_{14}$, we observe that 
		\begin{align*}
		&|B_{14}|\\
		\leq& \frac{b^2}{n}\sum_{i=1}^n\Big|\vgamma^T\big[\E(\vx_{i,-1}|\vx_i^T\vbeta_0)- \E(\vx_{i,-1}|\vx_i^T\vbeta)\big]  \big[ \vx_{i,-1}-\E(\vx_{i,-1}|\vx_i^T\vbeta_0)\big]^T\vgamma\Big|\\
		\leq & \frac{b^2}{n}\sum_{i=1}^n C\big[|\vx_{i,-1}^T\vgamma|  + (|\vx_i^T\vbeta|+|\vx_i^T\vbeta_0|)*||\vgamma||_2\big] \Big|\big[ \vx_{i,-1}-\E(\vx_{i,-1}|\vx_i^T\vbeta_0)\big]^T\vgamma\Big| \\
		\leq& b^2C||\vgamma||_2 \sqrt{\frac{2}{n}\sum_{i=1}^n  (\vx_{i,-1}^T\vgamma)^2+ (|\vx_i^T\vbeta|^2+|\vx_i^T\vbeta_0|^2)*||\vgamma||_2^2}\sqrt{\frac{1}{n}\sum_{i=1}^n  \Big\{\big[ \vx_{i,-1}-\E(\vx_{i,-1}|\vx_i^T\vbeta_0)\big]^T\vgamma\Big\}^2},
		\end{align*}
		where the second inequality applies Assumption~\ref{A2}-(c).
		Since $\vx_{i,-1} $  is mean-zero sub-Gaussian with variance proxy $\sigma_x^2$, similarly as previous steps, we have  
		$$\frac{1}{n}\sum_{i=1}^n  (\vx_{i,-1}^T\vgamma)^2\leq \left(\xi_3+c\sqrt{\frac{s\log p}{n}}\right)||\vgamma||_2^2,$$
		with probability at least $1-\exp(-c_1s\log p)$,   for some positive constants $c$, $c_1$. Since $[\vx_{i,-1}-\E(\vx_{i,-1}|\vx_i^T\vbeta_0)]$ is also sub-Gaussian by Lemma~\ref{lem:subg}, similarly we have $$\frac{1}{n}\sum_{i=1}^n  \Big\{\big[ \vx_{i,-1}-\E(\vx_{i,-1}|\vx_i^T\vbeta_0)\big]^T\vgamma\Big\}^2 \leq  c_2||\vgamma||_2^2,$$
		with probability at least $1-\exp(-c_1s\log p)$ for some positive constants $c_1$, $c_2$, and all $n$ sufficiently large. It follows that $|B_{14}|\leq c_0||\vgamma||_2^3 $. Since $\vx_{i,-1}$ and $[\vx_{i,-1}-\E(\vx_{i,-1}|\vx_i^T\vbeta_0)]$ are both sub-Gaussian, we apply the same techniques to $|B_{15}|$. Lemma~\ref{lem:cube_rate} ensures that $|B_{15}|\leq c_0\sigma_x^3||\vgamma||_2^3 $ with probability at least $1- \exp(-c_1 s\log p)$,  for some positive constants $c_0$, $c_1$, and all $n$ sufficiently large. 
		
		To bound $|B_{16}|$, Assumption~\ref{A2}-(a), \ref{K3} and \ref{K4}-(a)  imply that
		\begin{align*}
		|B_{16}|\leq& \frac{c}{n}\sum_{i=1}^n\Big|\vgamma^T\E (\vx_{i,-1} \vx_{i,-1}^T  |\vx_i^T\vbeta)\vgamma \big[ \vx_{i,-1}-\E(\vx_{i,-1}|\vx_i^T\vbeta_0)\big]^T\vgamma \Big|\\
		\leq &c\sqrt{ \frac{1}{n}\sum_{i=1}^n\left|\vgamma^T\E (\vx_{i,-1} \vx_{i,-1}^T  |\vx_i^T\vbeta)\vgamma\right|^2}*\sqrt{ \frac{1}{n}\sum_{i=1}^n\left| \vx_{i,-1}-\E(\vx_{i,-1}|\vx_i^T\vbeta_0)\big]^T\vgamma \right|^2}\\  
		\leq &c_0 \sqrt{\xi_4} ||\vgamma||_2^3,
		\end{align*}
		holds with probability at least  $1-\exp[-c_1\log(p\vee n)]- \exp(-c_1 s\log p)$ for some positive constants $c_0$, $c_1$, and all $n$ sufficiently large. In the above, the last inequality applies the sub-Gaussian property of $[\vx_{i,-1}-\E(\vx_{i,-1}|\vx_i^T\vbeta_0)]$, similarly as the derivation for $B_{14}$.
		%		\begin{align*}
		%		|B_{16}|\leq& \frac{c}{n}\sum_{i=1}^n\Big|\vgamma^T\E (\vx_i \vx_i^T  |\vx_i^T\vbeta_0)\vgamma \big[ \vx_i-\E(\vx_i|\vx_i^T\vbeta_0)\big]^T\vgamma \Big|\\
		%		& +\frac{c}{n}\sum_{i=1}^n\Big|\vgamma^T\big[\E (\vx_i \vx_i^T  |\vx_i^T\vbeta)-\E (\vx_i \vx_i^T  |\vx_i^T\vbeta_0)\big]\vgamma\big[ \vx_i-\E(\vx|\vx_i^T\vbeta_0)\big]^T\vgamma \Big|  \\
		%		\leq& \frac{c||\vgamma||_2^2}{n}\sum_{i=1}^n \Big|\big[ \vx_i-\E(\vx_i|\vx_i^T\vbeta_0)\big]^T\vgamma\Big|+\frac{c||\vgamma||_1^2}{n}\sum_{i=1}^n\big|(\vbeta+\vbeta_0)\vx_i\vx_i^T\vgamma\big|  * \Big| \big[ \vx_i-\E(\vx_i|\vx_i^T\vbeta_0)\big]^T\vgamma \Big|.\\ &+\frac{c||\vgamma||_1^3}{n}\sum_{i=1}^n  \big( |\vx_i^T\vbeta_0|^2+ |\vx_i^T\vbeta|^2+1\big)*\Big|\big[ \vx_i-\E(\vx_i|\vx_i^T\vbeta_0)\big]^T\vgamma \Big|.
		%		\end{align*}
		%		Similarly as $B_{14}$, $|B_{16}|\leq c_0 ||\vgamma||_2^3$ 	with the same probability bound.	
		
		Combining all the preceding results,  we conclude that $-B_1\geq  c_0\left(||\vgamma||_2^2-||\vgamma||_2^2\sqrt{\frac{\log p}{n}}\right)$ with probability at least $1-\exp(  -c_1  \log p ) $, and universal positive constants $c_0$ and $c_1$, for all $n$ sufficiently large.
		%with probability at least $1-\exp(-c_1\log n)$. Combining all previous results, we conclude that $J_{n1} \geq  \frac{a^2\xi_0}{2}||\vgamma||_2^2$ with probability at least $1-\exp[-c_1\log (p\wedge n)]$, with $0<r<1$ sufficiently small.
		
		%Note that $\sup_{\vbeta\in\bbB}\max_{1\leq i\leq n}|\vx_i^T\vbeta| \leq \sup_{\vbeta\in\bbB}||\vbeta||_1*\max_{1\leq i\leq n}||\vx_i||_\infty $. By the sub-Gaussian property for $\vx_i$, and consider $||\vbeta_0||_1$ as a constant and $r$ sufficiently small, we have
		%\begin{align*}
		%P\Big(\sup_{\vbeta\in\bbB}\max_{1\leq i\leq n}|\vx_i^T\vbeta|\geq\sqrt{\log  (p\vee n) }\Big)    \leq  P\Big( \max_{1\leq i\leq n}||\vx_i||_\infty\geq\frac{\sqrt{\log (p\vee n)}}{||\vbeta_0||_1+r} \Big)\leq \exp[-c_1\log (p\vee n)],
		%\end{align*} 
		%for universal constant $c_1>0$.  
		%Conditional on the event $\{\sup_{\vbeta\in\bbB}\max_{ 1\leq i \leq n}|\vx_i^T\vbeta| \leq\sqrt{ \log  (p\vee n)}\}$, 
		Assumption~\ref{K4}-(a), Lemma~\ref{lem:Gbound} and Lemma~\ref{lem:Ebound} imply that 
		\begin{align*}
		|B_{2}| &\leq  b\max_{ 1\leq i \leq n}\sup \limits_{\vbeta \in \bbB}\big|\big[\widehat{\E}(\vx_{i,-1}|\vx_i^T\vbeta)-\E(\vx_{i,-1}|\vx_i^T\vbeta)\big]^T\vgamma  \big|* n^{-1}\sum_{i=1}^n |G(\vx_i^T\vbeta_0|\vbeta_0)  - G(\vx_i^T\vbeta|\vbeta)|\\
		&  \leq d_1h^2||\vgamma||_2^2 ,
		\end{align*} 
		with probability at least $1- \exp[-c_2\log (p\vee n)]$, for some universal positive constants $d_1$, $c_2$, and all $n$ sufficiently large. %The sub-Gaussian properties of $\vx_i$ suggest that
		%	$$P\Big(||\vSigmah||_\infty \leq ||\E (\vx\vx^T)||_\infty +c_1\sqrt{\frac{\log p}{n}}\Big)\geq 1-\exp(-c_3\log p).$$ 
		Similarly, %conditional on the event   $\mathcal{F}_n $, as defined in Lemma~\ref{lem:events}, 
		we can  show that for some positive universal constants $d_1$ and $c_1$, 
		\begin{align*}
		&|A_{11}| \leq d_1 ||\vgamma||_2^2 \sqrt{\frac{\log p}{n}}   ,& &|A_{12}| \leq d_1h^2 ||\vgamma||_2 \sqrt{\frac{\log p}{n}} ,& &|A_{13}| \leq d_1h ||\vgamma||_2^2  \sqrt{\frac{\log p}{n}},\\
		&|A_{21}| \leq d_1 ||\vgamma||_2^2  \sqrt{\frac{\log p}{n}}  ,& & |A_{22}| \leq d_1h^2||\vgamma||_2^2\sqrt{\frac{\log p}{n}},& &|A_{23}| \leq  d_1h||\vgamma||_2\sqrt{\frac{\log p}{n}},\\
		&|A_{31}| \leq d_1 ||\vgamma||_2^3  ,& &|A_{32}| \leq d_1h^2 ||\vgamma||_2^2,& &|A_{33}| \leq d_1h ||\vgamma||_2^3,\\
		&|A_{34}| \leq d_1 h^2||\vgamma||_2^2  ,& &|B_3| \leq d_1h ||\vgamma||_2^2 ,& &|B_4| \leq d_1h^3 ||\vgamma||_2  ,\\
		&|B_5| \leq d_1 h^2||\vgamma||_2 ,& 
		\end{align*} 
		hold with probability at least $1-\exp(-c_1\log p)$,  for all $n$ sufficiently large. Since  $ n^{-1} \log p = O( h^5) $, there exist some universal positive constants $c_0$, $c_1$, $c_2$ and $r\leq 1$ such that 
		$$ \big\langle\vS_n(\vbeta,\widehat{G},\widehat{\E})-\vS_n(\vbeta_0,\widehat{G},\widehat{\E}),  \vgamma\big\rangle  \geq c_0 ||\vgamma||_2^2 - c_1 h^2||\vgamma||_2,$$
		with probability at least $1-\exp( -c_1  \log p)  $, for any $\vbeta\in\bbB$ and all $n$ sufficiently large.		
	\end{proof}
	
	\begin{proof}[Proof of Theorem~\ref{Lasso_error}]
		By the definition of $\vbetah=(1,\vbetah_{-1}^T)^T$, we have 
		\begin{align}
		\big\langle \vS_n(\vbetah,\widehat{G},\widehat{\E})+\lambda \vkappah,  \vbeta_{-1} - \vbetah_{-1}\big\rangle =0,\label{stationary_cond}
		\end{align} 
		for all feasible $\vbeta$, where $\vkappah\in\partial||\vbetah_{-1}||_1$. In particular, $\big\langle\vS_n(\vbetah,\widehat{G},\widehat{\E})+\lambda \vkappah,  \vbeta_{0,-1} - \vbetah_{-1}\big\rangle=0$. By the property of convex function, we know that $||\vbeta_0||_1-||\vbetah||_1\geq  \big\langle  \vkappah,  \vbeta_{0,-1} - \vbetah_{-1}\big\rangle $, for any $\vkappah\in\partial||\vbetah_{-1}||_1$. Combining this with (\ref{stationary_cond}), we have
		\begin{align}
		\big\langle\vS_n(\vbetah,\widehat{G},\widehat{\E}),   \vetah\big\rangle= \lambda \big\langle  \vkappah,  -\vetah \big\rangle \leq \lambda(||\vbeta_{0,-1}||_1-||\vbetah_{-1}||_1).\label{cond1}
		\end{align}
		where $\vetah = \vbetah_{-1}-\vbeta_{0,-1}$. Applying the local restricted strong convexity condition  established in
		Lemma~\ref{lem:local_LRC} to $\big\langle\vS_n(\vbeta,\widehat{G},\widehat{\E})-\vS_n(\vbeta_0,\widehat{G},\widehat{\E}),  \vetah\big\rangle$, we obtain
		\begin{align}
		c_0 ||\vetah||_2^2 - c_1 h^2||\vetah||_2 &\leq   \big\langle\vS_n(\vbetah,\widehat{G},\widehat{\E}),  \vetah\big\rangle - \big\langle \vS_n(\vbeta_0,\widehat{G},\widehat{\E}),  \vetah\big\rangle \nonumber\\
		&\leq   \lambda(||\vbeta_{0,-1}||_1-||\vbetah_{-1}||_1) - \big\langle \vS_n(\vbeta_0,\widehat{G},\widehat{\E}),  \vetah\big\rangle \nonumber \\
		&\leq  \lambda(||\vbeta_{0,-1}||_1-||\vbetah_{-1}||_1)+||\vS_n(\vbeta_{0},\widehat{G},\widehat{\E})||_\infty || \vetah||_1,\label{res1}
		\end{align}
		with probability at least $1-\exp(-c_1\log p)$, for some positive constant $c_1$ and all $n$ sufficiently large. In the above, the second inequality uses (\ref{cond1}). Note that $||\vetah||_2\leq ||\vetah||_1$.
		This implies that $$c_0 ||\vetah||_2^2 \leq \left(c_1 h^2 +||\vS_n(\vbeta_0,\widehat{G},\widehat{\E})||_\infty\right) ||\vetah||_1+\lambda\left(||\vbeta_{0,-1}||_1-||\vbetah_{-1}||_1\right) .$$ 
		By Lemma~\ref{Lgrad},  $ \lambda/4 \geq ||\vS_n(\vbeta_0,\widehat{G},\widehat{\E})||_\infty$ with probability at least $1-\exp[-c_1\log (p\vee n)]$, since $\sqrt{n^{-1}\log(p\vee n)} \leq c_0h^{5/2}\leq c_0h^2$ for some positive constant $c_0$.  Let $\vetah_{\mathcal{S}}$ and $\vetah_{\mathcal{S}^C}$ be the sub-vectors of $\vetah$  on the support $\mathcal{S} = \{j:\beta_{0,j+1}\neq 0,\ j=1,\cdots,p-1\}$, and $\mathcal{S}^C$, respectively.
		Then we have 
		\begin{align}
		c_0 ||\vetah||_2^2  &\leq  \frac{\lambda}{2}  \left(||\vetah_{\mathcal{S}}||_1+||\vetah_{\mathcal{S}^C}||_1\right) + \lambda \left(||\vetah_{\mathcal{S}}||_1-||\vetah_{\mathcal{S}^C}||_1\right) \leq \frac{3\lambda}{2}  ||\vetah_{\mathcal{S}}||_1- \frac{\lambda}{2}  ||\vetah_{\mathcal{S}^C}||_1,\label{cond2}
		\end{align}
		which implies that $||\vetah_{\mathcal{S}^C}||_1\leq 3 ||\vetah_{\mathcal{S}}||_1$. Then (\ref{cond2}) implies that 
		\begin{align*}
		c_0 ||\vetah||_2^2 \leq \left(c_1 h^2 +|| \vS_n(\vbeta_0,\widehat{G},\widehat{\E})||_\infty + \lambda\right) ||\vetah||_1 \leq \frac{3\lambda}{2}||\vetah||_1 \leq 6\lambda ||\vetah_{\mathcal{S}}||_1\leq 6\lambda \sqrt{s}||\vetah||_2.
		\end{align*}
		Hence $||\vbetah-\vbeta_0||_2=||\vetah||_2\leq \frac{6}{c_0} \lambda\sqrt{s}$. Since $||\vbetah-\vbeta_0||_1=||\vetah||_1 \leq  4||\vetah_{\mathcal{S}}||_1 \leq  4\sqrt{s} ||\vetah||_2$, 
		the bound of $||\vbetah-\vbeta_0||_1$ follows immediately.
	\end{proof}

	\section{Proofs of results in Section~\ref{sec:inf_theorem} of the main paper} \label{sec:proof3.2}
	
	\begin{proof} [Proof of Lemma~\ref{dbound}]
		(1)   To derive the uniform error bound for $\vd_j(\vbetah,\eta)$, $j=2,\ldots,p$, we first prove the following two results:\\
		(i) $\vd_{0j} =  (\vOmega_{-(j-1),-(j-1)})^{-1}\vOmega_{-(j-1),(j-1)}$ is feasible in the sense that it satisfies the constraint of the Dantzig problem in (\ref{dj_def}) of the main paper with high probability, uniformly  in $j=2,\cdots,p$;\\
		(ii) $\vOmegah = \frac{1}{n}\sum_{i=1}^n [\widehat{G}^{(1)}(\vx_i^T\vbetah|\vbetah)]^2 \vxh_{i,-1}\vxh_{i,-1}^T$ satisfies a restricted eigenvalue condition on  the support of $\vphi_{0j}=\tau_{0j}^2\vtheta_{j}$, denoted by $\mathcal{S}_{\vphi_j}$,  with high probability, uniformly  in $j=2,\cdots,p$, %such that $\vphi_{0j}\neq\ve_{j-1}$, where $\ve_{j-1}$ is the $(j-1)^{th}$ column of the identity matrix $\vI_{p-1}$, 
		as shown below in (\ref{rsc_dj}).
		
		To prove (i), %let $\vphi_{0j} = \tau^{2}_{0j}\vtheta_j$. 
		the assumptions of Theorem~\ref{Lasso_error} imply that  $\sigma_x^2\sqrt{\frac{\log p}{n}} = O(h^{5/2}) =  o(\eta)$. Lemma~\ref{lem:subexp} implies that 
		$$P\left(\max_{2\leq j\leq p} \Big|\Big|\frac{1}{n} \sum_{i=1}^n [G^{(1)}(\vx_i^T\vbeta_0|\vbeta_0)]^2 \vxw_{i,-1}^T\vphi_{0j}  \vxw_{i,-j*}\Big|\Big|_\infty\geq \eta/2\right)\leq \exp(-c_2\log p),$$
		for some  positive constant $c_2$, and all $n$ sufficiently large.
		
		Lemma~\ref{lem:G1xx} implies that 
		\bqa
		P\left(\max_{2\leq j\leq p}\Big|\Big|\frac{1}{n} \sum_{i=1}^n \big\{[\widehat{G}^{(1)}(\vx_i^T\vbetah|\vbetah)]^2 - [G^{(1)}(\vx_i^T\vbeta_0|\vbeta_0)]^2\big\} \vxw_{i,-1}^T\vphi_{0j}  \vxw_{i,-j*}\Big|\Big|_\infty\geq c_0h\right)\leq \exp(-c_1\log p),
		\eqa
		for some positive constants $c_0$ and $c_1$, and all $n$ sufficiently large. 
		Lemma~\ref{lem:Ex_err} implies there exist some positive constants $c_0$, $c_1$, such 
		that for all $n$ sufficiently large,
		$$P\left( \max_{2\leq j \leq p}\Big|\Big|\frac{1}{n} \sum_{i=1}^n  \big[\widehat{G}^{(1)}(\vx_i^T\vbetah|\vbetah)\big]^2 (\vxh_{i,-1}  \vxh_{i,-1}^T-\vxw_{i,-1} \vxw_{i,-1}^T)\vphi_{0j}\Big|\Big|_\infty \geq c_0\sqrt{s}h^2\right)\leq \exp(-c_1\log p). $$
		Hence there exist some universal positive constants $d_2$ and $c_1$,  such that 
		for $\eta=d_2h$, for all $n$ sufficiently large,
		\begin{align*}
		&P\left(\max_{2\leq j\leq p} \Big|\Big| \frac{1}{n}\sum_{i=1}^n [\widehat{G}^{(1)}(\vx_i^T\vbetah|\vbetah)]^2 \vxh_{i,-1}^T\vphi_{0j} \vxh_{i,-j*}\Big|\Big|_\infty\geq \eta \right) \\
		\leq& P\left(\max_{2\leq j\leq p} \Big|\Big|\frac{1}{n} \sum_{i=1}^n [G^{(1)}(\vx_i^T\vbeta_0|\vbeta_0)]^2 \vxw_{i,-1}^T\vphi_{0j} \vxw_{i,-j*}\Big|\Big|_\infty\geq \eta/2\right)\\
		&+P\left(\max_{2\leq j\leq p}\Big|\Big|\frac{1}{n} \sum_{i=1}^n \big\{[\widehat{G}^{(1)}(\vx_i^T\vbetah|\vbetah)]^2 - [G^{(1)}(\vx_i^T\vbeta_0|\vbeta_0)]^2\big\} \vxw_{i,-1}^T\vphi_{0j} \vxw_{i,-j*}\Big|\Big|_\infty\geq \eta/4\right)\\ 
		&+P\left( \max_{2\leq j \leq p}\Big|\Big|\frac{1}{n} \sum_{i=1}^n  [\widehat{G}^{(1)}(\vx_i^T\vbetah|\vbetah)]^2 (\vxh_{i,-1}  \vxh_{i,-1}^T-\vxw_{i,-1} \vxw_{i,-1}^T)\vphi_{0j}\Big|\Big|_\infty \geq \eta/4\right)\\
		\leq&\exp(-c_1\log p),
		\end{align*} 
		as $\sqrt{s}h^2\leq d_0h^2*\sqrt{nh^5}\leq d_0h\sqrt{nh^7} =o(\eta)$ for some positive constant $d_0$ by the assumptions of Theorem~\ref{Lasso_error}. Since $\vphi_{0j} = \Big(-(\vd_{0j})_{1:(j-2)}^T,1,-(\vd_{0j})_{(j-1):(p-2)}^T\Big)^T$ by Lemma~\ref{lem:thetaj}, it implies that $ \vd_{0j}$ satisfies the constraint in (\ref{dj_def}) in Section~\ref{sec:inf_method} of the main paper, that is,
		$$\Big|\Big|n^{-1}\sum_{i=1}^n [\widehat{G}^{(1)}(\vx_i^T\vbetah|\vbetah)]^2 (\xh_{i,j}-\vxh_{i,-j}^T \vd_{0j})\vxh_{i,-j*}\Big|\Big|_\infty\leq \eta,$$
		with probability at least $1-\exp(-c_1\log p)$ uniformly in $j$, with $\eta=d_2h$, for some positive constants $d_2$, $c_1$ and all $n$ sufficiently large. This ensures that $ \vd_{0j}$ is feasible for (\ref{dj_def}) with probability at least $1-\exp(-c_1\log p)$ uniformly in $j$, for all $n$ sufficiently large. By the definition of $\vd_j(\vbetah,\eta)$, we have %we have $P\big( ||\vd_j(\vbetah,\eta)||_1 \geq ||\vd_{0j}||_1  \big)\leq\exp(-c_1\log p)$, for some constants $d_2>0$, $c_1>1$ and all $n$ sufficiently large. 
		% the Dantzig Selector problem with high probability uniformly in $j$. Since $\vd_j(\vbetah,\eta)$ has the smallest $L_1-$norm among all the feasible solutions which satisfy the constraint in (\ref{dj_def}) in Section~\ref{sec:inf_method}, then we know that  
		$$P\Big(||\vd_j(\vbetah,\eta)||_1 \leq ||\vd_{0j}||_1 \mbox{ uniformly in }j  \Big)\geq1-\exp(-c_1\log p),$$ 
		with $\eta=d_2h$, for some positive constants $d_2$, $c_1$ and all $n$ sufficiently large. 
		
		Given (i), the event $\mathcal{E}_1 = \left\{ ||\vd_j(\vbetah,\eta)||_1 \leq ||\vd_{0j}||_1   \mbox{ uniformly in }j\right\}$
		holds with probability at least $1-\exp(-c_1\log p)$ for some positive constant $c_1$ and all $n$ sufficiently large. Define $\vw_j = \vphi_j(\vbetah, \eta)-\vphi_{0j}$.  
		Note that for any $j$ such that $\vphi_{0j}=\ve_{j-1}$, we have $\vd_{0j} =\vnull_{p-2}$, by Lemma~\ref{lem:thetaj}, where $\ve_{j-1}$ denotes the $(p-1)-$dimensional vector with the $(j-1)^{th}$ entry being one and all the other entries equal to zero, and $\vnull_{p-2}$ denotes the $(p-2)-$dimensional vector with all the entries equal to zero. On the event $\mathcal{E}_1$, if  $j$ is such that $\vphi_{0j}=\ve_{j-1}$,
		then $||\vw_j||_1=0$ (as we will have $||\vd_j(\vbetah,\eta)||_1 = ||\vd_{0j}||_1 =0$ for this case) and the results in Lemma~\ref{dbound}-(1) always hold. Therefore, without loss of generality, we assume that $\vphi_{0j}\neq\ve_{j-1}$ for any $j=2,\ldots,p$.   Recall that $\mathcal{S}_{\vphi_j}$ is the support set of $\vphi_{0j} = \tau_{0j}^2\vtheta_j$.
		On the event $\mathcal{E}_1$, 
		$$||\vw_{j,\mathcal{S}_{\vphi_j}^C}||_1 = \Big|\Big|\big[\vphi_j(\vbetah,\eta)\big]_{\mathcal{S}_{\vphi_j}^C} \Big|\Big|_1\leq \big|\big|\vphi_{0j, \mathcal{S}_{\vphi_j}} \big|\big|_1 - \Big|\Big| \big[\vphi_j(\vbetah,\eta) \big]_{\mathcal{S}_{\vphi_j}} \Big|\Big|_1\leq ||\vw_{j,\mathcal{S}_{\vphi_j}}||_1,$$
		where the first equality applies Lemma~\ref{lem:thetaj}; the second last inequality applies $||\vphi_j(\vbetah,\eta)||_1\leq||\vphi_{0j}||_1$; the last inequality applies $\left|\left|\vphi_{0j, \mathcal{S}_{\vphi_j}}\right|\right|_1=\left|\left| [\vphi_j(\vbetah,\eta)  ]_{\mathcal{S}_{\vphi_j}} -\vw_{j,\mathcal{S}_{\vphi_j}}\right|\right|_1\leq \left|\left| [\vphi_j(\vbetah,\eta)  ]_{\mathcal{S}_{\vphi_j}} \right|\right|+\left|\left|\vw_{j,\mathcal{S}_{\vphi_j}}\right|\right|_1$. Denote the set $\mathcal{V}_{2j}=\{\vv=(v_1,\cdots,v_{p-1})^T: ||\vv_{\mathcal{S}_{\vphi_j}^C}||_1 \leq ||\vv_{\mathcal{S}_{\vphi_j}}||_1, ||\vv||_2=1, v_{j-1}=0\}$, %where $\mathcal{S}_{\vphi_j}$ is the support set of $\vphi_{0j} = \tau_{0j}^2\vtheta_j$, 
		for any $j=2,\ldots,p$. % such that $\vphi_{0j}\neq\ve_{j-1}$. Note that if $\vphi_{0j}=\ve_{j-1}$, then $\mathcal{V}_{2j}=\emptyset$. 
		We observe that on the event $\mathcal{E}_1$,   $\frac{\vw_j}{||\vw_j||_2}\in\kV_{2j}$, for any $j=2,\ldots,p$. % such that $\vphi_{0j}\neq\ve_{j-1}$. 
		In the next step, we will prove
		\begin{align}
		P\left(\min_{2\leq j\leq p}\inf\limits_{\vv\in\mathcal{V}_{2j}} \vv^T \vOmegah \vv\leq \frac{\xi_2}{2} \right) \leq \exp(-d_0\log p), \label{rsc_dj}
		\end{align}	
		for some positive constant $d_0$ and all $n$ sufficiently large. 
		
		To prove (\ref{rsc_dj}), %consider the set $\mathcal{V}_{2j}=\{\vv=(v_1,\cdots,v_{p-1})^T: ||\vv_{\mathcal{S}_{\vphi_j}^C}||_1 \leq ||\vv_{\mathcal{S}_{\vphi_j}}||_1, ||\vv||_2=1, v_{j-1}=0\}$, where $\mathcal{S}_{\vphi_j}$ is the support set of $\vphi_{0j} = \tau_{0j}^2\vtheta_j$. Note that if $\vphi_{0j}=\ve_{j-1}$, then $\mathcal{V}_{2j}=\emptyset$. 
		Assumption~\ref{A2}-(a) indicates that  
		%for any support $\mathcal{I}$.
		\begin{align*}
		\min_{2\leq j\leq p}\inf\limits_{\vv\in\mathcal{V}_{2j}}\vv^T \vOmegah \vv & \geq 	\min_{2\leq j\leq p} \inf\limits_{\vv\in\mathcal{V}_{2j}}\vv^T \vOmega\vv - \max_{2\leq j\leq p}\sup \limits_{\vv\in\mathcal{V}_{2j}}\vv^T\{\vOmegah-\vOmega\}\vv \\
		%\inf\limits_{||\vv||_2=1}\vv^T \vOmegah \vv & \geq 	\inf\limits_{||\vv||_2=1}\vv^T \vOmega\vv -  \sup \limits_{||\vv||_2=1}\vv^T\{\vOmegah-\vOmega\}\vv \\
		&\geq \xi_2 - \max_{2\leq j\leq p}\sup \limits_{\vv\in\mathcal{V}_{2j}}\vv^T\{\vOmegah-\vOmega\}\vv ,%||\vOmegah - \vOmega||_\infty \sup \limits_{\vv\in\mathcal{C} (\mathcal{I})} ||\vv||_1^2\\&\geq  a^2\xi_1 -4|\mathcal{I}|* ||\vOmegah - \vOmega||_\infty.
		\end{align*}
		where $\xi_2$ is the positive constant defined in Assumption~\ref{A2}-(a).
		Note that 
		\begin{align*}
		\max_{2\leq j\leq p}\sup \limits_{\vv\in\mathcal{V}_{2j}}\big|\vv^T\{\vOmegah-\vOmega\} \vv\big| \leq&\max_{2\leq j\leq p}\sup \limits_{\vv\in\mathcal{V}_{2j}}\Big|\vv^T\Big( \frac{1}{n}\sum_{i=1}^n[\widehat{G}^{(1)}(\vx_i^T\vbetah|\vbetah)]^2[\vxh_i\vxh_i^T - \vxw_{i,-1}\vxw_{i,-1}^T]\Big)\vv\Big|\\
		&+\max_{2\leq j\leq p}\sup \limits_{\vv\in\mathcal{V}_{2j}} \Big| \frac{1}{n}\sum_{i=1}^n\Big\{ [\widehat{G}^{(1)}(\vx_i^T\vbetah|\vbetah)]^2 - [G^{(1)}(\vx_i^T\vbeta_0|\vbeta_0)]^2\Big\} (\vxw_{i,-1}^T\vv)^2\Big|  \\
		&+\max_{2\leq j\leq p}\sup \limits_{\vv\in\mathcal{V}_{2j}}\Big|\vv^T\Big( \frac{1}{n}\sum_{i=1}^n [G^{(1)}(\vx_i^T\vbeta_0|\vbeta_0)]^2 \vxw_{i,-1}\vxw_{i,-1}^T - \vOmega\Big)\vv\Big|.
		%\sup \limits_{||\vv||_2=1}\big|\vv^T\{\vOmegah-\vOmega\} \vv\big| \leq&\sup \limits_{||\vv||_2=1}\Big|\vv^T\Big( \frac{1}{n}\sum_{i=1}^n[\widehat{G}^{(1)}(\vx_i^T\vbetah|\vbetah)]^2[\vxh_i\vxh_i^T - \vxw_{i,-1}\vxw_{i,-1}^T]\Big)\vv\Big|\\
		%&+\sup \limits_{||\vv||_2=1} \Big| \frac{1}{n}\sum_{i=1}^n\Big\{ [\widehat{G}^{(1)}(\vx_i^T\vbetah|\vbetah)]^2 - [G^{(1)}(\vx_i^T\vbeta_0|\vbeta_0)]^2\Big\} (\vxw_{i,-1}^T\vv)^2\Big|  \\
		%&+\sup \limits_{||\vv||_2=1}\Big|\vv^T\Big( \frac{1}{n}\sum_{i=1}^n [G^{(1)}(\vx_i^T\vbeta_0|\vbeta_0)]^2 \vxw_{i,-1}\vxw_{i,-1}^T - \vOmega\Big)\vv\Big|.
		\end{align*}  
		Note that for any $\vv\in\kV_{2j}$, we have that $||\vv||_1\leq 2||\vv_{\mathcal{S}_{\vphi_j}}||_1\leq 2\sqrt{s_j}||\vv_{\mathcal{S}_{\vphi_j}}||_2\leq 2\sqrt{s_j}\leq 2\sqrt{\widetilde{s}}$, where $s_j = ||\vd_{0j}||_0$, and $\widetilde{s}=\max_{ 2\leq j\leq p}s_j$. Hence the proof of Lemma~\ref{lem:Ex_err} implies that $\max_{2\leq j\leq p}\sup \limits_{\vv\in\mathcal{V}_{2j}} \Big|\vv^T\Big( \frac{1}{n}\sum_{i=1}^n[\widehat{G}^{(1)}(\vx_i^T\vbetah|\vbetah)]^2[\vxh_{i,-1}\vxh_{i,-1}^T - \vxw_{i,-1}\vxw_{i,-1}^T]\Big)\vv\Big|\leq c\widetilde{s}sh^4\leq c_0\widetilde{s}h^4*nh^5\leq c_0h^2$ with probability at least $1-\exp(-c_1\log p)$, by the assumptions of Lemma~\ref{dbound}, for some positive constants $c$, $c_0$, $c_1$, and all $n$ sufficiently large.	The proof of Lemma~\ref{lem:G1xx} implies that with probability at least $1-\exp(-c_1\log p)$, 
		$$\max_{2\leq j\leq p}\sup \limits_{\vv\in\mathcal{V}_{2j}}  \Big| \frac{1}{n}\sum_{i=1}^n\Big\{ [\widehat{G}^{(1)}(\vx_i^T\vbetah|\vbetah)]^2 - [G^{(1)}(\vx_i^T\vbeta_0|\vbeta_0)]^2\Big\} (\vxw_{i,-1}^T\vv)^2\Big| \leq c_0h(1+\sqrt{n^{-1}\widetilde{s}\log p}),$$ for some positive constants $c_0$, $c_1$, and all $n$ sufficiently large. Note that $\sqrt{n^{-1}\widetilde{s}\log p} \leq 1$ by the assumptions of Lemma~\ref{dbound}.
		%Theorem 2.1 in \citet{subExp} implies that 
		Similarly as Lemma~\ref{lem:subg}, 
		$(2A_i-1)G^{(1)}(\vx_i^T\vbeta_0|\vbeta_0)  \vxw_{i,-1}$ is sub-Gaussian with variance proxy no larger than $b^2\sigma_x^2$, where $b$ is the positive constant defined in Assumption~\ref{A1}-(b). Similarly as Lemma~\ref{lem15NCL}, we have
		$$P\left(\max_{2\leq j\leq p}\sup \limits_{\vv\in\mathcal{V}_{2j}} \Big|\vv^T\Big( \frac{1}{n}\sum_{i=1}^n  [G^{(1)}(\vx_i^T\vbeta_0|\vbeta_0)]^2\vxw_{i,-1}\vxw_{i,-1}^T - \vOmega\Big)\vv\Big|\geq  c_0\sigma_x^2\sqrt{\frac{\widetilde{s}\log p}{n}}\right)\leq \exp(-c_1\widetilde{s}\log p),$$
		where $\widetilde{s}=\max_{2\leq j\leq p}||\vd_{0j}||_0 = \max_{2\leq j\leq p}||\vphi_{0j}||_0 -1$, by Lemma~\ref{lem:thetaj}, for some positive constants $c_0$, $c_1$, and all $n$ sufficiently large.
		Hence there exist some positive constants $c_0$, $c_1$, such that for all $n$ sufficiently large,
		\begin{align*}
		P\left(\max_{2\leq j\leq p}\sup \limits_{\vv\in\mathcal{V}_{2j}} \big|\vv^T\{\vOmegah-\vOmega\} \vv\big|\geq c_0\big(h+\sqrt{n^{-1}\widetilde{s}\log p}\big)\right)\leq \exp(-c_1\log p).%\label{omega_bound}
		\end{align*} 
		Note that $\sqrt{\frac{\widetilde{s}\log p}{n}} \leq c_2\sqrt{\widetilde{s}h^5}\leq c_2h \sqrt{\widetilde{s}h^3}  = o(h)$, for some positive constant $c_2$, by the assumptions of Theorem~\ref{Lasso_error} and Lemma~\ref{dbound}.
		Hence  we conclude 
		\begin{align*}
		P\left(\min_{2\leq j\leq p}\inf\limits_{\vv\in\mathcal{V}_{2j}} \vv^T \vOmegah \vv\leq \frac{\xi_2}{2} \right) \leq  &P\Big(\xi_2 -\max_{2\leq j\leq p}\sup \limits_{\vv\in\mathcal{V}_{2j}} \vv^T\{\vOmegah-\vOmega\}\vv\leq \frac{\xi_2}{2} \Big) \\
		\leq &P\left(\max_{2\leq j\leq p}	\sup \limits_{\vv\in\mathcal{V}_{2j}}
		\big|\vv^T\{\vOmegah-\vOmega\} \vv \big| \geq \frac{\xi_2}{2}\right) \leq \exp(-c_1\log p), 
		\end{align*}	
		for some positive constant $c_1$ and all $n$ sufficiently large.
		It proves (\ref{rsc_dj}), i.e., (ii).

		Given (i) and (ii), the event $$\mathcal{E} = \left\{ \min_{2\leq j\leq p}\inf_{\vv\in\mathcal{V}_{2j}} \vv^T \vOmegah\vv\geq \frac{\xi_2}{2}\mbox{, and }||\vd_j(\vbetah,\eta)||_1 \leq ||\vd_{0j}||_1   \mbox{ uniformly in }j\right\},$$
		holds with probability at least $1-\exp(-c_1\log p)$ for some positive constant $c_1$ and all $n$ sufficiently large. 
		%Define $\vw_j = \vphi_j(\vbetah, \eta)-\vphi_{0j}$. %Note that for any $j$ such that $\vphi_{0j}=\ve_{j-1}$, we have $\vd_{0j} =\vnull_{p-2}$, by Lemma~\ref{lem:thetaj}. Hence on the event $\mathcal{E}$, we have $||\vd_j(\vbetah,\eta)||_1 = ||\vd_{0j}||_1 =0$, and $||\vw_j||_1=||\vw_j||_2=0$. For any $j$ such that $\vphi_{0j}\neq\ve_{j-1}$,  on the event $\mathcal{E}$,
		%Let $\mathcal{S}_{\vphi_j}$ is the support set of $\vphi_{0j} = \tau_{0j}^2\vtheta_j$.
		%For any $j\in\{2,\cdots,p\}$,  on the event $\mathcal{E}$, 
		%$$||\vw_{j,\mathcal{S}_{\vphi_j}^C}||_1 = \Big|\Big|\big[\vphi_j(\vbetah,\eta)\big]_{\mathcal{S}_{\vphi_j}^C} \Big|\Big|_1\leq \big|\big|\vphi_{0j, \mathcal{S}_{\vphi_j}} \big|\big|_1 - \Big|\Big| \big[\vphi_j(\vbetah,\eta) \big]_{\mathcal{S}_{\vphi_j}} \Big|\Big|_1\leq ||\vw_{j,\mathcal{S}_{\vphi_j}}||_1,$$
		%where the first equality applies Lemma~\ref{lem:thetaj}; the second last inequality applies that $||\vphi_j(\vbetah,\eta)||_1\leq||\vphi_{0j}||_1$; the last inequality applies $\left|\left|\vphi_{0j, \mathcal{S}_{\vphi_j}}\right|\right|_1=\left|\left| [\vphi_j(\vbetah,\eta)  ]_{\mathcal{S}_{\vphi_j}} -\vw_{j,\mathcal{S}_{\vphi_j}}\right|\right|_1\leq \left|\left| [\vphi_j(\vbetah,\eta)  ]_{\mathcal{S}_{\vphi_j}} \right|\right|+\left|\left|\vw_{j,\mathcal{S}_{\vphi_j}}\right|\right|_1$.
		%It implies that $\vw_j\in \mathcal{V}_{2j}$.
		On the event $\mathcal{E}$,  $\vd_{0j}$ and $\vd_j(\vbetah,\eta)$  both satisfy the constraint of the Dantzig problem in (\ref{dj_def}) of the main paper with high probability. Then we have
		\begin{align*}
		||\vOmegah\vw_j||_\infty \leq  & \Big|\Big| \frac{1}{n}\sum_{i=1}^n [\widehat{G}^{(1)}(\vx_i^T\vbetah|\vbetah)]^2 \big[\xh_{i,j}-\vxh_{i,-j*}^T\vd_j(\vbetah,\eta)\big] \vxh_{i,-j*}\Big|\Big|_\infty \\
		&+  \Big|\Big| \frac{1}{n}\sum_{i=1}^n [\widehat{G}^{(1)}(\vx_i^T\vbetah|\vbetah)]^2 \big[\xh_{i,j}-\vxh_{i,-j*}^T\vd_{0j}\big] \vxh_{i,-j*}\Big|\Big|_\infty \leq 2\eta.
		\end{align*} 
		Note that on the event $\mathcal{E}$, %for any $j$ such that $\vphi_{0j}=\ve_{j-1}$, $||\vw_j||_1=0$; and for any $j$ such that $\vphi_{0j}\neq \ve_{j-1}$,
		$\frac{\vw_j}{||\vw_j||_2}\in\kV_{2j}$ holds uniformly in $j$. On the event $\mathcal{E}$,    we have
		\begin{align*}
		||\vw_j||_1^2\leq 4s_j ||\vw_j||_2^2\leq \frac{ 8s_j\vw_j^T\vOmegah\vw_j }{\xi_2}\leq \frac{8s_j||\vw_j||_1||\vOmegah\vw_j||_\infty}{\xi_2}
		\leq \frac{16s_j\eta||\vw_j||_1}{\xi_2}\leq \frac{32s_j^{3/2}\eta||\vw_j||_2}{\xi_2},
		\end{align*}
		uniformly in $j$, %such that $\vphi_{0j}\neq\ve_{j-1}$, 
		where the second last inequality applies the above result. It hence implies that $P\Big( ||\vw_j||_1 \leq \frac{16s_j\eta}{\xi_2} \mbox{ uniformly in }j\Big)\geq P(\mathcal{E})\geq 1-\exp(-c_1\log p)$, and $P\Big( ||\vw_j||_2 \leq \frac{8\sqrt{s_j}\eta}{\xi_2}\mbox{ uniformly in }j\Big)\geq 1-\exp(-c_1\log p)$, for some positive constant $c_1$ and all $n$ sufficiently large. Note that $||\vw_j||_1 = ||\vd_j(\vbetah,\eta)-\vd_{0j}||_1$ and $||\vw_j||_2=||\vd_j(\vbetah,\eta)-\vd_{0j}||_2$, (1) is proved.
		
		(2) Recall that $\tau^2_{0j} =\E\big\{[G^{(1)}(\vx_i^T\vbeta_0|\vbeta_0)]^2  \xw_{i,j}\vxw_{i,-1}^T\vphi_{0j}\big\}$. We have
		\begin{align*}
		\left|\tau^2_{0j}-\tau_j^{2}(\vbetah,\eta)\right| 
		\leq & \left|\frac{1}{n}\sum_{i=1}^n\left\{ [G^{(1)}(\vx_i^T\vbeta_0|\vbeta_0)]^2-[ \widehat{G}^{(1)}(\vx_i^T\vbetah|\vbetah)]^2 \right\} \xw_{i,j}\vxw_{i,-1}^T\vphi_{0j}\right|\\
		&+ \left|\frac{1}{n}\sum_{i=1}^n[\widehat{G}^{(1)}(\vx_i^T\vbetah|\vbetah)]^2[\xh_{i,j}  \vxh_{i,-1}^T-\xw_{i,j}  \vxw_{i,-1}^T]\vphi_{0j} \right|\\
		&+ \left|\frac{1}{n}\sum_{i=1}^n[\widehat{G}^{(1)}(\vx_i^T\vbetah|\vbetah)]^2\xh_{i,j}  \vxh_{i,-1}^T\vw_j \right|\\
		&+ \left|\frac{1}{n}\sum_{i=1}^n [G^{(1)}(\vx_i^T\vbeta_0|\vbeta_0)]^2 \xw_{i,j}\vxw_{i,-1}^T\vphi_{0j} -\tau^2_{0j} \right| \\
		\triangleq& |I_{j1}|+|I_{j2}|+ |I_{j3}|+ |I_{j4}|,
		\end{align*}
		where the definition of $I_{jk}$ is clear from the context, and $\vw_j = \vphi_j(\vbetah, \eta)-\vphi_{0j}$. Lemma~\ref{lem:G1xx} implies that there exist some positive constants $c_0$, $c_1$, such that for all $n$ sufficiently large,
		\begin{align*}
		&P\left( |I_{j1}|\geq c_0\sqrt{s_j}\eta\mbox{ uniformly in }j\right) \\
		\leq &P\left( \Big|\Big|\frac{1}{n} \sum_{i=1}^n \big\{[\widehat{G}^{(1)}(\vx_i^T\vbetah|\vbetah)]^2 - [G^{(1)}(\vx_i^T\vbeta_0|\vbeta_0)]^2\big\} \vxw_{i,-1}\vxw_{i,-1}^T\Big|\Big|_\infty \geq c_0\eta\right)\\
		\leq &\exp(-c_1\log p).
		\end{align*}  
		Lemma~\ref{lem:Ex_err} implies that 
		$P(|I_{j2}|\geq c_0\sqrt{s_j}\eta\mbox{ uniformly in }j) \leq \exp(-c_1\log p)$, for some positive constants $c_0$, $c_1$ and all $n$ sufficiently large,
		since $\sqrt{s}h^2\leq d_0h^2\sqrt{nh^5} = d_0h\sqrt{nh^7}=o(\eta)$, for some positive constant $d_0$ by the assumptions of Theorem~\ref{Lasso_error} and Lemma~\ref{dbound}. Lemma~\ref{G1func} and Assumption~\ref{K4}-(a) imply $ \max_{1\leq i\leq n}[\widehat{G}^{(1)}(\vx_i^T\vbetah|\vbetah)]^2\leq c_0(b+h)^2\leq 2c_0b^2$ with probability at least $1-\exp[-c_1\log(p\vee n)]$, for positive constants $c_0$, $c_1$, $b$, and all $n$ sufficiently large. Then there are some universal positive constants $c_0$ and $c_1$ such that for  all $n$ sufficiently large,
		\begin{align*}
		|I_{j3}|\leq& \max_{1\leq i\leq n}[\widehat{G}^{(1)}(\vx_i^T\vbetah|\vbetah)]^2*  \frac{1}{n}\sum_{i=1}^n |\xw_{i,j}  \vxw_{i,-1}^T\vw_j| +  \frac{1}{n}\sum_{i=1}^n [\widehat{G}^{(1)}(\vx_i^T\vbetah|\vbetah)]^2|(\xh_{i,j}  \vxh_{i,-1}^T-\xw_{i,j}  \vxw_{i,-1}^T)\vw_j|\\
		\leq&  c_0 \sqrt{ \frac{1}{n}\sum_{i=1}^n\xw_{i,j}^2}\sqrt{  \frac{1}{n}\sum_{i=1}^n (\vxw_{i,-1}^T\vw_j)^2} + c_0 ||\vw_j||_1* \frac{1}{n}\sum_{i=1}^n[\widehat{G}^{(1)}(\vx_i^T\vbetah|\vbetah)]^2||\vxh_{i,-1} \vxh_{i,-1}^T-\vxw_{i,-1} \vxw_{i,-1}^T||_\infty\\
		\leq& c_0 \sqrt{ \frac{1}{n}\sum_{i=1}^n\xw_{i,j}^2}\sqrt{ \vw_j^T\Big( \frac{1}{n}\sum_{i=1}^n\vxw_{i,-1}\vxw_{i,-1}^T\Big)\vw_j}\\
		&+ c_0 \eta\sqrt{\widetilde{s}} * \frac{1}{n}\sum_{i=1}^n[\widehat{G}^{(1)}(\vx_i^T\vbetah|\vbetah)]^2||\vxh_{i,-1} \vxh_{i,-1}^T-\vxw_{i,-1}\vxw_{i,-1}^T||_\infty,
		\end{align*} 
		uniformly in $j$ with probability at least $1-\exp[-c_1\log (p\vee n)]$, applying the bound of $||\vw_j||_1$ derived in the proof of (1).  Lemma~\ref{lem14NCL} and Lemma~\ref{lem15NCL} imply that for some universal positive constants $c_0$, $c_1$, and all  $n$ sufficiently large,
		\begin{align*}
		P\left( \max_{2\leq j\leq p} \frac{1}{n}\sum_{i=1}^n\xw_{i,j}^2\geq \max_{2\leq j\leq p}\E (\xw_{i,j}^2)+c_0\sigma_x^2\sqrt{\frac{\log p}{n}}\right)&\leq \exp(-c_1\log p),\\
		P\left( \Big|\Big|\frac{1}{n}\sum_{i=1}^n\vxw_{i,-1}\vxw_{i,-1}^T-\E(\vxw_{i,-1}\vxw_{i,-1}^T)\Big|\Big|_\infty\geq c_0 \sigma_x^2\sqrt{\frac{\log p}{n}}\right)& \leq\exp(-c_1\log p),\\
		P\bigg( \big|\vw_j^T\Big( \frac{1}{n}\sum_{i=1}^n\vxw_{i,-1}\vxw_{i,-1}^T-\E(\vxw_{i,-1}\vxw_{i,-1}^T)\Big)\vw_j\big| &\\
		\geq c_0 \sigma_x^2||\vw_j||_2^2\sqrt{\frac{\log p}{n}},\mbox{ uniformly in }j\bigg) &\leq  \exp(-c_1\log p),
		\end{align*} 
		with $\vw_j = \vphi_j(\vbetah, \eta)-\vphi_{0j}$. Since $\xw_{i,j}$ is sub-Gaussian with variance proxy at most $2\sigma_x^2$ uniformly in $j$ by Lemma~\ref{lem:subg}, we have $\max_{2\leq j\leq p}\E (\xw_{i,j}^2)\leq 2\sigma_x^2$. Note that $\E\left[\Cov(\vx_{-1}|\vx^T\vbeta_0)\right] = \E(\vxw_{-1}\vxw_{-1}^T)$. Assumption~\ref{A2}-(a) thus implies that $\xi_1$ is the largest eigenvalue of $\E(\vxw_{-1}\vxw_{-1}^T)$. The results in part (1) of the lemma indicate that there exist some positive constants $c_0$, $c_1$, such that for all $n$ sufficiently large,
		$$\vw_j^T\Big( \frac{1}{n}\sum_{i=1}^n\vxw_{i,-1}\vxw_{i,-1}^T\Big)\vw_j\leq 2 \vw_j^T\E(\vxw_{-1}\vxw_{-1}^T)\vw_j\leq 2\xi_1  ||\vw_j||_2^2\leq c_0\xi_1s_j\eta^2 \mbox{	uniformly in }j,$$  
		with probability at least  $1-\exp(-c_1\log p)$, where the first inequality applies Lemma~\ref{lem14NCL}. Lemma~\ref{lem:Ex_err} implies that there exist some positive constants $c_0$, $c_1$, such that for all $n$ sufficiently large, $$\frac{1}{n}\sum_{i=1}^n[\widehat{G}^{(1)}(\vx_i^T\vbetah|\vbetah)]^2||\vxh_{i,-1} \vxh_{i,-1}^T-\vxw_{i,-1}  \vxw_{i,-1}^T||_\infty\leq c_0\sqrt{s}h^2,$$ with probability at least  $1-\exp(-c_1\log p)$. The assumptions of Theorem~\ref{Lasso_error} and Lemma~\ref{dbound} imply that  $\sqrt{s\widetilde{s}}h^2\eta = \eta\sqrt{sh^3} \sqrt{\widetilde{s}h}\leq d_0\eta \sqrt{nh^8}=o(\eta)$, for some positive constant $d_0$, since $nh^6=O(1)$. 
		It follows that $ |I_{j3}|\leq c_0\sqrt{s_j}\eta$ uniformly in $j$ with probability at least  $1-\exp(-c_1\log p)$ for universal positive constants $c_0$, $c_1$, and all $n$ sufficiently large,
		
		To uniformly bound $|I_{4j}|$,  as in the proof of part (1) of the lemma, we observe that $G^{(1)}(\vx_i^T\vbeta_0|\vbeta_0) \xw_{i,j}$ and $G^{(1)}(\vx_i^T\vbeta_0)*\vxw_{i,-1}^T\vphi_{0j}$ are both sub-Gaussian with variance proxy at most $c_1b^2\sigma_x^2$. There exist some positive constants $c_0$, $c$, such that for all $n$ sufficiently large,
		\begin{align*}
		P\left( |I_{j4}|\geq c_0\sqrt{s_j}\eta\mbox{ uniformly in }j\right)\leq &\sum_{j=2}^pP\left( \Big|\frac{1}{n}\sum_{i=1}^n [G^{(1)}(\vx_i^T\vbeta_0|\vbeta_0)]^2 \xw_{i,j}\vxw_{i,-1}^T\vphi_{0j} -\tau^2_{0j} \Big| \geq c_0\sqrt{s_j}\eta \right)\\\leq &p\exp(-cn\eta^2).
		\end{align*} 
		Since $\log p\leq d_0 nh^5= o(n\eta^2)$ for some positive constant $d_0$, there exist some positive constants $c_0$, $c_1$, such that for all $n$ sufficiently large,
		\begin{align*}
		P\Big( |\tau^2_{0j}-\tau_j^{2}(\vbetah,\eta)| \leq c_0\sqrt{s_j}\eta\mbox{ uniformly in }j\Big)\leq  \exp(-c_1\log p).
		\end{align*} 
		The first result of part (2) of Lemma~\ref{dbound} is proved.  Note that
		\begin{align*}
		&P\Big( |\tau^{-2}_{0j}-\tau_j^{-2}(\vbetah,\eta)| \geq c_0\sqrt{s_j}\eta\mbox{ uniformly in }j\Big)\\
		=&P\Big( \tau^{-2}_{0j}*\tau_j^{-2}(\vbetah,\eta)|\tau^{2}_{0j}-\tau_j^{2}(\vbetah,\eta)| \geq c_0\sqrt{s_j}\eta\mbox{ uniformly in }j\Big)\\
		\leq &	P\Big( |\tau^2_{0j}-\tau_j^{2}(\vbetah,\eta)| \leq c_0\xi_2^{2}\sqrt{s_j}\eta/2 \mbox{ uniformly in }j\Big) + 	P\Big(\max_{2\leq j\leq p}\tau^{-2}_{0j}*\tau_j^{-2}(\vbetah,\eta) \leq 2\xi_2^{-2}\Big)\\
		\leq& \exp(-c_1\log p) + 	P\Big(\max_{2\leq j\leq p}\tau^{-2}_{0j}*\tau_j^{-2}(\vbetah,\eta) \leq 2\xi_2^{-2}\Big),
		\end{align*} 
		for some positive constants $c_0$ $c_1$, and all $n$ sufficiently large. Lemma~\ref{lem:thetaj} implies that  $\xi_2\leq \tau^{2}_{0j}\leq b^2\xi_1$, uniformly in $j$.  By the first result of part (2) of the lemma, we know that $\xi_2/2\leq\tau_j^{2}(\vbetah,\eta)\leq b^2\xi_1/2$ with probability at least $1-\exp(-c_1\log p)$, for some positive constant $c_1$ and all $n$ sufficiently large. Then we have $P\Big(\max_{2\leq j\leq p}\tau^{-2}_{0j}*\tau_j^{-2}(\vbetah,\eta) \leq 2\xi_2^{-2}\Big)\leq \exp(-c_1\log p)$ for some positive constant $c_1$ and all $n$ sufficiently large. Hence the second result of part (2) of Lemma~\ref{dbound} is proved.
		
		(3) Observe that  
		\begin{align*}
		||\vtheta_j(\vbetah,\eta)-\vtheta_{j} ||_1 
		= &	 ||\tau_j^{-2}(\vbetah,\eta)\vphi_j(\vbetah,\eta) - \tau_{0j}^{-2} \vphi_{0j}||_1 
		\leq  \tau_j^{-2}(\vbetah,\eta) ||\vw_j||_1 + |\tau_j^{-2}(\vbetah,\eta)  - \tau_{0j}^{-2}|* || \vphi_{0j}||_1 ,
		\end{align*}  
		where $\vw_j = \vphi_j(\vbetah, \eta)-\vphi_{0j}$.
		%Conditional on the event that results in (1) and (2) all hold, then
		Lemma~\ref{lem:thetaj} implies  
		$ || \vphi_{0j}||_1 =  \tau_{0j}^{2}|| \vtheta_{j}||_1\leq \sqrt{s_j}\tau_{0j}^{2}|| \vtheta_{j}||_2\leq b\sqrt{s_j}\xi_1\xi_2^{-1}$ uniformly in $j$, and $\max_{2\leq j\leq p}\tau_j^{-2}(\vbetah,\eta)\leq \max_{2\leq j\leq p}\tau_{0j}^{-2}+c_0\eta \leq 2\xi_2^{-2}$, with probability at least $1-\exp(-c_1\log p)$, for some positive constants $c_0$, $c_1$ and all $n$ sufficiently large. Results in (1) and (2) imply that there exist some positive constants $d_0$, $d_1$, $d_2$, such that for all $n$ sufficiently large,
		\begin{align*}
		||\vtheta_j(\vbetah,\eta)-\vtheta_{j} ||_1 
		&\leq 	d_1s_j\eta\xi_2^{-2} + d_1s_j\eta*b^2\xi_1\xi_2^{-1}\leq d_0s_j\eta \mbox{ uniformly in }j
		\end{align*} 
		with probability at least $1-\exp(-d_2\log p)$.
		Similar proofs can be applied for the uniform bound of $ ||\vtheta_j(\vbetah,\eta)-\vtheta_{j} ||_2$.
	\end{proof}

	\begin{proof}[Proof of Theorem~\ref{desparse}] 
		Recall that 
		\begin{align*}
		\vS_n(\vbetah,\widehat{G},\widehat{\E})	&=-n^{-1}\sum_{i=1}^n\Big[\widetilde{\epsilon}_i+G(\vx_i^T\vbeta_0|\vbeta_0)-\widehat{G}(\vx_i^T\vbetah|\vbetah)\Big]\widehat{G}^{(1)}(\vx_i^T\vbetah|\vbetah)\vxh_{i,-1},\\
		\vS_n(\vbeta_0,G, \E)	&=-n^{-1}\sum_{i=1}^n\widetilde{\epsilon}_iG^{(1)}(\vx_i^T\vbeta_0|\vbeta_0)\vxw_{i,-1},
		\end{align*}
		where $\vxw_{i,-1} = \vx_{i,-1}-\E(\vx_{i,-1}|\vx_i^T\vbeta_0)$, $\vxh_{i,-1} = \vx_{i,-1}-\widehat{\E}(\vx_{i,-1}|\vx_i^T\vbetah)$, and $G^{(1)}(\vx_i^T\vbeta_0|\vbeta_0) = f'_0(\vx_i^T\vbeta_0)$. Recall that the debiased estimator $\vbetaw_{-1}=\vbetah_{-1}-\vThetah\vS_n(\vbetah,\widehat{G},\widehat{\E})$. We have
		\begin{align*}
		&\sqrt{n}(\vbetaw_{-1}-\vbeta_{0,-1}) \\
		= &\sqrt{n} (\vbetah_{-1}-\vbeta_{0,-1}) -\sqrt{n}\vThetah^T \vS_n(\vbetah,\widehat{G},\widehat{\E}) \\
		=&-\sqrt{n}\vThetah^T\vS_n(\vbeta_0,G, \E) + \sqrt{n} (\vbetah_{-1}-\vbeta_{0,-1}) -\sqrt{n}\vThetah^T\big[\vS_n(\vbetah,\widehat{G},\widehat{\E}) -\vS_n(\vbeta_0,G, \E) \big]\\
		=&-\sqrt{n}\vThetah^T\vS_n(\vbeta_0,G, \E) + \sqrt{n}(\vI_{p-1}- \vThetah^T\vJ_1) (\vbetah-\vbeta_0) \\
		&-\sqrt{n}\vThetah^T\big[\vS_n(\vbetah,\widehat{G},\widehat{\E}) -\vS_n(\vbeta_0,G, \E)-  \vJ_1(\vbetah-\vbeta_0)  \big]\\ 
		\triangleq& \vA_{n1} +\vA_{n2} +\vA_{n3},
		\end{align*}
		where the definition of $\vA_{ni}$ is clear from the context, $\vI_{p-1}$ is the $(p-1)-$dimensional identity matrix, $\vJ_1 =  n^{-1}\sum_{i=1}^{n}  \big[\widehat{G}^{(1)}(\vx_i^T\vbetah|\vbetah)\big]^2\vxh_{i,-1}\vxh_{i,-1}^T$ is the leading term in the approximation to $\nabla \vS_n(\vbeta_0,G, \E) $. Let $\Delta_{n,p} =  s h^3 \sqrt{n} +  \widetilde{s}h \sqrt{\log  p}$. To prove the theorem, we will verify:\\
		(1) The $j^{th}$ component of $\vA_{n1}$, denoted as $ \vA_{n1j} $, is approximately normal, for any $2\leq j \leq p$;\\
		(2) $P(|| \vA_{n2} ||_\infty\geq c_0\Delta_{n,p})\leq \exp(-c_1\log p)$;\\
		(3) $P(|| \vA_{n3} ||_\infty\geq c_0\Delta_{n,p})\leq \exp(-c_1\log p)$;\\
		for some positive constants $c_0$ and $c_1$, and all $n$ sufficiently large.
		%$\Delta_{n,p} =  s h^3 \sqrt{n} +  \sqrt{h\log (p\vee n)}+ h [\log (p\vee n) ]^{3/2}+  \widetilde{s}h \sqrt{\log  p}$.
		
		To prove (1), we observe that $-\sqrt{n} \ve_{j-1}^T\vTheta^T\vS_n(\vbeta_0,G, \E) \rad N(\vnull, \ve_{j-1}^T\vTheta^T\vLam\vTheta \ve_{j-1})$ by the central limit theorem, with $\vLam = \E \Big\{\big[\widetilde{\epsilon}_iG^{(1)}(\vx_i^T\vbeta_0|\vbeta_0)\big]^2 \vxw_{i,-1}\vxw_{i,-1}^T\Big\}$, and $\vTheta\ve_{j-1} = \vtheta_j$, where $\ve_{j-1}$ is the $(j-1)^{th}$ column of the identity matrix $\vI_{p-1}$. It suffices to prove that  
		$$P\left(\max_{2\leq j \leq p}\Big|\sqrt{n}(\vthetah_j-\vtheta_j)^T\vS_n(\vbeta_0,G, \E)  \Big| \geq c_0\Delta_{n,p}\right)\leq \exp(-c_1\log p),$$
		for some positive constants $c_0$, $c_1$, and all $n$ sufficiently large.
		
		Note that $2(2A_i-1)$ is a Rademacher sequence and independent of $(\epsilon_i,\vx_i)$. %Given $\{(\epsilon_i,\vx_i)\}_{i=1}^n$, f
		We know that
		$$\sqrt{n}\vS_n(\vbeta_0,G, \E) = n^{-1/2}\sum_{i=1}^{n}2(2A_i-1)[\ep_i+g(\vx_i)] G^{(1)}(\vx_i^T\vbeta_0|\vbeta_0) \vxw_{i,-1},$$
		has mean zero; $2(2A_i-1)[\ep_i+g(\vx_i)] G^{(1)}(\vx_i^T\vbeta_0|\vbeta_0) $ and $\vxw_{i,-1}$ are both sub-Gaussian. Note that
		\begin{align*}
		\E\left[\vS_n(\vbeta_0,G, \E) \right]=&\E_{\vx_i^T\vbeta_0}\left[\vS_n(\vbeta_0,G, \E)|\vx_i^T\vbeta_0 \right]\\
		=&\E_{\vx_i^T\vbeta_0}\left\{2G^{(1)}(\vx_i^T\vbeta_0|\vbeta_0) \E\Big\{(2A_i-1)[\ep_i+g(\vx_i)] \vxw_{i,-1} \big|\vx_i^T\vbeta_0\Big\}\right\}\\
		=&\vnull_{p-1},
		\end{align*}
		where $\vnull_{p-1}$ is  a $(p-1)-$dimensional vector with all entries $0$. Hence		Lemma~\ref{lem14NCL} implies that
		$$P\Big(\big|\big|\sqrt{n}\vS_n(\vbeta_0,G, \E) \big|\big|_\infty\geq c_0  \sqrt{ \log p}\Big)\leq \exp(-c_1\log p),$$
		for some positive constants $c_0$, $c_1$ and all $n$ sufficiently large. Hence, according to Lemma~\ref{dbound}, there exist some positive constants $c_0$, $c_1$, $c_2$, such that for all $n$ sufficiently large,
		\begin{align*}
		&P\left(\max_{2\leq j \leq p}\Big|\sqrt{n}(\vthetah_j-\vtheta_j)^T\vS_n(\vbeta_0,G, \E) \Big| \geq c_0c_2\widetilde{s}h \sqrt{\log p}\right)\\
		\leq& P\Big(\big|\big|\sqrt{n}\vS_n(\vbeta_0,G, \E) \big|\big|_\infty\geq c_0  \sqrt{ \log p}\Big) +  \sum_{j=2}^pP \Big(||\vthetah_j-\vtheta_j||_1\geq c_2\widetilde{s}h  \Big)\\
		\leq& \exp(-c_1\log p),
		\end{align*}
		which completes the proof for (1).
		
		To prove (2), recall the definition of $\vthetah_j$ in (\ref{thetaj_def}) of the main paper.   Lemma~\ref{lem:thetaj} and Lemma~\ref{dbound} imply  that there exist some positive constants $c_0$, $c_1$, $c_2$, such that for all $n$ sufficiently large,
		\begin{align*}
		&\Big|\Big|n^{-1}\sum_{i=1}^{n}  [\widehat{G}^{(1)}(\vx_i^T\vbetah|\vbetah)]^2\vxh_i^T\vthetah_j\vxh_{i,-j*}\Big|\Big|_\infty\\
		\leq& \Big|\Big|n^{-1}\sum_{i=1}^{n}  [\widehat{G}^{(1)}(\vx_i^T\vbetah|\vbetah)]^2\vxh_i^T\vphi_j(\vbetah,\eta)\vxh_{i,-j*}\Big|\Big|_\infty \widehat{\tau}_j^{-2} \\
		\leq &\eta \widehat{\tau}_j^{-2}\leq  \eta \big(\tau_j^{-2}+c_0\sqrt{s_j}\eta\big)= c_1\eta,
		\end{align*}
		uniformly in $j$, with probability at least $1-\exp(-c_2\log p)$. By (\ref{tauj_def}) and (\ref{thetaj_def}) of the main paper, we have
		$\Big| 1-n^{-1}\sum_{i=1}^{n}  [\widehat{G}^{(1)}(\vx_i^T\vbetah|\vbetah)]^2\vxh_{i,-1}^T\vthetah_j\xh_{i,j}\Big| = |1-\widehat{\tau}_j^{-2}*\widehat{\tau}_j^{2}| = 0$. Hence  $||\vI_{p-1}-\vThetah^T \vJ_1||_\infty \leq d_0\eta$ with  probability at least $1-\exp(-d_1\log p)$, for some positive constants $d_0$, $d_1$ and all $n$ sufficiently large. Therefore, there exist some positive constants $c_0\geq d_0$, $c_1$, $c_2$, such that for all $n$ sufficiently large,
		\begin{align*}
		P(|| \vA_{n2} ||_\infty\geq c_0c_2\sqrt{n}s\lambda\eta)
		\leq &P(||\vI_{p-1}-\vThetah^T \vJ_1 ||_\infty\geq c_0 \eta)+P(  ||\vbetah-\vbeta_0||_1\geq c_2s\lambda)\\
		\leq& \exp(-c_1\log p),
		\end{align*} 
		where the second inequality applies the above result and Theorem~\ref{Lasso_error}. By the assumptions of Theorem~\ref{desparse}, we have $\sqrt{n}s\lambda\eta=d_0sh^3\sqrt{n}$ for some positive constant $d_0$. It completes the proof for (2).
		
		To prove (3), we have
		\begin{align*}
		&\vA_{n3}\\
		=&-\sqrt{n}\vThetah^T\big[\vS_n(\vbetah,\widehat{G},\widehat{\E}) -\vS_n(\vbeta_0,G, \E)-  \vJ_1(\vbetah-\vbeta_0)  \big]\\
		=& n^{-1/2} \vThetah^T \sum_{i=1}^{n}\widetilde{\epsilon}_i G^{(1)}(\vx_i^T\vbeta_0|\vbeta_0)  (\vxh_{i,-1}-\vxw_{i,-1})\\ 
		&+n^{-1/2} \vThetah^T\sum_{i=1}^{n}\widetilde{\epsilon}_i \big[\widehat{G}^{(1)}(\vx_i^T\vbetah|\vbetah) -G^{(1)}(\vx_i^T\vbeta_0|\vbeta_0)\big]\vxh_{i,-1}\\
		&+n^{-1/2} \vThetah^T \sum_{i=1}^{n}\big[G(\vx_{i}^{T} \vbeta_0|\vbeta_0)-\widehat{G}(\vx_{i}^{T} \vbetah|\vbetah) - \widehat{G}^{(1)}(\vx_i^T\vbetah|\vbetah) \vxh_{i,-1}^T(\vbeta_{0,-1}-\vbetah_{-1})\big] \widehat{G}^{(1)}(\vx_i^T\vbetah|\vbetah)\vxh_{i,-1} \\
		\triangleq&\sum_{k=1}^3\vA_{n3k},
		\end{align*}
		where the definition of $\vA_{n3k}$ ($k=1,2,3$) 
		is clear from the context. We observe that
		\begin{align*}
		&\vA_{n31}\\
		=&n^{-1/2} \vTheta^T \sum_{i=1}^{n}\widetilde{\epsilon}_i G^{(1)}(\vx_i^T\vbeta_0|\vbeta_0)  (\vxh_{i,-1}-\vxw_{i,-1}) +n^{-1/2} (\vThetah-\vTheta)^T \sum_{i=1}^{n}\widetilde{\epsilon}_i G^{(1)}(\vx_i^T\vbeta_0|\vbeta_0)  (\vxh_{i,-1}-\vxw_{i,-1})\\ 
		\triangleq&\vA_{n311}+\vA_{n312},
		\end{align*}
		where the definition of $\vA_{n311}$ and $\vA_{n312}$ is clear from the context.
		%Note that $\vbetah\in\bbB_1$ with probability at least $1-\exp(-c_1\log p)$, and Lemma~\ref{Gbetafunc} implies that $\widehat{E}(\vx|\vx^T\vbetah)\in\bbE$ with the same probability bound. Then 	
		Lemma~\ref{lem:Ehat_E} implies that $P(||\vA_{n311}||_\infty\geq c_0h\sqrt{s\log(p\vee  n)})\leq\exp(-c_1\log p)$, for some positive constants $c_0$, $c_1$ and all $n$ sufficiently large. Furthermore, 
		\begin{align*}
		&P\left(||\vA_{n312}||_\infty\geq c_0\widetilde{s}\eta h\sqrt{s\log(p\vee  n)}\right)\\
		\leq &P\left(\Big|\Big|n^{-1/2}  \sum_{i=1}^{n}\widetilde{\epsilon}_i G^{(1)}(\vx_i^T\vbeta_0|\vbeta_0)  (\vxh_{i,-1}-\vxw_{i,-1})\Big|\Big|_\infty\geq c_0 h\sqrt{s\log(p\vee  n)}\right) +P\left(\max_{2\leq j \leq p}||\vthetah_j-\vtheta_j||_1\geq \widetilde{s}\eta\right)\\
		\leq&\exp(-c_1\log p),
		\end{align*}
		for some positive constants $c_0$ and $c_1$, and all $n$ sufficiently large. In the above, the first inequality applies H\"older's inequality; the second inequality
		applies lemma~\ref{lem:Ehat_E} and Lemma~\ref{dbound}. The assumptions of Theorem~\ref{Lasso_error} and Lemma~\ref{dbound} imply that $h\sqrt{s\log(p\vee  n)} \leq d_0\sqrt{nh^7}\leq d_0  sh^3\sqrt{n}$, and $\widetilde{s}\eta \leq d_1$ for some positive constants $d_0$ and $d_1$.
		We have $P(||\vA_{n31}||_\infty\geq c_0\  sh^3\sqrt{n})\leq\exp(-c_1\log p)$, for some positive constants $c_0$, $c_1$, and all $n$ sufficiently large.
		
		To verify the bound for $\vA_{n32}$, we rewrite it as
		\begin{align*}
		\vA_{n32}=&n^{-1/2} \vTheta^T\sum_{i=1}^{n}\widetilde{\epsilon}_i \big[\widehat{G}^{(1)}(\vx_i^T\vbetah|\vbetah) -G^{(1)}(\vx_i^T\vbeta_0|\vbeta_0)\big]\vxw_{i,-1}\\
		&+n^{-1/2} (\vThetah-\vTheta)^T\sum_{i=1}^{n}\widetilde{\epsilon}_i \big[\widehat{G}^{(1)}(\vx_i^T\vbetah|\vbetah) -G^{(1)}(\vx_i^T\vbeta_0|\vbeta_0)\big]\vxw_{i,-1}\\
		&+n^{-1/2} \vThetah^T\sum_{i=1}^{n}\widetilde{\epsilon}_i \big[\widehat{G}^{(1)}(\vx_i^T\vbetah|\vbetah) -G^{(1)}(\vx_i^T\vbeta_0|\vbeta_0)\big](\vxh_{i,-1}-\vxw_{i,-1})\\
		\triangleq&\vA_{n321}+\vA_{n322}+\vA_{n323},
		\end{align*}
		with the definition of $\vA_{n32k}$ ($k=1,2,3$) is clear
		from the context. Similarly as the proof for $\vA_{n311}$ and $\vA_{n312}$, %we know that $\widehat{G}^{(1)}(\vx^T\vbetah|\vbetah)\in\bbM$ with probability at least $1-\exp(-c_1\log p)$. Then
		Lemma~\ref{lem:ghat1_g1} implies that $P\left(||\vA_{n321}||_\infty\geq c_0\left[h^2\log(p\vee n)\right]^{1/4}\right)\leq\exp(-c_1\log p)$, and $P\left(||\vA_{n322}||_\infty\geq c_0\widetilde{s}\eta\left[h^2\log(p\vee n)\right]^{1/4}\right)\leq\exp(-c_1\log p)$, for some positive constants $c_0$, $c_1$, and all $n$ sufficiently large. 
		%Note that the assumptions of Theorem~\ref{Lasso_error} and Lemma~\ref{dbound} imply that $\left[h^2 \log(p\vee  n)\right]^{1/4} \leq d_0\left(nh^7\right)^{1/4} \leq d_0  \sqrt{n}h^3 * (nh^5)^{-1/4}\leq  d_0  \sqrt{n}h^3$, and $\widetilde{s}\eta \leq d_1$ for some positive constants $d_0$ and $d_1$.
		%with $\gamma_1 = \sqrt{h\log (p\vee n)} + h [\log (p\vee n) ]^{3/2}$.
		For $\vA_{n323}$, Lemma~\ref{lem14NCL} implies that for sub-Gaussian random variables $\widetilde{\epsilon}_i$, we have
		$$P\left(\left| n^{-1}\sum_{i=1}^{n}\widetilde{\epsilon}_i^2-\E(\widetilde{\epsilon}_i^2)\right| \geq 4\sqrt{\sigma_\epsilon^2+M^2}\sqrt{\frac{\log p}{n}}\right)\leq \exp(-c\log p),$$
		for some positive constant $c$ and all $n$ sufficiently large. It implies that $P\Big( n^{-1}\sum_{i=1}^{n}|\widetilde{\epsilon}_i|^2 \geq c_0\Big)\leq \exp(-c_1\log p)$,	for some positive constants $c_0$, $c_1$, and all $n$ sufficiently large. In addition,
		\begin{align*}
		&P\Big( n^{-1}\sum_{i=1}^{n} |(\vxh_{i,-1}-\vxw_{i,-1})^T\vthetah_j|^2 \geq c_2sh^4 \Big)\\
		\leq &	P\Big( n^{-1}\sum_{i=1}^{n} |(\vxh_{i,-1}-\vxw_{i,-1})^T\vtheta_j|^2 \geq c_2sh^4/2 \Big)+P\Big( n^{-1}\sum_{i=1}^{n} |(\vxh_{i,-1}-\vxw_{i,-1})^T(\vthetah_j-\vtheta_j)|^2 \geq c_2sh^4 /2\Big)\\
		\leq &	\exp(-c_3\log p),+P\Big( \Big|\Big|n^{-1}\sum_{i=1}^{n} (\vxh_{i,-1}-\vxw_{i,-1})^T(\vxh_{i,-1}-\vxw_{i,-1}) \Big|\Big|_\infty ||\vthetah_j-\vtheta_j||_1^2 \geq c_2sh^4 /2\Big)\\
		\leq &\exp(-c_4\log p),
		\end{align*}
		for some positive constants $c_2$, $c_3$, $c_4$, and all $n$ sufficiently large.
		In the above, the second inequality applies Lemma~\ref{Efunc}; the last inequality applies Lemma~\ref{dbound} and Lemma~\ref{Efunc}. Hence There exist some positive constants $d_0$, $d_1$, $d_2$, $d_3$, such that for all $n$ sufficiently large,
		\begin{align*}
		&P\Big(||\vA_{n323}||_\infty\geq d_0d_2h^3\sqrt{sn}\Big)\\
		\leq &\sum_{j=2}^pP\left(\Big|n^{-1}\sum_{i=1}^{n}\widetilde{\epsilon}_i \big[\widehat{G}^{(1)}(\vx_i^T\vbetah|\vbetah) -G^{(1)}(\vx_i^T\vbeta_0|\vbeta_0)\big](\vxh_{i,-1}-\vxw_{i,-1})^T\vthetah_j\Big| \geq d_0d_2h^3\sqrt{s}\right) \\
		\leq&P\left(\max_{1\leq i\leq n}\big| \widehat{G}^{(1)}(\vx_i^T\vbetah|\vbetah) -G^{(1)}(\vx_i^T\vbeta_0|\vbeta_0)\big| \geq  d_2h\right)\\
		&+\sum_{j=2}^pP\left( n^{-1}\sum_{i=1}^{n}|\widetilde{\epsilon}_i| *|(\vxh_{i,-1}-\vxw_{i,-1})^T\vthetah_j| \geq d_0h^2\sqrt{s}\right)\\
		\leq&\exp(-d_1\log p) + P\Big( n^{-1}\sum_{i=1}^{n}|\widetilde{\epsilon}_i|^2 \geq d_0\Big) +\sum_{j=2}^pP\left( n^{-1}\sum_{i=1}^{n} |(\vxh_{i,-1}-\vxw_{i,-1})^T\vthetah_j|^2 \geq d_0sh^4 \right)\\
		\leq &\exp(-d_3\log p),
		\end{align*}
		according to Lemma~\ref{Gbetafunc}, where the second last inequality applies
		the Cauchy-Schwartz inequality.  %Now we obtain that $P(||\vA_{n323}||_\infty\geq d_0h^3\sqrt{sn})\leq\exp(-d_1\log p)$, for some positive constants $d_0$ and $d_1$, 	and all $n$ sufficiently large.
		Note that the assumptions of Theorem~\ref{Lasso_error} and Lemma~\ref{dbound} imply that $\left[h^2\log(p\vee n)\right]^{1/4} \leq \sqrt{h \log(p\vee n)}\leq d_2\sqrt{h* nh^5} = d_2h^3\sqrt{n}$, and $\widetilde{s}\eta\leq d_3$ for some positive constants $d_2$ and $d_3$.
		Hence there exist positive constants $c_0$, $c_1$ such that for all $n$ sufficiently large, we have
		$P\Big(||\vA_{n32}||_\infty\geq c_0 h^3\sqrt{sn}\Big)\leq\exp(-c_1\log p).$
		
		Finally, let's examine $\vA_{n33}$. Rewrite it as follows:
		\begin{align*}
		&\vA_{n33}\\
		=&n^{-1/2} \vThetah^T \sum_{i=1}^{n}\big[G(\vx_{i}^{T} \vbeta_0|\vbeta_0)-\widehat{G}(\vx_{i}^{T} \vbetah|\vbetah) - \widehat{G}^{(1)}(\vx_i^T\vbetah|\vbetah) \vxh_{i,-1}^T(\vbeta_{0,-1}-\vbetah_{-1})\big] \widehat{G}^{(1)}(\vx_i^T\vbetah|\vbetah)\vxh_{i,-1} \\
		=&n^{-1/2} \vTheta^T  \sum_{i=1}^{n}\big[G(\vx_{i}^{T} \vbeta_0|\vbeta_0)-G(\vx_{i}^{T} \vbetah|\vbetah) - G^{(1)}(\vx_i^T\vbetah|\vbetah) \vxh_{i,-1}^T(\vbeta_{0,-1}-\vbetah_{-1})\big] \widehat{G}^{(1)}(\vx_i^T\vbetah|\vbetah)\vxh_{i,-1}\\
		&+n^{-1/2} (\vThetah-\vTheta)^T  \sum_{i=1}^{n}\big[G(\vx_{i}^{T} \vbeta_0|\vbeta_0)-G(\vx_{i}^{T} \vbetah|\vbetah) - G^{(1)}(\vx_i^T\vbetah|\vbetah) \vxh_{i,-1}^T(\vbeta_{0,-1}-\vbetah_{-1})\big] \widehat{G}^{(1)}(\vx_i^T\vbetah|\vbetah)\vxh_{i,-1} \\ 
		&+n^{-1/2} \vTheta^T  \sum_{i=1}^{n}\big[G(\vx_{i}^{T} \vbetah|\vbetah) -\widehat{G}(\vx_{i}^{T} \vbetah|\vbetah) \big] \widehat{G}^{(1)}(\vx_i^T\vbetah|\vbetah)\vxh_{i,-1} \\
		&+n^{-1/2}  (\vThetah-\vTheta)^T  \sum_{i=1}^{n}\big[G(\vx_{i}^{T} \vbetah|\vbetah) -\widehat{G}(\vxh_{i}^{T} \vbetah|\vbetah) \big]\widehat{G}^{(1)}(\vx_i^T\vbetah|\vbetah)\vxh_{i,-1}  \\ 
		&+n^{-1/2} \vThetah ^T \sum_{i=1}^{n}\big[G^{(1)}(\vx_i^T\vbetah|\vbetah)- \widehat{G}^{(1)}(\vx_i^T\vbetah|\vbetah) \big]\widehat{G}^{(1)}(\vx_i^T\vbetah|\vbetah)\vxh_{i,-1}\vxh_{i,-1}^T(\vbeta_{0,-1}-\vbetah_{-1})\\ 
		\triangleq&\sum_{k=1}^5\vA_{n33k},
		\end{align*}
		where the definitions of $\vA_{n33k}$, $k=1,\cdots,5$ are clear from the context. %Similarly, $\widehat{G}(\vx^T\vbetah|\vbetah)\in\bbU$ with probability at least $1-\exp(-c_1\log p)$. 
		Lemma~\ref{lem:ghat_g} and Lemma~\ref{dbound}-(3) together imply that    $P(||\vA_{n331}||_\infty\geq c_0sh^3\sqrt{n})\leq\exp(-c_1\log p)$,   $P(||\vA_{n332}||_\infty\geq c_0\widetilde{s}\eta s h^3\sqrt{n})\leq\exp(-c_1\log p)$, $P(||\vA_{n333}||_\infty\geq c_0sh^3\sqrt{n})\leq\exp(-c_1\log p)$, and $P(||\vA_{n334}||_\infty\geq c_0\widetilde{s}\eta s h^3\sqrt{n})\leq\exp(-c_1\log p)$, for some positive constants $c_0$, $c_1$, and all $n$ sufficiently large. The assumptions of Lemma~\ref{dbound} imply that $\widetilde{s}\eta \leq d_0$ for some positive constant $d_0$. 
		%The assumptions of Theorem~\ref{Lasso_error} and Lemma~\ref{dbound} imply that $sh^2\sqrt{\log(p\vee n)} \leq c_0h^2*nh^5\leq c_0h^3\sqrt{n} *\sqrt{nh^8}\leq c_1h^3\sqrt{n}$, for some positive constants $c_0$ and $c_1$. 
		%Proof of Lemma~\ref{lem:Ex_err} also indicates that 
		To bound $||\vA_{n335}||_\infty$, we have
		\begin{align*}
		&P\Big(||\vA_{n335}||_\infty\geq c_0 \sqrt{n}hs\lambda \Big)\\
		\leq &P\left(\max_{2\leq j\leq p}\left|\left|n^{-1/2}  \sum_{i=1}^{n}   \widehat{G}^{(1)}(\vx_i^T\vbetah|\vbetah) \vxh_{i,-1}\vxh_{i,-1}^T\vthetah_j\right|\right|_\infty\geq c_0^{1/3} \sqrt{n} \right) \\
		&+P\Big(\max_{1\leq i \leq n}\big|G^{(1)}(\vx_i^T\vbetah|\vbetah)- \widehat{G}^{(1)}(\vx_i^T\vbetah|\vbetah)\big|\geq c_0^{1/3} h\Big)+P\Big(||\vbetah-\vbeta_0||_1\geq c_0^{1/3} s\lambda \Big)\\
		\leq&\exp(-c_1\log p),
		\end{align*} 
		for some positive constants $c_0$, $c_1$, and all $n$ sufficiently large. In the above, the first part of the second inequality applies the proof of Lemma~\ref{lem:Ex_err}; the second and third parts apply Lemma~\ref{G1func} and Theorem~\ref{Lasso_error}, respectively.
		The assumptions of Theorem~\ref{Lasso_error}   imply that  $\sqrt{n}hs\lambda \leq d_0  \sqrt{n}sh^3$, for some positive constant $d_0$.
		Combining the above results, we  conclude that 
		$$P\Big(|| \vA_{n3} ||_\infty\geq c_0\Delta_{n,p}\Big)\leq \exp(-c_1\log p),$$
		for some positive constants $c_0$, $c_1$, and all $n$ sufficiently large.
		Hence, the theorem is proved.
	\end{proof}

	\begin{proof}[Proof of Corollary~\ref{cor:marginal_conf}]
		Let $\sigma_j^2 = \E\big\{\big[\widetilde{\epsilon}_i   G^{(1)}(\vx_i^T\vbeta_0|\vbeta_0)\vxw_{i,-1}^T\vtheta_j\big]^2\big\}$, which is the  $(j-1)^{th}$ diagonal entry of $\vTheta^T\vLam\vTheta$. Recall that $\widetilde{\epsilon}_{i} = 2(2A_i-1)\big[\epsilon_i+g(\vx_i)\big]$, where $\epsilon_i$ is sub-Gaussian independent of $\vx_i$,  and $P\big(\max_{1\leq i\leq n}|g(\vx_i)|\leq M\big)=1$. Note that
		\begin{align*}
		\sigma_j^2 =&\E\big\{\ep_i^2 \big[G^{(1)}(\vx_i^T\vbeta_0|\vbeta_0)\vxw_{i,-1}^T\vtheta_j\big]^2\big\} + \E\big\{[g(\vx_i)]^2 \big[G^{(1)}(\vx_i^T\vbeta_0|\vbeta_0)\vxw_{i,-1}^T\vtheta_j\big]^2\big\}\\
		\leq &  b^2\E(\ep_i^2)\E[(\vxw_{i,-1}^T\vtheta_j)^2] + M^2b^2\E[(\vxw_{i,-1}^T\vtheta_j)^2]\\
		\leq & b^2(\sigma_\ep^2+M^2)\xi_1||\vtheta_j||_2^2,
		\end{align*}
		where $\max_{1\leq i \leq n}\big|G^{(1)}(\vx_i^T\vbeta_0|\vbeta_0)\big|\leq b$ by Assumption~\ref{A1}-(b), and $\lambda_{\max}(\E[\Cov(\vx_{-1}|\vx^T\vbeta_0)])\leq \xi_1$ by Assumption~\ref{A2}-(a). We have
		\begin{align*}
		\sigma_j^2 \geq \E\big\{\ep_i^2 \big[G^{(1)}(\vx_i^T\vbeta_0|\vbeta_0)\vxw_{i,-1}^T\vtheta_j\big]^2\big\}   \geq  \E(\ep_i^2)\xi_2\min_{2\leq j\leq p}||\vtheta_j||_2^2\geq \sigma_\epsilon^2\xi_2\min_{2\leq j\leq p}\tau_{0j}^{-2}\geq  \sigma_\epsilon^2\xi_2(b^2\xi_1)^{-2},
		\end{align*} 
		by the smallest eigenvalue condition for $\vOmega$ in Assumption~\ref{A2}-(a) and  Lemma~\ref{lem:thetaj}.
		Hence	we have 
		$$  \sigma_\epsilon^2\xi_2 (b^2\xi_1)^{-2}\leq \sigma_j^2 \leq (\sigma_\epsilon^2+M^2)b^2\xi_1||\vtheta_j||_2^2\leq (\sigma_\epsilon^2+M^2)b^2\xi_1\xi_2^{-2}.$$ 
		This suggests that $\max_{2\leq j\leq p} \sigma_j^{-1} \leq b^2\xi_1/( \sigma_\epsilon\sqrt{\xi_2})$.
		Since $\max_{2\leq j\leq p}\big|\widehat{\Sigma}_{jj}-\sigma_j^2 \big|=o_p(1)$ by  Lemma~\ref{cor_Sig}, we have $\max_{2\leq j\leq p}\big| \widehat{\Sigma}_{jj}^{-1/2}-\sigma_j^{-1}\big| =o_p(1) $. 
		
		Note that in the proof of Theorem~\ref{desparse}, we already showed that $\max_{2\leq j\leq p}\Big| \widetilde{\beta}_j-\beta_{0j}-\frac{1}{n}\sum_{i=1}^n\widetilde{\epsilon}_{i}  G^{(1)}(\vx_i^T\vbeta_0|\vbeta_0)\vxw_{i,-1}^T\vtheta_{j}\Big| = o_p(n^{-1/2})$. Hence  
		\begin{align*}
		&\max_{2\leq j\leq p}\Big| \widehat{\Sigma}_{j,j}^{-1/2}(\widetilde{\beta}_j-\beta_{0j})-\sigma_j^{-1}\frac{1}{n}\sum_{i=1}^n\widetilde{\epsilon}_{i}  G^{(1)}(\vx_i^T\vbeta_0|\vbeta_0)\vxw_{i,-1}^T\vtheta_{j}\Big| \\
		\leq& \max_{2\leq j\leq p}\sigma_j^{-1}*\max_{2\leq j\leq p}\Big|  \widetilde{\beta}_j-\beta_{0j}-\frac{1}{n}\sum_{i=1}^n\widetilde{\epsilon}_{i}  G^{(1)}(\vx_i^T\vbeta_0|\vbeta_0)\vxw_{i,-1}^T\vtheta_{j}\Big| \\
		&+\max_{2\leq j\leq p}\left| \widehat{\Sigma}_{jj}^{-1/2}-\sigma_j^{-1}\right|*\max_{2\leq j\leq p}\left|\widetilde{\beta}_j-\beta_{0j} \right| \\
		= &o_p(n^{-1/2})*O(1) + o_p(1)*O_p(n^{-1/2}) = o_p(n^{-1/2}).
		\end{align*}
		
		Applying the Berry-Esseen bound for CLT, there exists some universal constant $c_0>0$ such that  
		\begin{align*}
		&\max_{2\leq j\leq p} \sup_{\alpha\in(0,1)}\left|P\left(\sqrt{n}\left|\sigma_j^{-1} \frac{1}{n}\sum_{i=1}^n\widetilde{\epsilon}_{i}  G^{(1)}(\vx_i^T\vbeta_0|\vbeta_0)\vxw_{i,-1}^T\vtheta_{j}\right| \leq \Phi^{-1}(1-\alpha/2){\tiny }\right)-(1-\alpha)\right| \\
		\leq &\max_{2\leq j\leq p} \frac{c_0}{\sqrt{n}}\E\left[\left|\widetilde{\epsilon}_{i}  G^{(1)}(\vx_i^T\vbeta_0|\vbeta_0)\vxw_{i,-1}^T\vtheta_{j}\right|^3\right]\leq \max_{2\leq j\leq p} \frac{c_0b^3}{\sqrt{n}}\E\left(|\widetilde{\epsilon}_{i} \vxw_{i, -1}^T\vtheta_{j}|^3\right),
		\end{align*} 
		where $\Phi(\cdot)$ is the c.d.f of $N(0,1)$, and $\Phi^{-1}(\cdot)$ is the inverse function of $\Phi(\cdot)$. The above probability is  bounded by $\frac{c_1}{\sqrt{n}}\E\big(|\vxw_{i,-1}^T\vtheta_{j}|^3\big)$, where $c_1$ does not dependent on $n$, $p$ and $\vbeta_0$. Note that  $\vxw_{i,-1}^T\vtheta_{j}$ is sub-Gaussian with variance proxy $\sigma_x^2||\vtheta_j||_2^2$. The property of the sub-Gaussian distribution and Lemma~\ref{lem:thetaj} imply that $\max_{2\leq j\leq p} \E\big(|\vxw_{i,-1}^T\vtheta_{j}|^3\big)\leq c_2||\vtheta_j||_2^3\leq c_2 \xi_2^{-3}$, where $c_2>0$  does not depend  on $n$, $p$ and $\vbeta_0$. Then  for a universal constant $c>0$,
		\begin{align*}
		&\max_{2\leq j\leq p} \sup_{\alpha\in(0,1)}\left|P\left(\sqrt{n}\left|\sigma_j^{-1} \frac{1}{n}\sum_{i=1}^n\widetilde{\epsilon}_{i}  G^{(1)}(\vx_i^T\vbeta_0|\vbeta_0)\vxw_{i,-1}^T\vtheta_{j}\right| \leq \Phi^{-1}(1-\alpha/2){\tiny }\right)-(1-\alpha)\right|  \\
		\leq &\frac{c }{\sqrt{n}}=o(1).
		\end{align*} 
		Combining all the results above, we conclude the proof of the corollary.
	\end{proof}

	\begin{proof}[Proof of Theorem~\ref{multi_boots}]  
		Let $\widetilde{\delta}_j = n^{-1}\sum_{i=1}^n \widetilde{\epsilon}_i   G^{(1)}(\vx_i^T\vbeta_0|\vbeta_0)\vxw_{i,-1}^T\vtheta_{j}$, and $\vxi=(\xi_2,\cdots,\xi_p)^T$ be a multivariate mean zero Gaussian with covariance matrix $\vTheta^T\vLam\vTheta$. 
		
		Since $\widetilde{\epsilon}_i   G^{(1)}(\vx_i^T\vbeta_0|\vbeta_0)$ and $\vxw_{i,-1}^T\vtheta_j$ are both sub-Gaussian, Comment 2.2 in \citet{Chakraborty2014} implies that $ \widetilde{\epsilon}_i   G^{(1)}(\vx_i^T\vbeta_0|\vbeta_0)\vxw_{i,-1}^T\vtheta_j$ satisfies their condition (E.1) with $B_n = C_1$ for some universal constant $C_1$, which does not depend on $n$, $p$ and $\vbeta_0$. The order of $h$ implies that $\log p = o\big((\sqrt{n}h^3)^{-1}\big)$, and $(nh^5)^{-1}=o(1)$. Hence we can derive that 
		$\frac{[\log (pn)]^7}{n} = O\Big(\frac{[\log p]^7}{n} \Big) =  o\big((n^{9/2}h^{21})^{-1}\big) = o(n^{-0.3}).$
		Note that $\max_{2\leq j \leq p} |\widetilde{\delta}_{j}| = \max_{2\leq j \leq p}\left\{\widetilde{\delta}_{j}, - \widetilde{\delta}_{j}\right\}$, hence Corollary 2.1 in \citet{chernozhukov2013} indicates that 
		$\sup_{t\in\bbR}\big|P(\sqrt{n}\max_{2\leq j \leq p}|\widetilde{\delta}_{j}| \leq t)-P(\max_{2\leq j \leq p}|\xi_j| \leq t)\big|\leq \exp(-c_1\log n)$, for some universal constant $c_1>0$.
		
		Define the event $\mathbb{T}_n(\kG)=\left\{\sqrt{n}\left|\  \max_{j\in\kG}\left|\widetilde{\beta}_j-\beta_{0j}\right| -  \max_{j\in\kG}\left|\widetilde{\delta}_j\right|\   \right|> \Delta_{n,p}\right\}$. Note that Theorem~\ref{desparse} implies for some universal constant $c_1>0$, and all $n$ sufficiently large,
		\begin{align*}
		P\big(\mathbb{T}_n(\kG)\big) \leq P\left(\sqrt{n}  \max_{j\in\kG}\left|(\widetilde{\beta}_j-\beta_{0j})-\widetilde{\delta}_j\right|  > \Delta_{n,p} \right)  \leq P\Big( || \vDelta||_\infty \geq \Delta_{n,p}  \Big) \leq  \exp(-c_1\log p).
		\end{align*}  
		Applying Corollary 16 in \citet{Wasserman2014SteinS},  we have that for some universal positive  constants $C$, $c_1$, and all $n$ sufficiently large,
		\begin{align*}
		&P\left( \sqrt{n}\max_{j\in\kG}\left|\widetilde{\beta}_j-\beta_{0j}\right|\leq c^*_{1-\alpha}(\kG)  \right)\\
		\leq&  P\left(\max_{j\in\kG}\left|\widetilde{\delta}_j\right|  \leq c^*_{1-\alpha}(\kG) + \Delta_{n,p}\right) + P\big(\mathbb{T}_n(\kG)\big)\\
		\leq&  P\left(\max_{j\in\kG}\left|\widetilde{\delta}_j\right|  \leq c^*_{1-\alpha}(\kG) \right) + C\Delta_{n,p}\sqrt{1\vee\log(p/\Delta_{n,p})} +   \exp(-c_1\log p),
		\end{align*}  
		uniformly over $\alpha\in(0,1)$.
		Note that $\Delta_{n,p}\sqrt{ \log p} = o(1) $ implies that $\Delta_{n,p}\sqrt{1\vee\log(p/\Delta_{n,p})}=o(1)$. Hence it suffices to show $P\Big(\max_{j\in\kG}\left|\widetilde{\delta}_j\right| \leq c^*_{1-\alpha} \Big)\leq 1-\alpha+o(1)$ uniformly over $\alpha\in(0,1)$.
		
		Note that $\max_{j\in\kG}|\widetilde{\delta}_j|  = \max_{j\in\kG} \left\{\widetilde{\delta}_j,-\widetilde{\delta}_j\right\}$, and $\max_{j\in\kG}| \delta_j^*| = \max_{j\in\kG} \left\{\delta_j^*,-\delta_j^*\right\}$.  
		Observe that conditional on $\vw$, $ (\delta_1^*,\cdots,\delta_p^*,-\delta_1^*,\cdots,-\delta_p^*)^T$ is multivariate mean zero Gaussian with covariance matrix $\left(\begin{array}{cc}
		\vSigmah(\vbetah)	& -\vSigmah(\vbetah) \\ 
		-\vSigmah(\vbetah)	& \vSigmah(\vbetah) \end{array}\right) $. Then let $\Delta_0 = ||\vSigmah(\vbetah) - \vTheta^T\vLam\vTheta||_\infty$,  Gaussian comparison inequality suggests that for some positive constant $C$,  
		\begin{align*}
		P\left(\max_{j\in\kG}\left|\widetilde{\delta}_j\right|   \leq c^*_{1-\alpha}(\kG)\right)&\leq	P\left(\max_{j\in\kG}\left| \delta_j^*\right|   \leq c^*_{1-\alpha}(\kG)\ \big|\ \vw \right) + C\Delta_0^{1/3}\big[1\vee\log(p/\Delta_0)\big]^{2/3}\\
		&= 1-\alpha+ C\Delta_0^{1/3}\big[1\vee\log(p/\Delta_0)\big]^{2/3},
		\end{align*} 
		uniformly over $\alpha\in(0,1)$.
		Proof of Lemma~\ref{cor_Sig} implies that $\Delta_0\leq c_0\widetilde{s}^{1/2} h $ with probability
		at least $1-\exp(-c_1\log p)$, for universal positive constants $c_0$ and $c_1$. Hence, we can derive $\Delta_0 \log^2p=o_p(1)$. It implies that $\Delta_0^{1/3}\big[1\vee\log(p/\Delta_0)\big]^{2/3}=o(1)$  for all sufficiently large $n$. We obtain 
		$$\sup_{\alpha\in(0,1)}\left[P\Big( \sqrt{n}\max_{j\in\kG}\left|\widetilde{\beta}_j-\beta_{0j}\right|\leq c^*_{1-\alpha}(\kG) \Big)-(1-\alpha)\right]= o(1),$$ 
		for  all $n$ sufficiently large.
		Similarly, we can derive that $\sup_{\alpha\in(0,1)}\bigg[(1-\alpha) - P\Big( \sqrt{n}\max_{j\in\kG}|\widetilde{\beta}_j-\beta_{0j}|\leq c^*_{1-\alpha}(\kG) \Big)\bigg]=   o(1)$, for all $n$ sufficiently large. Note that all the universal constants do not depend on $n$, $p$ and $\vbeta_0$. We thus have 
		$$	\sup_{\vbeta_0\in\bbB_0:||\vbeta_0||_0\leq s}\sup_{\alpha\in(0,1)}\Big|P\Big( \sqrt{n}\max_{j\in\kG}|\widetilde{\beta}_j-\beta_{0j}|\leq c^*_{1-\alpha}(\kG)  \Big)-(1-\alpha)\Big| = o(1).$$
	\end{proof}	
	
	\section{Derivation of the Results in Section~\ref{sec:lemmas}} \label{sec:proof_append}
	\begin{proof}[Proof of Lemma~\ref{Lgrad}]
		By the model setup, we have:
		\begin{align*}
		& \vS_n(\vbeta_0,\widehat{G},\widehat{\E}) \\
		= &-n^{-1}\sum_{i=1}^n \left\{ \widetilde{\epsilon}_i + G(\vx_i^T\vbeta_0|\vbeta_0) - \widehat{G}(\vx_{i}^{T} \vbeta_0|\vbeta_0)\right\} \widehat{G}^{(1)}(\vx_{i}^T \vbeta_0|\vbeta_0) [\vx_{i,-1}-\widehat{\E}(\vx_{i,-1}|\vx_i^T\vbeta_0)]\\
		=& -n^{-1}\sum_{i=1}^n \widetilde{\epsilon}_i G^{(1)}(\vx_{i}^T \vbeta_0|\vbeta_0) [\vx_{i,-1}- \E(\vx_{i,-1}|\vx_i^T\vbeta_0)] \\
		&-n^{-1}\sum_{i=1}^n\widetilde{\epsilon}_i \left\{ \widehat{G}^{(1)}(\vx_{i}^T \vbeta_0|\vbeta_0) - G^{(1)}(\vx_{i}^T \vbeta_0|\vbeta_0)\right\} [\vx_{i,-1}-\E(\vx_{i,-1}|\vx_i^T\vbeta_0)]\\
		& -n^{-1}\sum_{i=1}^n \widetilde{\epsilon}_i \widehat{G}^{(1)}(\vx_{i}^T \vbeta_0|\vbeta_0) [ \E(\vx_{i,-1}|\vx_i^T\vbeta_0)-\widehat{\E}(\vx_{i,-1}|\vx_i^T\vbeta_0)] \\ 
		&  -n^{-1}\sum_{i=1}^n \left\{ G(\vx_i^T\vbeta_0|\vbeta_0) - \widehat{G}(\vx_{i}^{T} \vbeta_0|\vbeta_0)\right\} G^{(1)}(\vx_{i}^T \vbeta_0|\vbeta_0) [\vx_{i,-1}-\E(\vx_i|\vx_{i,-1}^T\vbeta_0)]\\
		&  -n^{-1}\sum_{i=1}^n \left\{ G(\vx_i^T\vbeta_0|\vbeta_0) - \widehat{G}(\vx_{i}^{T} \vbeta_0|\vbeta_0)\right\} \left\{ \widehat{G}^{(1)}(\vx_{i}^T \vbeta_0|\vbeta_0) - G^{(1)}(\vx_{i}^T \vbeta_0|\vbeta_0)\right\} [\vx_{i,-1}-\E(\vx_{i,-1}|\vx_i^T\vbeta_0)]\\
		&  -n^{-1}\sum_{i=1}^n \left\{ G(\vx_i^T\vbeta_0|\vbeta_0) - \widehat{G}(\vx_{i}^{T} \vbeta_0|\vbeta_0)\right\}  \widehat{G}^{(1)}(\vx_{i}^T \vbeta_0|\vbeta_0) [\E(\vx_{i,-1}|\vx_i^T\vbeta_0)-\widehat{\E}(\vx_{i,-1}|\vx_i^T\vbeta_0)]\\
		\triangleq&\sum_{j=1}^6\vI_{nj},
		\end{align*}
		where the definition of $\vI_{nj}$ is clear from the context. Since $G^{(1)}(\vx_i^T\vbeta_0|\vbeta_0)$ is bounded by Assumption~\ref{A1}-(b), and $\widetilde{\epsilon}_i$ is sub-Gaussian by Lemma~\ref{lem:subg},
		then $G^{(1)}(\vx_{i}^T \vbeta_0|\vbeta_0)\widetilde{\epsilon}_i$ is also sub-Gaussian by Lemma~\ref{lem:subg}. Lemma~\ref{lem14NCL} in Section~\ref{sec:proof_auxil} implies that there exist positive constants $c_0$, $c_1$ and $c_2$ such that for all $n$ sufficiently large,
		$$P\left(||\vI_{n1}||_\infty\geq  c_0 \sqrt{\frac{\log (p\vee n)}{n}}\,\right)\leq \exp[-c_1\log (p\vee n)] .$$
		
		Note that $\vbeta_0\in\bbB_1$. Lemma~\ref{G1func} implies that $\max_{1\leq i \leq n}\left|\widehat{G}^{(1)}(\vx^T\vbeta_0|\vbeta_0) - G^{(1)}(\vx^T\vbeta_0|\vbeta_0)\right|\leq c_0h$, with probability at least $1 -\exp[-c_1\log (p\vee n)]$, for some positive constants $c_0$, $c_1$, and all $n$ sufficiently large. Then we can apply the proof of Lemma~\ref{lem:ghat1_g1} to show that with probability at least $1 -\exp[-c_1\log (p\vee n)]$, 
		$ \sqrt{n}||\vI_{n2}||_\infty\leq  c_0\left[h^2\log(p\vee n)\right]^{1/4}\leq  c_0\sqrt{\log (p\vee n)}$. Hence we have that with probability at least $1 -\exp[-c_1\log (p\vee n)]$, 
		$ ||\vI_{n2}||_\infty\leq   c_0\sqrt{\frac{\log (p\vee n)}{n}}$.

		To bound $||\vI_{n3}||_\infty$, observe that 
		\begin{align*} 
		||\vI_{n3}||_\infty \leq & \Big|\Big|n^{-1}\sum_{i=1}^n \widetilde{\epsilon}_i G^{(1)}(\vx_{i}^T \vbeta_0|\vbeta_0) [ \E(\vx_{i,-1}|\vx_i^T\vbeta_0)-\widehat{\E}(\vx_{i,-1}|\vx_i^T\vbeta_0)]  \Big|\Big|_\infty\\
		&+ \Big|\Big|n^{-1}\sum_{i=1}^n \widetilde{\epsilon}_i [\widehat{G}^{(1)}(\vx_{i}^T \vbeta_0|\vbeta_0) -G^{(1)}(\vx_{i}^T \vbeta_0|\vbeta_0)] * [ \E(\vx_{i,-1}|\vx_i^T\vbeta_0)-\widehat{\E}(\vx_{i,-1}|\vx_i^T\vbeta_0)]  \Big|\Big|_\infty\\
		=&||\vI_{n31}||_\infty+||\vI_{n32}||_\infty,
		\end{align*}
		where the definitions of $\vI_{n31}$ and $\vI_{n32}$ are clear from the context. Similarly, Lemma~\ref{lem:Ebound} indicates that 
		$\max_{1\leq i \leq n}\left|\left[\widehat{\E}(\vx_{i,-1}|\vx_i^T\vbeta_0) - \E(\vx_{i,-1}|\vx_i^T\vbeta_0)\right]^T\veta\right|\leq c_0h^2||\veta||_2$, with probability at least $1 -\exp[-c_1\log (p\vee n)]$, for some positive constants $c_0$, $c_1$, and all $n$ sufficiently large. Then the proof of Lemma~\ref{lem:Ehat_E} implies that $ \sqrt{n}||\vI_{n31}||_\infty\leq c_0 h\sqrt{s\log(p\vee n)}$ holds with probability at least $1 -\exp[-c_1\log (p\vee n)]$,  for some positive constants $c_0$, $c_1$, and all $n$ sufficiently large.  For $||\vI_{n32}||_\infty$,  Lemma~\ref{G1func} and Lemma~\ref{lem:Ebound} imply that
		\begin{align*} 
		||\vI_{n32}||_\infty\leq &c_0hn^{-1}\sum_{i=1}^n |\widetilde{\epsilon}_i| *\big|\big|   \E(\vx_{i,-1}|\vx_i^T\vbeta_0)-\widehat{\E}(\vx_{i,-1}|\vx_i^T\vbeta_0)\big|\big|_\infty\\
		\leq&c_0h^3\sqrt{n^{-1}\sum_{i=1}^n \widetilde{\epsilon}_i^2}  \leq c_1h^3\leq c_0n^{-1/2},
		\end{align*}
		with probability at least  $1 -\exp[-c_1\log (p\vee n)]$, for some positive constants $c_0$, $c_1$, and all $n$ sufficiently large.  Hence we have $P\Big(||\vI_{n3}||_\infty\geq  c_0 \sqrt{\frac{\log (p\vee n)}{n}}\,\Big)\leq\  \exp [-c_1\log (p\vee n)]$, for some positive constants $c_0$, $c_1$, and all $n$ sufficiently large.

		The proof of Lemma~\ref{lem:ghat_g} implies that  $P\left(\sqrt{n}||\vI_{n4}||_\infty\geq h[\log(p\vee n)]^{1/4}\,\right)\leq \ \exp[-c_1\log (p\vee n)]$. Hence   $||\vI_{n4}||_\infty\leq d_0 n^{-1/2}$  with probability at least $1- \exp[-c_1\log (p\vee n)] $.
		
		For $\vI_{n5}$, note that $\vx_{i,-1}-\E(\vx_{i,-1}|\vx_i^T\vbeta_0)$ is sub-Gaussian. Proof of Lemma~\ref{lem14NCL} indicates that for universal constant $c_1>0$,
		\begin{align*}
		&P\left(\frac{1}{n}\sum_{i=1}^n \big|\big|[\vx_{i,-1}-\E(\vx_{i,-1}|\vx_i^T\vbeta_0)][\vx_{i,-1}-\E(\vx_{i,-1}|\vx_i^T\vbeta_0)]^T\big|\big|_\infty\geq ||\Cov(\vx_{-1}|\vx^T\vbeta_0)||_\infty+\sigma_x^2\sqrt{\frac{\log (p\vee n)}{n}}\right) \\
		&\leq  \exp[-c_1\log (p\vee n)].
		\end{align*}
		Lemma~\ref{Gfunc} and Lemma~\ref{G1func} together indicate  
		$||\vI_{n5}||_\infty\leq ch^3n^{-1}\sum_{i=1}^n||\vx_{i,-1}-\E(\vx_{i,-1}|\vx_i^T\vbeta_0)||_\infty \leq c_0n^{-1/2}$ with probability at least $1- \exp[-c_1\log (p\vee n)]$.  Similarly, Lemma~\ref{Gfunc} and Lemma~\ref{lem:Ebound} indicate that  $||\vI_{n6}||_\infty \leq c_0n^{-1/2}$ with the same probability bound.	Combining all the previous results, we conclude the lemma.
	\end{proof}

	\begin{proof}[Proof of Lemma~\ref{lem:subg}]
		To prove the first part of the lemma, note that for any unit vector $\vv\in\bbR^p$ and $c\in\bbR$,  Jensen's inequality and the sub-Gaussian property of $\vx$ imply that
		\begin{align*}
		\E\big\{\exp[s\vv^T\E(\vx\big|\vx^T\vbeta)]\big\}&=\E\big\{\exp[\E(s\vx^T\vv\big|\vx^T\vbeta)]\big\} \\
		&\leq \E\big\{\E[\exp(s\vx^T\vv)\big|\vx^T\vbeta]\big\}\\
		&=\E[\exp(s\vx^T\vv)] \leq \exp\left(\frac{s^2\sigma_x^2}{2}\right).
		\end{align*}
		For $\vx-\E(\vx|\vx^T\vbeta)$, we apply H\"older's inequality, which indicates that for any $q_1, q_2>0$ such that $q_1^{-1}+q_2^{-1}=1$,
		\begin{align*}
		\E\Big\{\exp\big\{s\vv^T[\vx-\E(\vx|\vx^T\vbeta)]\big\}\Big\}\leq& \big[ \E \exp(sq_1\vx^T\vv)\big]^{1/q_1} \Big\{ \E \exp\big[sq_2\vv^T \E(\vx|\vx^T\vbeta)\big] \Big\}^{1/q_2}\\
		\leq&  \E \exp(sq_1\vx^T\vv)   * \E \exp\big[sq_2\vv^T \E(\vx|\vx^T\vbeta)\big] \\
		\leq&  \exp\left[\frac{s^2\sigma_x^2(q_1^2+q_2^2)}{2}\right].
		\end{align*}
		Let $q_1=q_2=2$, then we have $ \E\Big\{\exp\big\{s\vv^T[\vx-\E(\vx|\vx^T\vbeta)]\big\}\Big\}  \leq  \exp(2s^2\sigma_x^2).$
		
		For $\widetilde{\epsilon} = 2(2A-1)\big(\epsilon+g(\vx)\big)$, note that $\epsilon$ and $A$ are independent. Easy to show that $2(2A-1)\epsilon$ is sub-Gaussian, we note that $\E\big\{\exp[2s(2A-1)\epsilon ]\} = \frac{1}{2} \big(\E e^{2s \epsilon }+\E e^{-2s \epsilon }\big)\leq \exp( 2s^2\sigma_\ep^2 ).$
		Since $g(\cdot)$ is bounded by $M$ almost everywhere, Exercise 2.4 in \citet{subExp} implies that $2(2A-1)g(\vx)$ is sub-Gaussian with variance proxy at most $4M^2$. Then similarly as the previous step, we conclude that $\widetilde{\epsilon}$ is sub-Gaussian.
		
		To prove the second part of Lemma~\ref{lem:subg}, denote $z = xy-\E(xy)$. Note that $\E z=0$. Since  $y$ is a sub-Gaussian, for any integer $k\geq 1$, we have $\E(|y|^k)\leq (2\sigma_y^2)^{k/2}k\Gamma(k/2)$.  Hence  for any $s>0$,  
		\begin{align*}
		&\E[\exp(sz)] \\
		\leq &1+\sum_{k=2}^\infty\frac{s^k\E(|z|^k)}{k!}\\
		\leq &1+\sum_{k=2}^\infty\frac{s^k2^{k-1}\big\{\E(|xy|^k) + [\E(|xy|)]^k\big\}}{k!}\\
		\leq &1+\sum_{k=2}^\infty\frac{s^k2^{k-1}\sigma_x^k\big[\E(|y|^k) + (\E|y|)^k]}{k!}\\ 
		\leq &1+\sum_{k=2}^\infty\frac{(2 s\sigma_x)^k\E(|y|^k)}{k!}\\
		\leq &1+\sum_{n=1}^\infty\frac{ (2 s\sigma_x)^{2n} (2\sigma_y^2)^{n}(2n)\Gamma(n)}{(2n)!}+ \sum_{n=1}^\infty\frac{(2 s\sigma_x)^{2n+1}(2\sigma_y^2)^{n+1/2}(2n+1) \Gamma(n+1/2)}{(2n+1)!}\\
		\leq&1+(1+2\sqrt{2}s\sigma_x\sigma_y)\sum_{n=1}^\infty\frac{ (2\sqrt{2} s\sigma_x\sigma_y)^{2n} 2(n!)}{(2n)!} \\
		\leq&1+(1+2\sqrt{2}s\sigma_x\sigma_y)\sum_{n=1}^\infty\frac{ (2\sqrt{2} s\sigma_x\sigma_y)^{2n} }{ n!} \\
		= &\exp(8s^2\sigma_x^2\sigma_y^2) + 2\sqrt{2}s\sigma_x\sigma_y \big[ \exp(8s^2\sigma_x^2\sigma_y^2) -1\big]\\
		\leq &\exp(16s^2\sigma_x^2\sigma_y^2),
		\end{align*} 
		where the second and the fourth inequalities apply Jensen's inequality, the second last inequality applies that $2(n!)^2\leq (2n)!$ for any $n\geq 1$. The conclusion follows by the definition of sub-Gaussian random variables.
	\end{proof}

	\begin{proof}[Proof of Lemma~\ref{lem:events}]
		Note that $\vx_i $,  $[\vx_{i,-1}-\E(\vx_{i,-1}|\vx_i^T\vbeta)]^T\vtheta_j$ and $\widetilde{\epsilon}_i$ are all sub-Gaussian with variance proxy $\sigma_x^2$, $2||\vtheta_j||_2^2\sigma_x^2$, and $4(\sigma_\epsilon^2+M^2)$, respectively, by Lemma~\ref{lem:subg}.   Lemma~\ref{lem:thetaj} implies that $\max_{2\leq j \leq p}||\vtheta_j||_0\leq \widetilde{s}+1$, and $||\vtheta_j||_2\leq \xi_2^{-1}$ uniformly in $j$,  where $\xi_2$ is defined in Assumption~\ref{A2}-(a).  Denote $\widecheck{\vx}_{i,-1} =\vx_{i,-1}-\E(\vx_{i,-1}|\vx_i^T\vbeta)$, and $\vxw_{i,-1}=\vx_{i,-1}-\E(\vx_{i,-1}|\vx_i^T\vbeta_0)$.  Assumption~\ref{A2}-(c) and Lemma~\ref{lem15NCL}  imply that 
		\begin{align*}
		& \sup_{\vbeta\in\bbB }n^{-1}\sum_{i=1}^n\left|  \widecheck{\vx}_{i,-1}^T\vtheta_j - \vxw_{i,-1}^T\vtheta_j\right|^2\\
		\leq& C\xi_2^{-2}\sup_{\vbeta\in\bbB }n^{-1}\sum_{i=1}^n\left(|\vx_i^T\vbeta-\vx_i^T\vbeta_0|+\max(|\vx_i^T\vbeta|,|\vx_i^T\vbeta_0|)||\vbeta-\vbeta_0||_2\right)^2\\
		\leq& 2C\xi_2^{-2}\sup_{\vbeta\in\bbB }n^{-1}\sum_{i=1}^n\left(|\vx_i^T\vbeta-\vx_i^T\vbeta_0|^2+ (|\vx_i^T\vbeta|^2+|\vx_i^T\vbeta_0|^2)||\vbeta-\vbeta_0||_2^2\right) \\
		\leq& 2C\xi_2^{-2}\sup_{\vbeta\in\bbB }n^{-1}\sum_{i=1}^n |\vx_i^T(\vbeta-\vbeta_0)|^2+ 4Cr^2\xi_2^{-2}\sup_{\vbeta\in\bbB }n^{-1}\sum_{i=1}^n|\vx_i^T\vbeta|^2  \\
		\leq &c_0\xi_2^{-2}\sigma_x^2,
		\end{align*}
		with probability at least $1-\exp(-c_1n)$, for some positive constants $C$, $c_0$, $c_1$ and all $n$ sufficiently large, where the last inequality applies Lemma~\ref{lem15NCL}, since the sub-Gaussian property implies that $\E[(\vx_i^T\vv)^2]\leq \sigma_x^2$ for any $||\vv||_2=1$.
		We thus have
		\begin{align*}
		&P\left(\max_{2\leq j \leq p} \sup_{\vbeta\in\bbB }n^{-1}\sum_{i=1}^n\left| \widecheck{\vx}_{i,-1}^T\vtheta_j\right|^2  \geq d_0\xi_2^{-2}\sigma_x^2 \right) \\
		\leq& \sum_{j=2}^{p}P\left( \sup_{\vbeta\in\bbB }n^{-1}\sum_{i=1}^n\left|  \widecheck{\vx}_{i,-1}^T\vtheta_j\right|^2  \geq d_0\xi_2^{-2}\sigma_x^2   \right)\\
		\leq& \sum_{j=2}^{p}P\left(  \sup_{\vbeta\in\bbB }2n^{-1}\sum_{i=1}^n\left|  \widecheck{\vx}_{i,-1}^T\vtheta_j - \vxw_{i,-1}^T\vtheta_j\right|^2  \geq 2c_0\xi_2^{-2} \sigma_x^2 \right)+P\left(  2n^{-1}\sum_{i=1}^n\left|  \vxw_{i,-1}^T\vtheta_j  \right|^2   \geq (d_0-2c_0)\xi_2^{-2} \sigma_x^2  \right)\\
		\leq&(p-1)\exp(-c_2n) +P\left( \left|n^{-1}\sum_{i=1}^n\left|  \vxw_{i,-1}^T\vtheta_j  \right|^2  -\E\left(\left|  \vxw_{i,-1}^T\vtheta_j\right|^2\right) \right| \geq (d_0-2c_0-4)\xi_2^{-2} \sigma_x^2/2 \right)\\
		\leq& p\exp(-c_2n) = \exp(-c_2n+\log p)=  \exp(-cn), 
		\end{align*}  
		for some positive constants $d_0>4$, $c_2$, $c$, and all $n$ sufficiently large. In the above, %the third inequality applies Assumption~\ref{A2}-(c) and Lemma~\ref{lem15NCL};
		the last inequality applies Lemma~\ref{lem14NCL}, with $\max_{2\leq j\leq p}\E(|\vxw_{i,-1}^T\vtheta_j|^2)\leq 2\xi_2^{-2}\sigma_x^2$ by its sub-Gaussian property. It derives the probability bound for $\mathcal{G}_n$. The probability bound for  $\mathcal{K}_n$ follows from Lemma~\ref{lem15NCL} with similar technique.   For $\mathcal{H}_n$, note that $\Big| 
		2[\ep_i+g(\vx_i)]*	 \big[\vx_{i,-1}-\E(\vx_{i,-1}|\vx_i^T\vbeta)\big]^T\vtheta_j\Big|^2 = \Big| \widetilde{\epsilon}_i  \left[\vx_{i,-1}-\E(\vx_{i,-1}|\vx_i^T\vbeta)\right]^T\vtheta_j\Big|^2$, where $\widetilde{\epsilon}_i  = 2(2A_i-1)[\ep_i+g(\vx_i)]$.
		Lemma~\ref{lem:subg} implies that $\widecheck{\vx}_{i,-1}^T\vtheta_j$ and $\widetilde{\epsilon}_i$ are both sub-Gaussian. Hence we have that $\E(\widetilde{\epsilon}_i^4)\leq 16*\left[4(\sigma_\epsilon^2+M^2)\right]^2$, and $\max_{2\leq j \leq p}\E\left[(\vxw_{i,-1}^T\vtheta_j)^4\right]\leq 16*\left(2\xi_2^{-2}\sigma_x^2\right)^2$.
		Similar as the above analysis for $\mathcal{G}_n$,
		Assumption~\ref{A2}-(c) implies that  
		\begin{align*}
		\sup_{\vbeta\in\bbB }n^{-1}\sum_{i=1}^n\left|  \widecheck{\vx}_{i,-1}^T\vtheta_j - \vxw_{i,-1}^T\vtheta_j\right|^4 
		\leq& C\xi_2^{-4}\sup_{\vbeta\in\bbB }n^{-1}\sum_{i=1}^n |\vx_i^T(\vbeta-\vbeta_0)|^4+2Cr^4\xi_2^{-4}\sup_{\vbeta\in\bbB }n^{-1}\sum_{i=1}^n|\vx_i^T\vbeta|^4 \\
		\leq& c_0\xi_2^{-4}\sigma_x^4,
		\end{align*}
		with probability at least $1-\exp(-c_1\sqrt{n})$, for some positive constants $C$, $c_0$, $c_1$ and all $n$ sufficiently large, where the last inequality applies Lemma~\ref{lem:cube_rate} with $s_0=ks$ and $t= c_3\xi_2^{-4}\sigma_x^4$ for some positive constant $c_3$. Hence  by Lemma~\ref{lem:cube_rate}, there exist some  positive constants $d_1>256\sqrt{2}$, $d_2>16$, $c_2$, $c$, such that for all $n$ sufficiently large,
		\begin{align*}
		&P\left(\max_{2\leq j \leq p}\sup_{\vbeta\in\bbB }n^{-1}\sum_{i=1}^n\big|\widetilde{\epsilon}_i \widecheck{\vx}_{i,-1}^T\vtheta_j \big|^2\geq d_1\xi_2^{-2}\sigma_x^2(\sigma_\ep^2+M^2) \right)\\
		\leq &P\left(\sqrt{n^{-1}\sum_{i=1}^n \widetilde{\epsilon}_i ^4}\geq  d_2(\sigma_\ep^2+M^2) \right) + P\left(\max_{2\leq j \leq p}\sup_{\vbeta\in\bbB }\sqrt{n^{-1}\sum_{i=1}^n\left|\widecheck{\vx}_{i,-1}^T\vtheta_j \right|^4}\geq  d_1d_2^{-1}\xi_2^{-2}\sigma_x^2 \right)\\
		\leq& P\left( n^{-1}\sum_{i=1}^n \widetilde{\epsilon}_i^4\geq  d_2^2(\sigma_\ep^2+M^2)^2 \right) +\sum_{j=2}^{p}P\left( \sup_{\vbeta\in\bbB }n^{-1}\sum_{i=1}^n\big|\widecheck{\vx}_{i,-1}^T\vtheta_j \big|^4\geq d_1^2d_2^{-2}\xi_2^{-4}\sigma_x^4\right)\\
		\leq& P\left( \left| n^{-1}\sum_{i=1}^n\widetilde{\epsilon}_i^4 - \E\left(\widetilde{\epsilon}_i^4\right)\right|\geq (d_2^2-256) (\sigma_\ep^2+M^2)^2 \right) +P\left( 8 n^{-1}\sum_{i=1}^n\left(\vxw_{i,-1}^T\vtheta_j\right)^4  \geq  (d_1^2d_2^{-2}-8c_0)\xi_2^{-4}\sigma_x^4 \right)\\
		&+\sum_{j=2}^{p}P\left( \sup_{ \vbeta \in \bbB } 8n^{-1}\sum_{i=1}^n\left(\widecheck{\vx}_{i,-1}^T\vtheta_j -\vxw_{i,-1}^T\vtheta_j\right)^4  \geq  8c_0\xi_2^{-4}\sigma_x^4 \right)\\
		\leq& p \exp(-c_2 \sqrt{n} )+P\left( 8\left|n^{-1}\sum_{i=1}^n\left(\vxw_{i,-1}^T\vtheta_j\right)^4 - \E\left[\left(\vxw_{i,-1}^T\vtheta_j\right)^4\right]\right|\geq  (d_1^2d_2^{-2}-512-8c_0)\xi_2^{-4}\sigma_x^4\right)\\
		\leq& (p+1) \exp(-c_2 \sqrt{n} )=  \exp(-c\sqrt{n}).
		\end{align*}  
		Hence we prove that $P(\mathcal{H}_n)\geq 1-\exp(-c\sqrt{n})$, for some positive constant $c$ and all $n$ sufficiently large.
		
		For $\mathcal{J}_n$, note that
		$$\max_{1\leq i\leq n}\sup_{\vbeta\in\bbB_1}|\vx_i^T\vbeta|\leq \max_{1\leq i\leq n}|\vx_i^T\vbeta_0|+ \max_{1\leq i\leq n}\sup_{\vbeta\in\bbB_1}\left|\vx_i^T(\vbeta-\vbeta_0)\right|.$$
		The sub-Gaussian property of $\vx_i^T\vbeta_0$ implies that $P\left( \max_{  1\leq i\leq n}|\vx_i^T\vbeta_0|\geq ||\vbeta_0||_2\sigma_x\sqrt{\log(p\vee n)}\right)\leq \exp[-c\log(p\vee n)],$ for some positive constant $c$, and all $n$ sufficiently large. Note that for any $\vbeta\in\bbB_1$, we have that $||\vbeta-\vbeta_0||_2\leq c_0\sqrt{s}h^2$, and $||\vbeta-\vbeta_0||_0\leq (k+1)s$. Then we have that $ \max_{1\leq i\leq n}\sup_{\vbeta\in\bbB_1}\left|\vx_i^T(\vbeta-\vbeta_0)\right|\leq \max_{1\leq i\leq n}||\vx_i||_\infty\sup_{\vbeta\in\bbB_1}||\vbeta-\vbeta_0||_1\leq c_0s^{3/2}h^2\sqrt{\log (p\vee n)}$ with probability at least $1-\exp[-c_1\log(p\vee n)]$, for some positive constants $c_0$, $c_1$ and all $n$ sufficiently large,  according to (\ref{subG_prop}).
		%		 Lemma~\ref{lem15NCL} implies that
		%		\begin{align*}
		%		&P\left(\max_{  1\leq i\leq n}\sup_{\vbeta\in\bbB_1}\left|\vx^T(\vbeta-\vbeta_0)\right|\geq \sqrt{s}h^2\left[\sigma_x + s\log(p\vee n)\right]\right)\\
		%		\leq &\sum_{i=1}^n P\left(\sup_{\vbeta\in\bbB_1}\left|\vx^T(\vbeta-\vbeta_0)\right|\geq \sqrt{s}h^2\left[\sigma_x + s\log(p\vee n)\right]\right)\\
		%		\leq & \sum_{i=1}^n P\left(\sup_{\vbeta\in\bbB_1}\left|\left[\vx^T(\vbeta-\vbeta_0)\right]-\E\left\{\left[\vx^T(\vbeta-\vbeta_0)\right]^2\right\}\right|\geq \sqrt{s}h^2 *s\log(p\vee n)\right)\\
		%		\leq &n\exp[-c_0s\log(p\vee n)]\leq \exp[-c_1s\log(p\vee n)],
		%		\end{align*}
		%		for some positive constants $c_0$, $c_1$, and all $n$ sufficiently large. 
		The assumptions of Theorem~\ref{Lasso_error} imply that $s^{3/2}h^2\leq 1$. Hence we conclude that 
		$$P\left( \max_{1\leq i\leq n}\sup_{\vbeta\in\bbB_1}|\vx_i^T\vbeta|\geq 2||\vbeta_0||_2\sigma_x\sqrt{\log(p\vee n)}\right)\leq \exp[-c\log(p\vee n)],$$
		for some positive constant $c$, and all $n$ sufficiently large. It concludes the proof of the lemma.
	\end{proof}

	\begin{proof}[Proof of Lemma~\ref{lem:Gbound}]
		Recall $\vgamma = \vbeta_{-1}-\vbeta_{0,-1}$. By Taylor expansion, we can derive that
		\begin{align}
		G(t|\vbeta)  = &\E\big[f_0(\vx^T\vbeta_0)|\vx^T\vbeta=t\big] \nonumber\\
		=& \E\left[f_0(\vx^T\vbeta) -f_0'(\vx^T\vbeta)\vx_{-1}^T\vgamma+ \int_{\vx^T\vbeta}^{\vx^T\vbeta_0} f''_0(u) (\vx^T\vbeta_0-u ) du |\vx^T\vbeta=t\right] \nonumber\\
		=&f_0(t)  - f_0'(t)\E(\vx_{-1}|\vx^T\vbeta=t)^T\vgamma  + \E\left[  \int_{0}^{\vx_{-1}^T\vgamma}af''_0(a+\vx^T\vbeta_0)da |\vx^T\vbeta=t\right].\label{form:Gnew_def}
		\end{align}
		Hence we have 
		\begin{align*}
		G(\vx^T\vbeta|\vbeta)  -G(\vx^T\vbeta_0|\vbeta_0)
		=&\big[f_0(\vx^T\vbeta)-f_0(\vx^T\vbeta_0) \big] - f_0'(\vx^T\vbeta)\E(\vx|\vx^T\vbeta)^T(\vbeta-\vbeta_0 )  \\
		&+ \E\left[  \int_{0}^{\vx^T(\vbeta-\vbeta_0)}af''_0(a+\vx^T\vbeta_0)da |\vx^T\vbeta\right]  \\
		=&  f_0'(\vx^T\vbeta)\left[\vx_{-1}^T\vgamma  -\E(\vx_{-1}^T\vgamma  |\vx^T\vbeta)\right] -  \int_{0}^{\vx_{-1}^T\vgamma}af''_0(a+\vx^T\vbeta_0)da\\
		&+ \E\left[  \int_{0}^{\vx_{-1}^T\vgamma}af''_0(a+\vx^T\vbeta_0)da |\vx^T\vbeta\right] .
		\end{align*}
		It proves (\ref{G-1}). 	To prove (\ref{G1-1}),  according to equation (\ref{form:Gnew_def}), we have that 
		\begin{align*}
		G^{(1)}(t|\vbeta)  = \frac{d}{dt}G(t|\vbeta)
		=&f_0'(t)  - f_0'(t)\E^{(1)}(\vx_{-1}|\vx^T\vbeta=t)^T\vgamma  - f_0''(t)\E(\vx_{-1}|\vx^T\vbeta=t)^T\vgamma \\
		& +\frac{d}{dt} \E\left[  \int_{0}^{\vx_{-1}^T\vgamma}af''_0(a+\vx^T\vbeta_0)da |\vx^T\vbeta=t\right]\\
		=&f_0'(t)  - f_0'(t)\E^{(1)}(\vx_{-1}|\vx^T\vbeta=t)^T\vgamma - f_0''(t)\E(\vx_{-1}|\vx^T\vbeta=t)^T\vgamma \\
		&+ \E^{(1)}\left[  \int_{0}^{\vx_{-1}^T\vgamma}af''_0(a+\vx^T\vbeta_0)da |\vx^T\vbeta=t\right],
		\end{align*}
		where $\E^{(1)}(\cdot|\vx^T\vbeta=t)$ is the first derivative of $\E(\cdot|\vx^T\vbeta=t)$ with respect to $t$.  	Hence we have 
		\begin{align*}
		&G^{(1)}(\vx^T\vbeta|\vbeta)  -G^{(1)}(\vx^T\vbeta_0|\vbeta_0)\\
		=&f''_0(\vx^T\vbeta) \big[ \vx_{-1}^T\vgamma-\E( \vx_{-1}^T\vgamma|\vx^T\vbeta)\big] - \int^{\vx_{-1}^T\vgamma}_{0} af'''_0(a+\vx^T\vbeta_0)da\\
		&-   f_0'(\vx^T\vbeta)\E^{(1)}(\vx_{-1}|\vx^T\vbeta=t)^T\vgamma + \E^{(1)}\left[  \int_{0}^{\vx_{-1}^T\vgamma}af''_0(a+\vx^T\vbeta_0)da |\vx^T\vbeta=t\right],
		\end{align*}
		It proves (\ref{G1-1}).
		
		Denote the event 
		\begin{align*}
		\mathcal{E}_0= &\left\{\max_{1\leq i \leq n}\sup_{\vbeta\in\bbB}|\vx_i^T\vbeta|\geq c_0 \sqrt{s\log(p\vee n)}\right\} \\
		&\bigcap \left\{\max_{1\leq i \leq n} |\vx_i^T(\vbeta_1-\vbeta_2)| \geq c_0 \sqrt{s\log(p\vee n)}||\vbeta_1-\vbeta_2||_2,\ \forall\ \vbeta_1,\vbeta_2\in\bbB\right\},
		\end{align*}
		for some constant $c_0>0$.	 		
		Taking $t=d_0\sigma_x^2\log(p\vee n)$ and $s_0=ks$, Lemma~\ref{lem15NCL} implies that 
		\begin{align*}
		P\left(\max_{1\leq i \leq n}\sup_{\vv\in\bbK(2ks)}(\vx_i^T\vv)^2\leq (d_0+1)  \sigma_x^2s \log(p\vee n)\right) 
		\leq &\sum_{i=1}^nP\left(\sup_{\vv\in\bbK(2ks)} (\vx_i^T\vv)^2\leq (d_0+1)\sigma_x^2  s\log(p\vee n)\right)\\
		\leq &\exp[-d_1s\log(p\vee n)],
		\end{align*} 
		for some positive constants $d_0$, $d_1$, and all $n$ sufficiently large. Combining this result with the definition of $\bbB$, we obtain that $P(\mathcal{E}_0)\geq 1-\exp[-c_2s\log(p\vee n)]$, for some positive constants $c_0$, $c_2$, and all $n$ sufficiently large. 
		Note that $G(t|\vbeta)$ is twice-differentiable with respect to $t$, and the derivatives are bounded by Assumption~\ref{K4}-(a). Hence on the event $\mathcal{E}_0$, we have
		\begin{align*}
		&\max_{1\leq i \leq n} \left[  G(\vx_i^T\vbeta_1|\vbeta_1)-G(\vx_i^T\vbeta_2|\vbeta_2) \right]^2\\
		\leq& \max_{1\leq i \leq n} 2\left[  G(\vx_i^T\vbeta_1|\vbeta_1)-G(\vx_i^T\vbeta_1|\vbeta_2) \right]^2+\max_{1\leq i \leq n} 2\left[  G(\vx_i^T\vbeta_1|\vbeta_2)-G(\vx_i^T\vbeta_2|\vbeta_2) \right]^2\\
		\leq& c_3||\vbeta_1-\vbeta_2||_2s\log(p\vee n)+c_3\max_{1\leq i \leq n} \left[ \vx_i^T(\vbeta_1-\vbeta_2)\right]^2\\
		\leq& c_1||\vbeta_1-\vbeta_2||_2s\log(p\vee n),
		\end{align*}
		for any $\vbeta_1,\vbeta_2\in\bbB$, some positive constants $c_1$ and $c_3$, where the first part of the second inequality applies Assumption~\ref{K4}-(c). It proves (\ref{Gt}). We can conclude (\ref{Gtt}) with similar techniques.

		To prove (\ref{G-2}), observe that 
		\begin{align*}
		n^{-1}\sum_{i=1}^{n}\big[G(\vx_i^T\vbeta|\vbeta) - G(\vx_i^T\vbeta_0|\vbeta_0)\big]^2\leq& 3 \sum_{k=1}^3 A_{k} , 
		\end{align*} 
		where
		\begin{align*}
		A_1&= n^{-1} \sum_{i=1}^{n}\big[ f'_0(\vx_i^T\vbeta)\big]^2 \big[ \vx_{i,-1}^T\vgamma-\E( \vx_{i,-1}^T\vgamma|\vx_i^T\vbeta)\big]^2,\\
		A_2&=n^{-1} \sum_{i=1}^{n} \big[h(\vx_{i,-1}^T\vgamma)\big]^2,\\
		A_3&=n^{-1} \sum_{i=1}^{n} \left\{\E\big[h(\vx_{i,-1}^T\vgamma)|\vx_i^T\vbeta\big]\right\}^2,
		\end{align*}
		with $h(u) = \int_{0}^uaf_0''(a+\vx^T\vbeta_0)da$.
		It is sufficient to bound $A_k$ for $k=1,2,3$. To bound $A_1$, we have 
		\begin{align*}
		A_1\leq &b^2n^{-1}\sum_{i=1}^{n} \big[ \vx_{i,-1}^T\vgamma-\E( \vx_{i,-1}^T\vgamma|\vx_i^T\vbeta)\big] ^2.
		\end{align*}
		Lemma~\ref{lem:Ex_eigen} in Section~\ref{sec:proof_auxil} implies that $A_1\leq c_1||\vgamma||_2^2$ with probability at least $1-\exp(-c_2n)$, for positive constants $c_1$, $c_2$.  For $A_{2}$, note that 
		\begin{align*}
		A_2=&(4n)^{-1}\sum_{i=1}^{n}\Big[ f''_0(z_i)  (\vx_{i,-1}^T\vgamma )^2 \Big]^2\leq c_0n^{-1}\sum_{i=1}^{n}  (\vx_{i,-1}^T\vgamma )^4,
		%n^{-1}\sum_{i=1}^{n}\Big[\int_{0}^{\vx_i^T\veta} af''_0(\vx_i^T\vbeta+a) da\Big]^2c_0n^{-1}\sum_{i=1}^{n}\big[\veta^T\E ( \vx_i\vx_i^T|\vx_i^T\vbeta)\veta \big]^2\\
		%\leq &n^{-1}\sum_{i=1}^{n}\Big| \int_{0}^{\vx_i^T\veta} [f''_0(\vx_i^T\vbeta+a)]^2 a^2 da\Big|\\
		%\leq &c_1n^{-1}\sum_{i=1}^{n}\Big| \int_{0}^{\vx_i^T\veta}   a^2 da\Big|=\frac{c_1}{3n} \sum_{i=1}^{n} |\vx_i^T\veta |^3  ,
		\end{align*}
		where $z_i$ is between $\vx^T\vbeta$ and $\vx^T\vbeta_0$,
		for positive constant $c_0$. The  last inequality applies the assumption that $f_0''(\cdot)$ is bounded. Lemma~\ref{lem:cube_rate} indicates that  $A_2\leq c_1||\vgamma||_2^4$ with probability at least $1-\exp(-c_2\sqrt{n})$, since $s\log p \leq d_0nh^5\leq d_1n^{1/6} \leq d_1\sqrt{n}  $, for some positive constants $d_0$, $d_1$ and $n\geq 1$.
		
		For $A_3$, observe that
		\begin{align*}
		A_3= (4n)^{-1}\sum_{i=1}^{n}\left\{ \E\big[ f''_0(z_i) (\vx_{i,-1}^T\vgamma )^2   |\vx_i^T\vbeta\big]\right\}^2 
		\leq c_0n^{-1}\sum_{i=1}^{n}\big[\vgamma^T\E ( \vx_{i,-1}\vx_{i,-1}^T|\vx_i^T\vbeta)\vgamma \big]^2,
		\end{align*}  
		where $z_i$ is between $\vx^T\vbeta$ and $\vx^T\vbeta_0$.
		Assumption~\ref{A2}  implies that 
		\begin{align*}
		n^{-1}\sum_{i=1}^{n}\big[\vgamma^T\E ( \vx_{i,-1}\vx_{i,-1}^T|\vx_i^T\vbeta) \vgamma\big]^2  
		\leq   \xi_4||\vgamma||_2^4.
		\end{align*} 
		%with probability at least $1-\exp[-c\log(p\vee n)]$, for some positive constant $c$ and all $n$ sufficiently large.
		%		with probability at  least $1-\exp(-c_2s\log p)$, where the second last inequality applies Assumption~\ref{A2}-(c). 
		%		Note that %Assumption~\ref{A2}-(a) implies that 
		%		\begin{align*}
		%		&n^{-1}\sum_{i=1}^{n}\big[\veta^T\E ( \vx_i\vx_i^T|\vx_i^T\vbeta) \veta\big]^2  
		%		\\	\leq  &2n^{-1}\sum_{i=1}^{n}\big[\veta^T\E ( \vx_i\vx_i^T|\vx_i^T\vbeta_0) \veta  \big]^2 +2c_0n^{-1}\sum_{i=1}^{n} \big\{\veta^T[\E ( \vx_i\vx_i^T|\vx_i^T\vbeta) -\E ( \vx_i\vx_i^T|\vx_i^T\vbeta_0) ]\veta \big\}^2\\
		%		\leq &c_0'||\veta||_2^4  \Big\{ 1+n^{-1}\sum_{i=1}^{n}\big[ (\vx_i^T\vbeta)^4+(\vx_i^T\vbeta_0)^4 +  (1+|\vx_i^T\vbeta|^4) ||\veta||_2^2  \big]\Big\}\\
		%		\leq &c_1||\veta||_2^4,
		%		\end{align*} 
		%		with probability at  least $1-\exp(-c_2s\log p)$, where the second last inequality applies Assumption~\ref{A2}-(c). 
		Hence we obtain the high probability upper bound of $A_3$. Then combining all the above results, we complete the proof for (\ref{G-2}).    
	\end{proof}

	\begin{proof}[Proof of Lemma~\ref{Gfunc}]
		Note that 
		\begin{align*}
		\widehat{G}(\vx_i^T\vbeta|\vbeta)-G(\vx_i^T\vbeta|\vbeta)=&\sum_{i=1}^n W_{ni}(t|\vbeta)\big(\widetilde{Y}_i - G(\vx_i^T\vbeta|\vbeta)\big)\\
		=&\frac{ (n-1)^{-1}\sum_{j=1,j\neq i}^n K_h(\vx_i^T\vbeta-\vx_j^T\vbeta) \big[\widetilde{Y}_j - G(\vx_i^T\vbeta|\vbeta)\big]}{ (n-1)^{-1}\sum_{j=1,j\neq i}^n K_h(\vx_i^T\vbeta-\vx_j^T\vbeta)}\\
		\triangleq& \frac{A_{n1}(\vx_i^T\vbeta|\vbeta)+A_{n2}(\vx_i^T\vbeta|\vbeta)}{A_{n3}(\vx_i^T\vbeta|\vbeta)},
		\end{align*}
		where
		\begin{align*}
		A_{n1}(\vx_i^T\vbeta|\vbeta) = &(n-1)^{-1}\sum_{j=1,j\neq i}^n K_h(\vx_i^T\vbeta-\vx_j^T\vbeta)\widetilde{\epsilon}_j,\\
		A_{n2}(\vx_i^T\vbeta|\vbeta) = &(n-1)^{-1}\sum_{j=1,j\neq j}^n K_h(\vx_i^T\vbeta-\vx_j^T\vbeta) \big[ f_0(\vx_j^T\vbeta_0) - G(\vx_i^T\vbeta|\vbeta)\big],\\
		A_{n3}(\vx_i^T\vbeta|\vbeta) =& (n-1)^{-1}\sum_{j=1,j\neq i}^n K_h(\vx_i^T\vbeta-\vx_i^T\vbeta).
		\end{align*} 
		Then Lemma~\ref{lem:An1bound}--\ref{lem:An3bound} provide the  high probability bounds for $A_{ni}$, $i=1,2,3$, as following:
		\begin{align*}
		P\left(\max_{  1\leq i\leq n}\sup \limits_{\vbeta \in \bbB} |A_{n1}(\vx_i^T\vbeta|\vbeta)| \geq c_0 h^2 \right)&\leq  \exp(-c_1nh^5), \\
		P\left(\max_{  1\leq i\leq n}\sup \limits_{\vbeta \in \bbB} |A_{n2}(\vx_i^T\vbeta|\vbeta)| \geq c_0h^2\right)&\leq \exp\big[-c_1\log(p\vee n)\big],\\
		P\left(\max_{1\leq i \leq n}\sup \limits_{\vbeta \in \bbB} \left|A_{n3}(\vx_i^T\vbeta|\vbeta)-\E\left[A_{n3}(\vx_i^T\vbeta|\vbeta)\right]\right| \geq c_0h^2\right)&\leq \exp(-c_1nh^5).
		\end{align*}
		for universal positive constants $c_0$, $c_1$, and all $n$ sufficiently large.
		
		We denote  the p.d.f of $\vx^T\vbeta$ as $f_{\vbeta}(\cdot)$. For $\E\left[A_{n3}(\vx_i^T\vbeta|\vbeta)\right]$, we have 
		\begin{align*}
		\E\left[A_{n3}(\vx_i^T\vbeta|\vbeta)\right]&= h^{-1} \int K\Big(\frac{\vx_i^T\vbeta-y}{h}\Big)f_{\vbeta}(y)dy \\
		&= \int K(-z)f_{\vbeta}(\vx_i^T\vbeta+hz)dz \\
		&=  \int K(-z)\Big[f_{\vbeta}(\vx_i^T\vbeta) + hz f_{\vbeta}'(\vx_i^T\vbeta) + \frac{h^2z^2}{2}f_{\vbeta}''(\widetilde{t})\Big]dz \\
		& = f_{\vbeta}(\vx_i^T\vbeta) + \frac{h^2}{2}\int  z^2K(-z) f_{\vbeta}''(\widetilde{t})dz,
		\end{align*}	
		where $\widetilde{t}$ is between $\vx_i^T\vbeta$ and $\vx_i^T\vbeta+hz$. In the above, the  second equality employs the transformation $z = (y-\vx_i^T\vbeta)/h$. Assumption~\ref{K1}--\ref{K3} imply that 
		$$\sup \limits_{\vbeta \in \bbB} \left| \frac{h^2}{2}\int  z^2K(-z) f_{\vbeta}''(\widetilde{t})dz\right| \leq c_0h^2,$$ 
		for some positive constant $c_0$.  Hence $f_{\vbeta}(\vx_i^T\vbeta)-c_0h^2\leq \E\left[A_{n3}(\vx_i^T\vbeta|\vbeta)\right]\leq f_{\vbeta}(\vx_i^T\vbeta)+c_0h^2$.
		Assumption~\ref{K3} implies that $\max_{ 1\leq i \leq n}\sup_{\vbeta\in\bbB}  f^{-1}_{\bm{\beta}}(\vx_i^T\vbeta) \leq M$, for some positive constant $M$. It ensures that 
		\begin{align}
		P\left(\max_{1\leq i \leq n}\sup \limits_{\vbeta \in \bbB} [A_{n3}(\vx_i^T\vbeta|\vbeta)]^{-1} \geq 2M\right)\leq \exp(-c_1nh^5),\label{A3bound}
		\end{align} 
		for universal positive constant $c_1$, and all $n$ sufficiently large.
		%$f^{-1}_{\vbeta}(t) = O(1)$ for any $ |t|\leq \sqrt{\log(p\vee n)}$ and $\vbeta$. According to the property of sub-Gaussian random variables, we know that $P(\max_{ 1\leq i \leq n}\sup_{\vbeta\in\bbB}|\vx_i^T\vbeta|\geq \sqrt{\log(p\vee n)})\leq \exp[-c_1\log(p\vee n)]$. 
		Then we conclude that for universal positive constants $c_0$ and $c_1$,  and all $n$ sufficiently large,
		\begin{align*}
		P\left(\max_{ 1\leq i \leq n} \sup\limits_{\vbeta \in \bbB}\big|\widehat{G}(\vx_i^T\vbeta|\vbeta)-G(\vx_i^T\vbeta|\vbeta)  \big| \geq c_0h^2\right)\leq  \exp\big[-c_1\log(p\vee n)\big].
		\end{align*} 
	\end{proof}

	\begin{proof}[Proof of Lemma~\ref{G1func}]
		Since 
		$$	\widehat{G}(\vx_i^T\vbeta|\vbeta)-G(\vx_i^T\vbeta|\vbeta)=  \frac{A_{n1}(\vx_i^T\vbeta|\vbeta)+A_{n2}(\vx_i^T\vbeta|\vbeta)}{A_{n3}(\vx_i^T\vbeta|\vbeta)},$$
		we have
		\begin{align*}
		\widehat{G}^{(1)}(\vx_i^T\vbeta|\vbeta)-G^{(1)}(\vx_i^T\vbeta|\vbeta)\triangleq&  \frac{A_{n1}^{(1)}(\vx_i^T\vbeta|\vbeta)+A_{n2}^{(1)}(\vx_i^T\vbeta|\vbeta)}{A_{n3}(\vx_i^T\vbeta|\vbeta)}\\
		& + \frac{A_{n1}(\vx_i^T\vbeta|\vbeta)+A_{n2}(\vx_i^T\vbeta|\vbeta)}{A_{n3}(\vx_i^T\vbeta|\vbeta)}* \frac{A_{n3}^{(1)}(\vx_i^T\vbeta|\vbeta)}{A_{n3}(\vx_i^T\vbeta|\vbeta)},
		\end{align*}
		where $G^{(1)}(t|\vbeta) = \frac{d}{d t}G(t|\vbeta)$, $\widehat{G}(t|\vbeta) = \frac{d}{d t}\widehat{G}(t|\vbeta)$, $A_{nk}^{(1)}(t|\vbeta) = \frac{d}{d t}A_{nk}(t|\vbeta) $, for $k=1,2,3$. Let $K_h(z) = h^{-1}K(z/h)$, and $K'_h(z) = h^{-2}K'(z/h)$. We have	
		\begin{align*}
		A_{n1}^{(1)}(\vx_i^T\vbeta|\vbeta) =& (n-1)^{-1}\sum_{j=1,j\neq i}^n K'_h(\vx_i^T\vbeta-\vx_j^T\vbeta)\widetilde{\epsilon}_j,\\
		A_{n2}^{(1)}(\vx_i^T\vbeta|\vbeta) =& (n-1)^{-1}\sum_{j=1,j\neq i}^n K'_h(\vx_i^T\vbeta-\vx_j^T\vbeta) \big[f_0(\vx_j^T\vbeta_0) - G(\vx_i^T\vbeta|\vbeta)\big] \\
		& -G^{(1)}(\vx_i^T\vbeta|\vbeta)n^{-1}\sum_{j=1,j\neq i}^n K_h(\vx_i^T\vbeta-\vx_j^T\vbeta)  \\
		\triangleq&  	A_{n21}^{(1)}(\vx_i^T\vbeta|\vbeta) -	A_{n22}^{(1)}(\vx_i^T\vbeta|\vbeta), \\
		A_{n3}^{(1)}(\vx_i^T\vbeta|\vbeta) =& (n-1)^{-1}\sum_{i=1}^n K'_h(\vx_i^T\vbeta-\vx_j^T\vbeta).
		\end{align*}
		
		First, similarly as in the proof of Lemma~\ref{lem:An1bound}, we can derive that
		$$P\left(\max_{1\leq i \leq n}\sup \limits_{ \vbeta \in \bbB} |A_{n1}^{(1)}(\vx_i^T\vbeta|\vbeta)| \geq c_0 h\right)\leq  \exp(-c_1nh^5).$$ 
		We denote  the p.d.f of $\vx^T\vbeta$ as $f_{\vbeta}(\cdot)$.  Note that $\vx_j^T\vbeta$ is independent of $\vx_i^T\vbeta$. We thus have
		\begin{align*}
		\E_{\vx_i^T\vbeta} \left[K'_h(\vx_i^T\vbeta-\vx_j^T\vbeta)\, \big|\vx_i^T\vbeta\right]	& =  h^{-2} \int K'\Big(\frac{\vx_i^T\vbeta-y}{h}\Big)   f_{\vbeta}(y) dy \\
		& =  h^{-1}\int K'(-z)   \Big[f_{\vbeta}(\vx_i^T\vbeta) + hzf_{\vbeta}'(\vx_i^T\vbeta) +  \frac{h^2z^2}{2}f_{\vbeta}''(\widetilde{t})   \Big] dz\\
		& =  f_{\vbeta}'(\vx_i^T\vbeta) + \frac{h}{2} \int z^2K'(-z)   f_{\vbeta}''(\widetilde{t}) dz,
		\end{align*}
		where $\widetilde{t}$ is between $\vx_i^T\vbeta$ and $\vx_i^T\vbeta +hz$. In the above, the second equality considers Taylor expansion at point $\vx_j^T\vbeta$, with the notation $z = -(\vx_j^T\vbeta-y)/h$, which is followed by $y = \vx_i^T\vbeta+hz$.
		Similarly as in the proof of Lemma~\ref{lem:An3bound}, we can show that for universal positive constants $c_0$ and $c_1$,
		$$P\left(\max_{1\leq i \leq n}\sup \limits_{\vbeta \in \bbB} \big|A_{n3}^{(1)}(\vx_i^T\vbeta|\vbeta)-f'_{\vbeta}(\vx_i^T\vbeta)\big| \geq c_0h \right)\leq  \exp(-c_1nh^5).$$
		Then the techniques in the proof of Lemma~\ref{lem:An1bound} and Lemma~\ref{lem:An2bound} can be applied to analyze $A_{n21}^{(1)}(\vx_i^T\vbeta|\vbeta)$ and $A_{n22}^{(1)}(\vx_i^T\vbeta|\vbeta)$. Observe that  
		\begin{align*}
		&\E \Big\{K'_h(\vx_i^T\vbeta-\vx_j^T\vbeta)\big[ f_0(\vx_j^T\vbeta_0) - G(\vx_i^T\vbeta|\vbeta)\big] \Big\}	\\ 
		= &	\E_{(\vx_i^T\vbeta,\vx_j^T\vbeta)} \left\{\E\left\{K'_h(\vx_i^T\vbeta-\vx_j^T\vbeta)\big[ f_0(\vx_j^T\vbeta_0) - G(\vx_i^T\vbeta|\vbeta)\big] \Big| \vx_i^T\vbeta, \vx_j^T\vbeta \right\} \right\}\\
		= &	\E_{(\vx_i^T\vbeta,\vx_j^T\vbeta)} \Big\{K'_h(\vx_i^T\vbeta-\vx_j^T\vbeta)\big[ G(\vx_j^T\vbeta|\vbeta)- G(t|\vbeta)\big]\Big\} \\
		=& h^{-1} \int K'(-z)\big[ G(\vx_i^T\vbeta+hz|\vbeta)- G(\vx_i^T\vbeta|\vbeta)\big]   f_{\vbeta}(\vx_i^T\vbeta+hz) dz \\
		=&   -\int K'(z)  \Big[z G^{(1)}(\vx_i^T\vbeta|\vbeta)+\frac{hz^2}{2} G^{(2)}(t_1|\vbeta)\Big]   \Big[f_{\vbeta}(\vx_i^T\vbeta) + hz f_{\vbeta}'(t_2 ) \Big]  dz\\
		= & G^{(1)}(\vx_i^T\vbeta|\vbeta) f_{\vbeta}(\vx_i^T\vbeta)  - \frac{h f_{\vbeta}(\vx_i^T\vbeta)}{2} \int z^2 K'(z) G^{(2)}(t_1|\vbeta) dz\\
		&  - h G^{(1)}(\vx_i^T\vbeta|\vbeta) \int z^2 K'(z) f_{\vbeta}'(t_2) dz - \frac{h^2}{2} \int z^3 K'(z) G^{(2)}(t_1|\vbeta)f_{\vbeta}'(t_2)  dz,
		\end{align*}
		where $t_1$ and $t_2$ are both between $\vx_i^T\vbeta$ and $\vx_i^T\vbeta+hz$. In the above, the second equality applies the independence between $\vx_i^T\vbeta$ and $\vx_j^T\vbeta$, and $G(\vx_j^T\vbeta|\vbeta) = \E\big[f_0(\vx_j^T\vbeta_0)|\vx_j^T\vbeta\big]$.
		Then Assumption \ref{K1}--\ref{K4} and the proofs in Lemma~\ref{lem:An2bound} and  Lemma~\ref{lem:An3bound} imply that for some constants $c_0$, $c_1$, and all $n$ sufficiently large,
		$$	P\left(\max_{1\leq i \leq n}\sup \limits_{\vbeta \in \bbB} \big|A_{n21}^{(1)}(\vx_i^T\vbeta|\vbeta) - G^{(1)}(\vx_i^T\vbeta|\vbeta) f_{\vbeta}(\vx_i^T\vbeta)\big| \geq c_0h\right)  \leq  \exp[-c_1\log(p\vee n)],$$ 
		$$	P\left(\max_{1\leq i \leq n}\sup \limits_{ \vbeta \in \bbB} \big|A_{n22}^{(1)}(\vx_i^T\vbeta|\vbeta) - G^{(1)}(\vx_i^T\vbeta|\vbeta) f_{\vbeta}(\vx_i^T\vbeta)\big| \geq c_0h\right)  \leq  \exp[-c_1\log(p\vee n)].$$
		This implies that 
		$	P\Big(\max_{1\leq i \leq n}\sup \limits_{\vbeta \in \bbB} \big|A_{n2}^{(1)}(\vx_i^T\vbeta|\vbeta) \big| \leq c_0h\Big)  \geq 1- 2  \exp[-c_1\log(p\vee n)].$ Assumption~\ref{K3} implies that $\max_{ 1\leq i \leq n}\sup_{\vbeta\in\bbB}  f^{-1}_{\bm{\beta}}(\vx_i^T\vbeta) \leq M$, for some positive constant $M$. by noting  the high probability bounds for $A_{ni}(\vx_i^T\vbeta|\vbeta) $ and $A_{ni}^{(1)}(\vx_i^T\vbeta|\vbeta) $, we conclude the lemma.
	\end{proof}

	\begin{proof}[Proof of Lemma~\ref{lem:Ebound}]  
		We will prove the first part of the claim below. The proof of the second and third parts is similar.
		
		Note that 
		\begin{align*}
		\widehat{\E}(\vx_i|\vx_i^T\vbeta)-\E(\vx_i|\vx_i^T\vbeta) = \frac{(n-1)^{-1}\sum_{j=1,j\neq i}^nK_h (\vx_i^T\vbeta-\vx_j^T\vbeta) \big[\vx_j- \E(\vx_i|\vx_i^T\vbeta)\big]}{A_{n3}(\vx_i^T\vbeta|\vbeta)}.
		\end{align*} 
		
		Let $B_{n}(\vx_i^T\vbeta,\vv|\vbeta) = [(n-1)h]^{-1}\sum_{j=1,j\neq i }^n \vgamma_i(z_j)^T\vv,$ where $\vgamma_i(z_j) =K \big(\frac{\vx_i^T\vbeta-\vx_j^T\vbeta}{h}\big) \big[\vx_j- \E(\vx_i|\vx_i^T\vbeta)
		\big]$. 	Lemma~\ref{lem:An3bound} and inequality (\ref{A3bound}) already provide a high probability bound for the denominator.
		It suffices to prove the high probability bound for $\sup \limits_{ \substack{\vbeta \in \bbB\\\vv\in\bbK(p,2ks)}} |B_{n}(\vx_i^T\vbeta,\vv|\vbeta)|$.
		
		We first derive the bound of $\big|\E\big[B_n(\vx_i^T\vbeta,\vv|\vbeta)\big] \big|$.
		Let $f_{\vbeta}(\cdot)$ denote the p.d.f of $\vx^T\vbeta$. Note  
		\begin{align*}
		\E [\vgamma_i(z_j)^T\vv ]=& \E \Big\{K \Big(\frac{\vx_i^T\vbeta-\vx_j^T\vbeta}{h}\Big) \big[\vx_j- \E(\vx_i|\vx_i^T\vbeta)\big]^T\vv\Big\}\\
		= & \E_{(\vx_i^T\vbeta,\vx_j^T\vbeta)} \Big\{\E \Big[K \Big(\frac{\vx_i^T\vbeta-\vx_j^T\vbeta}{h}\Big) \big[\vx_j- \E(\vx_i|\vx_i^T\vbeta)\big]^T\vv \big| \vx_i^T\vbeta, \vx_j^T\vbeta\Big]\Big\}\\
		= &  \E_{(\vx_i^T\vbeta,\vx_j^T\vbeta)} \Big\{K \Big(\frac{\vx_i^T\vbeta-\vx_j^T\vbeta}{h}\Big) \big[\E(\vx_j|\vx_j^T\vbeta)- \E(\vx_i|\vx_i^T\vbeta)\big]^T\vv\Big\}\\
		=&  \E_{\vx_i^T\vbeta }  \left\{h\int K(-z)  \big[ \E(\vx|\vx^T\vbeta=\vx_i^T\vbeta+hz) - 
		\E(\vx_i|\vx_i^T\vbeta)\big]^T\vv f_{\vbeta}(\vx_i^T\vbeta+hz) dz \right\}\\
		=& \E_{\vx_i^T\vbeta }  \bigg\{h \int K(-z)  \Big[\E^{(1)}(\vx|\vx^T\vbeta=\vx_i^T\vbeta)hz \\
		&+  \frac{h^2z^2}{2}  \E^{(2)}(\vx|\vx^T\vbeta=t_1)\Big]^T\vv \Big[f_{\vbeta}(\vx_i^T\vbeta) + hzf_{\vbeta}'(\widetilde{t})  \Big]\bigg\} dz\\
		= & \E_{\vx_i^T\vbeta }  \left\{\frac{h^3f_{\vbeta}(\vx_i^T\vbeta)}{2} \int z^2 K(-z)\E^{(2)}(\vx^T\vv|\vx^T\vbeta=t_1)dz\right\} \\
		&+ \E_{\vx_i^T\vbeta }  \left\{h^3\E^{(1)}(\vx^T\vv|\vx^T\vbeta=\vx_i^T\vbeta) \int z^2 K(-z) f_{\vbeta}'(\widetilde{t})  dz\right\} \\
		& +  \E_{\vx_i^T\vbeta }  \left\{\frac{h^4 }{2} \int z^3 K(-z) \E^{(2)}(\vx^T\vv|\vx^T\vbeta=t_1)f_{\vbeta}'(\widetilde{t})  dz\right\},
		\end{align*}
		where $t_1$ and $\widetilde{t}$ are both between $\vx_i^T\vbeta$ and $\vx_i^T\vbeta+hz$.
		In the above, the third equality uses the independence between $\vx_i^T\vbeta$ and $\vx_j^T\vbeta$.
		The  last equality uses $\int zK(-z)dz=0$. %the transformation $z = (\vx_j-\vx_i)^T\vbeta/h$.
		Assumptions~\ref{A2}--\ref{K3} imply that $ \big|\E\big[B_n(\vx_i^T\vbeta,\vv|\vbeta)\big] \big|\leq ch^2||\vv||_2$ for some positive constant $c$. Applying the same techniques as those in the proof of Lemma~\ref{lem:An1bound}, we can derive that  $ \E\{ [\vgamma_i(z_j)^T\vv ]^2\}\leq ch||\vv||_2^2$. 
		
		Next, we derive the high probability bound of $ |B_{n}(\vx_i^T\vbeta,\vv|\vbeta)|$.  
		Note that $\vx_j^T\vv$ is independent of $\vx_i^T\vbeta$. The sub-Gaussian property of $\vx_j^T\vv$ and $\E(\vx_i^T\vv|\vx_i^T\vbeta)$ implies that $\vx_j^T\vv - \E(\vx_i^T\vv|\vx_i^T\vbeta)$ is also sub-Gaussian. Since $K(\cdot)$ is bounded on the real line,  Lemma~\ref{lem:subg} implies that  $[ \vgamma_i(z_j) -  \E \vgamma_i(z_j)]^T\vv $ is sub-Gaussian. Then the tail probability inequality for sub-Gaussian implies that 
		\begin{align*}
		P\left( \left| \sum_{j=1,j\neq i}^n[ \vgamma_i(z_j) -  \E \vgamma_i(z_j)]^T\vv \right|  \geq t \ \Big|\ \vx_i^T\vbeta\right)  \leq 2\exp\left[ -\frac{t^2}{2c^2(n-1)||\vv||_2^2 h}\right] ,
		\end{align*} 
		for some positive constant $c$, where applies $ \E\{ [\vgamma_i(z_j)^T\vv ]^2\}\leq c^2h||\vv||_2^2$.
		%		
		%		 It implies that 
		%		$\E\left\{ \left|\vgamma_i(z_j)^T\vv -  \E [\vgamma_i(z_j)^T\vv]\right|^k\right\} \leq \frac{1}{2}L^{k-2}(k!)\E\left\{ \left|\vgamma_i(z_j)^T\vv -  \E [\vgamma_i(z_j)^T\vv]\right|^2\right\} $, for some positive constant $L$ and any integer $k\geq 1$. By Bernstein inequality, $\forall\ 0\leq v\leq (2L)^{-1}\sqrt{(n-1)\E\left\{ \left|\vgamma_i(z_j)^T\vv -  \E [\vgamma_i(z_j)^T\vv]\right|^2\right\} }$, we have that  
		%		\begin{align*}
		%		P\left( \left| \sum_{j=1,j\neq i}^n[ \vgamma_i(z_j) -  \E \vgamma_i(z_j)]^T\vv \right|  \geq 2v||\vv||_2 \sqrt{c(n-1)h}  \ \Big|\ \vx_i^T\vbeta\right)  \leq \exp(-v^2) ,
		%		\end{align*} 
		%		for some positive constant $c$, where applies $ \E\{ [\vgamma_i(z_j)^T\vv ]^2\}\leq ch||\vv||_2^2$.
		Taking $t=c(n-1) h^3$, we have that 
		\begin{align*}
		&	P\left( \left| \sum_{j=1,j\neq i}^n[ \vgamma_i(z_j) -  \E \vgamma_i(z_j)]^T\vv \right|  \geq c(n-1) h^3   \right)   \\
		=&\E_{\vx_i^T\vbeta}\left\{	P\left( \left| \sum_{j=1,j\neq i}^n[ \vgamma_i(z_j) -  \E \vgamma_i(z_j)]^T\vv \right|  \geq  c(n-1)h^3   \ \Big|\ \vx_i^T\vbeta\right)   \right\} \\
		\leq &2 \exp(-cnh^5),
		\end{align*} 
		for some positive constant $c$ and all $n$ sufficiently large, since $||\vv||_2\leq 1$. 
		Combining this with the bound of $\big|\E\big[B_n(\vx_i^T\vbeta,\vv|\vbeta)\big] \big|$, we conclude that 
		there exist some positive constants $c_0$ and $c_1$ such that for all $n$ sufficiently large,
		$$P\Big(  \big|B_n(\vx_i^T\vbeta,\vv|\vbeta)\big| \geq c_0h^2 \Big)  \leq  \exp(-c_1nh^5).$$
		
		%		\begin{align*}
		%		B_{n}(\vx_i^T\vbeta|\vbeta) =& [(n-1)h]^{-1}\sum_{j=1,j\neq i }^n K \big(\frac{\vx_i^T\vbeta-\vx_j^T\vbeta}{h}\big) \vx_j^T\vv    \\
		%		&- [(n-1)h]^{-1}\sum_{j=1,j\neq i }^n K \big(\frac{\vx_i^T\vbeta-\vx_j^T\vbeta}{h}\big)\E(\vx_i^T\vv |\vx_i^T\vbeta)\\
		%		=& 		B_{n1}(\vx_i^T\vbeta|\vbeta)  + B_{n2}(\vx_i^T\vbeta|\vbeta),
		%		\end{align*}
		%		where the definitions of $B_{n1}(\vx_i^T\vbeta|\vbeta) $ and $B_{n2}(\vx_i^T\vbeta|\vbeta) $ are clear from the context. Assumption~\ref{K1} implies that $K(\cdot)$ is bounded. Hence Hoeffding’s inequality indicates that 
		%		$$P\Big(||\Big)$$
		
		%	Let $a_1,\cdots,a_n$ be a series of  random variables independent of $(\vx_1,\cdots,\vx_n)$, such that $P(a_i=1)=P(a_i=-1)=1/2$.  Assumption~\ref{K1} implies that $K(\cdot)$ is bounded, hence by the definition of sub-Gaussian random variables, $a_j K \big(\frac{\vx_i^T\vbeta-\vx_j^T\vbeta}{h}\big)$ is also sub-Gaussian. Lemma~\ref{lem:subg} implies that $\vx_j$ and $ \E(\vx_i|\vx_i^T\vbeta)$ are independent and both sub-Gaussian. Hence $a_j \big[\vx_j- \E(\vx_i|\vx_i^T\vbeta)	\big]^T\vv$ is sub-Gaussian.
		
		%Therefore, Lemma~\ref{lem14NCL} implies that there exist universal positive constants $c_0$ and $c_1$ such that for all $n$ sufficiently large,
		%$$P\Big(  \big|B_n(\vx_i^T|\vbeta)\big| \geq c_0h^2||\vv||_2 \Big)  \leq  \exp [-c_1\log  (p\vee n)].$$
		
		To obtain the uniform bound, we will cover $\bbB$ with $N_1$ $L_2-$balls of radius $\delta$. Denote the centers by $\vbeta_1^*,\cdots,\vbeta_{N_1}^*$. Similarly in the proof of Lemma~\ref{lem:An1bound}, we can cover $\bbK(p,2ks)$ with $N_2$ $L_2-$balls of radius $\delta$. Denote their centers by $\vv_1^*,\cdots,\vv_{N_2}^*$. Let $\mathcal{N}_{\delta}$ be the this joint cover  of $\bbB\times\bbK(p,2ks)$. We can construct the covers such that $N \triangleq |\mathcal{N}_{\delta}|= N_1*N_2 \leq c p^{4ks}\delta^{-4ks}$ for some positive constant $c$. 
		%we cover $\bbB\times\bbK(2ks)$ with $L_2-$balls with radius $\delta\times\delta$, where $\bbK(2ks) = \{\vv\in\bbR^p:||\vv||_2\leq 1,||\vv||_0\leq 2ks\}$.  
		%Let $\mathcal{N}_{\delta}$ be the $\delta\times\delta-$cover  of $\bbB\times\bbK(2ks)$. 	The covering number $N = |\mathcal{N}_{\delta}|$ satisfies   $N \leq c p^{4ks}\delta^{-4ks}$ for some positive constant $c$, as shown in the proof of Lemma~\ref{lem:An1bound}.  
		Given any $\vbeta\in\bbB$ and $\vv\in\bbK(p,2ks)$, we can find $(\vbeta^*,\vv^*)\in\mathcal{N}_{\delta}$ such that $||\vbeta-\vbeta^*||_2\leq \delta$, and $||\vv-\vv^*||_2\leq \delta$. We have
		%For any $\vbeta $ in such a ball with center $\vbeta^*$, we have
		\begin{align*}
		&\Big|(n-1)^{-1}\sum_{j=1,j\neq i}^nK_h(\vx_i^T\vbeta-\vx_j^T\vbeta)\big[\vx_j-\E(\vx_i|\vx_i^T\vbeta)\big]^T\vv\\
		&-(n-1)^{-1}\sum_{j=1,j\neq i}^nK_h(\vx_i^T\vbeta^*-\vx_j^T\vbeta^*)\big[\vx_j-\E(\vx_i|\vx_i^T\vbeta^*)\big]^T \vv^*\Big|\\
		\leq&\Big|(n-1)^{-1}\sum_{j=1,j\neq i}^n \Big[K_h(\vx_i^T\vbeta-\vx_j^T\vbeta)-K_h(\vx_i^T\vbeta^*-\vx_j^T\vbeta^*) \Big] \big[\vx_j- \E(\vx_j|\vx_j^T\vbeta) \big]^T\vv\Big|\\  
		&+\Big|(n-1)^{-1}\sum_{j=1,j\neq i}^n \Big[K_h(\vx_i^T\vbeta-\vx_j^T\vbeta)-K_h(\vx_i^T\vbeta^*-\vx_j^T\vbeta^*) \Big] \big[ \E(\vx_j|\vx_j^T\vbeta)- \E(\vx_i|\vx_i^T\vbeta) \big]^T\vv\Big|\\   
		& +\Big| (n-1)^{-1}\sum_{j=1,j\neq i}^n  K_h(\vx_i^T\vbeta^*-\vx_j^T\vbeta^*)\big[ \E(\vx_i|\vx_i^T\vbeta^*)  -  \E(\vx_i|\vx_i^T\vbeta) \big]^T\vv\Big|  \\ 
		&+\Big| (n-1)^{-1}\sum_{j=1,j\neq i}^nK_h(\vx_i^T\vbeta^*-\vx_j^T\vbeta^*)\big[\vx_j-\E(\vx_i|\vx_i^T\vbeta^*)\big]^T (\vv-\vv^*)\Big|\\
		\triangleq&\sum_{i=1}^4|I_{ni}|, 
		\end{align*} 
		where the definition of $I_{ni}$ is clear from the context. 	 
		%The sub-Gaussian property of $\vx_i^T\vbeta$, $\E(\vx_i^T\vv|\vx_i^T\vbeta) $ and $\vx_i^T(\vbeta-\vbeta^*)$ 
		Lemma~\ref{lem15NCL} implies that 
		\begin{align}
		&P\left(\max_{1\leq i\leq n} |\vx_i^T\vbeta | \geq  \sigma_x   \sqrt{s\log(p\vee n)}||\vbeta||_2,\ \forall \ \vbeta\in\bbB\right) 
		\leq \exp[-cs\log(p\vee n)],\label{A7-1}\\
		&P\left(\max_{1\leq i\leq n}\sup_{\vv\in\bbK(p,4ks)}\big| \E(\vx_i^T\vv|\vx_i^T\vbeta) \big|\geq\sigma_x   \sqrt{s\log(p\vee n)},\ \forall \ \vbeta\in\bbB\right)\leq \exp[-cs\log(p\vee n)], \label{A7-2}\\
		&P\left(\max_{1\leq i\leq n}  |\vx_i^T(\vbeta-\vbeta^*)| \geq  \sigma_x  \sqrt{s\log(p\vee n)}||\vbeta-\vbeta^*||_2,\ \forall\ \vbeta,\vbeta^*\in\bbB\right) 
		\leq \exp[-cs\log(p\vee n)],\label{A7-3}
		\end{align} 
		for some positive constant $c$, and all $n$ sufficiently large, where the analysis of (\ref{A7-2}) is similar as the proofs of Lemma~\ref{lem:events} and Lemma~\ref{lem:Ex_eigen}. Since $\vx_j^T(\vbeta-\vbeta^*)$ are independent sub-Gaussian random variables, Lemma~\ref{lem15NCL} implies that  
		\begin{align*}
		&P\bigg( (n-1)^{-1}\left|\sum_{j=1,j\neq i}^n |\vx_j^T(\vbeta-\vbeta^*)|^2 - (\vbeta-\vbeta^*)^T\E(\vx\vx^T)(\vbeta-\vbeta^*)\right| \\
		&\qquad\qquad\qquad\qquad\quad\geq c_0 \sigma_x^2  \sqrt{\frac{s\log(p\vee n)}{n}}||\vbeta-\vbeta^*||_2^2,\forall \ \vbeta,\vbeta^*\in\bbB\bigg)\\
		\leq &\exp[-cs\log(p\vee n)],
		\end{align*}
		for some positive constants $c_0$, $c$, and all $n$ sufficiently large. Assumption~\ref{A2} implies that $(\vbeta-\vbeta^*)^T\E(\vx\vx^T)(\vbeta-\vbeta^*)\leq \xi_3||\vbeta-\vbeta^*||_2^2$. 
		Hence we have that 
		\begin{align}
		(n-1)^{-1}\sum_{j=1,j\neq i}^n |\vx_j^T(\vbeta-\vbeta^*)|^2
		&\leq \left(\xi_3+c_0\sigma_x^2\sqrt{\frac{s\log(p\vee n)}{n}}\right) ||\vbeta-\vbeta^*||_2^2\nonumber\\
		&\leq \sigma_x^2s\log (p\vee n) ||\vbeta-\vbeta^*||_2^2 ,\ \forall \ \vbeta,\vbeta^*\in\bbB\label{A7-4}
		\end{align}
		with probability at least $1-\exp[-c_1\log(p\vee n)]$, 	for some positive constants $c_0$, $c_1$, and all $n$ sufficiently large.   
		Similarly the sub-Gaussian property of   $\E(\vx_j^T\vv|\vx_j^T\vbeta)$ and $\vx_j^T\vv- \E(\vx_j^T\vv|\vx_j^T\vbeta)$,  Lemma~\ref{lem14NCL} and Lemma~\ref{lem15NCL} imply that  
		\begin{align}
		&P\left(\sup_{\vv\in\bbK(p,4ks)}\frac{1}{n-1}\sum_{j=1,j\neq i}^n  \big[ \E(\vx_j^T\vv|\vx_j^T\vbeta) \big]^2\geq  c_0\sigma_x^2,\ \forall \ \vbeta\in\bbB  \right)\leq \exp[-c\log(p\vee n)],\label{A7-5} \\
		&P\left(\sup_{\vv\in\bbK(p,4ks)}\frac{1}{n-1}\sum_{j=1,j\neq i}^n  \big( \vx_j^T\vv\big)^2\geq  c_0\sigma_x^2  \right)\leq \exp[-c\log(p\vee n)],\label{A7-7} \\
		&P\left(\sup_{\vv\in\bbK(p,4ks)}\frac{1}{n-1}\sum_{j=1,j\neq i}^n  \big[\vx_j^T\vv- \E(\vx_j^T\vv|\vx_j^T\vbeta) \big]^2\geq c_0\sigma_x^2,\ \forall \ \vbeta\in\bbB  \right)\leq \exp[-c\log(p\vee n)], \label{A7-6}
		\end{align} 
		for some positive constants $c_0$, $c$, and all $n$ sufficiently large, where the analysis of (\ref{A7-5}) and (\ref{A7-6}) is similar as the proofs of Lemma~\ref{lem:events} and Lemma~\ref{lem:Ex_eigen}. Denote the event 
		\begin{align*}
		\mathcal{E}=\bigg\{   &\max_{ 1\leq i \leq n} |\vx_i^T(\vbeta-\vbeta^*)|^2+\frac{1}{n-1}\sum_{j=1,j\neq i}^n |\vx_j^T(\vbeta-\vbeta^*)|^2   
		\leq 2\sigma_x^2s\log(p\vee n)||\vbeta-\vbeta^*||_2^2,\ \forall \ \vbeta,\vbeta^* \in \bbB \bigg\}\\
		\bigcap\Big\{&\max_{ 1\leq i \leq n}\sup_{\substack{\vbeta\in\bbB\\\vv\in\bbK(p,4ks)}}\Big[\big|\E(\vx_i^T\vv|\vx_i^T\vbeta) \big|^2+\frac{1}{n-1}\sum_{j=1,j\neq i}^n \big|\E(\vx_j^T\vv|\vx_j^T\vbeta) \big|^2\Big] 
		\leq 2c_0\sigma_x^2s\log(p\vee n) \Big\}\\
		\bigcap\Big\{&  \sup_{\substack{\vbeta\in\bbB\\\vv\in\bbK(p,4ks)}} \frac{1}{n-1}\sum_{j=1,j\neq i}^n \big[ \vx_j^T\vv-\E(\vx_j^T\vv|\vx_j^T\vbeta) \big]^2   \leq 2\sigma_x^2\Big\}\\
		\bigcap\Big\{& \sup_{\vv\in\bbK(p,4ks)} \frac{1}{n-1}\sum_{j=1,j\neq i}^n (\vx_j^T\vv)  ^2  \leq 2\sigma_x^2\Big\} 
		\bigcap\Big\{ \max_{1\leq i \leq n}  |\vx_i^T\vbeta|\leq \sigma_x \sqrt{s\log(p\vee n)} ||\vbeta||_2,\ \forall \ \vbeta \in \bbB\Big\}.
		\end{align*} 
		Combining (\ref{A7-1}) -- (\ref{A7-6}), we have $P(\mathcal{E})\geq 1- 7\exp[-c_1s\log(p\vee n)]$, for some positive constants $c_0$, $c_1$, and all $n$ sufficiently large.%, uniformly in $\vbeta,\vbeta^*\in\bbB$, $\vv\in\bbK(4ks)$.

		Take $\delta = \frac{ h^4}{8s\log (p\vee n)}$.	 According to Assumption~\ref{K1}, %the Lipschitz condition for $K(\cdot)$, 
		we observe
		\begin{align*}
		|I_{n1}| \leq &  \sqrt{(n-1) ^{-1}\sum_{j=1,j\neq i}^n\Big| K_h(\vx_i^T\vbeta-\vx_j^T\vbeta)-K_h(\vx_i^T\vbeta^*-\vx_j^T\vbeta^*)\Big|^2}\\
		&*\sqrt{(n-1)^{-1}\sum_{j=1,j\neq i}^n  \Big| \big[\vx_j- \E(\vx_j|\vx_j^T\vbeta) \big]^T\vv\Big|^2}\\
		\leq& 2h^{-2} \sqrt{|\vx_i^T(\vbeta-\vbeta^*)|^2 + (n-1) ^{-1}\sum_{j=1,j\neq i}^n  |\vx_j^T(\vbeta-\vbeta^*)|^2}\\
		&*\sqrt{(n-1)^{-1}\sum_{j=1,j\neq i}^n  \big[\vx_j^T\vv- \E(\vx_j^T\vv|\vx_j^T\vbeta) \big]^2}\\
		\leq & c_0h^{-2}\delta||\vv||_2\sqrt{s\log(p\vee n)}\leq \frac{c_0h^2}{8\sqrt{s\log (p\vee n)}},
		\end{align*} 
		on the event $\mathcal{E}$, for some constant $c_0>0$, and all $n$ sufficiently large, since $||\vv||_2\leq 1$. In the above, the last inequality applies the first and the third events in $\mathcal{E}$.

		Similarly for $I_{n2}$,  we have
		\begin{align*}
		|I_{n2}| \leq&  \sqrt{(n-1) ^{-1}\sum_{j=1,j\neq i}^n\Big| K_h(\vx_i^T\vbeta-\vx_j^T\vbeta)-K_h(\vx_i^T\vbeta^*-\vx_j^T\vbeta^*)\Big|^2}\\
		&*\sqrt{(n-1)^{-1}\sum_{j=1,j\neq i}^n  \Big| \big[\E(\vx_j|\vx_j^T\vbeta)- \E(\vx_i|\vx_i^T\vbeta)\big]^T\vv\Big|^2}\\
		\leq& 2h^{-2} \sqrt{|\vx_i^T(\vbeta-\vbeta^*)|^2 + (n-1) ^{-1}\sum_{j=1,j\neq i}^n  |\vx_j^T(\vbeta-\vbeta^*)|^2}\\
		&*\sqrt{ \big| \E(\vx_i^T\vv|\vx_i^T\vbeta)\big|^2  + (n-1)^{-1}\sum_{j=1,j\neq i}^n \big| \E(\vx_j^T\vv|\vx_j^T\vbeta)\big|^2}\\
		\leq &c_0h^{-2}\delta||\vv||_2 s\log(p\vee n)\leq c_0h^2/8,
		\end{align*}
		on the event $\mathcal{E}$, for some constant $c_0>0$, and all $n$ sufficiently large. In the above, the last inequality applies the first and the second events in $\mathcal{E}$.

		Assumption~\ref{A2}-(c) implies
		\begin{align*}
		|I_{n3}| &\leq ch^{-1} \Big| \big[ \E(\vx_i|\vx_i^T\vbeta^*)  -  \E(\vx_i|\vx_i^T\vbeta) \big]^T\vv\Big| \\
		&\leq ch^{-1}||\vv||_2 \Big[|\vx_i^T(\vbeta-\vbeta^*)| + (|\vx_i^T\vbeta|+|\vx_i^T\vbeta^*|)*||\vbeta-\vbeta^*||_2 \Big]\\
		&\leq c_0h^{-1}\delta||\vv||_2\sqrt{s\log(p\vee n)}\leq \frac{c_0h^3}{8\sqrt{s\log (p\vee n)}},
		\end{align*}
		on the event $\mathcal{E}$, for some positive constants $c$, $c_0$, and all $n$ sufficiently large. In the above, the second last inequality applies the first and the last events in $\mathcal{E}$.
		
		Assumption~\ref{K1}  implies that 
		\begin{align*}
		|I_{n4}| &\leq ch^{-1} (n-1)^{-1}\sum_{j=1,j\neq i}^n\Big|\big[\vx_j-\E(\vx_i|\vx_i^T\vbeta^*)\big]^T (\vv-\vv^*)\Big|\\
		&\leq \sqrt{2}ch^{-1} \sqrt{ \big[\E(\vx_i|\vx_i^T\vbeta^*)^T(\vv-\vv^*)\big]^2+(n-1)^{-1}\sum_{j=1,j\neq i}^n \big[\vx_j ^T (\vv-\vv^*)\big]^2  } \\
		&\leq c_0h^{-1} \sqrt{s\log(p\vee n)}||\vv-\vv^*||_2 \leq \frac{c_0h^3}{8\sqrt{s\log (p\vee n)}},
		\end{align*}
		on the event $\mathcal{E}$, for some positive constants $c$, $c_0$, and all $n$ sufficiently large. In the above, the second last inequality applies the second and the fourth events in $\mathcal{E}$, since $\big[\E(\vx_i|\vx_i^T\vbeta^*)^T(\vv-\vv^*)\big]^2\leq \big|\E(\vx_i|\vx_i^T\vbeta)^T(\vv-\vv^*) \big|^2+\frac{1}{n-1}\sum_{j=1,j\neq i}^n \big|\E(\vx_j|\vx_j^T\vbeta)^T(\vv-\vv^*) \big|^2$, for any $\vv-\vv^*\in\bbK(p,4ks)$.

		%Denote $B_{n}(\vx_i^T\vbeta^*|\vbeta^*) = [(n-1)h]^{-1}\sum_{j=1,j\neq i }^nK \big(\frac{\vx_i^T\vbeta-\vx_j^T\vbeta}{h}\big) \big[\vx_j- \E(\vx_i|\vx_i^T\vbeta) \big]^T\vv^*$.
		Combining  the above results, we conclude that on the event $\mathcal{E}$,
		\begin{align*}
		&\Big|(n-1)^{-1}\sum_{j=1,j\neq i}^nK_h(\vx_i^T\vbeta-\vx_j^T\vbeta)\big[\vx_j-\E(\vx_i|\vx_i^T\vbeta)\big]^T\vv\\
		&-(n-1)^{-1}\sum_{j=1,j\neq i}^nK_h(\vx_i^T\vbeta^*-\vx_j^T\vbeta^*)\big[\vx_j-\E(\vx_i|\vx_i^T\vbeta^*)\big]^T \vv^*\Big|\\
		\leq& \frac{c_0h^2}{2},
		\end{align*} 
		for some positive constant $c_0$, and all $n$ sufficiently large. Hence,
		\begin{align*}
		&P\Big(\sup \limits_{\substack{\vbeta \in \bbB\\\vv\in\bbK(p,2ks)}} |B_{n}(\vx_i^T\vbeta,\vv|\vbeta)| \geq c_0h^2 \Big) \\
		\leq &P\Big(\bigcup\limits_{(\vbeta^*,\vv^*)\in \mathcal{N}_{\delta}}\big|B_{n}(\vx_i^T\vbeta^*,\vv^*|\vbeta^*) \big|\geq  c_0h^2 /2\Big)\\
		&+ P\Big(\sup\limits_{( \vbeta^*,\vv^*)\in \mathcal{N}_{\delta}}\sup\limits_{ \substack{ ||\vbeta-\vbeta^*||_2\leq \delta\\ ||\vv-\vv^*||_2\leq \delta}}|B_{n}(\vx_i^T\vbeta,\vv|\vbeta)-B_{n}(\vx_i^T\vbeta^*,\vv^*|\vbeta^*)| \geq c_0h^2/2 \Big) \\
		\leq& \sum_{( \vbeta^*,\vv^*)\in \mathcal{N}_{\delta} }P\Big(\big|B_{n}(\vx_i^T\vbeta^*,\vv^*|\vbeta^*) \big|\geq   c_0h^2 /2\Big)+7\exp[-c_1s\log (p\vee n)] \\
		\leq& c p^{4ks}\delta^{-4ks}\exp(-c_1nh^5) + 7\exp[-c_1s\log (p\vee n)] \\
		=&\exp[-c_2\log(p\vee n)],
		\end{align*}   
		for some positive constants $c_0$,  $c_1$, $c_2>1$, and all $n$ sufficiently large,

		Hence there exist some universal positive constants $d_0$ and $d_1$,  such that for all $n$ sufficiently large,
		\begin{align*}
		P\Big(\max_{ 1\leq i \leq n} \sup\limits_{\substack{\vbeta \in \bbB\\\vv\in\bbK(p,2ks)}}|B_{n}(\vx_i^T\vbeta,\vv|\vbeta)| \geq d_0h^2 \Big)\leq&\sum_{i=1}^n P\Big(\sup \limits_{\substack{\vbeta \in \bbB\\\vv\in\bbK(p,2ks)}} |B_{n}(\vx_i^T\vbeta,\vv|\vbeta)| \geq d_0h^2 \Big) \\
		\leq & \exp[-d_1\log(p\vee n)],
		\end{align*} 
		which concludes the lemma.
	\end{proof}

	\begin{proof}[Proof of Lemma~\ref{lem:ghat1_g1}]
		
		We will prove (\ref{ghat1_g1}) below. The proof of (\ref{ghat1_g1_inf}) is similar.
		
		Theorem~\ref{Lasso_error} implies that $P\left(\vbetah\in\bbB_1\right)\geq 1-\exp(-c\log p)$, for some positive constant $c$, and all $n$ sufficiently large. Lemma~\ref{Gbetafunc} implies that 
		$$P\left(\max_{ 1\leq i \leq n}\left|\widehat{G}^{(1)} (\vx_i^T\vbetah|\vbetah)-G^{(1)} (\vx_i^T\vbeta_0|\vbeta_0)  \right| \geq c_0h \right) \leq  \exp(-c_1\log p),$$
		for some positive constants $c_0$, $c_1$, and all $n$ sufficiently large.
		Define
		\begin{align*}
		\bbM = &\Big\{m(\cdot|\vbeta): \vbeta\in\bbB_1,\ m(\cdot|\vbeta)\in C_1^1(T),\ \forall\ \vbeta,\mbox{ and }\\
		&\max_{  1\leq i\leq n}\sup_{\vbeta\in\bbB_1}\left|m(\vx_i^T\vbeta|\vbeta)-G^{(1)}(\vx_i^T\vbeta_0|\vbeta_0 )\right|\leq c_1h\Big\},
		\end{align*} 
		for some positive constant $c_1$, where $T=\{t\in\bbR: |t|\leq 2||\vbeta_0||_2\sigma_x\sqrt{\log (p\vee n)}\}$, and $C_1^1(T)$ is the set of all continuous and Lipschitz functions $f:T\ra \bbR$. 
		To prove (\ref{ghat1_g1}), it is sufficient to prove that there exist some positive constants $c_0$, $c_1$, such that for all $n$ sufficiently large,
		\begin{align*}
		P\Big(\max_{2\leq j \leq p} \sup_{\vbeta\in\bbB_1,m\in\bbM} \big|n^{-1/2}\sum_{i=1}^n\vtheta_j^T\vgamma(Z_i,\vbeta,m) \big| \geq c_0\left[h^2\log(p\vee n)\right]^{1/4}\Big)\leq\exp[-c_1\log (p\vee n)],  
		\end{align*} 
		where $Z_i=(\vx_i,\epsilon_i,A_i)$, $\vgamma(Z_i,\vbeta,m) = \big[m(\vx_i^T\vbeta|\vbeta) - G^{(1)}(\vx_i^T\vbeta_0|\vbeta_0)\big] \widetilde{\epsilon}_i \big[\vx_{i,-1}-\E(\vx_{i,-1}|\vx_i^T\vbeta_0)\big] $,  $\widetilde{\epsilon}_i=2(2A_i-1)[\ep_i+g(\vx_i)]$, and $m(\vx^T\vbeta|\vbeta)$ depends on $\vx$ only through $\vx^T\vbeta$.

		We have
		\begin{align*}
		&P\left(\max_{2\leq j \leq p}\sup_{\vbeta\in\bbB_1,m\in\bbM} \Big|n^{-1/2}\sum_{i=1}^n\vtheta_{j}^T \vgamma(Z_i,\vbeta,m) \Big|>t\right)\\
		\leq &P\left(\max_{2\leq j \leq p}\sup_{\vbeta\in\bbB_1,m\in\bbM}\Big|n^{-1/2}\sum_{i=1}^n\vtheta_{j}^T \vgamma(Z_i,\vbeta,m) \Big|>t\ \bigg|\mathcal{H}_n \cap\mathcal{J}_n\cap \mathcal{K}_n\right) +\exp[-c\log (p\vee n)]\\
		\leq &\sum_{j=2}^pP\left(\sup_{\vbeta\in\bbB_1,m\in\bbM}\Big|n^{-1/2}\sum_{i=1}^n\vtheta_{j}^T \vgamma(Z_i,\vbeta,m) \Big|>t\ \bigg|\mathcal{H}_n \cap \mathcal{J}_n\cap \mathcal{K}_n \right) +\exp[-c\log(p\vee n)],
		\end{align*}
		for some positive constant $c$ and all $n$ sufficiently large, where the events $\mathcal{H}_n$, $\mathcal{J}_n$, and  $\mathcal{K}_n$ are defined in Lemma~\ref{lem:events}. 
		In the proof below, we write $\vgamma(Z_i,\vbeta,m) $ as $\vgamma_{i}(\vbeta,m) $ for simplicity.
		
		Note that $\vtheta_{j}^T\vgamma_{i}(\vbeta,m) = 2(2A_i-1)\left[\epsilon_i+g(\vx_i)\right]\big[m(\vx_i^T\vbeta|\vbeta) - G^{(1)}(\vx_i^T\vbeta_0|\vbeta_0)\big] \big[\vx_{i,-1} -\E(\vx_{i,-1}|\vx_i^T\vbeta_0)\big]^T\vtheta_j $, where $(2A_i-1)$ is a Rademacher sequence, and independent of $(\vx_i,\epsilon_i)$.  Hence given $\{(\vx_i,\ep_i)\}_{i=1}^n$, on the event $\mathcal{H}_n $, we have 
		$ \sup_{\vbeta\in\bbB_1,m\in\bbM} n^{-1} \sum_{i=1}^n\left| \vtheta_{j}^T\vgamma_{i}(\vbeta,m) \right|^2\leq Ch^2,$
		for some positive constant  $C$, and any $j$. Therefore, %$ h ^{-1} \vtheta_{j}^T\vgamma_{i}(\vbeta,m) $ satisfies the conditions of
		by Massart's concentration inequality (e.g., Theorem 14.2, \citet{buhlmann2011}), given $\{(\vx_i,\ep_i)\}_{i=1}^n$, $\forall\ t>0$,
		\begin{align*} 
		&P\bigg( \sup_{\vbeta\in\bbB_1,m\in\bbM} \Big|(nh)^{-1}\sum_{i=1}^n \vtheta_{j}^T\vgamma_{i}(\vbeta,m)\Big|\geq (\sqrt{n}h )^{-1}\E_{sup}\big[\vtheta_{j}^T\vgamma_i(\vbeta,m)\big] + t\ \bigg| \{(\vx_i,\ep_i)\}_{i=1}^n, \mathcal{H}_n  \cap\mathcal{J}_n\cap \mathcal{K}_n\bigg)\\
		\leq& \exp\left(-\frac{nt^2}{8}\right),
		\end{align*}
		where $\E_{sup}\big[\vtheta_{j}^T\vgamma_i(\vbeta,m)\big] = \E\left[\sup_{\vbeta\in\bbB_1,m\in\bbM}\Big| n^{-1/2}\sum_{i=1}^n\vtheta_{j}^T\vgamma_{i}(\vbeta,m)\Big| \ \bigg|\{(\vx_i,\ep_i)\}_{i=1}^n,\mathcal{H}_n \cap\mathcal{J}_n\cap \mathcal{K}_n\right].$
		Equivalently,  $\forall\ t>0$,
		\begin{align} 
		&P\bigg( \sup_{\vbeta\in\bbB_1,m\in\bbM} \Big|n^{-1/2}\sum_{i=1}^n \vtheta_{j}^T\vgamma_{i}(\vbeta,m) \Big|\geq\E_{sup}\big[\vtheta_{j}^T\vgamma_i(\vbeta,m)\big] +\sqrt{n} ht   \ \bigg|\{(\vx_i,\ep_i)\}_{i=1}^n,\mathcal{H}_n \cap\mathcal{J}_n\cap \mathcal{K}_n \bigg) \nonumber\\
		\leq &\exp\left(-\frac{nt^2}{8}\right).\label{mass_1}
		\end{align}
		Next we will derive  an upper bound for $\E_{sup}\big[\vtheta_{j}^T\vgamma_i(\vbeta,m)\big].$ Let $M_1$, $\cdots$, $M_{m(s)}$ denote all possible subsets of $\{1,\cdots,p\}$, corresponding to different submodels of sizes at most $ks$. Note that $m(s)\leq {p\choose ks}$. Let $S_{M_l} = \{\vbeta\in\bbB_1:
		%\bbR^p: ||\vbeta-\vbeta_0||_2\leq c_0\sqrt{s} h^2, 
		\mbox{supp}(\vbeta) = M_l\}$, where supp$(\vbeta)$ denotes the support set of $\vbeta$. Then $\bbB_1 = \bigcup_{l=1}^{m(s)}S_{M_l} $.
		
		Given any $\delta_n>0$, (w.l.o.g., $\delta_n\leq c\sqrt{s}h^2$), for each $S_{M_l} $, $l=1,\cdots,m(s)$, we can cover it by $L_2-$balls of radius $\delta_n$. Note that this cover has cardinality %can be constructed such that 
		$$N_l\leq \left(1+\frac{2c_0\sqrt{s}h^2}{\delta_n}\right)^{ks},$$
		for some positive constant $c_0$.
		Denote the centers of these $L_2-$balls by $\vbeta^{\circ}_{l0},\cdots,\vbeta^{\circ}_{lN_l}$.
		Denote the collection of these $L_2-$balls by $\mathbb{C}(\vbeta^{\circ}_{ll'})$, $l=1,\cdots,m(s)$, $l'=1,\cdots,N_l$.

		Observe that on the event $\mathcal{J}_n$,  $\max_{1\leq i \leq n}\sup_{\vbeta\in\bbB_1}|\vx_i^T\vbeta|\in T$, then $\bigcup_{\vbeta\in\bbB_1}m(\cdot|\vbeta)\in C_1^1(T)$. By Theorem 2.7.1 in \citet{van1996weak}, the entropy of the $\delta_n-$covering number of $C_1^1(T)$ satisfies
		$$\log N\left(\delta_n, C_1^1(T),L_2(\bbP_n)\right)\leq C||\vbeta_0||_2\frac{\sqrt{\log(p\vee n)}}{\delta_n},$$
		for some positive constant $C$. So we can cover $C_1^1(T)$ with  $N_2\leq \exp\big[C||\vbeta_0||_2\delta_n^{-1}\sqrt{\log(p\vee n)}\big]$ $L_2-$balls of radius $\delta_n$.  Let the centers of these $L_2(\bbP_n)-$balls of the cover be $m_a^{\circ} (\cdot)$, $a=1,\cdots,N_2$. %Then for any $\vbeta\in\bbB_1$, and any $m(\cdot|\vbeta)\in\bbM$, 
		Hence $\forall\ \vbeta\in\bbB_1$, $m(\cdot|\vbeta)\in C_1^1(T)$,  we can find $l,\ l'$, and  a function $m_a^{\circ} (\cdot): T\ra \bbR$ such that $\vbeta\in \mathbb{C}(\vbeta^{\circ}_{ll'})$, and  
		\begin{align}
		n^{-1} \sum_{i=1}^n  [m(\vx_i^T\vbeta_1|\vbeta)-m_a^{\circ}(\vx_i^T\vbeta_1)]^2 \leq \delta_n^2,\ \ \forall \vbeta_1\in\bbB_1 .\label{cover_m}
		\end{align} 
		On the event $\mathcal{H}_n\cap \mathcal{J}_n \cap \mathcal{K}_n$, there exist some positive constants $c_0$, $c_1$ and $c_2$, such that
		\begin{align*}
		&n^{-1}\left| \sum_{i=1}^n\vtheta_j^T \vgamma_{i}(\vbeta,m)-\sum_{i=1}^n \vtheta_j^T\vgamma_{i}(\vbeta^{\circ}_{ll'},m_a^{\circ})\right|\\
		\leq &n^{-1}\left| \sum_{i=1}^n\left[m(\vx_i^T\vbeta|\vbeta) - m_a^{\circ}(\vx_i^T\vbeta^{\circ}_{ll'})\right]\widetilde{\epsilon}_i [\vx_{i,-1} -\E(\vx_{i,-1}|\vx_i^T\vbeta_0)] ^T\vtheta_{j}\right| \\
		\leq &\sqrt{n^{-1}\sum_{i=1}^n\left[m(\vx_i^T\vbeta|\vbeta) - m_a^{\circ}(\vx_i^T\vbeta^{\circ}_{ll'})\right]^2} * \sqrt{n^{-1}\sum_{i=1}^n\left\{\widetilde{\epsilon}_i [\vx_{i,-1} -\E(\vx_{i,-1}|\vx_i^T\vbeta_0)] ^T\vtheta_{j}\right\}^2} \\
		\leq & c_0\sqrt{n^{-1}\sum_{i=1}^n\left[m(\vx_i^T\vbeta|\vbeta) - m(\vx_i^T\vbeta^{\circ}_{ll'}|\vbeta)\right]^2 + n^{-1}\sum_{i=1}^n\left[m(\vx_i^T\vbeta^{\circ}_{ll'}|\vbeta) - m_a^{\circ}(\vx_i^T\vbeta^{\circ}_{ll'})\right]^2}\\
		\leq & c_1\sqrt{n^{-1}\sum_{i=1}^n\left[\vx_i^T(\vbeta-\vbeta^{\circ}_{ll'})\right]^2 +\delta_n^2}\\
		\leq & c_2\delta_n,
		\end{align*} 
		where  the third inequality applies the event $\mathcal{H}_n$; the second last inequality applies (\ref{cover_m}) and the differentiability condition; the last inequality applies the event $\mathcal{K}_n$. 
		%Then $\left[ n^{-1}\sum_{i=1}^n\vtheta_j^T\vgamma_{i}(\vbeta^{\circ}_{ll'},m_a^{\circ})-c_2\delta_n, n^{-1}\sum_{i=1}^n\vtheta_j^T\vgamma_{i}(\vbeta^{\circ}_{ll'},m_a^{\circ})+c_2\delta_n\right]$ forms a $c_2\delta_n-$cover for $n^{-1}\sum_{i=1}^n\vtheta_j^T\vgamma_{i}(\vbeta,m)$. 
		Hence, the $\delta_n-$covering number of the class of functions $\Gamma_j=\{\vtheta_j^T\vgamma(Z,\vbeta,m):\vbeta\in\bbB_1, m\in\bbM\}$ satisfies 
		\begin{align}
		N (\delta_n,\Gamma_j,L_1(\mathbb{P}_n))\leq c{p\choose ks}\left(1+\frac{2c_0\sqrt{s}h^2}{\delta_n}\right)^{ks}\exp\big[C\delta_n^{-1}\sqrt{\log(p\vee n)}\big], \ \ \forall j.\label{cover_no0}
		\end{align} 
		
		Recall that  $\vtheta_{j}^T\vgamma_{i}(\vbeta,m) = 2(2A_i-1)\left[\epsilon_i+g(\vx_i)\right]\big[m(\vx_i^T\vbeta|\vbeta) - G^{(1)}(\vx_i^T\vbeta_0|\vbeta_0)\big] \big[\vx_{i,-1} -\E(\vx_{i,-1}|\vx_i^T\vbeta_0)\big] ^T\vtheta_j$, where $(2A_i-1)$ is a Rademacher sequence independent of $(\vx_i,\epsilon_i)$.  
		Note that on the event $\mathcal{H}_n \cap\mathcal{J}_n\cap \mathcal{K}_n$, we have that $\sup_{\vbeta\in\bbB_1,m\in\bbM} \sqrt{n^{-1}\sum_{i=1}^{n}\left[\vtheta_{j}^T\vgamma_{i}(\vbeta,m)^2\right]} =c_0h\triangleq R_n$.
		Let $L=\min\{l:l\geq 1,\ 2^{-l}\leq 4/\sqrt{n}\}$.
		Therefore, Lemma~14.18  in \citet{van1996weak} implies that 
		\begin{align*}
		&\E_{sup}\big[\vtheta_{j}^T\vgamma_i(\vbeta,m)\big] \\
		= &\E\left[\sup_{\vbeta\in\bbB_1,m\in\bbM}\Big| n^{-1/2}\sum_{i=1}^n\vtheta_{j}^T\vgamma_{i}(\vbeta,m)\Big| \ \bigg|\{(\vx_i,\ep_i)\}_{i=1}^n,\mathcal{H}_n \cap\mathcal{J}_n\cap \mathcal{K}_n\right]\\
		\leq &Ch\left\{4+6\sum_{l=1}^L 2^{-l}\sqrt{\log \left[N (2^{-l}h,\Gamma_j,L_1(\mathbb{P}_n))+1\right]}   \right\}   \\
		\leq& Ch \left\{4+c_1\sum_{l=1}^L 2^{-l}\sqrt{ ks\log p + ks \log \left[1+\frac{2c_0\sqrt{s}h^2}{2^{-l}h} \right] + \frac{\sqrt{\log(p\vee n)}}{2^{-l}h }}\right\},
		\end{align*}
		for some positive constants $c_0$, $c_1$, where  the last inequality applies (\ref{cover_no0}). Hence we have $2^{-l}\geq 2/\sqrt{n}$, for any $1\leq l\leq L$. Note that %Recall that $c(n )=\frac{1}{4}[\log (p\vee n)]^{-3/2}$, we have 
		$$\frac{2\sqrt{s}h^2}{ 2^{-l}h}\leq 4h\sqrt{ns}\leq 4c_2\sqrt{n^2h^7}\leq 4c_2n,  $$
		for some positive constant $c_2$, by the assumptions of Theorem~\ref{Lasso_error}. Furthermore, the assumptions of Theorem~\ref{Lasso_error} imply  $s\sqrt{\log (p\vee n)}\leq c_1nh^5\leq c_2h^{-1}$, for some positive constants $c_1$, $c_2$. We thus have
		\begin{align*}
		&\E_{sup}\big[\vtheta_{j}^T\vgamma_i(\vbeta,m)\big]\\%\E\left[\sup_{\vbeta\in\bbB_1,m\in\bbM}\Big|   n^{-1/2}\sum_{i=1}^n\vtheta_{j}^T\vgamma_{i}(\vbeta,m)\Big|\ \bigg|\{(\vx_i,\ep_i)\}_{i=1}^n,\mathcal{H}_n \cap\mathcal{J}_n\cap \mathcal{K}_n \right] \\
		\leq& Ch \left[8+d_1\sum_{l=1}^L 2^{-l}\sqrt{s\log (p\vee n)  +2^{l}h^{-1}\sqrt{\log(p\vee n)}}\right]\\
		\leq& Ch \left[8+d_1\sum_{l=1}^L 2^{-l}\sqrt{c_2h^{-1}\sqrt{\log(p\vee n)} +2^{l}h^{-1}\sqrt{\log(p\vee n)}}\right]\\
		\leq& Ch \left[8+d_2\big[h^{-1}\sqrt{\log(p\vee n)}\big]^{1/2}\sum_{l=1}^L 2^{-l}(1+2^{l/2})\right]\\
		\leq& Ch \left[8+2d_2\big[h^{-1}\sqrt{\log(p\vee n)}\big]^{1/2}\right]\\
		\leq& d_3 \left[h^2 \log(p\vee n)\right]^{1/4},
		\end{align*}
		for some positive constants $d_1$, $d_2$, $d_3$, and all $n$ sufficiently large, since $s\log(p\vee n)\leq s\sqrt{\log(p\vee n)} *\sqrt{\log(p\vee n)}\leq c_1h^{-1}\sqrt{\log(p\vee n)}$, and $\big[h^{-1}\sqrt{\log(p\vee n)}\big]^{1/2}\geq 8$ for some positive constant $c_1$, and all $n$ sufficiently large.
		It follows from (\ref{mass_1}) that $\forall\ t>0$,
		\begin{align*}
		&	P\left( \sup_{\vbeta\in\bbB_1,m\in\bbM} \Big| n^{-1/2}\sum_{i=1}^n\vtheta_{j}^T\vgamma_{i}(\vbeta,m)\Big|
		\geq c \left[h^2 \log(p\vee n)\right]^{1/4} +\sqrt{n} ht \ \bigg|  \{(\vx_i,\ep_i)\}_{i=1}^n,\mathcal{H}_n \cap\mathcal{J}_n \cap \mathcal{K}_n \right) \\
		\leq &\exp\left(-\frac{nt^2}{8}\right).
		\end{align*} 
		Take $t = 4\sqrt{n^{-1}\log (p\vee n)}$.  Note that the assumptions of Theorem~\ref{Lasso_error} imply $h \sqrt{\log (p\vee n) }\leq c_1\sqrt{nh^7}\leq c_1 $, for some positive constant $c_1$. Hence we have 
		%Recall that $\gamma_1 = \sqrt{h\log(p\vee n)} + h[ \log (p\vee n)]^{3/2}  $, then
		\begin{align*}
		&P\left( \sup_{\vbeta\in\bbB_1,m\in\bbM}\Big| n^{-1/2}\sum_{i=1}^n\vtheta_{j}^T\vgamma_{i}(\vbeta,m)\Big|
		\geq c \left[h^2 \log(p\vee n)\right]^{1/4}\ \bigg|  \mathcal{H}_n \cap\mathcal{J}_n\cap \mathcal{K}_n  \right) \\
		=&\E_{(\vx_i,\ep_i)}\left\{P\left( \sup_{\vbeta\in\bbB_1,m\in\bbM}\Big| n^{-1/2}\sum_{i=1}^n\vtheta_{j}^T\vgamma_{i}(\vbeta,m)\Big|
		\geq c \left[h^2 \log(p\vee n)\right]^{1/4}\ \bigg| \{(\vx_i,\ep_i)\}_{i=1}^n,\mathcal{H}_n \cap\mathcal{J}_n\cap \mathcal{K}_n  \right) \right\}\\
		\leq &\exp\left[-2\log (p\vee n)  \right].
		\end{align*} 
		Therefore, there exist positive constants $c_0$, $c_1$, such that for all $n$ sufficiently large,
		$$P\left(\max_{2\leq j \leq p} \sup_{\vbeta\in\bbB_1,m\in\bbM}\Big| n^{-1/2}\sum_{i=1}^n\vtheta_{j}^T\vgamma_{i}(\vbeta,m)\Big|\geq c_0\left[h^2 \log(p\vee n)\right]^{1/4}\right)\leq\exp[-c_1\log (p\vee n)].$$
	\end{proof}

	\begin{proof}[Proof of Lemma~\ref{lem:ghat_g}]
		
		We will prove (\ref{ghat_g1}) and (\ref{ghat_g2}) below. The proofs of (\ref{ghat_g1_inf}) and (\ref{ghat_g2_inf})  are similar.
		
		Recall that $\vxw_i = \vx_i-\E(\vx_i|\vx_i^T\vbeta_0)$, and $\vxh_i = \vx_i-\widehat{\E}(\vx_i|\vx_i^T\vbetah)$. To prove (\ref{ghat_g1}), observe that 
		\begin{align*}
		&\max_{2\leq j \leq p}\left|n^{-1/2}\sum_{i=1}^n\vtheta_j^T\vnu_1(Z_i,\vbetah,\widehat{G}^{(1)})\right|\\
		\leq &\max_{2\leq j \leq p}\left| n^{-1/2}\sum_{i=1}^n  \Big[G(\vx_{i}^{T} \vbeta_0|\vbeta_0)-G(\vx_{i}^{T} \vbetah|\vbetah) - G^{(1)}(\vx_i^T\vbetah|\vbetah) \vxw_{i,-1}^T(\vbeta_{0,-1}-\vbetah_{-1})\Big] G^{(1)}(\vx_i^T\vbeta_0|\vbeta_0) \vxw_{i,-1}^T\vtheta_j\right| \\
		&+\max_{2\leq j \leq p}\bigg| n^{-1/2}\sum_{i=1}^n  \Big[G(\vx_{i}^{T} \vbeta_0|\vbeta_0)-G(\vx_{i}^{T} \vbetah|\vbetah) - G^{(1)}(\vx_i^T\vbetah|\vbetah) \vxw_{i,-1}^T(\vbeta_{0,-1}-\vbetah_{-1})\Big]\\
		&\qquad\qquad\qquad* \Big[\widehat{G}^{(1)}(\vx_i^T\vbetah|\vbetah)-G^{(1)}(\vx_i^T\vbeta_0|\vbeta_0)\Big]\vxw_{i,-1}^T\vtheta_j\bigg| \\
		&+\max_{2\leq j \leq p}\left| n^{-1/2}\sum_{i=1}^n  \Big[G(\vx_{i}^{T} \vbeta_0|\vbeta_0)-G(\vx_{i}^{T} \vbetah|\vbetah)\Big] \widehat{G}^{(1)}(\vx_i^T\vbeta_0|\vbeta_0)( \vxh_{i,-1}-\vxw_{i,-1})^T \vtheta_j\right| \\
		&+\max_{2\leq j \leq p}\left| n^{-1/2}\sum_{i=1}^n G^{(1)}(\vx_i^T\vbetah|\vbetah) \widehat{G}^{(1)}(\vx_i^T\vbeta_0|\vbeta_0)(\vbeta_{0,-1}-\vbetah_{-1})^T( \vxh_{i,-1}\vxh_{i,-1}^T-\vxw_{i,-1}\vxw_{i,-1}^T)\vtheta_j\right| \\
		=&\sum_{k=1}^{4}V_{n1k} ,
		\end{align*}
		where the definition of $V_{n1k}$, $k=1,\cdots,4$, is clear from the context.
		
		We first bound $V_{n13}$. %Theorem~\ref{Lasso_error} implies that $P\left(\vbetah\in\bbB_1\right)\geq 1-\exp(-c\log p)$, for some positive constant $c$, and all $n$ sufficiently large. Lemma~\ref{Gbetafunc} implies that 
		%$$P\left(\max_{ 1\leq i \leq n}\left|\widehat{G} (\vx_i^T\vbetah|\vbetah)-G  (\vx_i^T\vbeta_0|\vbeta_0)  \right| \geq c_0sh^2\sqrt{\log(p\vee n)} \right) \leq  \exp(-c_1\log p).$$ 
		Let $\mathcal{R}_n = \left\{b/2\leq\max_{1\leq i \leq n}\left|\widehat{G}^{(1)}(\vx_i^T\vbeta_0|\vbeta_0)\right|\leq2b\right\}$.  Note that $\max_{1\leq i \leq n}\left|G^{(1)}(\vx_i^T\vbeta_0|\vbeta_0)\right|\leq b$ by Assumption~\ref{A1}-(b). Lemma~\ref{G1func} implies that $P(\mathcal{R}_n )\geq 1-\exp[-c\log(p\vee n)]$ for some positive constant $c$, and all $n$ sufficiently large.    Then there exist positive constants $c_0$, $c_1$, $c_2$, $c_3$, such that for all $n$ sufficiently large,
		%$\mathcal{R}_n = \left\{b/2\leq\max_{1\leq i \leq n}\sup_{\vbeta\in\bbB}\left|\widehat{G}^{(1)}(\vx_i^T\vbeta|\vbeta)\right|\leq2b\right\}$, where $\max_{1\leq i \leq n}\sup_{\vbeta\in\bbB}\left|G^{(1)}(\vx_i^T\vbeta|\vbeta)\right|\leq b$ by Assumption~\ref{K4}-(a). Lemma~\ref{Gbetafunc} implies that $P(\mathcal{R}_n )\geq 1-\exp(-c_1\log p)$ for some positive constant $c_1$, and all $n$ sufficiently large.    Then there exist positive constants $c_0$, $c_1$, $c_2$, $c_3$, such that for all $n$ sufficiently large,
		\begin{align*}
		V_{n13}\leq & 2b\sqrt{n}\sqrt{n^{-1}\sum_{i=1}^n   \Big[G(\vx_{i}^{T} \vbeta_0|\vbeta_0)-G(\vx_{i}^{T} \vbetah|\vbetah)\Big]^2} *\max_{2\leq j \leq p} \sqrt{n^{-1}\sum_{i=1}^n  \big[   (\vxh_{i,-1}-\vxw_{i,-1})^T\vtheta_j \big]^2}\\
		\leq &c_0\sqrt{n}||\vbeta_0-\vbetah||_2 *sh^2\sqrt{\log (p\vee n)}\\
		\leq &c_1\sqrt{ns}h^3*sh\sqrt{\log (p\vee n)}\leq c_2\sqrt{ns}h^3,
		\end{align*} 
		with probability at least $ 1-\exp(-c_3\log p)$. In the above, the second inequality applies (\ref{G-2}) in Lemma~\ref{lem:Gbound}, Lemma~\ref{Gbetafunc} and Lemma~\ref{lem:thetaj}, the third and fourth inequalities apply Theorem~\ref{Lasso_error} and its assumptions.
		
		We next bound $V_{n14}$. Similarly as in the proof of Lemma~\ref{lem:Ex_err}, we can show that 
		\begin{align*}
		\left|\left| n^{-1}\sum_{i=1}^n G^{(1)}(\vx_i^T\vbetah|\vbetah) \widehat{G}^{(1)}(\vx_i^T\vbeta_0|\vbeta_0)( \vxh_{i,-1}\vxh_{i,-1}^T-\vxw_{i,-1}\vxw_{i,-1}^T)\vtheta_j\right| \right|_\infty\leq c_0\sqrt{s}h^2,
		\end{align*} 
		with probability at least $1-\exp(-c_1\log p)$, for some positive constants $c_0$, $c_1$ and all $n$ sufficiently large. Hence Theorem~\ref{Lasso_error} and its assumptions imply that
		$ V_{n14}\leq c_0\sqrt{n}s^{3/2}h^4\leq d_0 \sqrt{ns}h^4*nh^5\leq d_1\sqrt{ns}h^3$, with  probability at least $1-\exp(-c_1\log p)$, for some positive constants $c_0$, $c_1$, $d_0$, $d_1$  and all $n$ sufficiently large.
		
		To bound $V_{n12}$, we note that 
		\begin{align*}
		V_{n12}\leq &\max_{2\leq j \leq p} \bigg| n^{-1/2}\sum_{i=1}^n  \Big[G(\vx_{i}^{T} \vbeta_0|\vbeta_0)-G(\vx_{i}^{T} \vbetah|\vbetah) \Big]  \Big[\widehat{G}^{(1)}(\vx_i^T\vbetah|\vbetah)-G^{(1)}(\vx_i^T\vbeta_0|\vbeta_0)\Big]\vxw_{i,-1}^T\vtheta_j\bigg|\\
		&+\max_{2\leq j \leq p}\bigg| n^{-1/2}\sum_{i=1}^n   G^{(1)}(\vx_i^T\vbetah|\vbetah)  \Big[\widehat{G}^{(1)}(\vx_i^T\vbetah|\vbetah)-G^{(1)}(\vx_i^T\vbeta_0|\vbeta_0)\Big](\vbeta_{0,-1}-\vbetah_{-1})^T\vxw_{i,-1}\vxw_{i,-1}^T\vtheta_j\bigg|\\
		=&	V_{n121}+	V_{n122},
		\end{align*} 
		where the definitions of $V_{n121}$ and $V_{n122}$ are clear from the context.
		For $V_{n121}$, we have
		\begin{align*}
		V_{n121}\leq & \max_{1\leq i \leq n}\Big|\widehat{G}^{(1)}(\vx_i^T\vbetah|\vbetah)-G^{(1)}(\vx_i^T\vbeta_0|\vbeta_0)\Big| * \sqrt{ n^{-1/2}\sum_{i=1}^n  \Big[G(\vx_{i}^{T} \vbeta_0|\vbeta_0)-G(\vx_{i}^{T} \vbetah|\vbetah) \Big]^2}\\
		&*\max_{2\leq j \leq p}\sqrt{n^{-1/2}\sum_{i=1}^n (\vxw_{i,-1}^T\vtheta_j)^2}\\
		\leq & c_0 h*||\vbetah-\vbeta_0||_2 \leq  c_1\sqrt{ns}h^3,
		\end{align*}
		with probability at least $1-\exp(-c_2\log p)$, for some positive constants $c_0$, $c_1$, $c_2$ and all $n$ sufficiently large. In the above, the second inequality applies   Lemma~\ref{lem:events}, Lemma~\ref{lem:Gbound} and Lemma~\ref{Gbetafunc}; the third inequality applies Theorem~\ref{Lasso_error}.
		Similarly, Theorem~\ref{Lasso_error} and Lemma~\ref{Gbetafunc} imply that  
		\begin{align*}
		V_{n122}\leq & \max_{1\leq i \leq n}b\Big|\widehat{G}^{(1)}(\vx_i^T\vbetah|\vbetah)-G^{(1)}(\vx_i^T\vbeta_0|\vbeta_0)\Big| * ||\vbetah_{-1}-\vbeta_{0,-1}||_1 * \max_{2\leq j \leq p} \left|\left|n^{-1/2}\sum_{i=1}^n \vxw_{i,-1} \vxw_{i,-1}^T\vtheta_j\right|\right|_\infty\\  
		\leq & c_0 h*||\vbetah-\vbeta_0||_1\leq  c_1s\sqrt{n}h^3,
		\end{align*}
		with probability at least $1-\exp(-c_2\log p)$, for some positive constants $c_0$, $c_1$, $c_2$ and all $n$ sufficiently large. In the above, the second inequality applies  Assumption~\ref{K4}-(a), Lemma~\ref{lem:events} and Lemma~\ref{lem14NCL}; the third inequality applies Theorem~\ref{Lasso_error}.
		
		Finally we bound $V_{n11}$.  Note that $\forall\vbeta\in\bbB_1$, $G^{(1)}(\cdot|\vbeta)$ is differentiable by Assumption~\ref{K4}-(a). Hence, there exists some positive constant $L$ such that $$\left|G(\vx_{i}^{T} \vbeta_0|\vbetah)-G(\vx_{i}^{T} \vbetah|\vbetah) - G^{(1)}(\vx_i^T\vbetah|\vbetah) \vx_{i,-1}^T(\vbeta_{0,-1}-\vbetah_{-1}) \right|\leq L \big[\vx_i^T(\vbeta_0-\vbetah)\big]^2.$$
		It implies that 
		\begin{align*}
		V_{n11}\leq  &\max_{2\leq j \leq p}  \bigg| n^{-1/2}\sum_{i=1}^n  \Big[G(\vx_{i}^{T} \vbeta_0|\vbeta_0)-G(\vx_{i}^{T} \vbeta_0|\vbetah) \Big] G^{(1)}(\vx_i^T\vbeta_0|\vbeta_0)\vxw_{i,-1}^T\vtheta_j\bigg|\\
		&+\max_{2\leq j \leq p} L\bigg| n^{-1/2}\sum_{i=1}^n  \big[\vx_i^T(\vbeta_0-\vbetah)\big]^2 G^{(1)}(\vx_i^T\vbeta_0|\vbeta_0) \vxw_{i,-1}^T\vtheta_j\bigg|\\
		&+\max_{2\leq j \leq p}  \bigg| n^{-1/2}\sum_{i=1}^n  \big[G^{(1)}(\vx_i^T\vbeta_0|\vbeta_0)\big]^2 (\vbeta_{0,-1}-\vbetah_{-1})^T\E(\vx_{i,-1}|\vx_i^T\vbeta)\vxw_{i,-1}^T\vtheta_j\bigg|\\
		&+\max_{2\leq j \leq p}  \bigg| n^{-1/2}\sum_{i=1}^n \big[G^{(1)}(\vx_i^T\vbetah|\vbetah)-G^{(1)}(\vx_i^T\vbeta_0|\vbeta_0) \big] G^{(1)}(\vx_i^T\vbeta_0|\vbeta_0)\\
		&\qquad\qquad\qquad\quad* (\vbeta_{0,-1}-\vbetah_{-1})^T\E(\vx_{i,-1}|\vx_i^T\vbeta)\vxw_{i,-1}^T\vtheta_j\bigg|\\
		=&\sum_{l=1}^4V_{n11l},
		\end{align*} 
		where the definition of $V_{n11l}$, $l=1,\cdots,4$, is clear from the context. To bound $V_{n112}$, note that 
		\begin{align*}
		V_{n112}\leq & Lb\sqrt{n}\sqrt{n^{-1}\sum_{i=1}^n   \big[\vx_i^T(\vbeta_0-\vbetah)\big]^4} *\max_{2\leq j \leq p}  \sqrt{n^{-1}\sum_{i=1}^n  (\vxw_{i,-1}^T\vtheta_j)^2}.
		\end{align*} 
		Note that $\frac{\vbeta_0-\vbetah}{||\vbeta_0-\vbetah||_2}\in \bbK(p,2ks)$ for $\vbetah\neq \vbeta_0$, where $\bbK(p,s_0)=\{\vv\in\bbR^p:||\vv||_2\leq 1,||\vv||_0\leq s_0\}$. 
		In Lemma~\ref{lem:cube_rate}, take $t = d_0\sigma_x^4||\vbeta_0-\vbetah||_2^2$ and $s_0=ks$, for some positive constant $d_0$, then Theorem~\ref{Lasso_error} and  Lemma~\ref{lem:cube_rate} imply that $n^{-1}\sum_{i=1}^n  \big[  \vx_i ^T(\vbeta_0-\vbetah)  \big]^4\leq c_0\sigma_x^4s^2h^8$ with probability at least $1-\exp(-c_1\log p)$, for  some positive constants $c_0$, $c_1$, and all $n$ sufficiently large. We thus have $V_{n112}\leq d_0s\sqrt{n}h^4$ with probability at least $1-\exp[-d_1\log(p\vee n)]$, for  some positive constants $d_0$, $d_1$, and all $n$ sufficiently large.
		
		To bound $V_{n113}$, note that $\E(\vx_{i,-1}|\vx_i^T\vbeta_0)$ and $\vxw_{i,-1}^T\vtheta_j$ are both sub-Gaussian by Lemma~\ref{lem:subg},  and $G^{(1)}(\vx_i^T\vbeta_0|\vbeta_0)$ is bounded by Assumption~\ref{K4}-(a). It's easy to show   $G^{(1)}(\vx_i^T\vbeta_0|\vbeta_0)\big]^2 \vxw_{i,-1}^T\vtheta_j$ is also sub-Gaussian by Lemma~\ref{lem:subg}. We observe that 
		\begin{align*}
		&	\E\Big\{\big[ G^{(1)}(\vx_i^T\vbeta_0|\vbeta_0)\big]^2 \E(\vx_{i,-1}|\vx_i^T\vbeta_0) \vxw_{i,-1}^T\vtheta_j\Big\}\\
		=&\E_{\vx_i^T\vbeta_0}\Big\{\big[ G^{(1)}(\vx_i^T\vbeta_0|\vbeta_0)\big]^2   \E(\vx_{i,-1}|\vx_i^T\vbeta_0)\E\big( \vxw_{i,-1}^T\vtheta_j\big| \vx_i^T\vbeta_0\big)\Big\}\\
		=&\vnull_{p-1},
		\end{align*}
		where $\vnull_{p-1}$ is a $(p-1)-$dimensional vector with all entries $0$, since $\E\big( \vxw_{i,-1} \big| \vx_i^T\vbeta_0\big)= \vnull_{p-1}$.
		Hence, Lemma~\ref{lem14NCL} implies that 
		\begin{align*}
		V_{n113}\leq & \sqrt{n}||\vbeta_0-\vbetah||_1*\max_{2\leq j \leq p} \Big|\Big|n^{-1}\sum_{i=1}^n \big[ G^{(1)}(\vx_i^T\vbeta_0|\vbeta_0)\big]^2 \E(\vx_{i,-1}|\vx_i^T\vbeta_0)  \vxw_{i,-1}^T\vtheta_j\Big|\Big|_\infty\\
		\leq &c_0\sqrt{n}sh^2*\sqrt{\frac{\log p}{n}} \leq c_0s\sqrt{n}h^{9/2},
		\end{align*}
		with probability at least $1-\exp(-c_1\log p)$, for  some positive constants $c_0$, $c_1$, and all $n$ sufficiently large, where the second inequality also applies Theorem~\ref{Lasso_error}, and the third inequality applies the assumptions of Theorem~\ref{Lasso_error}.
		
		Theorem~\ref{Lasso_error} and Assumption~\ref{K4}-(b) imply that
		\begin{align*}
		V_{n114}\leq &c_0\sqrt{n}||\vbetah-\vbeta_0||_2 \Big\{n^{-1}\sum_{i=1}^n   \big[(\vbeta_0-\vbetah)^T\E(\vx_i|\vx_i^T\vbeta_0)\big]^4  \Big\}^{1/4}*\max_{2\leq j \leq p} \Big[n^{-1}\sum_{i=1}^n  ( \vxw_{i,-1}^T\vtheta_j)^4\Big]^{1/4},
		\end{align*}
		with probability at least $1-\exp(-c_1\log p)$, for  some positive constants $c_0$, $c_1$, and all $n$ sufficiently large.  Lemma~\ref{lem:cube_rate} implies that $n^{-1}\sum_{i=1}^n   \big[(\vbeta_0-\vbetah)^T\E(\vx_{i,-1}|\vx_i^T\vbeta_0)\big]^4  \ \leq d_0||\vbeta_0-\vbetah||_2^4$, and $n^{-1}\sum_{i=1}^n  ( \vxw_{i,-1}^T\vtheta_j)^4\leq d_0$, 	with probability at least $1-\exp(-d_1\sqrt{n})$, for  some positive constants $d_0$, $d_1$, and all $n$ sufficiently large.  Hence combining the results with Theorem~\ref{Lasso_error}, we have $V_{n114}\leq c_0\sqrt{n}sh^4$, 	with probability at least $1-\exp(-c_1\log p)$, for  some positive constants $c_0$, $c_1$, and all $n$ sufficiently large. 
		
		We finally bound $V_{n111}$. Theorem~\ref{Lasso_error} implies that $P\left(\vbetah\in\bbB_1\right)\geq 1-\exp(-c\log p)$, for some positive constant $c$, and all $n$ sufficiently large.  	It is sufficient to show that there exist some positive constants $c_0$, $c_1$, such that for all $n$ sufficiently large,
		\begin{align*}
		P\left(\max_{2\leq j \leq p}\sup_{\vbeta\in\bbB_1} \Big|n^{-1/2}\sum_{i=1}^n \vtheta_j^T\vnu_{11}(Z_i,\vbeta) \Big|\geq c_0h^2s\sqrt{\log(p\vee n)} \right)\leq\exp[-c_1\log (p\vee n)], 
		\end{align*} 
		where $Z_i=(\vx_i,\epsilon_i,A_i)$,  
		$\vnu_{11}(Z_i,\vbeta) = \big[G(\vx_i^T\vbeta_0|\vbeta_0) -G(\vx_i^T\vbeta_0|\vbeta)\big] G^{(1)}(\vx_i^T\vbeta_0|\vbeta_0) \vxw_{i,-1} $. 	Note that the assumptions of Theorem~\ref{Lasso_error} imply that $h^2s\sqrt{\log(p\vee n)} \leq d_0h^2\sqrt{snh^5}\leq d_0h^3\sqrt{sn}$ for some positive constant $d_0$.  
		Note that for any $\vbeta\in\bbB_1$, we have
		\begin{align*}
		\E\big[\vtheta_j^T\vnu_{11}(Z_i,\vbeta) \big] =& \E_{\vx_i^T\vbeta_0}\E\Big\{ \big[G(\vx_i^T\vbeta_0|\vbeta_0) -G(\vx_i^T\vbeta_0|\vbeta)\big] G^{(1)}(\vx_i^T\vbeta_0|\vbeta_0) \vxw_{i,-1}^T\vtheta_j  |\vx_i^T\vbeta_0 \Big\}\\
		=& \E_{\vx_i^T\vbeta_0}\Big\{ \big[G(\vx_i^T\vbeta_0|\vbeta_0) -G(\vx_i^T\vbeta_0|\vbeta)\big] G^{(1)}(\vx_i^T\vbeta_0|\vbeta_0)\vtheta_j^T \E(\vxw_{i,-1} |\vx_i^T\vbeta_0) \Big\}\\
		=&0.
		\end{align*}   
		Denote the event $$\mathcal{T}_n = \Big\{\max_{2\leq j \leq p}\frac{1}{n}\sum_{i=1}^n(\vxw_{i,-1}^T\vtheta_j)^4\leq 20\xi_2^{-4}\sigma_x^4\mbox{, and } \sup_{\vbeta\in\bbB_1}\frac{1}{n}\sum_{i=1}^n \big[\E(\vx_{i}|\vx_i^T\vbeta)^T(\vbeta-\vbeta_{0})\big]^4 \leq 5\sigma_x^4s^2h^8 \Big\}.$$ Similarly as the above analysis, Lemma~\ref{lem:subg} and Lemma~\ref{lem:cube_rate} imply that $P(\mathcal{T}_n)\geq 1-\exp(-d_0\sqrt{n})$, for some positive constant $d_0$ and all $n$ sufficiently large.
		
		To prove (\ref{ghat_g_1}), we have
		\begin{align*}
		&P\left(\max_{2\leq j \leq p}\sup_{\vbeta\in\bbB_1}\Big|n^{-1/2}\sum_{i=1}^n\vtheta_j^T \vnu_{11}(Z_i,\vbeta) \Big|>t\right)\\
		\leq &P\left(\max_{2\leq j \leq p}\sup_{\vbeta\in\bbB_1}\Big|n^{-1/2}\sum_{i=1}^n\vtheta_j^T \vnu_{11}(Z_i,\vbeta) \Big|>t\ \bigg|\mathcal{G}_n \cap\mathcal{K}_n\cap\mathcal{T}_n  \right) +\exp[-c\log (p\vee n)]\\
		\leq& \sum_{j=2}^{p}P\left( \sup_{\vbeta\in\bbB_1} \Big|n^{-1/2}\sum_{i=1}^n \vtheta_j^T \vnu_{11}(Z_i,\vbeta) \Big|>t\ \bigg|\mathcal{G}_n  \cap\mathcal{K}_n \cap\mathcal{T}_n  \right)+\exp[-c\log (p\vee n)],
		\end{align*}
		for some positive constant $c$, and all $n$ sufficiently large, where  the events $\mathcal{G}_n$ and $\mathcal{K}_n$ are defined in Lemma~\ref{lem:events}.  

		We observe that there exists some $\vbeta^r$ between $\vbeta$ and $\vbeta_0$, such that%on the event $\mathcal{G}_n \cap\mathcal{J}_n$,
		\begin{align*}
		&\sup_{\vbeta\in\bbB_1}\frac{1}{n}\sum_{i=1}^n\ [ \vtheta_j^T \vnu_{11}(Z_i,\vbeta)]^2 \\
		\leq& b^2
		\sup_{\vbeta\in\bbB_1}  \frac{1}{n}\sum_{i=1}^n \big[G(\vx_i^T\vbeta_0|\vbeta_0) -G(\vx_i^T\vbeta_0|\vbeta)\big]^2  ( \vxw_{i,-1}^T\vtheta_j  )^2 \ \\
		\leq& b^2	\sup_{\vbeta\in\bbB_1}  \frac{1}{n}\sum_{i=1}^n \Big\{-f_0'(\vx_i^T\vbeta_0)\E(\vx_{i,-1}|\vx_i^T\vbeta)^T(\vbeta_{-1}-\vbeta_{0,-1}) \\
		&\qquad\qquad\qquad+ \frac{1}{2}\E\big\{f_0''(\vx_i^T\vbeta^r)[\vx_{i,-1}^T(\vbeta_{-1}-\vbeta_{0,-1})]^2|\vx_i^T\vbeta\big\}\Big\}^2  ( \vxw_{i,-1}^T\vtheta_j  )^2 \\ 
		\leq& 2b^2
		\sup_{\vbeta\in\bbB_1}  \frac{1}{n}\sum_{i=1}^n\big[f_0'(\vx_i^T\vbeta_0)\E(\vx_{i,-1}|\vx_i^T\vbeta)^T(\vbeta_{-1}-\vbeta_{0,-1})\big]^2  ( \vxw_{i,-1}^T\vtheta_j  )^2 \\
		&+c_0b^2
		\sup_{\vbeta\in\bbB_1}  \frac{1}{n}\sum_{i=1}^n\big[(\vbeta_{-1}-\vbeta_{0,-1})^T\E(\vx_{i,-1}\vx_{i,-1}^T|\vx_i^T\vbeta)^T(\vbeta_{-1}-\vbeta_{0,-1})\big] ( \vxw_{i,-1}^T\vtheta_j  )^2   \\
		\leq  & 2b^2
		\sup_{\vbeta\in\bbB_1}  \frac{1}{n}\sum_{i=1}^n\big[f_0'(\vx_i^T\vbeta_0)\E(\vx_{i}|\vx_i^T\vbeta)^T(\vbeta -\vbeta_0)\big]^2  ( \vxw_{i,-1}^T\vtheta_j  )^2  \\
		&+c_0b^2
		\sup_{\vbeta\in\bbB_1}\sqrt{ \frac{1}{n}\sum_{i=1}^n\big[ \lambda_{\max}(\E(\vx_i\vx_i^T|\vx_i^T\vbeta)) \big] ^2 ||\vbeta-\vbeta_0||_2^4  }\sqrt{  \frac{1}{n}\sum_{i=1}^n ( \vxw_{i,-1}^T\vtheta_j  )^4  }\\
		\leq  & b^4\sup_{\vbeta\in\bbB_1}  \frac{1}{n}\sum_{i=1}^n \big[\E\big(\vx_i|\vx_i^T\vbeta\big)^T(\vbeta-\vbeta_0)\big]^2  ( \vxw_{i,-1}^T\vtheta_j  )^2   +c_0b^2\sqrt{\xi_4}\sup_{\vbeta\in\bbB_1}||\vbeta-\vbeta_0||_2^2 \sqrt{ \frac{1}{n}\sum_{i=1}^n (\vxw_{i,-1}^T\vtheta_j)^4 },
		\end{align*}
		with probability at least $1-\exp[-c_2\log(p\vee n)]$, for some positive constants $c_0$, $c_1$, $c_2$, and all $n$ sufficiently large. In the above, the second inequality applies (\ref{form:Gnew_def}) in the proof of Lemma~\ref{lem:Gbound}; the third inequality applies Assumption~\ref{A1}-(b); the last inequality applies Assumption~\ref{A2}-(a) and \ref{A1}-(b). On the event $\mathcal{T}_n $, we have
		\begin{align*}
		&	\sup_{\vbeta\in\bbB_1}\frac{1}{n}\sum_{i=1}^n \big[\E\big(\vx_i|\vx_i^T\vbeta\big)^T(\vbeta-\vbeta_0)\big]^2  ( \vxw_{i,-1}^T\vtheta_j  )^2 \\
		\leq& \sqrt{\sup_{\vbeta\in\bbB_1}\frac{1}{n}\sum_{i=1}^n \big[\E(\vx_i|\vx_i^T\vbeta)^T(\vbeta-\vbeta_0)\big]^4  }\sqrt{\frac{1}{n}\sum_{i=1}^n( \vxw_{i,-1}^T\vtheta_j  )^4}\\
		\leq& d_0sh^4,
		\end{align*} 
		for some positive constant $d_0$. 
		%Note that $\E(\vx_i|\vx_i^T\vbeta)^T(\vbeta-\vbeta_0)$ and $\vxw_i^T\vtheta_j $ are both sub-Gaussian by Lemma~\ref{lem:subg}, then their sub-Gaussian properties imply that they both have finite fourth moments. Hence we have $\E \big[ (\vxw_i^T\vtheta_j)^4 \big]\leq d_0$, and $\E\Big\{ \big[\E(\vx_i|\vx_i^T\vbeta)^T(\vbeta-\vbeta_0)\big]^4 \Big\}\leq d_1||\vbeta-\vbeta_0||_2^4$	for some positive constants $d_0$,  $d_1$. Then we have
		%		\begin{align*}
		%		&\E\Big\{ \big[\E(\vx_i|\vx_i^T\vbeta)^T(\vbeta-\vbeta_0)\big]^2  ( \vxw_i^T\vtheta_j  )^2 \Big\}\\
		%		\leq& \sqrt{\E\Big\{ \big[\E(\vx_i|\vx_i^T\vbeta)^T(\vbeta-\vbeta_0)\big]^4 \Big\}*\E[( \vxw_i^T\vtheta_j  )^4]}\\
		%		\leq& \sqrt{d_0d_1}||\vbeta-\vbeta_0||_2^2.
		%		\end{align*} 
		Hence   we obtain that $\sup_{\vbeta\in\bbB_1}\frac{1}{n}\sum_{i=1}^n [ \vtheta_j^T \vnu_{11}(Z_i,\vbeta)]^2 \leq d_1sh^4$, on the event $\mathcal{T}_n $,  with probability at least $1-\exp[-d_2\log(p\vee n)]$, for some positive constants $d_1$, $d_2$, and all $n$ sufficiently large. 
		%In the above, the second inequality applies Lemma~\ref{lem:thetaj}. 
		%		 On the event  $\mathcal{T}_n $, similarly as the analysis above, we have 
		%		 \begin{align*}
		%		 &\sup_{\vbeta\in\bbB_1} n^{-1} \sum_{i=1}^n\left| \vtheta_{j}^T\vnu_{11}(Z_i,\vbeta)  \right|^2\\
		%		 \leq& b^2
		%		 \sup_{\vbeta\in\bbB_1}  \frac{1}{n}\sum_{i=1}^n \big[G(\vx_i^T\vbeta_0|\vbeta_0) -G(\vx_i^T\vbeta_0|\vbeta)\big]^2  ( \vxw_i^T\vtheta_j  )^2 \\
		%		 \leq& 2b^2
		%		 \sup_{\vbeta\in\bbB_1}   \frac{1}{n}\sum_{i=1}^n\big[f_0'(\vx_i^T\vbeta_0)\E(\vx_i|\vx_i^T\vbeta)^T(\vbeta-\vbeta_0)\big]^2  ( \vxw_i^T\vtheta_j  )^2 \\
		%		 &+c_0b^2
		%		 \sup_{\vbeta\in\bbB_1}   \frac{1}{n}\sum_{i=1}^n\big[(\vbeta-\vbeta_0)^T\E(\vx_i\vx_i^T|\vx_i^T\vbeta)^T(\vbeta-\vbeta_0)\big] ( \vxw_i^T\vtheta_j  )^2  \\
		%		 \leq  & c_1sh^4,
		%		 \end{align*} 
		%		 with probability at least $1-\exp[-c_2\log(p\vee n)]$, 
		%		 for some positive constants  $c_0$, $c_1$, $c_2$, any $j$ and all $n$ sufficiently large. In the above, the last inequality applies Cauchy-Schwartz inequality, the event $\mathcal{T}_n$, and Assumption~\ref{A2}-(a). Hence  
		Therefore, by Massart's concentration inequality (e.g., Theorem 14.2, \citet{buhlmann2011}) on the event $\mathcal{G}_n \cap\mathcal{T}_n \cap\mathcal{K}_n $,  $\forall\ t>0$, 
		\begin{align} 
		P\bigg( \sup_{\vbeta\in\bbB_1} \Big|n^{-1/2} \sum_{i=1}^n \vtheta_j^T \vnu_{11}(Z_i,\vbeta)  \Big| 	\geq   E_{sup,1} +th^2\sqrt{ sn}  \bigg|\mathcal{G}_n \cap\mathcal{K}_n \cap\mathcal{T}_n \bigg)   \leq &\exp\left( - \frac{nt^2}{8}\right),\label{mass_2}
		\end{align}  
		where $E_{sup,1}=\E\left[\sup_{\vbeta\in\bbB_1} \Big| n^{-1/2}\sum_{i=1}^n \vtheta_j^T \vnu_{11}(Z_i,\vbeta) \Big|\ \bigg|\mathcal{G}_n \cap\mathcal{K}_n \cap\mathcal{T}_n  \right]$.
		
		Similarly as the proof of Lemma~\ref{lem:ghat1_g1}, we can cover $\bbB_1$   by ${p\choose ks}\left(1+\frac{2c_0sh^2\sqrt{\log(p\vee n)} }{\delta_n }\right)^{ks}$ $L_2-$balls of radius $\frac{\delta_n}{\sqrt{s\log(p\vee n)}} $.
		Denote the centers of these $L_2-$balls by $\vbeta^{\circ}_{l0},\cdots,\vbeta^{\circ}_{lN_l}$.
		Denote these $L_2-$balls by $\mathbb{C}(\vbeta^{\circ}_{ll'})$, $l=1,\cdots,m(s)$, $l'=1,\cdots,N_l$. 
		Hence $\forall\ \vbeta\in\bbB_1$,  we can find $l,\ l'$ such that $\vbeta\in \mathbb{C}(\vbeta^{\circ}_{ll'})$.  On the event $\mathcal{G}_n  \cap\mathcal{K}_n$, there exist some positive constants $c_0$, $c_1$, such that for all $n$ sufficiently large,
		%  	 \bfred{ \begin{align*}
		%  	 &n^{-1}\left| \sum_{i=1}^na_i\vtheta_j^T \vnu_{11}(Z_i,\vbeta) -\sum_{i=1}^na_i \vtheta_j^T\vnu_{11}(Z_i,\vbeta^{\circ}_{ll'})  \right|\\
		%  	 \leq &n^{-1}\left| \sum_{i=1}^n\left[G(\vx_i^T\vbeta|\vbeta) - G(\vx_i^T\vbeta^{\circ}_{ll'}|\vbeta^{\circ}_{ll'})\right]a_iG^{(1)}(\vx_i^T\vbeta_0|\vbeta_0) [\vx_i -\E(\vx_i|\vx_i^T\vbeta_0)] ^T\vtheta_{j}\right| \\
		%  	 &+n^{-1}\left| \sum_{i=1}^n\left[G(\vx_i^T\vbeta_0|\vbeta) - G(\vx_i^T\vbeta|\vbeta)\right]a_iG^{(1)}(\vx_i^T\vbeta_0|\vbeta_0) [\vx_i -\E(\vx_i|\vx_i^T\vbeta_0)] ^T\vtheta_{j}\right| \\
		%  	 &+n^{-1}\left| \sum_{i=1}^n\left[G(\vx_i^T\vbeta_0|\vbeta^{\circ}_{ll'}) - G(\vx_i^T\vbeta^{\circ}_{ll'}|\vbeta^{\circ}_{ll'})\right]a_iG^{(1)}(\vx_i^T\vbeta_0|\vbeta_0) [\vx_i -\E(\vx_i|\vx_i^T\vbeta_0)] ^T\vtheta_{j}\right| 
		%  	 %\\  \leq & c\delta_n,
		%  	 \end{align*}  }
		\begin{align*}
		&n^{-1}\left| \sum_{i=1}^n\vtheta_j^T \vnu_{11}(Z_i,\vbeta) -\sum_{i=1}^n \vtheta_j^T\vnu_{11}(Z_i,\vbeta^{\circ}_{ll'})  \right|\\
		\leq &n^{-1}\left| \sum_{i=1}^n\left[G(\vx_i^T\vbeta_0|\vbeta) - G(\vx_i^T\vbeta_0|\vbeta^{\circ}_{ll'})\right]G^{(1)}(\vx_i^T\vbeta_0|\vbeta_0) [\vx_{i,-1} -\E(\vx_{i,-1}|\vx_i^T\vbeta_0)] ^T\vtheta_{j}\right| \\
		\leq & \sqrt{n^{-1}  \sum_{i=1}^n\left[G(\vx_i^T\vbeta_0|\vbeta) - G(\vx_i^T\vbeta_0|\vbeta^{\circ}_{ll'})\right]^2 }*\sqrt{n^{-1}  \sum_{i=1}^n\left[G^{(1)}(\vx_i^T\vbeta_0|\vbeta_0) [\vx_{i,-1} -\E(\vx_{i,-1}|\vx_i^T\vbeta_0)] ^T\vtheta_{j} \right]^2} \\
		\leq &c_0||\vbeta-\vbeta^{\circ}_{ll'}||_2\sqrt{s\log(p\vee n)} =c_0\delta_n,
		\end{align*}  
		with probability at least $1-\exp[-c_1\log(p\vee n)]$. In the above, the last inequality applies the event $\mathcal{G}_n \cap\mathcal{K}_n$, Assumption~\ref{A1}-(b) and Assumption~\ref{K4}-(c)  %(\ref{Gt}) in Lemma~\ref{lem:Gbound}.
		Hence, the $\delta_n-$covering number of the class of functions $V_{1j}=\{\vtheta_j^T\vnu_{11}(Z_i,\vbeta) :\vbeta\in\bbB_1\}$ is bounded by 
		\begin{align}
		N (\delta_n,V_{1j},L_1(\mathbb{P}_n))\leq c{p\choose ks}\left(1+\frac{2c_0sh^2\sqrt{\log(p\vee n)} }{\delta_n }\right)^{ks}.\label{cover_no1}
		\end{align} 
		
		Let $a_1,\cdots,a_n$ be a Rademacher sequence that is independent of the data. The symmetrization theorem (Theorem~14.3 in \citet{buhlmann2011}) implies that 
		\begin{align*} 
		E_{sup,1}=&\E\left[\sup_{\vbeta\in\bbB_1} \Big| n^{-1/2}\sum_{i=1}^n \vtheta_j^T \vnu_{11}(Z_i,\vbeta) \Big|\ \bigg|\mathcal{G}_n  \cap\mathcal{K}_n\cap\mathcal{T}_n  \right]\\
		\leq&2 \E_{ \{\vx_i\}_{i=1}^n}\left\{\E\left[\sup_{\vbeta\in\bbB_1} \Big| n^{-1/2}\sum_{i=1}^n a_i\vtheta_j^T \vnu_{11}(Z_i,\vbeta) \Big|\ \bigg| \{\vx_i\}_{i=1}^n,\mathcal{G}_n \cap\mathcal{K}_n \cap\mathcal{T}_n  \right]\right\}.
		\end{align*}
		
		Note that on the event $\mathcal{T}_n$, we have   $\sup_{\vbeta\in\bbB_1}\sqrt{n^{-1}\sum_{i=1}^n[a_i \vtheta_j^T \vnu_{11}(Z_i,\vbeta)]^2 }\leq ch^2\sqrt{s}\triangleq R_n$.   Lemma~14.18  in \citet{van1996weak} implies that %on the event $\mathcal{T}_n$,
		\begin{align*}
		E_{sup,1}\leq &2 \E_{ \{\vx_i\}_{i=1}^n}\left\{\E\left[\sup_{\vbeta\in\bbB_1} \Big| n^{-1/2}\sum_{i=1}^n a_i\vtheta_j^T \vnu_{11}(Z_i,\vbeta) \Big|\ \bigg| \{\vx_i\}_{i=1}^n,\mathcal{G}_n  \cap\mathcal{K}_n\cap\mathcal{T}_n  \right]\right\} \\ 
		\leq &Ch^2\sqrt{s}\left\{4+6\sum_{l=1}^L 2^{-l}\sqrt{\log \left[N (2^{-l}  h^2\sqrt{s},V_{j},L_1(\mathbb{P}_n))+1\right]}   \right\}   \\
		\leq& c_0h^2s\sqrt{\log(p\vee n)},
		\end{align*}
		for some positive constants $c_0$, $C$, where $L=\min\{l:l\geq 1,\ 2^{-l}\leq 4/\sqrt{n}\}$, and the last inequality applies (\ref{cover_no1}).   The analysis is similar as that  in the proof of Lemma~\ref{lem:ghat1_g1}.

		It follows from (\ref{mass_2}) that 
		\begin{align*}
		&P\bigg( \sup_{\vbeta\in\bbB_1} \Big| n^{-1/2}\sum_{i=1}^n \vtheta_j^T \vnu_{11}(Z_i,\vbeta)  \Big| 	\geq c_0h^2s\sqrt{\log(p\vee n)}+th^2\sqrt{ns}  \bigg| \mathcal{G}_n   \cap\mathcal{K}_n\cap\mathcal{T}_n \bigg)  \\
		\leq &\exp\left(-\frac{nt^2}{8}\right).
		\end{align*} 
		Take $t = 4\sqrt{n^{-1}\log (p\vee n)}$.  Then there exist some positive constants $c_0$, $c_1$, such that for all $n$ sufficiently large,
		\begin{align*}
		&P\left(\max_{2\leq j \leq p} \sup_{\vbeta\in\bbB_1}\left| n^{-1/2}\sum_{i=1}^n\vtheta_{j}^T\vnu_{11}(Z_i,\vbeta)\right|\geq    4ch^2s\sqrt{\log(p\vee n)} \right) 
		\leq \exp[-c_1\log (p\vee n)].
		\end{align*} 
		Combining all the above results, we prove (\ref{ghat_g1}).
		
		To prove (\ref{ghat_g2}), we note that 
		\begin{align*}
		&\max_{2\leq j \leq p}\left|n^{-1/2}\sum_{i=1}^n\vtheta_j^T\vnu_2(Z_i,\vbetah,\widehat{G},\widehat{G}^{(1)})\right|\\
		\leq &\max_{2\leq j \leq p}\left| n^{-1/2}\sum_{i=1}^n  \Big[G(\vx_{i}^{T} \vbetah|\vbetah)-\widehat{G}(\vx_{i}^{T} \vbetah|\vbetah)  \Big] G^{(1)}(\vx_i^T\vbetah|\vbetah) \big[\vx_{i,-1}-\E(\vx_{i,-1}|\vx_i^T\vbetah)\big]^T\vtheta_j\right| \\
		&+\max_{2\leq j \leq p}\left| n^{-1/2}\sum_{i=1}^n  \Big[G(\vx_{i}^{T} \vbetah|\vbetah)-\widehat{G}(\vx_{i}^{T} \vbetah|\vbetah)  \Big] \big[\widehat{G}^{(1)}(\vx_i^T\vbetah|\vbetah)-G^{(1)}(\vx_i^T\vbetah|\vbetah)\big] \big[\vx_{i,-1}-\E(\vx_{i,-1}|\vx_i^T\vbeta_0)\big]^T\vtheta_j\right| \\ 
		&+\max_{2\leq j \leq p}\bigg| n^{-1/2}\sum_{i=1}^n  \Big[G(\vx_{i}^{T} \vbetah|\vbetah)-\widehat{G}(\vx_{i}^{T} \vbetah|\vbetah)  \Big] \big[\widehat{G}^{(1)}(\vx_i^T\vbetah|\vbetah)-G^{(1)}(\vx_i^T\vbetah|\vbetah)\big]\\
		&\qquad\qquad\qquad\quad* \big[\E(\vx_{i,-1}|\vx_i^T\vbetah)-\E(\vx_{i,-1}|\vx_i^T\vbeta_0)\big]^T\vtheta_j\bigg| \\ 
		&+\max_{2\leq j \leq p}\left| n^{-1/2}\sum_{i=1}^n  \Big[G(\vx_{i}^{T} \vbetah|\vbetah)-\widehat{G}(\vx_{i}^{T} \vbetah|\vbetah)  \Big]\widehat{G}^{(1)}(\vx_i^T\vbetah|\vbetah) \big[\E(\vx_{i,-1}|\vx_i^T\vbetah)-\widehat{\E}(\vx_{i,-1}|\vx_i^T\vbetah)\big]^T\vtheta_j\right| \\
		=&\sum_{k=1}^{4}V_{n2k} ,
		\end{align*}
		where the definition of $V_{n2k}$, $k=1,\cdots,4$, is clear from the context. 
		
		We first bound $V_{n22}$.  Lemma~\ref{Gfunc} and Lemma~\ref{G1func} imply that $\max_{1\leq i \leq n}\Big|G(\vx_{i}^{T} \vbetah|\vbetah)-\widehat{G}(\vx_{i}^{T} \vbetah|\vbetah) \Big|\leq c_0h^2$, and $\max_{1\leq i \leq n}\Big|\widehat{G}^{(1)}(\vx_i^T\vbetah|\vbetah)-G^{(1)}(\vx_i^T\vbetah|\vbetah) \Big|\leq c_0h$, with probability at least $1-\exp(-c_1\log p)$, for some positive constants $c_0$, $c_1$ and all $n$ sufficiently large.  Note that $\vx_{i,-1}-\E(\vx_{i,-1}|\vx_i^T\vbeta_0)$ is sub-Gaussian by Lemma~\ref{lem:subg}. Then  Lemma~\ref{lem:thetaj} and Lemma~\ref{lem14NCL} imply that $\max_{2\leq j \leq p}n^{-1}\sum_{i=1}^n \Big\{ \big[\vx_{i,-1}-\E(\vx_{i,-1}|\vx_i^T\vbeta_0)\big]^T\vtheta_j\Big\}^2\leq c_2$,  with probability at least $1-\exp(-c_3n)$, for some positive constants $c_2$, $c_3$ and all $n$ sufficiently large.  Combining all these results, we have that $V_{n22}\leq d_0\sqrt{n}h^3$, with probability at least $1-\exp(-d_1\log p)$, for some positive constants $d_0$, $d_1$ and all $n$ sufficiently large.  
		
		Next, Assumption~\ref{A2}-(c) implies that 
		\begin{align*}
		&\max_{2\leq j \leq p}n^{-1}\sum_{i=1}^n \Big\{ \big[\E(\vx_{i,-1}|\vx_i^T\vbetah)-\E(\vx_{i,-1}|\vx_i^T\vbeta_0)\big]^T\vtheta_j\Big\}^2\\
		\leq&	\max_{2\leq j \leq p}Cn^{-1}\sum_{i=1}^n  \big[|\vx_i^T\vbetah-\vx_i^T\vbeta_0| + (|\vx_i^T\vbetah|+|\vx_i^T\vbeta_0|)||\vbetah-\vbeta_0||_2\big]^2||\vtheta_j||_2^2\\
		\leq&	\max_{2\leq j \leq p}2C\xi_2^{-2}n^{-1}\sum_{i=1}^n  \big[|\vx_i^T\vbetah-\vx_i^T\vbeta_0|^2 + (|\vx_i^T\vbetah|+|\vx_i^T\vbeta_0|)^2||\vbetah-\vbeta_0||_2^2\big] \\
		\leq&c_0||\vbetah-\vbeta_0||_2^2,
		\end{align*}
		with probability at least $1-\exp(-c_1n)$, for some positive constants $c_0$, $c_1$ and all $n$ sufficiently large. In the above, the second inequality applies Lemma~\ref{lem:thetaj}; the last inequality applies  Lemma~\ref{lem14NCL}.
		Combining Theorem~\ref{Lasso_error} with the above results, we have  $V_{n23}\leq d_0\sqrt{ns}h^5$, with probability at least $1-\exp(-d_1\log p)$, for some positive constants $d_0$, $d_1$ and all $n$ sufficiently large. 
		
		Observe that Theorem~\ref{Lasso_error}, Lemma~\ref{lem:Ebound} and Lemma~\ref{lem:thetaj} imply that $\max_{1\leq i \leq n}\Big|\big[\E(\vx_{i,-1}|\vx_i^T\vbetah)-\widehat{\E}(\vx_{i,-1}|\vx_i^T\vbetah)\big]^T\vtheta_j\Big|\leq c_0h^2$, with probability at least $1-\exp(-c_1\log p)$, for some positive constants $c_0$, $c_1$ and all $n$ sufficiently large. Then the above results and Assumption~\ref{K4}-(a) indicate that $V_{n24}\leq d_0\sqrt{n}h^4$, with probability at least $1-\exp(-d_1\log p)$, for some positive constants $d_0$, $d_1$ and all $n$ sufficiently large. 
		
		Finally we bound $V_{n21}$.  Define
		\begin{align*}
		\bbU  = \Big\{u(\cdot|\vbeta):\vbeta\in\bbB_1,\ u(\cdot|\vbeta)\in C_1^1(T),\ \forall\ \vbeta,\mbox{ and }  \max_{1\leq i \leq n}\sup_{\vbeta\in\bbB_1}\left|u(\vx_i^T\vbeta|\vbeta)-G(\vx_i^T\vbeta|\vbeta)\right|\leq c_1h^2\Big\},
		\end{align*}  
		for some positive constant $c_1$, where $T=\{t\in\bbR: |t|\leq 2||\vbeta_0||_2\sigma_x\sqrt{\log (p\vee n)}\}$, and $C_1^1(T)$ is the set of all continuous and Lipschitz functions $f:T\ra \bbR$. 
		It is sufficient to show that there exist some positive constants $c_0$, $c_1$, such that for all $n$ sufficiently large,
		\begin{align}
		P\left(\max_{2\leq j \leq p}\sup_{\vbeta\in\bbB_1,u\in\bbU} \Big|n^{-1/2}\sum_{i=1}^n \vtheta_j^T\vnu_{21}(Z_i,\vbeta,u) \Big|\geq c_0h[\log (p\vee n)]^{1/4} \right)\leq\exp[-c_1\log (p\vee n)], \label{ghat_g_1}
		\end{align} 
		where $Z_i=(\vx_i,\epsilon_i,A_i)$,  
		$\vnu_{21}(Z_i,\vbeta, u) = \big[G(\vx_i^T\vbeta|\vbeta) -u(\vx_i^T\vbeta|\vbeta)\big] G^{(1)}(\vx_i^T\vbeta|\vbeta) \big[\vx_{i,-1}-\E(\vx_{i,-1}|\vx_i^T\vbeta) \big]$,   $u(\vx^T\vbeta|\vbeta)$ depends on $\vx$ only through $\vx^T\vbeta$. 
		Note   the assumptions of Theorem~\ref{Lasso_error} imply that $ h [\log(p\vee n)]^{1/4} \leq d_0(nh^9)^{1/4}\leq d_0h^3\sqrt{n}*(nh^3)^{-1/4}\leq d_0sh^3\sqrt{n}$ for some positive constant $d_0$. 
		
		To prove (\ref{ghat_g_1}), we have
		\begin{align*}
		&P\left(\max_{2\leq j \leq p}\sup_{\vbeta\in\bbB_1,u\in\bbU}\Big|n^{-1/2}\sum_{i=1}^n\vtheta_j^T \vnu_{21}(Z_i,\vbeta, u) \right|>t\Big)\\
		\leq &P\left(\max_{2\leq j \leq p}\sup_{\vbeta\in\bbB_1,u\in\bbU}\Big|n^{-1/2}\sum_{i=1}^n\vtheta_j^T \vnu_{21}(Z_i,\vbeta, u) \Big|>t\ \bigg|\mathcal{G}_n \cap \mathcal{J}_n\cap \mathcal{K}_n \right) +\exp[-c\log (p\vee n)]\\
		\leq& \sum_{j=2}^{p}P\left( \sup_{\vbeta\in\bbB_1,u\in\bbU} \Big|n^{-1/2}\sum_{i=1}^n \vtheta_j^T \vnu_{21}(Z_i,\vbeta, u) \Big|>t\ \bigg|\mathcal{G}_n \cap \mathcal{J}_n \cap \mathcal{K}_n \right)+\exp[-c\log (p\vee n)],
		\end{align*}
		for some positive constant $c$, and all $n$ sufficiently large, where the events $\mathcal{G}_n$, $\mathcal{J}_n$ and $\mathcal{K}_n$ are defined in Lemma~\ref{lem:events}.  
		Note that %$\sup_{\vbeta\in\bbB_1,u\in\bbU}\E\{[ \vtheta_j^T \vnu_{21}(Z_i,\vbeta,u)]^2\}\leq c h^4$,   
		$\E[ \vtheta_j^T \vnu_{21}(Z_i,\vbeta,u)] =0$, and $ \sup\limits_{\vbeta\in\bbB_1,u\in\bbU} n^{-1} \sum_{i=1}^n\big| \vtheta_{j}^T\vnu_{21}(Z_i,\vbeta, u)  \big|^2\leq C h^4$, on the event $\mathcal{G}_n$, 
		for some positive constant  $C$, and any $j$.  Hence %$ \left[s^2h^4\log(p\vee n) \right]^{-1/2} a_i\vtheta_{j}^T\vnu_1(Z_i,\vbeta, u) $ satisfies the conditions of
		by Massart's concentration inequality (e.g., Theorem 14.2, \citet{buhlmann2011}),
		% on the event $\mathcal{G}_n  \cap\mathcal{J}_n \cap \mathcal{K}_n $    
		$\forall\ t>0$, 
		\begin{align} 
		P\bigg( \sup_{\vbeta\in\bbB_1,u\in\bbU} \Big|n^{-1/2} \sum_{i=1}^n \vtheta_j^T \vnu_{21}(Z_i,\vbeta, u)  \Big| 	\geq  E_{sup,2} + th^2\sqrt{ n}  \bigg| \mathcal{G}_n \cap\mathcal{J}_n \cap \mathcal{K}_n  \bigg) 
		\leq &\exp\left( - \frac{nt^2}{8}\right),\label{mass_22}
		\end{align}  
		where $E_{sup,2} =  \E\left[\sup_{\vbeta\in\bbB_1,u\in\bbU} \Big| n^{-1/2}\sum_{i=1}^n \vtheta_j^T \vnu_{21}(Z_i,\vbeta, u) \Big|\ \bigg|\mathcal{G}_n \cap\mathcal{J}_n \cap \mathcal{K}_n \right]$.
		
		Similarly as the proof of Lemma~\ref{lem:ghat1_g1}, we can cover $\bbB_1$   by ${p\choose ks}\left(1+\frac{2c_0s^{3/2}h^2\log(p\vee n)}{\delta_n^2}\right)^{ks}$ $L_2-$balls of radius $\frac{\delta_n^2}{s\log(p\vee n)}$.
		Denote the centers of these $L_2-$balls by $\vbeta^{\circ}_{l0},\cdots,\vbeta^{\circ}_{lN_l}$.
		Denote these $L_2-$balls by $\mathbb{C}(\vbeta^{\circ}_{ll'})$, $l=1,\cdots,m(s)$, $l'=1,\cdots,N_l$.
		
		%Let $T=\{t\in\bbR: |t|\leq \sqrt{\log (p\vee n)}\}$, and $C_1^1(T)$ be the set of all continuous and Lipschitz functions $f:T\ra \bbR$. 
		Observe that on the event $\mathcal{J}_n$, $\max_{1\leq i \leq n}\sup_{\vbeta\in\bbB_1}|\vx_i^T\vbeta|\in T$, then $\bigcup_{\vbeta\in\bbB_1}u(\cdot|\vbeta)\in C_1^1(T)$. By Theorem 2.7.1 in \citet{van1996weak}, we have 
		$$\log N\left(\delta_n, C_1^1(T),L_2(\bbP_n)\right)\leq C||\vbeta_0||_2\frac{\sqrt{\log(p\vee n)}}{\delta_n},$$
		for some positive constant $C$. So we can find $N_2\leq \exp\big[C||\vbeta_0||_2\delta_n^{-1}\sqrt{\log(p\vee n)}\big]$ balls with the centers $u_a^{\circ} (\cdot)$, $a=1,\cdots,N_2$, to cover $ C_1^1(T)$. %for any $u\in\bbU$, there exists $l,\ l'$, and $u_a^{\circ} (\cdot)$ such that $\vbeta\in \mathbb{C}(\vbeta^{\circ}_{ll'})$, and
		Hence $\forall\ \vbeta\in\bbB_1$, $u(\cdot|\vbeta)\in C_1^1(T)$,  we can find $l,\ l'$ and $a$ such that $\vbeta\in \mathbb{C}(\vbeta^{\circ}_{ll'})$ and $u_a^{\circ} (\cdot)$ satisfies 
		\begin{align}
		n^{-1} \sum_{i=1}^n  \left[u(\vx_i^T\vbeta_1|\vbeta_1)-u_a^{\circ}(\vx_i^T\vbeta_1)\right]^2 \leq \delta_n^2,\ \forall \vbeta_1\in\bbB_1.\label{cover_u}
		\end{align} 
		On the event $\mathcal{G}_n\cap \mathcal{J}_n \cap \mathcal{K}_n$, we have % there exists some positive constant $c_0$, such that 
		\begin{align*}
		&n^{-1}\left| \sum_{i=1}^n\vtheta_j^T \vnu_{21}(Z_i,\vbeta, u) -\sum_{i=1}^n \vtheta_j^T\vnu_{21}(Z_i,\vbeta^{\circ}_{ll'}, u_a^{\circ})  \right|\\
		\leq &n^{-1}\left| \sum_{i=1}^n\left[G(\vx_i^T\vbeta|\vbeta) - G(\vx_i^T\vbeta^{\circ}_{ll'}|\vbeta^{\circ}_{ll'})\right]G^{(1)}(\vx_i^T\vbeta|\vbeta) [\vx_{i,-1} -\E(\vx_{i,-1} |\vx_i^T\vbeta)] ^T\vtheta_{j}\right| \\ 
		&+n^{-1}\left| \sum_{i=1}^n\left[u(\vx_i^T\vbeta|\vbeta) - u_a^{\circ}(\vx_i^T\vbeta)\right]G^{(1)}(\vx_i^T\vbeta|\vbeta) [\vx_{i,-1}  -\E(\vx_{i,-1} |\vx_i^T\vbeta)] ^T\vtheta_{j}\right| \\ 
		&+n^{-1}\left| \sum_{i=1}^n\left[u_a^{\circ}(\vx_i^T\vbeta)- u_a^{\circ}(\vx_i^T\vbeta^{\circ}_{ll'})\right]G^{(1)}(\vx_i^T\vbeta|\vbeta) [\vx_{i,-1}  -\E(\vx_{i,-1} |\vx_i^T\vbeta)] ^T\vtheta_{j}\right| \\ 
		&+n^{-1}\left| \sum_{i=1}^n\left[G(\vx_i^T\vbeta^{\circ}_{ll'}|\vbeta^{\circ}_{ll'})- u(\vx_i^T\vbeta^{\circ}_{ll'}|\vbeta^{\circ}_{ll'})\right]\big[G^{(1)}(\vx_i^T\vbeta|\vbeta)-G^{(1)}(\vx_i^T\vbeta^{\circ}_{ll'}|\vbeta^{\circ}_{ll'})\big] [\vx_{i,-1}  -\E(\vx_{i,-1} |\vx_i^T\vbeta)] ^T\vtheta_{j}\right| \\ 
		&+n^{-1}\left| \sum_{i=1}^n\left[ u(\vx_i^T\vbeta^{\circ}_{ll'}|\vbeta^{\circ}_{ll'})- u_a^{\circ}(\vx_i^T\vbeta^{\circ}_{ll'})\right]\big[G^{(1)}(\vx_i^T\vbeta|\vbeta)-G^{(1)}(\vx_i^T\vbeta^{\circ}_{ll'}|\vbeta^{\circ}_{ll'})\big] [\vx_{i,-1}  -\E(\vx_{i,-1} |\vx_i^T\vbeta)] ^T\vtheta_{j}\right| \\  
		&+n^{-1}\left| \sum_{i=1}^n\left[G(\vx_i^T\vbeta^{\circ}_{ll'}|\vbeta^{\circ}_{ll'})-  u(\vx_i^T\vbeta^{\circ}_{ll'}|\vbeta^{\circ}_{ll'})\right]G^{(1)}(\vx_i^T\vbeta^{\circ}_{ll'}|\vbeta^{\circ}_{ll'}) [\E(\vx_{i,-1} |\vx_i^T\vbeta^{\circ}_{ll'}) -\E(\vx_{i,-1} |\vx_i^T\vbeta)] ^T\vtheta_{j}\right|\\
		&+n^{-1}\left| \sum_{i=1}^n\left[ u(\vx_i^T\vbeta^{\circ}_{ll'}|\vbeta^{\circ}_{ll'})- u_a^{\circ}(\vx_i^T\vbeta^{\circ}_{ll'})\right]G^{(1)}(\vx_i^T\vbeta^{\circ}_{ll'}|\vbeta^{\circ}_{ll'}) [\E(\vx_{i,-1} |\vx_i^T\vbeta^{\circ}_{ll'}) -\E(\vx_{i,-1} |\vx_i^T\vbeta)] ^T\vtheta_{j}\right|\\
		=&\sum_{j=1}^{7}R_{nj}, 
		\end{align*}  
		where the definition of $R_{nj}$, $j=1,\cdots,7$, is clear from the context. (\ref{Gt})  in Lemma~\ref{lem:Gbound} and the event $\mathcal{G}_n$ implies that 
		$R_{n1}\leq c_1\sqrt{||\vbeta-\vbeta^{\circ}_{ll'}||_2s\log (p\vee n)} \leq d_0\delta_n,$
		with probability at least $1-\exp[-c_2\log(p\vee n)]$, for some positive constants $d_0$, $c_1$, $c_2$, and all $n$ sufficiently large. 
		The Cauchy-Schwartz inequality implies that
		\begin{align*}
		R_{n2}\leq \sqrt{n^{-1} \sum_{i=1}^n\left[u(\vx_i^T\vbeta|\vbeta) - u_a^{\circ}(\vx_i^T\vbeta\right] ^2}\sqrt{n^{-1}  \sum_{i=1}^n\left\{ G^{(1)}(\vx_i^T\vbeta|\vbeta) [\vx_{i,-1}  -\E(\vx_i|\vx_{i,-1}^T\vbeta)] ^T\vtheta_{j}\right\}^2 }  \leq d_0\delta_n,
		\end{align*}
		for some positive constant $d_0$, which applies the event $\mathcal{G}_n$, Assumption~\ref{K4}-(a) and (\ref{cover_u}). Since $u_a^{\circ}(\cdot)\in C_1^1(T)$, we have  
		\begin{align*}
		R_{n3}\leq d_1\sqrt{n^{-1} \sum_{i=1}^n\left[\vx_i^T(\vbeta-\vbeta^{\circ}_{ll'})\right] ^2}\sqrt{n^{-1}  \sum_{i=1}^n\left\{ G^{(1)}(\vx_i^T\vbeta|\vbeta) [\vx_{i,-1} -\E(\vx_i|\vx_{i,-1}^T\vbeta)] ^T\vtheta_{j}\right\}^2 } \leq d_0\delta_n,
		\end{align*}
		for some positive constants $d_0$, $d_1$, which applies the event  $\mathcal{G}_n\cap \mathcal{K}_n$, and Assumption~\ref{K4}-(a). Similarly, we can derive that $R_{n4}\leq d_0h^2\delta_n$ and $R_{n5}\leq d_0\delta_n^2$,  with probability at least $1-\exp[-c_2\log(p\vee n)]$, for some positive constants $d_0$, $c_2$, and all $n$ sufficiently large, which applies (\ref{Gtt}) in Lemma~\ref{lem:Gbound} and the event $\mathcal{G}_n$. In addition, we have $R_{n6}\leq d_0h^2\delta_n$ and $R_{n7}\leq d_0\delta_n^2$, which both apply Assumption~\ref{A2}-(c) and  Assumption~\ref{K4}-(a).

		Hence, the $\delta_n-$covering number of the class of functions $V_{2j}=\{\vtheta_j^T\vnu_{21}(Z_i,\vbeta, u) :\vbeta\in\bbB_1, u\in\bbU\}$ is bounded by 
		\begin{align}
		N (\delta_n,V_{2j},L_1(\mathbb{P}_n))\leq c{p\choose ks}\left[1+\frac{2c_0s^{3/2}h^2\log (p\vee n)}{\delta_n^2}\right]^{ks}\exp\big[C||\vbeta_0||_2\delta_n^{-1}\sqrt{\log(p\vee n)}\big], \ \forall j.\label{cover_no2}
		\end{align} 
		
		Let $a_1,\cdots,a_n$ be a Rademacher sequence that is independent of the data.  The symmetrization theorem (Theorem~14.3 in \citet{buhlmann2011}) implies that 
		\begin{align*} 
		E_{sup,2} = & \E\left[\sup_{\vbeta\in\bbB_1,u\in\bbU} \Big| n^{-1/2}\sum_{i=1}^n \vtheta_j^T \vnu_{21}(Z_i,\vbeta, u) \Big|\ \bigg|\mathcal{G}_n \cap\mathcal{J}_n\cap\mathcal{K}_n \right]\\
		\leq &2\E_{ \{\vx_i\}_{i=1}^n}\left\{\E\left[\sup_{\vbeta\in\bbB_1,u\in\bbU} \Big| n^{-1/2}\sum_{i=1}^n a_i\vtheta_j^T \vnu_{21}(Z_i,\vbeta, u) \Big|\ \bigg|\{\vx_i\}_{i=1}^n,\mathcal{G}_n \cap\mathcal{J}_n\cap\mathcal{K}_n \right]\right\}. 
		\end{align*}

		Note that on the event $\mathcal{G}_n\cap\mathcal{J}_n$, we have that $\sup_{\vbeta\in\bbB_1,u\in\bbU}\sqrt{n^{-1}\sum_{i=1}^n[a_i \vtheta_j^T \vnu_{21}(Z_i,\vbeta,u)]^2 }\leq ch^2\triangleq R_n$.  Lemma~14.18  in \citet{van1996weak} implies that %on the event $\mathcal{G}_n$,
		\begin{align*}
		E_{sup,2}\leq& 2\E_{ \{\vx_i\}_{i=1}^n}\left\{\E\left[\sup_{\vbeta\in\bbB_1,u\in\bbU} \Big| n^{-1/2}\sum_{i=1}^n a_i\vtheta_j^T \vnu_{21}(Z_i,\vbeta, u) \Big|\ \bigg|\{\vx_i\}_{i=1}^n,\mathcal{G}_n \cap\mathcal{J}_n\cap\mathcal{K}_n \right]\right\} \\ 
		\leq &ch^2 \left\{4+6\sum_{l=1}^L 2^{-l}\sqrt{\log \left[N (2^{-l}h^2,V_{j},L_1(\mathbb{P}_n))+1\right]}   \right\}   \\
		\leq& c_0h [\log(p\vee n)]^{1/4},
		\end{align*}
		with probability at least $1-\exp[-c_1\log(p\vee n)]$, for some positive constants $c$, $c_0$, $c_1$, and all $n$ sufficiently large, where $L=\min\{l:l\geq 1,\ 2^{-l}\leq 4/\sqrt{n}\}$, and the last inequality applies (\ref{cover_no2}).   The analysis is similar as that in the proof of Lemma~\ref{lem:ghat1_g1}. 
		
		It follows from (\ref{mass_22}) that 
		\begin{align*}
		&P\bigg( \sup_{\vbeta\in\bbB_1,u\in\bbU} \Big| n^{-1/2}\sum_{i=1}^n \vtheta_j^T \vnu_{21}(Z_i,\vbeta, u)  \Big| 	\geq ch [\log(p\vee n)]^{1/4} +th^2\sqrt{n} \bigg| \mathcal{G}_n \cap\mathcal{J}_n\cap\mathcal{K}_n   \bigg) \\
		\leq &\exp\left(-\frac{nt^2}{8}\right).
		\end{align*} 
		Take $t = 4\sqrt{n^{-1}\log (p\vee n)}$.  Note that the assumptions of Theorem~\ref{Lasso_error} imply $ h [\log(p\vee n)]^{1/4} \leq  d_1  (nh^9)^{1/4}\leq d_1$, for some positive constant $d_1$. Hence there exist some positive constants $c_0$, $c_1$, such that for all $n$ sufficiently large, 
		\begin{align*}
		&P\left(\max_{2\leq j \leq p} \sup_{\vbeta\in\bbB_1,u\in\bbU}\left| n^{-1/2}\sum_{i=1}^n\vtheta_{j}^T\vnu_{21}(Z_i,\vbeta,u)\right|\geq   c_0  h [\log(p\vee n)]^{1/4} \right) 
		\leq \exp[-c_1\log (p\vee n)].
		\end{align*} 
		Therefore, combining all the above results, we prove (\ref{ghat_g2}).

	\end{proof}

	\begin{proof}[Proof of Lemma~\ref{lem:Ehat_E}] 
		We will prove (\ref{Ehat_E}) below. The proof of  (\ref{Ehat_E_inf}) is similar.  
		
		Theorem~\ref{Lasso_error} implies that $P\left(\vbetah\in\bbB_1\right)\geq 1-\exp(-c\log p)$, for some positive constant $c$, and all $n$ sufficiently large. Lemma~\ref{Gbetafunc} implies that 
		$$P\left(\max_{ 1\leq i \leq n}\sup_{\vv\in\mathbb{K}(p,2ks+\widetilde{s})}\left|\widehat{\E}(\vx_i^T\vv|\vx_i^T\vbetah)-\E (\vx_i^T\vv|\vx_i^T\vbeta_0)  \right| \geq c_0sh^2\sqrt{\log(p\vee n)} \right) \leq  \exp(-c_1\log p),$$
		for some positive constants $c_0$, $c_1$ %, any fixed unit vector $\vv\in\bbR^p$, 
		and all $n$ sufficiently large. Define
		\begin{align*}
		\bbE  = \Big\{\vE(\cdot|\vbeta):\vbeta\in\bbB_1;  \mbox{ for any  unit vector }\vv\in\bbR^p, \vE(\cdot|\vbeta)^T\vv\in C_1^1(T),\  \forall\ \vbeta, \\\mbox{and } \max_{1\leq i \leq n}\sup_{\substack{\vbeta\in\bbB_1\\\vv\in\mathbb{K}(p,2ks+\widetilde{s})}}\Big|[\vE(\vx_i^T\vbeta|\vbeta)-\E(\vx_i |\vx_i^T\vbeta_0)]^T\vv\Big|\leq c_1  sh^2\sqrt{\log(p\vee n)} 
		\Big\},
		\end{align*} 
		for some positive constant $c_1$, where $T=\big\{t\in\bbR: |t|\leq 2||\vbeta_0||_2\sigma_x\sqrt{\log (p\vee n)}\big\}$, and $C_1^1(T)$ is the set of all continuous and Lipschitz functions $f:T\ra \bbR$. 
		To prove (\ref{Ehat_E}), it is sufficient to prove that there exist some positive constants $c_0$, $c_1$, such that for all $n$ sufficiently large,
		\begin{align*}
		P\left(\max_{2\leq j \leq p} \sup_{\vbeta\in\bbB_1,\vE\in\bbE} \Big|n^{-1/2}\sum_{i=1}^n\vtheta_j^T\vxi(Z_i,\vbeta,\vE) \Big| \geq c_0h\sqrt{s\log(p\vee n)}\right)\leq\exp[-c_1\log (p\vee n)],  
		\end{align*} 
		where $\vxi(Z_i,\vbeta,\vE) =  \widetilde{\epsilon}_i G^{(1)}(\vx_i^T\vbeta_0|\vbeta_0) [\vE_{-1}(\vx_i^T\vbeta|\vbeta) -\E(\vx_{i,-1}|\vx_i^T\vbeta_0)\big] $,    $\vE(\vx^T\vbeta|\vbeta)$ depends on $\vx$ only through $\vx^T\vbeta$, and $\vE_{-1}(\vx_i^T\vbeta|\vbeta)$ denotes the $(p-1)-$subvector of $\vE(\vx^T\vbeta|\vbeta)$ that excludes the $1^{st}$ entry.
		
		Note that $\max_{ 1\leq i \leq n}\left|G^{(1)}(\vx_i^T\vbeta_0|\vbeta_0)\right|\leq b$ according to assumption \ref{A1}-(b). Using the same technique as that in Lemma~\ref{lem:events}, the event $\mathcal{Q}_n = \big\{n^{-1}\sum_{i=1}^{n}\left[\widetilde{\epsilon}_iG^{(1)}(\vx_i^T\vbeta_0)\right]^2\leq 5b^2(\sigma_\ep^2+M^2)\big\}$ holds with probability at least $1-\exp(-cn)$, for some universal constant $c>0$.  Then  there exist some positive constant $c$, such that for all $n$ sufficiently large,
		\begin{align*}
		&P\left(\max_{2\leq j \leq p}\sup_{\vbeta\in\bbB_1,\vE\in\bbE} \Big|n^{-1/2}\sum_{i=1}^n\vtheta_{j}^T \vxi(Z_i,\vbeta,\vE) \Big|>t\right)\\ 
		\leq &\sum_{j=2}^pP\left(\sup_{\vbeta\in\bbB_1,\vE\in\bbE}\Big|n^{-1/2}\sum_{i=1}^n\vtheta_{j}^T \vxi(Z_i,\vbeta,\vE) \Big|>t\ \bigg|\mathcal{Q}_n \cap \mathcal{J}_n\cap \mathcal{K}_n \right) +\exp[-c\log(p\vee n)].
		\end{align*}
		
		Note that $\vtheta_j^T\vxi(Z_i,\vbeta,\vE) = 2(2A_i-1)\left[\epsilon_i+g(\vx_i)\right]G^{(1)}(\vx_i^T\vbeta_0|\vbeta_0)\vtheta_j^T\left[\vE_{-1}(\vx_i^T\vbeta|\vbeta) -\E(\vx_{i,-1}|\vx_i^T\vbeta_0)\right]$, where $(2A_i-1)$ is a Rademacher sequence, and independent of $(\vx_i,\epsilon_i)$. Hence given $\{(\vx_i,\ep_i)\}_{i=1}^n$, on the event $\mathcal{Q}_n $, we have 
		$$ \sup_{\vbeta\in\bbB_1,\vE\in\bbE} n^{-1} \sum_{i=1}^n\left| \vtheta_{j}^T\vxi(Z_i,\vbeta,\vE) \right|^2\leq Cs^2h^4\log(p\vee n),$$ 
		for some positive constant  $C$, and any $j$. Hence %$ \left[sh^2\sqrt{\log(p\vee n)}\right]^{-1}\vtheta_{j}^T\vxi(Z_i,\vbeta,\vE)$ satisfies the conditions of
		by Massart's concentration inequality (e.g., Theorem 14.2, \citet{buhlmann2011}) on $\mathcal{Q}_n $, given $\{(\vx_i,\ep_i)\}_{i=1}^n$,  $\forall\ t>0$, 
		\begin{align} 
		&P\bigg( \sup_{\vbeta\in\bbB_1,\vE\in\bbE} \Big|n^{-1/2}\sum_{i=1}^n \vtheta_{j}^T\vxi(Z_i,\vbeta,\vE) \Big|\geq\E_{sup}\big[\vtheta_{j}^T\vxi(Z_i,\vbeta,\vE)\big]+tsh^2\sqrt{n\log(p\vee n)} \nonumber\\
		&\qquad\qquad\qquad\qquad\qquad\qquad\qquad\qquad\qquad\qquad\qquad\quad  \ \bigg| \{(\vx_i,\ep_i)\}_{i=1}^n,\mathcal{Q}_n \cap\mathcal{J}_n\cap \mathcal{K}_n \bigg) \nonumber\\
		\leq &\exp\left(-\frac{nt^2}{8}\right),\label{mass_3}
		\end{align}
		where $\E_{sup}\big[\vtheta_{j}^T\vxi(Z_i,\vbeta,\vE)\big] =\E\left[ \sup_{\vbeta\in\bbB_1,\vE\in\bbE} \left|n^{-1/2}\sum_{i=1}^n \vtheta_{j}^T\vxi(Z_i,\vbeta,\vE) \right|\ \Big| \{(\vx_i,\ep_i)\}_{i=1}^n,\mathcal{Q}_n \cap\mathcal{J}_n\cap \mathcal{K}_n\right] $
		It remains to find an upper bound for $\E_{sup}\big[\vtheta_{j}^T\vxi(Z_i,\vbeta,\vE)\big] $. 
		
		Use the $\delta_n-$cover for $\bbB_1$ and $C_1^1(T)$ constructed in the proof of Lemma~\ref{lem:ghat1_g1}. Hence $\forall\ \vbeta\in\bbB_1$, $\left[\vtheta_j^T\vE_{-1}(\cdot|\vbeta)\right]\in C_1^1(T)$ for any $j=2,\cdots,p$,  we can find $l,\ l'$ and $a$ such that $\vbeta\in \mathbb{C}(\vbeta^{\circ}_{ll'})$ and $\vE_a^{\circ} (\cdot)$ satisfies 
		\begin{align*}
		n^{-1} \sum_{i=1}^n  \left\{[\vE_{-1}(\vx_i^T\vbeta_1|\vbeta)- \vE_{a,-1}^{\circ}(\vx_i^T\vbeta_1)]^T\vtheta_j\right\}^2\leq \delta_n^2,\ \ \forall\vbeta_1\in\bbB_1. 
		\end{align*}  
		On the event $\mathcal{Q}_n\cap \mathcal{J}_n \cap \mathcal{K}_n$, there exists some positive constant $c$, such that 
		\begin{align*}
		&n^{-1}\left| \sum_{i=1}^n\vtheta_j^T\vxi(Z_i,\vbeta,\vE)-\sum_{i=1}^n \vtheta_j^T\vxi(Z_i,\vbeta^{\circ}_{ll'},\vE_a^{\circ})\right|\\
		\leq &n^{-1}\left| \sum_{i=1}^n\vtheta_j^T\left[\vE_{-1}(\vx_i^T\vbeta|\vbeta) - \vE_{a,-1}^{\circ}(\vx_i^T\vbeta^{\circ}_{ll'})\right]\widetilde{\epsilon}_i G^{(1)}(\vx_i^T\vbeta_0|\vbeta_0)\right| \\
		\leq &\sqrt{n^{-1}\sum_{i=1}^n\left\{\vtheta_j^T\left[\vE_{-1}(\vx_i^T\vbeta|\vbeta) - \vE_{-1}(\vx_i^T\vbeta^{\circ}_{ll'}|\vbeta)\right]\right\}^2+\left\{ \vtheta_j^T\left[\vE_{-1}(\vx_i^T\vbeta^{\circ}_{ll'}|\vbeta) - \vE_{a,-1}^{\circ}(\vx_i^T\vbeta^{\circ}_{ll'})\right]\right\}^2} \\
		&* \sqrt{n^{-1}\sum_{i=1}^n\left[\widetilde{\epsilon}_i G^{(1)}(\vx_i^T\vbeta_0|\vbeta_0)\right]^2} \\
		\leq & c\delta_n.
		\end{align*} 
		%Then $\left[n^{-1}\sum_{i=1}^n \vtheta_j^T\vxi(Z_i,\vbeta^{\circ}_{ll'},\vE_a^{\circ})-c\delta_n, n^{-1}\sum_{i=1}^n\vtheta_j^T\vxi(Z_i,\vbeta^{\circ}_{ll'},\vE_a^{\circ})+c\delta_n\right]$ forms a $c\delta_n-$cover for $n^{-1}\sum_{i=1}^n\vtheta_j^T\vxi(Z_i,\vbeta,\vE)$. 
		Hence, the $\delta_n-$covering number of the class of functions $\Phi_j=\{\vtheta_j^T\vxi(Z_i,\vbeta,\vE):\vbeta\in\bbB_1, \vE\in\bbE\}$ satisfies
		\begin{align}
		N (\delta_n,\Phi_j,L_1(\mathbb{P}_n))\leq c{p\choose ks}\left(1+\frac{2c_0\sqrt{s}h^2}{\delta_n}\right)^{ks}\exp\big[C||\vbeta_0||_2\delta_n^{-1}\sqrt{\log(p\vee n)}\big],\ \ \forall j.\label{cover_no3}
		\end{align}

		Recall $\vtheta_j^T\vxi(Z_i,\vbeta,\vE) = 2(2A_i-1)\left[\epsilon_i+g(\vx_i)\right]G^{(1)}(\vx_i^T\vbeta_0|\vbeta_0)\vtheta_j^T\left[\vE_{-1}(\vx_i^T\vbeta|\vbeta) -\E(\vx_{i,-1}|\vx_i^T\vbeta_0)\right]$, where $(2A_i-1)$ is a Rademacher sequence, and independent of $(\vx_i,\epsilon_i)$.
		Note that on the event $\mathcal{Q}_n \cap\mathcal{J}_n\cap \mathcal{K}_n$, we have that $\sup_{\vbeta\in\bbB_1,\vE\in\bbE} \sqrt{n^{-1}\sum_{i=1}^{n}\left[\vtheta_{j}^T\vxi(Z_i,\vbeta,\vE)^2\right]} =c_0sh^2\sqrt{\log(p\vee n)}\triangleq R_n$.
		Therefore, Lemma~14.18  in \citet{van1996weak} implies that 
		\begin{align*}
		&\E_{sup}\big[\vtheta_{j}^T\vxi(Z_i,\vbeta,\vE)\big] \\
		=&\E\left[\sup_{\vbeta\in\bbB_1,\vE\in\bbE}\Big| n^{-1/2}\sum_{i=1}^n \vtheta_j^T\vxi(Z_i,\vbeta,\vE)\Big| \ \bigg| \{(\vx_i,\ep_i)\}_{i=1}^n,\mathcal{Q}_n \cap\mathcal{J}_n\cap \mathcal{K}_n \right] \\
		\leq &Cs h^2\sqrt{\log(p\vee n)}\left\{4+6\sum_{l=1}^L 2^{-l}\sqrt{\log \left[N (2^{-l}sh^2\sqrt{\log(p\vee n)},\Phi_j,L_1(\mathbb{P}_n))+1\right]}   \right\}   \\
		\leq& c_0h\sqrt{s\log (p\vee n)},
		\end{align*} 
		for some positive constants $C$, $c_0$, and all $n$ sufficiently large, where $L=\min\{l:l\geq 1,\ 2^{-l}\leq 4/\sqrt{n}\}$, and the last inequality applies (\ref{cover_no3}). 
		It follows from (\ref{mass_3}) that 
		\begin{align*}
		&	P\bigg( \sup_{\vbeta\in\bbB_1,\vE\in\bbE} \Big| n^{-1/2}\sum_{i=1}^n\vtheta_{j}^T\vxi(Z_i,\vbeta,\vE)\Big|
		\geq c h\sqrt{s\log (p\vee n)}+tsh^2\sqrt{n\log(p\vee n)} \\
		&\qquad\qquad\qquad\qquad\qquad\qquad\qquad\qquad\quad \bigg| \{(\vx_i,\ep_i)\}_{i=1}^n,\mathcal{Q}_n \cap\mathcal{J}_n \cap \mathcal{K}_n \bigg)  \leq  \exp\left(-\frac{nt^2}{8}\right).
		\end{align*} 
		Take $t = 4\sqrt{n^{-1}\log (p\vee n)}$.  Note that the assumptions of Theorem~\ref{Lasso_error} imply $h\sqrt{s\log (p\vee n)}\leq c_1\sqrt{nh^7}\leq c_1 $, for some positive constant $c_1$. Hence we have  
		\begin{align*}
		&P\left( \sup_{\vbeta\in\bbB_1,\vE\in\bbE}\Big| n^{-1/2}\sum_{i=1}^n\vtheta_{j}^T\vxi(Z_i,\vbeta,\vE)\Big|
		\geq c h\sqrt{s\log (p\vee n)}\ \bigg|\mathcal{Q}_n \cap\mathcal{J}_n\cap \mathcal{K}_n  \right) \\
		= &\E\left\{P\left( \sup_{\vbeta\in\bbB_1,\vE\in\bbE}\Big| n^{-1/2}\sum_{i=1}^n\vtheta_{j}^T\vxi(Z_i,\vbeta,\vE)\Big|
		\geq c h\sqrt{s\log (p\vee n)}\ \bigg|\{(\vx_i,\ep_i)\}_{i=1}^n,\mathcal{Q}_n \cap\mathcal{J}_n\cap \mathcal{K}_n  \right) \right\}\\
		\leq &\exp\left[-2\log (p\vee n)  \right].
		\end{align*} 
		Therefore, there exist some positive constants $c_0$, $c_1$, such that for all $n$ sufficiently large,
		$$P\left(\max_{2\leq j \leq p} \sup_{\vbeta\in\bbB_1,\vE\in\bbE}\Big| n^{-1/2}\sum_{i=1}^n\vtheta_{j}^T\vxi(Z_i,\vbeta,\vE)\Big| \geq c_0h\sqrt{s\log (p\vee n)}\right)\leq\exp[-c_1\log (p\vee n)].$$ 
	\end{proof}

	\begin{proof}[Proof of Lemma~\ref{Gbetafunc}] 
		We will prove (\ref{Gbetabound}) and (\ref{G1betabound}) below. The proof of (\ref{Ebetabound})  is similar.   
		\begin{align*}
		&\max_{ 1\leq i \leq n}\sup \limits_{\vbeta\in\bbB_1 n}\big|\widehat{G}(\vx_i^T\vbeta|\vbeta)-G(\vx_i^T\vbeta_0|\vbeta_0)  \big|\\
		\leq &\max_{ 1\leq i \leq n}\sup \limits_{ \vbeta\in\bbB_1}\big|\widehat{G}(\vx_i^T\vbeta|\vbeta)-G (\vx_i^T\vbeta|\vbeta)  \big|\\
		&+\max_{ 1\leq i \leq n}\sup \limits_{ \vbeta\in\bbB_1}\big|G(\vx_i^T\vbeta|\vbeta)-G(\vx_i^T\vbeta_0|\vbeta_0)  \big|.
		\end{align*} 
		Lemma~\ref{lem:Gbound} implies that %there exists positive constants $c_1$ and $c_2$ such that
		\begin{align*}
		&\max_{ 1\leq i \leq n}\sup \limits_{ \vbeta\in\bbB_1}\big|G(\vx_i^T\vbeta|\vbeta)-G(\vx_i^T\vbeta_0|\vbeta_0)  \big| \\
		\leq&\max_{ 1\leq i \leq n}\sup \limits_{ \vbeta\in\bbB_1}\big|f_0'(\vx_i^T\vbeta)[\vx_{i,-1}-\E(\vx_{i,-1}|\vx_i^T\vbeta)]^T\vgamma\big| + \frac{1}{2}\max_{ 1\leq i \leq n}\sup \limits_{ \vbeta\in\bbB_1}\big|f_0''(\vx_i^T\vbeta_1) (\vx_{i,-1}^T\vgamma)^2\big| \\	
		&+ \frac{1}{2}\max_{ 1\leq i \leq n}\sup \limits_{ \vbeta\in\bbB_1}\big|\E\big[f_0''(\vx_i^T\vbeta_2) (\vx_{i,-1}^T\vgamma)^2|\vx_i^T\vbeta\big]\big| \\	
		\leq &C\max_{ 1\leq i \leq n}\sup \limits_{ \vbeta\in\bbB_1}\big|[\vx_{i,-1}-\E(\vx_{i,-1}|\vx_i^T\vbeta)]^T\vgamma\big| + C\max_{ 1\leq i \leq n}\sup \limits_{ \vbeta\in\bbB_1}  (\vx_{i,-1}^T\vgamma)^2  \\	
		&+C\max_{ 1\leq i \leq n}\sup\limits_{ \vbeta\in\bbB_1}\big|\E\big[(\vx_{i,-1}^T\vgamma)^2|\vx_i^T\vbeta\big]\big| \\	
		\leq &C\max_{ 1\leq i \leq n}\big|\big|\vx_i-\E(\vx_i|\vx_i^T\vbeta)\big|\big|_\infty \sup \limits_{ \vbeta\in\bbB_1}||\vgamma||_1+ C\max_{ 1\leq i \leq n}||\vx_i||_\infty^2  \sup \limits_{ \vbeta\in\bbB_1}||\vgamma||_1^2  +C||\vgamma||_2^2\log(p\vee n),%\max_{ 1\leq i \leq n}\big|\big|\E(\vx_i\vx_i^T |\vx_i^T\vbeta)\big|\big|_\infty \sup \limits_{ \vbeta\in\bbB_1}||\vgamma||_1^2,
		\end{align*}
		where $\vgamma=\vbeta_{-1}-\vbeta_{0,-1}$, $\vbeta_1$ and $\vbeta_2$ are between $\vbeta$ and $\vbeta_0$.  The last part of the last inequality comes from Assumption~\ref{A2}-(b). 
		
		For a sub-Gaussian random vector $\vx\in\bbR^p$, its property implies that $$P\left(\max_{ 1\leq i \leq n} ||\vx_i||_\infty\geq c_1 \sqrt{\log (p\vee n)}\right)\leq \exp\left[-c_2\log (p\vee n)\right],$$ for positive constants $c_1$ and $c_2$. Similar bounds also work for $\max_{ 1\leq i \leq n}\big|\big|\vx_i-\E(\vx_i|\vx_i^T\vbeta)\big|\big|_\infty $.% and $\max_{ 1\leq i \leq n}\big|\big|\E(\vx_i\vx_i^T |\vx_i^T\vbeta)\big|\big|_\infty$.
		The assumptions of Theorem~\ref{Lasso_error} imply that $||\vgamma||_2^2\log(p\vee n)\leq d_0sh^4\log(p\vee n) =  d_0sh^2\sqrt{\log(p\vee n)} *\sqrt{h^4\log(p\vee n) }\leq  d_1sh^2\sqrt{\log(p\vee n)} \sqrt{nh^9}\leq  d_1sh^2\sqrt{\log(p\vee n)}$, for some positive constants $d_0$, $d_1$.
		Hence we have 
		$$\max_{ 1\leq i \leq n}\sup \limits_{\vbeta\in\bbB_1}\big|G(\vx_i^T\vbeta|\vbeta) -G(\vx_i^T\vbeta_0|\vbeta_0)  \big| \leq Csh^2\sqrt{ \log  (p\vee n)},$$ for positive constant $C$, with probability at least $1- \exp[-c_2\log (p\vee n)]$. 	Hence we can apply Lemma~\ref{Gfunc} and derive that 
		$$P\left(\max_{1\leq i \leq n}\sup \limits_{\vbeta\in\bbB_1}\big|\widehat{G}(\vx_i^T\vbeta|\vbeta)-G(\vx_i^T\vbeta_0|\vbeta_0)  \big| \geq c_0sh^2\sqrt{\log  (p\vee n)} \right)\leq\exp[-c_1\log  (p\vee n)].$$ 
		Hence we can conclude (\ref{Gbetabound}). 
		
		Similarly,  it's sufficient to bound $\max_{ 1\leq i \leq n}\sup \limits_{ \vbeta\in\bbB_1}\big|G^{(1)}(\vx_i^T\vbeta|\vbeta)-G^{(1)}(\vx_i^T\vbeta_0|\vbeta_0)  \big|$ to prove (\ref{G1betabound}). Lemma~\ref{lem:Gbound} implies that there exist $\vbeta_1$ and $\vbeta_2$ between $\vbeta$ and $\vbeta_0$, such that
		\begin{align*}
		&\max_{ 1\leq i \leq n}\sup \limits_{ \vbeta\in\bbB_1}\big|G^{(1)}(\vx_i^T\vbeta|\vbeta)-G^{(1)}(\vx_i^T\vbeta_0|\vbeta_0)  \big| \\
		\leq&\max_{ 1\leq i \leq n}\sup \limits_{ \vbeta\in\bbB_1}\big|f_0''(\vx_i^T\vbeta)[\vx_{i,-1}-\E(\vx_{i,-1}|\vx_i^T\vbeta)]^T\vgamma\big| + \frac{1}{2}\max_{ 1\leq i \leq n}\sup \limits_{ \vbeta\in\bbB_1}\big|f_0'''(\vx_i^T\vbeta_1) (\vx_{i,-1}^T\vgamma)^2\big| \\	
		&+  \max_{ 1\leq i \leq n}\sup_{\vbeta\in\bbB_1} \left|f_0'(\vx_i^T\vbeta)\E^{(1)}(\vx_{i,-1}^T\vgamma|\vx_i^T\vbeta)\right|+ \frac{1}{2}\max_{ 1\leq i \leq n}\sup \limits_{ \vbeta\in\bbB_1}\big|\E^{(1)}\big[f_0''(\vx_i^T\vbeta_2) (\vx_{i,-1}^T\vgamma)^2|\vx_i^T\vbeta\big]\big| \\	
		\leq &C\max_{ 1\leq i \leq n}\sup \limits_{ \vbeta\in\bbB_1}\big|[\vx_{i,-1}-\E(\vx_{i,-1}|\vx_i^T\vbeta)]^T\vgamma\big| + C\max_{ 1\leq i \leq n}\sup \limits_{ \vbeta\in\bbB_1}  (\vx_{i,-1}^T\vgamma)^2  \\	
		&+C\max_{ 1\leq i \leq n}\sup_{\vbeta\in\bbB_1} \left| \E^{(1)}(\vx_{i,-1}^T\vgamma|\vx_i^T\vbeta)\right|+ C\max_{ 1\leq i \leq n}\sup\limits_{ \vbeta\in\bbB_1}\big|\E^{(1)}\big[(\vx_{i,-1}^T\vgamma)^2|\vx_i^T\vbeta\big]\big| \\	
		\leq &C\max_{ 1\leq i \leq n}\big|\big|\vx_i-\E(\vx_i|\vx_i^T\vbeta)\big|\big|_\infty \sup \limits_{ \vbeta\in\bbB_1}||\vgamma||_1+ C\max_{ 1\leq i \leq n}||\vx_i||_\infty  \sup \limits_{ \vbeta\in\bbB_1}||\vgamma||_1^2  +C||\vgamma||_2+C||\vgamma||_2^2 \sqrt{\log(p\vee n)}\\
		\leq &c_0sh^2 \sqrt{\log(p\vee n)},
		\end{align*}
		with probability at least $1-\exp[-c_1\log (p\vee n)]$, for some positive constants $C$,  $c_0$, $c_1$, and all $n$ sufficiently large. The last two parts of the second last inequality come from Assumption~\ref{A2}-(b). Note that $sh^2\sqrt{\log(p\vee n)} \leq d_0nh^7\leq d_0h$ for some positive constant $d_0$. We thus have
		$$P\left(\max_{ 1\leq i \leq n}\sup \limits_{ \vbeta\in\bbB_1}\big|G^{(1)}(\vx_i^T\vbeta|\vbeta)-G^{(1)}(\vx_i^T\vbeta_0|\vbeta_0)  \big|\geq c_0h\right)\leq\exp[-c_1\log  (p\vee n)],$$ 
		for some positive constants $c_0$, $c_1$, and all $n$ sufficiently large. Combining this result with Lemma~\ref{G1func}, we can conclude (\ref{G1betabound}). 
	\end{proof}

	\begin{proof}[Proof of Lemma~\ref{Efunc}]
		We will prove (\ref{Eboundphi}) below. The proof of  (\ref{Ebound}) is similar.    
		\begin{align*}
		&\sup \limits_{\substack{\vbeta \in \bbB_1\\\vv\in\bbK(p,2ks+\widetilde{s})}}\Big|\frac{1}{n}\sum_{i=1}^n\vv^T[\widehat{\E} (\vx_i|\vx_i^T\vbeta)-\E(\vx_i|\vx_i^T\vbeta_0)] [\widehat{\E} (\vx_i|\vx_i^T\vbeta)-\E(\vx_i|\vx_i^T\vbeta_0)]^T\vv \Big|   \\
		\leq& \sup \limits_{\substack{\vbeta \in \bbB_1\\\vv\in\bbK(p,2ks+\widetilde{s})}}\Big|\frac{2}{n}\sum_{i=1}^n\vv^T[\widehat{\E} (\vx_i|\vx_i^T\vbeta)-\E(\vx_i|\vx_i^T\vbeta_0)] [\widehat{\E} (\vx_i|\vx_i^T\vbeta)-\E(\vx_i|\vx_i^T\vbeta_0)]^T\vv \Big|   \\
		&+\sup \limits_{\substack{\vbeta \in \bbB_1\\\vv\in\bbK(p,2ks+\widetilde{s})}}\Big|\frac{2}{n}\sum_{i=1}^n\vv^T[\E (\vx_i|\vx_i^T\vbeta)-\E(\vx_i|\vx_i^T\vbeta_0)] [\E (\vx_i|\vx_i^T\vbeta)-\E(\vx_i|\vx_i^T\vbeta_0)]^T\vv \Big|  .
		\end{align*} 
		Lemma~\ref{lem:Ebound} implies that $\sup \limits_{\substack{\vbeta \in \bbB\\\vv\in\bbK(p,2ks+\widetilde{s})}}\frac{1}{n}\sum_{i=1}^n\big| [\widehat{\E} (\vx_i|\vx_i^T\vbeta)-\E(\vx_i|\vx_i^T\vbeta )]^T\vv \big|^2\leq c_0h^4$ holds with probability at least $1-\exp[-c_1\log(p\vee n)]$, for some positive constants $c_0$, $c_1$, and all $n$ sufficiently large. 
		
		Assumption~\ref{A2} implies that there exists positive constant $C$ such that
		\begin{align*}
		&\sup \limits_{\substack{\vbeta \in \bbB_1\\\vv\in\bbK(p,2ks+\widetilde{s})}}\frac{1}{n}\sum_{i=1}^n\big| [\E (\vx_i|\vx_i^T\vbeta)-\E(\vx_i|\vx_i^T\vbeta_0)]^T\vv \big|^2 \\
		\leq &C  \sup \limits_{\vbeta\in\bbB_1}\frac{1}{n}\sum_{i=1}^n\big| \vx_i^T\vbeta -\vx_i^T\vbeta_0|^2 + C \sup \limits_{\vbeta\in\bbB_1}\frac{1}{n}\sum_{i=1}^n( | \vx_i^T\vbeta|^2+ | \vx_i^T\vbeta_0|^2)||\vbeta-\vbeta_{0}||_2^2.
		\end{align*}
		By Lemma~\ref{lem15NCL}, we have that  $\sup \limits_{\vbeta\in\bbB_1}\frac{1}{n}\sum_{i=1}^n\big| \vx_i^T\vbeta|^2*||\vbeta-\vbeta_{0}||_2^2 \leq c_0sh^4$ with probability at least $1-\exp[-c_1\log(p\vee n)]$, for some positive constants $c_0$, $c_1$, and all $n$ sufficiently large. Similarly, we can derive that $\sup \limits_{\vbeta\in\bbB_1}\frac{1}{n}\sum_{i=1}^n\big| \vx_i^T\vbeta -\vx_i^T\vbeta_0|^2 \leq c_1sh^4$with probability at least $1-\exp[-c_1\log(p\vee n)]$, for some positive constants $c_0$, $c_1$, and all $n$ sufficiently large. 
		Combining these two results,  we can derive that 
		$$P\Big(\sup \limits_{\substack{\vbeta \in \bbB\\\vv\in\bbK(p,2ks+\widetilde{s})}}\frac{1}{n}\sum_{i=1}^n\big|[\widehat{\E} (\vx_i|\vx_i^T\vbeta)-\E(\vx_i|\vx_i^T\vbeta_0)]^T\vv \big|^2
		\geq c_0 sh^4 \Big)\leq  \exp(-c_1\log p),$$ 
		for some positive constants $c_0$, $c_1$, and all $n$ sufficiently large. 
		Hence we can conclude (\ref{Eboundphi}). 
	\end{proof}

	\begin{proof} [Proof of Lemma~\ref{lem:thetaj}]
		First, we define $\vphi_{0j}^*\triangleq  \Big(-(\vd_{0j})_{1:(j-2)}^T,1,-(\vd_{0j})_{(j-1):(p-2)}^T\Big)^T$. By the definitions of $\vd_{0j}$ and $ \tau^2_{0j}$, we know that
		\begin{align*}
		\vOmega_{-(j-1)}\vphi_{0j}^* &=\vOmega_{-(j-1),(j-1)}-\vOmega_{-(j-1),-(j-1)} \vd_{0j}=  \vnull_{p-2},\\
		\vOmega_{j-1}^T\vphi_{0j}^* &=\Omega_{(j-1),(j-1)}-\vOmega_{(j-1),-(j-1)} \vd_{0j}=   \tau^2_{0j},
		\end{align*} 
		where $	\vOmega_{j-1}\in\bbR^{p-1}$ is the $(j-1)^{th}$ column of $\vOmega$, and $\vOmega_{-(j-1)}\in\bbR^{(p-2)\times (p-1)}$ is the submatrix of $\vOmega$ with its $(j-1)^{th}$ row removed. Given these two facts, we can derive that 
		\begin{align*}
		\vOmega\vphi_{0j}^* = \tau^2_{0j}\ve_{j-1} = \vOmega\vtheta_j \tau^2_{0j},
		\end{align*}
		since $\vOmega\vTheta=\vI_{p-1}$, where $\ve_{j-1}$ is the $(j-1)^{th}$ column of $\vI_{p-1}$. Assumption~\ref{A2}-(a) indicates that $\lam_{min}(\vOmega)\geq \xi_2>0$, then we have $\vphi_{0j}^* = \Big(-(\vd_{0j})_{1:(j-2)}^T,1,-(\vd_{0j})_{(j-1):(p-2)}^T\Big)^T =\vtheta_j \tau^2_{0j}=\vphi_{0j} $. We thus have that $||\vtheta_j ||_0=||\vd_{0j}||+1\leq \widetilde{s}+1$.
		
		To prove the second part of the lemma, note that $\E(\vxw_{-1}\vxw_{-1}^T) = \E[\Cov(\vx_{-1}|\vx^T\vbeta_0)]$.
		Assumption \ref{A2} indicates that $\sup_{||\vv||_2=1}\vv^T\E(\vxw_{-1}\vxw_{-1}^T) \vv= \xi_1$. We can derive that $$\sup_{||\vv||_2=1}\vv^T\vOmega\vv = \sup_{||\vv||_2=1}\E \big\{[G^{(1)}(\vx^T\vbeta_0|\vbeta_0)(\vxw_{-1}^T\vv)]^2\big\}\leq b^2\sup_{||\vv||_2=1}\vv^T\E(\vxw_{-1}\vxw_{-1}^T)\vv=b^2\xi_1.$$ 
		Recall that $\tau_{0j}^2 = \Omega_{(j-1),(j-1)}-(\vOmega_{-(j-1),(j-1)})^T(\vOmega_{-(j-1),-(j-1)})^{-1}\vOmega_{-(j-1),(j-1)}$. Since $\vOmega$ is positive definite, we thus have that $(\vOmega_{-(j-1),-(j-1)})^{-1}$ is positive definite, and  $\tau_{0j}^2\leq \Omega_{(j-1),(j-1)}= \ve_{j-1}^T\vOmega\ve_{j-1} \leq b^2\xi_1$, uniformly in $j=2,\cdots,p$.
		
		%By the definition of $\kV_{2j}$, $j=2,\cdots,p$, in Assumption~\ref{A2}-(a), we have that $\vphi_{0j}\neq \ve_j$. Hence $\vOmega_{-j,j}\neq \vnull_{p-1}$, and $(\vOmega_{-j,j})^T(\vOmega_{-j,-j})^{-1}\vOmega_{-j,j}>0$, since $(\vOmega_{-j,-j})^{-1}$ is positive definite. We thus have  $\tau_{0j}^2  <\Omega_{j,j}= \ve_j^T\vOmega\ve_j \leq b^2\xi_1$.
		
		Since $\vTheta = \vOmega^{-1}$, we know that $\vtheta_j^T\vOmega \vtheta_j = \Theta_{(j-1),(j-1)}$. We observe that $ \tau^{-2}_{0j}=\Theta_{(j-1),(j-1)}\leq ||\vtheta_j||_2 $. % Note that  if $\vphi_{0j}=\ve_{j-1}$, we have that $||\vtheta_j||_2 =  \tau^{-2}_{0j}=\Omega_{(j-1),(j-1)}^{-1}\leq \xi_2^{-1}$ by Assumption~\ref{A2}-(a).  If $\vphi_{0j} \neq \ve_{j-1}$, we denote $\vv_j\triangleq\frac{\vphi_{0j} - \ve_{j-1}}{||\vphi_{0j} - \ve_{j-1}||_2}$. Note that $\vv_j\in\kV_{2j}$, hence $  \vv_j^T\vOmega \vv_j \geq \xi_2 $. 
		Note that $\tau^{-2}_{0j}=\Omega_{(j-1),(j-1)}^{-1}\leq \xi_2^{-1}$ by Assumption~\ref{A2}-(a). We observe that
		$	||\vtheta_j||_2\geq \vtheta_j^T\vOmega \vtheta_j =\tau_{0j}^{-4} (\vphi_{0j}^T\vOmega \vphi_{0j} ) \geq  \xi_2\tau_{0j}^{-4}|| \vphi_{0j}||_2^2 =\xi_2	||\vtheta_j||_2^2,$
		where the second inequality applies Assumption~\ref{A2}-(a).
		It implies that   $\tau^{-2}_{0j}\leq ||\vtheta_j||_2 \leq \xi_2^{-1}$ uniformly in $j$,  which completes the proof of the lemma.  
	\end{proof}

	\begin{proof}[Proof of Lemma~\ref{cor_Sig}]
		It suffices to show that 
		$$ \max_{2\leq j,k \leq p}\bigg|  \frac{1}{n}\sum_{i=1}^n [\widetilde{Y}_i-\widehat{G}(\vx_i^T\vbetah|\vbetah)]^2 [\widehat{G}^{(1)}(\vx_i^T\vbetah|\vbetah)]^2 \vthetah_{j} ^T\vxh_{i,-1}\vxh_{i,-1}^T\vthetah_{k}   - \vtheta_j^T\vLam\vtheta_k\bigg|=o_p(1).$$
		Rewrite that
		\begin{align*}
		& \max_{2\leq j,k \leq p}\bigg|  \frac{1}{n}\sum_{i=1}^n [\widetilde{Y}_i-\widehat{G}(\vx_i^T\vbetah|\vbetah)]^2 [\widehat{G}^{(1)}(\vx_i^T\vbetah|\vbetah)]^2 \vthetah_{j} ^T\vxh_{i,-1}\vxh_{i,-1}^T\vthetah_{k}   - \vtheta_j^T\vLam\vtheta_k\bigg|\\
		\leq& \max_{2\leq j,k \leq p}\bigg|\frac{1}{n}\sum_{i=1}^n  \widetilde{\epsilon}_i^2 [\widehat{G}^{(1)}(\vx_i^T\vbetah|\vbetah)]^2 \vthetah_{j} ^T(\vxh_{i,-1}\vxh_{i,-1}^T-\vxw_{i,-1}\vxw_{i,-1}^T)\vthetah_{k} \bigg| \\ 
		& + \max_{2\leq j,k \leq p}\bigg|\frac{1}{n}\sum_{i=1}^n  \widetilde{\epsilon}_i^2\big\{[\widehat{G}^{(1)}(\vx_i^T\vbetah|\vbetah)]^2-[G^{(1)}(\vx_i^T\vbeta_0|\vbeta_0)]^2 \big\}  \vthetah_{j} ^T\vxw_{i,-1}\vxw_{i,-1}^T\vthetah_{k} \bigg| \\
		&+\max_{2\leq j,k \leq p}\bigg|\frac{1}{n}\sum_{i=1}^n  \widetilde{\epsilon}_i^2[G^{(1)}(\vx_i^T\vbeta_0|\vbeta_0)]^2  \vthetah_{j} ^T\vxw_{i,-1}\vxw_{i,-1}^T\vthetah_{k}  -\vtheta_j^T\vLam\vtheta_k\bigg|\\
		&+\max_{2\leq j,k \leq p}\bigg|\frac{1}{n}\sum_{i=1}^n [G(\vx_i^T\vbeta_0|\vbeta_0)-\widehat{G}(\vx_i^T\vbetah|\vbetah)]^2 [\widehat{G}^{(1)}(\vx_i^T\vbetah|\vbetah)]^2   \vthetah_{j} ^T\vxh_{i,-1}\vxh_{i,-1}^T\vthetah_{k} \bigg|\\
		&+\max_{2\leq j,k \leq p}\bigg|\frac{2}{n}\sum_{i=1}^n \widetilde{\epsilon}_i[G(\vx_i^T\vbeta_0|\vbeta_0)-\widehat{G}(\vx_i^T\vbetah|\vbetah)]  [\widehat{G}^{(1)}(\vx_i^T\vbetah|\vbetah)]^2  \vthetah_{j} ^T\vxh_{i,-1}\vxh_{i,-1}^T\vthetah_{k} \bigg|\\
		\triangleq&\sum_{l=1}^5\max_{2\leq j,k \leq p}|J_{njkl}|,
		\end{align*} 
		where $J_{njkl}$'s are defined in the context. Note that 
		\begin{align*}
		\max_{2\leq j,k \leq p}|J_{njk1}|  \leq& \max_{2\leq j,k \leq p}\bigg|\frac{1}{n}\sum_{i=1}^n   \widetilde{\epsilon}_i^2 [\widehat{G}^{(1)}(\vx_i^T\vbetah|\vbetah)]^2 \vtheta_{j} ^T(\vxh_{i,-1}\vxh_{i,-1}^T-\vxw_{i,-1}\vxw_{i,-1}^T)\vtheta_{k}\bigg| \\  
		&+ \max_{2\leq j,k \leq p}\bigg|\frac{1}{n}\sum_{i=1}^n  \widetilde{\epsilon}_i^2 [\widehat{G}^{(1)}(\vx_i^T\vbetah|\vbetah)]^2 (\vthetah_{j}-\vtheta_{j}) ^T(\vxh_{i,-1}\vxh_{i,-1}^T-\vxw_{i,-1}\vxw_{i,-1}^T)\vtheta_{k} \bigg| \\  
		&+ \max_{2\leq j,k \leq p}\bigg|\frac{1}{n}\sum_{i=1}^n  \widetilde{\epsilon}_i^2 [\widehat{G}^{(1)}(\vx_i^T\vbetah|\vbetah)]^2 \vtheta_{j} ^T(\vxh_{i,-1}\vxh_{i,-1}^T-\vxw_{i,-1}\vxw_{i,-1}^T)(\vthetah_{k}-\vtheta_{k}) \bigg| \\  
		&+ \max_{2\leq j,k \leq p}\bigg|\frac{1}{n}\sum_{i=1}^n  \widetilde{\epsilon}_i^2 [\widehat{G}^{(1)}(\vx_i^T\vbetah|\vbetah)]^2 (\vthetah_{j}-\vtheta_{j}) ^T(\vxh_{i,-1}\vxh_{i,-1}^T-\vxw_{i,-1}\vxw_{i,-1}^T)(\vthetah_{k}-\vtheta_{k}) \bigg| \\  
		\triangleq&\max_{2\leq j,k \leq p}|J_{njk11}|+\max_{2\leq j,k \leq p}|J_{njk12}|+\max_{2\leq j,k \leq p}|J_{njk13}|+\max_{2\leq j,k \leq p}|J_{njk14}|.
		\end{align*} 
		
		Consider the event $\mathcal{F}_n=\{\max_{1\leq i \leq n} |\widetilde{\epsilon}_i|\leq  (\sigma_\ep+M) \sqrt{\log (p\vee n)}\} $, which holds with probability at least $1-\exp[-c \log (p\vee n)]$ by the sub-Gaussian property for $\widetilde{\epsilon}$, , for some constants $c>0$, and all $n$ sufficiently large.  Then by the proofs of Lemma~\ref{lem:Ex_err}, we can derive that 
		\begin{align*}
		\max_{2\leq j,k \leq p}|J_{njk11}|&\leq   O_p(\log (p\vee n)) *O_p(\sqrt{s}h^2)= O_p(\sqrt{s}h^2\log (p\vee n)) = o_p(nh^7)=o_p(1),
		\end{align*}
		with probability at least  $1-\exp(-c_1\log p)$,  for some constant $c_1>0$, and all $n$ sufficiently large. 
		Similarly, we have that $\max_{2\leq j,k \leq p}|J_{njk12}|\leq O_p(\sqrt{s}h^2\log (p\vee n)) * \max_{2\leq j \leq p}||\vthetah_{j}-\vtheta_{j}||_1= o_p(nh^7)*O_p(\widetilde{s}\eta) =o_p(1) $, and $\max_{2\leq j,k \leq p}|J_{njk13}|=o_p(1) $, with the same probability bound. Finally, we can derive that $\max_{2\leq j,k \leq p}|J_{njk14}|\leq O_p(\sqrt{s}h^2\log (p\vee n)) * \max_{2\leq j \leq p}||\vthetah_{j}-\vtheta_{j}||_1^2= o_p(nh^7)*O_p(\widetilde{s}^2\eta^2) =o_p(1) $. 
		
		To bound $\max_{2\leq j,k \leq p}|J_{njk2}|$, Lemma~\ref{Gbetafunc} and Lemma~\ref{lem:cube_rate} together imply that
		\begin{align*}
		&\max_{2\leq j,k \leq p}|J_{njk2}|\\
		\leq&  \max_{1\leq i \leq n}\big|[\widehat{G}^{(1)}(\vx_i^T\vbetah|\vbetah)]^2-[G^{(1)}(\vx_i^T\vbeta_0|\vbeta_0)]^2 \big| *\Big[\frac{1}{n} \sum_{i=1}^n\widetilde{\epsilon}_i^6\Big]^{1/3} * \max_{2\leq j,k \leq p}\Big[\frac{1}{n}\sum_{i=1}^n| \vthetah_{j} ^T\vxw_{i,-1}\vxw_{i,-1}^T\vthetah_{k}  |^{3/2}\Big]^{2/3}\\
		\leq &O_p(h) * \Big[\frac{1}{n} \sum_{i=1}^n\widetilde{\epsilon}_i^6\Big]^{1/3}*\max_{2\leq j \leq p} \Big[\frac{1}{n}\sum_{i=1}^n| \vxw_{i,-1}^T\vthetah_{j} |^3\Big]^{2/3}\\
		\leq & O_p(h) *O_p(1) = O_p(h)=o_p(1),
		\end{align*}
		with probability at least $1-\exp[-c_1\log (p\wedge n)]$, for some positive constant $c_1$, and all $n$ sufficiently large. We can also derive that 
		\begin{align*}
		\max_{2\leq j,k \leq p}|J_{njk4}|&\leq  \max_{1\leq i \leq n} [\widehat{G}(\vx_i^T\vbetah|\vbetah)]- G(\vx_i^T\vbeta_0|\vbeta_0)]^2  *\max_{1\leq i \leq n}[\widehat{G}^{(1)}(\vx_i^T\vbetah|\vbetah)]^2*\max_{2\leq j \leq n}\frac{1}{n}\sum_{i=1}^n(\vxh_{i,-1}^T\vtheta_j )^2\\
		&\leq O_p(s^2h^4\log (p\vee n)) * O_p(1)*O_p(1) =O_p(s^2h^4\log (p\vee n))=o_p(1),
		\end{align*}
		with probability at least $1-\exp(-c_1\log p)$for some positive constant $c_1$, and all $n$ sufficiently large, by Lemma~\ref{lem:thetax}, Lemma~\ref{Gfunc} and Lemma~\ref{Gbetafunc}. 
		Conditional on the event $\mathcal{F}_n$, we can see $\max_{2\leq j,k \leq p}|J_{njk5}|\leq O_p (sh^2\log(p\vee n))= o_p(1)$. To bound $\max_{2\leq j,k \leq p}|J_{njk4}|$, let $\widehat{\vLam} = \frac{1}{n}\sum_{i=1}^n  \widetilde{\epsilon}_i^2 [G^{(1)}(\vx_i^T\vbeta_0|\vbeta_0)]^2  \vxw_{i,-1}\vxw_{i,-1}^T$, then we can rewrite it as  
		\begin{align*}
		\max_{2\leq j,k \leq p}|J_{njk4}|\leq& \max_{2\leq j,k \leq p} \bigg| \vtheta_j^T \{\widehat{\vLam} -\vLam\}\vtheta_k \bigg| +\max_{2\leq j,k \leq p}\big| (\vthetah_j- \vtheta_j)^T\widehat{\vLam}\vtheta_k\big|+\max_{2\leq j,k \leq p}\big| \vthetah_j^T \widehat{\vLam}(\vthetah_k-\vtheta_k) \big|\\
		=&\max_{2\leq j,k \leq p}| J_{njk41}|+\max_{2\leq j,k \leq p}|J_{njk42}|+\max_{2\leq j,k \leq p}|J_{njk43}|,
		\end{align*}
		where the definition of $J_{njk4l}$ is clear. Given $(\epsilon_i,\vx_i)$, for any $\vtheta_j$, $ \widetilde{\epsilon}_iG^{(1)}(\vx_i^T\vbeta_0|\vbeta_0)   \vxw_{i,-1}^T\vtheta_j$ is sub-Gaussian with variance proxy at most $b^2 [\epsilon_i+g(\vx_i)]^2 (\vxw_{i,-1}^T\vtheta_j)^2$. Hence, we can conclude that $ \frac{1}{n}\sum_{i=1}^n \widetilde{\epsilon}_i G^{(1)}(\vx_i^T\vbeta_0|\vbeta_0)   \vxw_{i,-1}^T\vtheta_j$ is sub-Gaussian with variance proxy at most $ \frac{b^2}{n}\sum_{i=1}^n [\epsilon_i+g(\vx_i)]^2 (\vxw_{i,-1}^T\vtheta_j)^2\rap c$  for constant $c>0$. 
		Lemma~\ref{lem14NCL} implies that 
		\begin{align*}
		P\left(\max_{2\leq j,k \leq p}|J_{njk41}|\geq c \sqrt{\frac{\log p}{n}} \right)\leq \sum_{j,k}P\left( |J_{njk41}|\geq c \sqrt{\frac{\log p}{n}} \right) \leq p^2\exp(-c_1\log p),
		\end{align*}
		for some positive constant $c_1$, and all $n$ sufficiently large, 	where $ \sqrt{\frac{\log p}{n}} = o(h^{5/2})$. For $J_{njk42}$, we can conclude that
		\begin{align*}
		\max_{2\leq j,k \leq p}|J_{njk42}|&\leq b^2 \Big[\frac{1}{n} \sum_{i=1}^n\widetilde{\epsilon}_i^6\Big]^{1/3} *\max_{1\leq k \leq p} \Big[\frac{1}{n}\sum_{i=1}^n| \vxw_{i,-1}^T\vthetah_{k} |^3\Big]^{1/3}*\max_{2\leq j \leq p} \Big[\frac{1}{n}\sum_{i=1}^n| \vxw_{i,-1}^T(\vthetah_j- \vtheta_j) |^3\Big]^{1/3}\\
		&\leq O_p(\widetilde{s}^{1/2}h) =o_p(1).
		\end{align*}
		Similar proof works to bound $\max_{2\leq j,k \leq p}|J_{njk43}|$. The lemma is proved.
	\end{proof}

	\section{Auxiliary results} \label{sec:proof_auxil}

	\setcounter{lemma}{0}
	\renewcommand{\thelemma}{B\arabic{lemma}} 
	
	\blem[Lemma 14 in \citet{NCL}]
	\label{lem14NCL} 
	Let $p_1$, $p_2$ be two arbitrary positive integers. If $\{\vx_i\in\bbR^{p_1}:i=1,\cdots,n\}$ are independent zero-mean sub-Gaussian random vectors with variance proxy $\sigma_x^2$, then for any fixed unit vector $\vv\in\bbR^{p_1}$, $\forall\ t>0$,
	\begin{align}
	P\left(\Big|\frac{1}{n}\sum_{i=1}^n\big[(\vx_i^T\vv)^2-\emph{E}(\vx_{i}^T\vv)^2\big]\Big|\geq t\right)\leq 2\exp\left[-cn\min\left(\frac{t^2}{\sigma_x^4},\frac{t}{\sigma_x^2}\right)\right],\label{lem14-1}
	\end{align}
	for some universal constant $c>0$.
	Moreover, if $\{\vy_i\in\bbR^{p_2}:i=1,\cdots,n\}$ are independent zero-mean sub-Gaussian random vectors with variance proxy $\sigma_y^2$, then $\forall\ t>0$,
	\begin{align}
	P\left(\Big|\Big|\frac{1}{n}\sum_{i=1}^n\big[\vx_i\vy_i^T -\emph{E}(\vx_i\vy_i^T) \big]\Big|\Big|_\infty\geq t \right)\leq 6p_1p_2\exp\left[-cn\min\left(\frac{t^2}{\sigma_x^2\sigma_y^2},\frac{t}{\sigma_x\sigma_y}\right)\right].\label{lem14-2}
	\end{align}
	Let if $p=p_1\vee p_2$. If $\log p = O(n)$, then there exist universal positive constants $c_0$, $c_1$ and $c_2$ such that
	\begin{align}
	P\left(\Big|\Big|\frac{1}{n}\sum_{i=1}^n\big[\vx_i\vy_i^T -\emph{E}(\vx_i\vy_i^T) \big]\Big|\Big|_\infty\geq c_0\sigma_x\sigma_y\sqrt{\frac{\log p}{n}}\right)\leq c_1\exp(-c_2\log p) .\label{lem14-3}
	\end{align}
	\elem
	
	\blem[Lemma 15 in \citet{NCL}]
	\label{lem15NCL} 
	Let $\bbK(s_0) = \{\vv\in\bbR^p:||\vv||_2\leq 1,||\vv||_0\leq s_0\}$.  If $\{\vx_i\in\bbR^{p}:i=1,\cdots,n\}$ are independent zero-mean sub-Gaussian random vectors with variance proxy $\sigma_x^2$, then there is a universal constant $c > 0$ such that for any $s_0\geq 1$,
	\begin{align*}
	P\Big(\sup_{\vv\in\mathbb{K}(2s_0)}\Big|\frac{1}{n}\sum_{i=1}^n\big[(\vx_i^T\vv)^2-\emph{E}(\vx_{i}^T\vv)^2\big]\Big|\geq t\Big)\leq 2\exp\left[-cn\min\left(\frac{t^2}{\sigma_x^4},\frac{t}{\sigma_x^2}\right)+2s_0\log p\right].%\label{lem15}
	\end{align*}
	\elem

	\blem \label{lem:cube_rate} 
	Let $\vx_1,\cdots,\vx_n$ be independent sub-Gaussian random vectors with variance proxy $\sigma_x^2$. For any $s_0\geq 1$, there exists a universal constant $c>0$ such that for all $n$ sufficiently large.
	\begin{align*}
	P\Big(\sup_{\vgamma\in\mathbb{K}(p,2s_0)}\Big|\frac{1}{n}\sum_{i=1}^n\big[|\vx_i^T\vgamma|^3-\emph{E}|\vx_{i}^T\vgamma|^3\big]\Big|\geq t\Big)\leq \exp\left\{-c\min\left[\frac{nt^2}{\sigma_x^6}, \frac{(nt)^{2/3}}{\sigma_x^2}\right]+2s_0\log p \right\};\\% \label{cube_rate}
	P\Big(\sup_{\vgamma\in\mathbb{K}(p,2s_0)}\Big|\frac{1}{n}\sum_{i=1}^n\big[|\vx_i^T\vgamma|^4-\emph{E}|\vx_{i}^T\vgamma|^4\big]\Big|\geq t\Big)\leq \exp\left\{-c\min\left[\frac{nt^2}{\sigma_x^8}, \frac{(nt)^{1/2}}{\sigma_x^2}\right]+2s_0\log p \right\}.
	\end{align*}
	\elem

	\begin{proof} 
		For any fixed $\vgamma\in\bbR^p$ such that $||\vgamma||_2\leq 1$, $\vx_i^T\vgamma$ is also sub-Gaussian with variance proxy bounded by $\sigma^2$. Applying the result on concentration inequality for the polynomial functions of independent sub-Gaussian random variables,  Theorem 1.4 of \citet{Adamczak2015} and the example in their section 3.1.2, we have $\forall\ t>0$, there exist universal positive constants $c_1$ and $c_2$ such that
		\begin{align}
		P\Big(\frac{1}{n}\sum_{i=1}^n\big(  |\vx_i^T\vgamma|^3-\E |\vx_i^T\vgamma|^3\big) \geq t\Big)\leq  c_1\exp\left\{-c_2\min\left[\frac{nt^2}{\sigma_x^6}, \frac{(nt)^{2/3}}{\sigma_x^2}\right]\right\} ,\label{cube_rate1}\\
		P\Big(\frac{1}{n}\sum_{i=1}^n\big(  |\vx_i^T\vgamma|^4-\E |\vx_i^T\vgamma|^4\big) \geq t\Big)\leq  c_1\exp\left\{-c_2\min\left[\frac{nt^2}{\sigma_x^8}, \frac{(nt)^{1/2}}{\sigma_x^2}\right]\right\} .
		\end{align} 
		
		Next, we apply the covering technique of Lemma~\ref{lem15NCL} to extend the above probability bound to uniformly on $\bbK(p,2s_0) = \{\vv\in\bbR^p:||\vv||_2\leq 1,||\vv||_0\leq 2s_0\}$. %Since the steps are similar, we only show for the first inequality.
		
		For any subset $\mathcal{U} \subseteq\{1,\cdots,p\}$, define $\kS_\kU = \{\vgamma\in \bbR^p: ||\vgamma||_2\leq 1,\mbox{supp}(\vgamma)\subseteq \kU\}$. Then $\bbK(p,2s_0) = \bigcup_{|\kU |\leq 2s_0}\kS_\kU $. 
		Let $\kA=\{u_1,\cdots,u_m\}$ be a $\frac{1}{4}-$cover of $\kS_\kU$, that is $\forall\ \vgamma\in \kS_\kU$, there exists some $\vxi\in\kA$ such that $||\vgamma-\vxi||_2\leq 1/4$.  We can construct $\kA$ such that $|\kA|\leq 16^{2s_0}$. We observe that 
		\begin{align*}
		\Big|\frac{1}{n}\sum_{i=1}^n\big(  |\vx_i^T\vgamma|^3-  |\vx_i^T\vxi|^3\big)\Big|\leq & \frac{1}{n}\sum_{i=1}^n|\vx_i^T\vgamma|^2  |\vx_i^T(\vgamma-\vxi)| + \frac{1}{n}\sum_{i=1}^n|\vx_i^T\vgamma| *|\vx_i^T\vxi| *  |\vx_i^T(\vgamma-\vxi)| \\
		&+ \frac{1}{n}\sum_{i=1}^n|\vx_i^T\vxi|^2  |\vx_i^T(\vgamma-\vxi)|.
		\end{align*}
		By H\"older inequality,
		\begin{align*}
		\frac{1}{n}\sum_{i=1}^n|\vx_i^T\vgamma|^2  |\vx_i^T(\vgamma-\vxi)| \leq \Big(\frac{1}{n}\sum_{i=1}^n|\vx_i^T\vgamma|^3 \Big)^{2/3}*\Big(\frac{1}{n}\sum_{i=1}^n|\vx_i^T(\vgamma-\vxi)|^3 \Big)^{1/3}.
		\end{align*}
		Since $4(\vgamma-\vxi)\in\kS_\kU$, we have
		$$\sup_{\vgamma\in\kS_\kU}\sup_{\vxi\in\kA} \frac{1}{n}\sum_{i=1}^n|\vx_i^T\vgamma|^2  |\vx_i^T(\vgamma-\vxi)| \leq \frac{1}{4n}\sup_{\vgamma\in\kS_\kU}\sum_{i=1}^n|\vx_i^T\vgamma|^3.$$
		Similarly analysis applies to the other two terms. Note that $\max_{\vxi\in\kA}\frac{1}{n}\sum_{i=1}^n|\vx_i^T\vxi|^3 \leq  \sup_{\vgamma\in\kS_\kU}\\\frac{1}{n}\sum_{i=1}^n|\vx_i^T\vgamma|^3$, then we have 
		\begin{align*}
		\sup_{\vgamma\in\kS_\kU}\frac{1}{n}\sum_{i=1}^n|\vx_i^T\vgamma|^3\leq &  \max_{\vxi\in\kA}\frac{1}{n}\sum_{i=1}^n|\vx_i^T\vxi|^3 + \frac{3}{4n}\sup_{\vgamma\in\kS_\kU}\sum_{i=1}^n|\vx_i^T\vgamma|^3,
		\end{align*}   
		which implies that $\sup_{\vgamma\in\kS_\kU}\frac{1}{n}\sum_{i=1}^n|\vx_i^T\vgamma|^3\leq 4  \max_{\vxi\in\kA}\frac{1}{n}\sum_{i=1}^n|\vx_i^T\vxi|^3 $. Combining this with (\ref{cube_rate1}) and applying the union bound, there exists a universal constant $c>0$ such that 
		\begin{align*}
		P\Big(\sup_{\vgamma\in\kS_\kU}\frac{1}{n}\sum_{i=1}^n\big(  |\vx_i^T\vgamma|^3-\E |\vx_i^T\vgamma|^3\big) \geq 4t\Big) &\leq   P\Big(\max_{\vxi\in\kA}\frac{1}{n}\sum_{i=1}^n  \big(|\vx_i^T\vxi|^3-\E |\vx_i^T\vgamma|^3\big) \geq t\Big)\\
		&\leq 16^{2s_0}\exp\left\{-c \min\left[\frac{nt^2}{\sigma_x^6}, \frac{(nt)^{2/3}}{\sigma_x^2}\right]\right\}.
		\end{align*} 
		Taking a union bound over the ${p\choose 2s_0}\leq p^{2s_0}$ choices of $\kU$ for $\bbK(p,2s_0)$ yields that for all $n$ sufficiently large,
		\begin{align*}
		P\Big(\sup_{\vgamma\in\bbK(p,2s_0)}\Big|\frac{1}{n}\sum_{i=1}^n\big(  |\vx_i^T\vgamma|^3-\E |\vx_i^T\vgamma|^3\big)\Big|  \geq t\Big)\leq    \exp\left\{-c\min\left[\frac{nt^2}{\sigma_x^6}, \frac{(nt)^{2/3}}{\sigma_x^2}\right]+2s_0\log p \right\}.  
		\end{align*}  
		Hence, the first claim of Lemma~\ref{lem:cube_rate} is proved. The second claim can be probed similarly. 
	\end{proof}

	\blem \label{lem:An1bound}  
	Under the assumptions of Theorem~\ref{Lasso_error}, there exist some positive constants $c_0$, $c_1$ such that for all $n$ sufficiently large,
	\begin{align*}
	P\left(\max_{  1\leq i\leq n}\sup \limits_{\vbeta \in \bbB} |A_{n1}(\vx_i^T\vbeta|\vbeta)| \geq c_0 h^2 \right)\leq  \exp(-c_1nh^5).%\label{An1bound} 
	\end{align*} 
	%where $\bbB$ is defined as in Lemma~\ref{lem:local_LRC} and Lemma~\ref{Gfunc}.
	\elem
	\begin{proof}	
		%Consider the event $\mathcal{F}_n$ defined in Lemma~\ref{lem:events}. Let $\epsilon_j' = \frac{2(\epsilon_j+g(\vx_j))}{6\sigma\log n}$, then on the event $\mathcal{F}_n$, $|\epsilon_j'|\leq 1$. 
		Write $a_{nj}(t|\vbeta) = K\Big(\frac{t-\vx_j^T\vbeta}{h}\Big)\widetilde{\epsilon}_j$.
		Then
		$	A_{n1}(\vx_i^T\vbeta|\vbeta) =  [(n-1)h]^{-1} \sum_{j=1,j\neq i}^n a_{nj}(\vx_i^T\vbeta|\vbeta) $. 
		%Note that $\forall\ \eta>0$, 	
		%\begin{align*}
		%P\Big(\sup \limits_{t \in \bbT, \vbeta \in \bbB} |A_{n1}(t|\vbeta)| \geq \eta\Big)
		%&\leq 	P\Big(\sup \limits_{t \in \bbT, \vbeta \in \bbB} |A_{n1}(t|\vbeta)| \geq \eta | \mathcal{F}_n \Big) + \exp(-c_1\log n).
		%\end{align*}
		Let  $f_{\vbeta}(\cdot)$ denote the p.d.f of $\vx^T\vbeta$. Assumption~\ref{A1}, \ref{K1} and \ref{K3} together imply that %for some constant $c$ large enough,
		\begin{align*}
		\E [a_{nj}^2(t|\vbeta) ]&\leq (\sigma_\epsilon^2+M^2)\E\Big\{ K^2\Big(\frac{t-\vx_j^T\vbeta}{h}\Big) \Big\} =(\sigma_\epsilon^2+M^2)h  \int K^2(z) f_{\vbeta}(t-hz) dz \leq ch,
		\end{align*}
		for some positive constant $c$.
		Note that $\E [a_{nj}(t|\vbeta)] =0$. Note that $\epsilon_i$ is sub-Gaussian, $K(\cdot)$ and $g(\cdot)$ are bounded almost everywhere. It is easy to see that $ \E\big[ |a_{nj}(t|\vbeta) |^k\big]\leq \frac{1}{2}\E[a_{nj}^2(t|\vbeta) ] L^{k-2}k!$, for some positive real $L$ and every integer $k\geq 2$. For any fixed $\vbeta$ and $0\leq v\leq \frac{1}{2L}\sqrt{(n-1)\E[a_{nj}^2(\vx_i^T\vbeta|\vbeta) ]}$,  by Bernstein's inequality,  
		\begin{align*}
		P\left(  \Big|\sum_{j=1,j\neq i}^n a_{nj}(\vx_i^T\vbeta|\vbeta)\Big| \geq 2v\sqrt{c(n-1)h} \right)\leq 2\exp (-v^2).
		\end{align*}
		Taking $v = \sqrt{c(n-1)h^5}$, we have 
		$$P\left(  \big|A_{n1}(\vx_i^T\vbeta|\vbeta)\big| \geq ch^2 \right)  \leq 2\exp(-cnh^5).$$
		
		%To cover $\bbT$ with $ \delta\sqrt{\log (p\vee n)} - $balls, the covering number $N_1\leq c_1\delta^{-1}$ for sufficiently large $c_1$. 
		To obtain the uniform bound, we cover $\bbB$ with  $L_2-$balls with radius $\delta$. Denote the unit Euclidean sphere as $\mathcal{S}^{ks} = \{\vv\in\bbR^{ks}: ||\vv||_2 = 1\}$. Let the covering number $N(\delta,\mathcal{S}^{ks},\rho) $ be the minimum $n$ such that there exists an $\delta-$cover of $\mathcal{S}^{ks}$ of size $n$, with respect to the $L_2$ distance $\rho$. 	It is well known that $N(\delta,\mathcal{S}^{ks},\rho) \leq (1+\frac{2}{\delta})^{ks}$. Consider the  decomposition %Together with the decomposition inspired in \citet{fan2018},
		$$\big\{\vbeta\in \mathcal{S}^{p} : ||\vbeta||_0=ks\big\} = \bigcup\limits_{\mathcal{S}\subseteq[p]: |\mathcal{S}| = ks}\big\{\vbeta\in \mathcal{S}^p : \mbox{supp}(\vbeta) =\mathcal{S} \big\},$$
		where $|\mathcal{S}|$ is the cardinal number of $\mathcal{S}$.
		Let $\mathcal{N}_{\delta}$ be an $\delta-$cover  of $\bbB = \{\vbeta\in\bbB_0:||\vbeta-\vbeta_0||_2\leq r,  ||\vbeta||_0\leq ks\}$. 	it is easy to show the covering number $N = |\mathcal{N}_{\delta}|$ satisfies $$N\leq \Bigg\{ \Bigg( \begin{array}{c} p \\ ks \end{array} \Bigg) \Big(1+\frac{2r}{\delta}\Big)^{ks}\Bigg\}^2\leq \bigg\{\Big(1+\frac{2r}{\delta}\Big)\frac{ep}{ks}\bigg\}^{2ks} \leq c_2\Big(\frac{p }{\delta }\Big)^{2ks},$$
		for sufficiently large $c_2$. % Hence to cover $\bbT\times\bbB$ with  $\delta\sqrt{\log (p\vee n)} \times \delta-$balls, the covering number $N = N_1*N_2\leq c p^{2ks}\delta^{-(2ks+1)}$.
		For any $\vbeta$ in such a ball with center $\vbeta^*$, let us take $\delta = \frac{h^4}{2\sqrt{n}}$, then %conditional on the event $\mathcal{F}_n=\big\{\max_{1\leq i \leq n} |\widetilde{\epsilon}_i|\leq  (\sigma_\ep+M) \sqrt{\log (p\vee n)}\big\}$, 
		the Lipschitz condition for $K(\cdot)$ implies that
		%$(\vbeta, t)$ in such a ball with center $(\vbeta^*,t^*)$, let us take $\delta = \frac{h^4}{2\sqrt{\log (p\vee n)}} $, then conditional on the event $\mathcal{F}_n$, the Lipschitz condition for $K(\cdot)$ implies that
		\begin{align*}
		&	\Big|(n-1)^{-1}\sum_{j=1,j\neq i}^n [K_h(\vx_i^T\vbeta-\vx_j^T\vbeta)-K_h(\vx_i^T\vbeta^*-\vx_j^T\vbeta^*)]\widetilde{\epsilon}_j\Big|\\
		\leq &c_0[(n-1)h^2]^{-1}\sum_{j=1,j\neq i}^n |(\vx_i-\vx_j)^T(\vbeta-\vbeta^*)|*|\widetilde{\epsilon}_j|\\
		\leq&2c_0h^{-2}\sqrt{|\vx_i^T(\vbeta-\vbeta^*)|^2+(n-1)^{-1}\sum_{j=1,j\neq i}^n |\vx_j^T(\vbeta-\vbeta^*)|^2} *\sqrt{(n-1)^{-1}\sum_{j=1,j\neq i}^n\widetilde{\epsilon}_j^2},
		\end{align*} 
		for some positive constant $c_0$.
		Lemma~\ref{lem:subg} and Lemma~\ref{lem14NCL} imply that $P\Big(\Big| (n-1)^{-1}\sum_{j=1,j\neq i}^n\widetilde{\epsilon}_j^2 - 4(\sigma_\ep^2+M^2)\Big| \geq \sigma_\ep^2+M^2 \Big)\leq \exp(-c_1n)$ for some constant $c_1>0$. %Since $\vx_j ^T(\vbeta-\vbeta^*)$ is sub-Gaussian, 
		Lemma~\ref{lem15NCL} suggests
		$$P\Big(\Big| (n-1)^{-1}\sum_{j=1,j\neq i}^n |\vx_j^T(\vbeta-\vbeta^*)|^2 - (\vbeta-\vbeta^*)^T\E(\vx\vx^T) (\vbeta-\vbeta^*) \Big| \geq \sigma_x^2||\vbeta-\vbeta^*||_2^2,\ \forall\ \vbeta,\vbeta^*\in\bbB \Big)\leq \exp(-c_1n),$$ 
		for some constant $c_1>0$ and all $n$ sufficiently large.
		Taking $t=(n-1)\sigma_x^2$ and $s_0=ks$, Lemma~\ref{lem15NCL} suggests
		%The sub-Gaussian random property of $\vx_i^T(\vbeta-\vbeta^*)$ implies
		$$P\left(  |\vx_i^T(\vbeta-\vbeta^*)| \geq \sqrt{n}\sigma_x||\vbeta-\vbeta^*||_2,\ \forall\ \vbeta,\vbeta^*\in\bbB\right)\leq \exp(-c_1n),$$ 
		for some positive constant $c_1$ and all $n$ sufficiently large.
		Define the event
		\begin{align*}
		\mathcal{E}_1 = \Big\{& (n-1)^{-1}\sum_{j=1,j\neq i}^n\widetilde{\epsilon}_j^2 \leq 5 (\sigma_\ep^2+M^2),\   |\vx_i^T(\vbeta-\vbeta^*)| \geq \sqrt{ n}\sigma_x^2 ||\vbeta-\vbeta^*||_2,\ \forall\ \vbeta,\vbeta^*\in\bbB\Big\}\\
		&\bigcap\Big\{  (n-1)^{-1}\sum_{j=1,j\neq i}^n |\vx_j^T(\vbeta-\vbeta^*)|^2  \leq  (\xi_3+\sigma_x^2)||\vbeta-\vbeta^*||_2^2,\ \forall\ \vbeta,\vbeta^*\in\bbB   \Big\}.
		\end{align*} 
		We have  that $P(\mathcal{E}_1)\geq 1- 3\exp(-c_1n)$ for all $n$ sufficiently large, according to the above discussions.  Hence on the $\mathcal{E}_1$,  we have 
		\begin{align*}
		&\Big|(n-1)^{-1}\sum_{j=1,j\neq i}^n [K_h(\vx_i^T\vbeta-\vx_j^T\vbeta)-K_h(\vx_i^T\vbeta^*-\vx_j^T\vbeta^*)]\widetilde{\epsilon}_j\Big|\\
		\leq& 2c_0h^{-2}\sqrt{5 (\sigma_\ep^2+M^2)}*\sqrt{n\sigma_x^2 + \sigma_x^2+\xi_3}*||\vbeta-\vbeta^*||_2\\
		\leq& c_2h^{-2} \sqrt{n}||\vbeta-\vbeta^*||_2 \leq c_2h^{-2} \sqrt{n}\delta= c_2h^2/2,
		\end{align*} 
		for some positive constant $c_2$ and all $n$ sufficiently large. 
		We thus have
		\begin{align*}
		P\left(\sup \limits_{\vbeta \in \bbB} |A_{n1}(\vx_i^T\vbeta|\vbeta)| \geq c_2h^2 \right)
		\leq &P\left(\bigcup\limits_{\vbeta^*\in \mathcal{N}_\delta} |A_{n1}(\vx_i^T\vbeta^*|\vbeta^*)| \geq c_2h^2/2 \right)\\\
		&+ P\left(\sup\limits_{ \vbeta^*\in \mathcal{N}_\delta}\sup\limits_{  ||\vbeta-\vbeta^*||_2\leq \delta}|A_{n1}(\vx_i^T\vbeta|\vbeta)-A_{n1}(\vx_i^T\vbeta^*|\vbeta^*)| \leq c_2h^2/2 \right) \\
		\leq &\sum_{\vbeta^*\in \mathcal{N}_\delta}P\left( |A_{n1}(\vx_i^T\vbeta^*|\vbeta^*)| \geq  c_2h^2/2 \right)+  1- P(\mathcal{E}_1)\\
		\leq &c p^{2ks}\delta^{-2ks}\exp ( -c_2nh^5  )+ 3\exp(-c_1n)\\
		=&c  \exp\big( 2ks\log p - 2ks \log \delta -c_2nh^5 \big)+ 3\exp(-c_1n).
		%	\\	&=\exp(-c_1 nh^5),
		\end{align*}  
		By the  assumptions of Theorem~\ref{Lasso_error}, $d_0s\log (p\vee n) \leq  nh^5$ for some constant $d_0$.
		Thus $- s\log \delta = c*s\log(h^{-1}) + c*s\log n\leq s\log (p\vee n) $.  It is followed that 
		$$	P\left(\sup \limits_{\vbeta \in \bbB} |A_{n1}(\vx_i^T\vbeta|\vbeta)| \geq  c_0h^2\right)\leq \exp(-cnh^5),$$
		for some positive constants $c_0$ and $c$. We therefore have  
		\begin{align*}
		P\left(\max_{  1\leq i\leq n}\sup \limits_{\vbeta \in \bbB} |A_{n1}(\vx_i^T\vbeta|\vbeta)| \geq  c_0h^2\right)\leq& \sum_{i=1}^n	P\left(\sup \limits_{\vbeta \in \bbB} |A_{n1}(\vx_i^T\vbeta|\vbeta)| \geq  c_0h^2\right)\\
		\leq &n \exp(-cnh^5) =  \exp [-(cnh^5-\log n)]\\
		\leq&   \exp(-c_1nh^5),
		\end{align*} 
		for some  positive constants $c_0$ $c_1$, and all $n$ sufficiently large.
	\end{proof}

	\blem \label{lem:An2bound} 
	%Under the assumptions \ref{A1} -- \ref{A4} and \ref{K1} -- \ref{K4},
	Under the assumptions of Theorem~\ref{Lasso_error}, there exist some positive constants $c_0$, $c_1$ such that  for all $n$ sufficiently large,
	\begin{align*}
	P\left(\max_{  1\leq i\leq n}\sup \limits_{\vbeta \in \bbB} |A_{n2}(\vx_i^T\vbeta|\vbeta)| \geq c_0h^2\right)\leq \exp[-c_1\log(p\vee n)].%\exp(-c_1nh^5) . %\label{An2bound} 
	\end{align*}  
	\elem
	\begin{proof}
		Let $A_{n2}(\vx_i^T\vbeta|\vbeta) = [(n-1)h]^{-1}\sum_{j=1,j\neq i}^n \gamma_i(z_j),$ where $\gamma_i(z_j) = K \Big(\frac{\vx_i^T\vbeta-\vx_j^T\vbeta}{h}\Big) \big[ f_0(\vx_j^T\vbeta_0) - G(\vx_i^T\vbeta|\vbeta)\big]$.  Recall that $f_{\vbeta}(\cdot)$ denotes the p.d.f of $\vx^T\vbeta$. Note that
		\begin{align*}
		\E \gamma_i(z_j) = &\E \Big\{K \Big(\frac{\vx_i^T\vbeta-\vx_j^T\vbeta}{h}\Big) \big[ f_0(\vx_j^T\vbeta_0) - G(\vx_i^T\vbeta|\vbeta)\big]\Big\}\\
		=&\E_{(\vx_i^T\vbeta,\vx_j^T\vbeta)} \Big\{\E \Big\{K \Big(\frac{\vx_i^T\vbeta-\vx_j^T\vbeta}{h}\Big) \big[ f_0(\vx_j^T\vbeta_0) - G(\vx_i^T\vbeta|\vbeta)\big] \Big| \vx_i^T\vbeta,\vx_j^T\vbeta \Big\}\Big\}\\
		=&    \E_{(\vx_i^T\vbeta,\vx_j^T\vbeta)} \Big\{K \Big(\frac{\vx_i^T\vbeta-\vx_j^T\vbeta}{h}\Big) \big[ G(\vx_j^T\vbeta|\vbeta) - G(\vx_i^T\vbeta|\vbeta)\big]\Big\}\\
		= &  \E_{\vx_i^T\vbeta} \Big\{h\int K(-z)  \big[ G(\vx_i^T\vbeta+hz|\vbeta) - G(\vx_i^T\vbeta|\vbeta)\big] f_{\vbeta}(\vx_i^T\vbeta+hz) dz\Big\} \\
		=&   \E_{\vx_i^T\vbeta} \Big\{h \int K(-z)  \Big[G^{(1)}(\vx_i^T\vbeta|\vbeta)hz + \frac{h^2z^2}{2}  G^{(2)}(t_1|\vbeta)\Big] \Big[f_{\vbeta}(\vx_i^T\vbeta) + hzf_{\vbeta}'(\widetilde{t})  \Big] dz\Big\}\\
		=&  \E_{\vx_i^T\vbeta} \Big\{ \frac{h^3f_{\vbeta}(\vx_i^T\vbeta)}{2} \int z^2 K(-z) G^{(2)}(t_1|\vbeta) dz \Big\}+ \E_{\vx_i^T\vbeta} \Big\{ h^3G^{(1)}(\vx_i^T\vbeta|\vbeta) \int z^2 K(-z) f_{\vbeta}'(\widetilde{t})  dz \Big\}\\
		& +  \E_{\vx_i^T\vbeta} \Big\{ \frac{h^4 }{2} \int z^3 K(-z) G^{(2)}(t_1|\vbeta)f_{\vbeta}'(\widetilde{t})  dz\Big\},
		\end{align*}
		where $t_1$ and $\widetilde{t}$ are both between $\vx_i^T\vbeta$ and $\vx_i^T\vbeta+hz$. In the above, the third equality applies the independence between $\vx_i^T\vbeta$ and  $\vx_j^T\vbeta$, and  $G(\vx_j^T\vbeta|\vbeta) = \E[f_0(\vx_j^T\vbeta_0)|\vx_j^T\vbeta]$.
		According to \ref{K1}--\ref{K4}, we know that $\sup \limits_{\vbeta \in \bbB} \E\big[A_{n2}(\vx_i^T\vbeta|\vbeta)\big] \leq ch^2$ for some $c$ large enough.
		
		Observing that $\sup \limits_{\vbeta \in \bbB} \E[\gamma_i(z_j)^2 ]\leq c_1h$,  for some positive constant $c_1$. Since $K(\cdot)$ is bounded on the real line,  for any fixed $\vbeta$,  by Bernstein's inequality, $\forall\ t\geq 0$, there exists some constant  $c >0$   such that 
		\begin{align*}
		P\left(  \left|\sum_{j=1,j\neq i}^n \gamma_i(z_j) -  \E \gamma_i(z_j) \right| \geq t\ \big|\ \vx_i^T\vbeta \right)  \leq 2\exp\left[\frac{-c t^2 }{(n-1)h}\right].
		\end{align*}
		Note that $	\E\gamma(z_i) \leq ch^3$ for some positive constant $c$. Take $t =(n-1)h^3$, then we have that 
		\begin{align*}
		&P\left(  \left|\sum_{j=1,j\neq i}^n \gamma_i(z_j) -  \E \gamma_i(z_j) \right| \geq (n-1)h^3  \right) \\
		=&\E_{\vx_i^T\vbeta}\left\{P\left(  \left|\sum_{j=1,j\neq i}^n \gamma_i(z_j) -  \E \gamma_i(z_j) \right| \geq (n-1)h^3\ \big|\ \vx_i^T\vbeta \right) \right\} \\
		\leq &2\exp\left[-c(n-1)h^5\right].
		\end{align*}
		Hence we can conclude that $P\Big(  \big|A_{n2}(t|\vbeta)\big| \geq c_0h^2\Big)  \leq 2 \exp (-c_1nh^5)$ for some positive constants $c_0$, $c_1$, and all $n$ sufficiently large.
		
		%Use  $ \delta\sqrt{\log (p\vee n)}  \times \delta^2-$balls to cover $\bbT\times\bbB$, with the covering number $N  \leq c p^{2ks}\delta^{-(4ks+1)}$, as shown in the proof of Lemma~\ref{lem:An1bound}.
		%For any $(\vbeta, t)$ in such a ball with center $(\vbeta^*,t^*)$, we need to bound
		
		To obtain the uniform bound, we cover $\bbB$ with $L_2-$balls with radius $\delta^2$.  
		Let $\mathcal{N}_{\delta^2}$ be the $\delta^2-$cover  of $\bbB$. 	The covering number $N_2 = |\mathcal{N}_{\delta^2}|$ satisfies   $N_2  \leq c p^{2ks}\delta^{-4ks }$ for sufficiently large $c$, as shown in the proof of Lemma~\ref{lem:An1bound}.
		For any $\vbeta $ in such a ball with center $\vbeta^*$, we need to bound
		\begin{align*}
		&\Big|(n-1)^{-1}\sum_{j=1,j\neq i}^n  K_h(\vx_i^T\vbeta-\vx_j^T\vbeta)\big[f_0(\vx_j^T\vbeta_0) - G(\vx_i^T\vbeta|\vbeta)\big]\\
		&-(n-1)^{-1}\sum_{j=1,j\neq i}^n K_h(\vx_i^T\vbeta^*-\vx_j^T\vbeta^*)\big[ f_0(\vx_j^T\vbeta_0) - G(\vx_i^T\vbeta^*|\vbeta^*)\big] \Big|\\
		\leq& \Big|(n-1)^{-1}\sum_{j=1,j\neq i}^n  \Big[K_h(\vx_i^T\vbeta-\vx_j^T\vbeta)-K_h(\vx_i^T\vbeta^*-\vx_j^T\vbeta^*) \Big] \big[ G(\vx_j^T\vbeta_0|\vbeta_0) - G(\vx_j^T\vbeta|\vbeta )\big]\Big|\\ 
		&+\Big|(n-1)^{-1}\sum_{j=1,j\neq i}^n  \Big[K_h(\vx_i^T\vbeta-\vx_j^T\vbeta)-K_h(\vx_i^T\vbeta^*-\vx_j^T\vbeta^*) \Big] \big[ G(\vx_j^T\vbeta|\vbeta) - G(\vx_i^T\vbeta|\vbeta )\big]\Big|\\ 
		&+\Big|(n-1)^{-1}\sum_{j=1,j\neq i}^n  K_h(\vx_i^T\vbeta^*-\vx_j^T\vbeta^*)\big[  G(\vx_i^T\vbeta|\vbeta )  - G(\vx_i^T\vbeta^*|\vbeta^*) \big] \Big| \\ 
		=&\sum_{k=1}^3 |I_{nk}|,%	\leq&  c(nh)^{-1}\sum_{i=1}^n \delta\log p = ch^2/2.
		\end{align*} 
		where the definition of $I_{nk}$ is clear from the context. %Note that for any $\vbeta\in\bbB$, we have that $||\vbeta||_1\leq ||\vbeta_0||_1+r\leq c_0$ for positive constant $c_0$, with $||\vbeta_0||_1$ as a constant.  According to Assumption~\ref{K4}-(d), we have
		%\begin{align*}
		%\max_{ 1\leq i \leq n}|G(\vx_i^T\vbeta|\vbeta) - G(\vx_i^T\vbeta^*|\vbeta^*)|\leq& 	\max_{ 1\leq i \leq n}||\vx_i||_\infty||\vbeta-\vbeta^*||_1 + (1+	c_0\max_{ 1\leq i \leq n}||\vx_i||_\infty )\sqrt{||\vbeta-\vbeta^*||_2}\\
		%\leq&\delta^2\sqrt{\log(p\vee n)} +2c_0 \delta\sqrt{\log(p\vee n)} \leq c_1 \delta\sqrt{\log(p\vee n)},
		%\end{align*} 
		%for any $\delta\leq 1$ and suitable $c_1>0$, with probability at least $1-\exp[-c\log (p\vee n)]$. 
		%In addition, 
		%\begin{align*}
		%\max_{ 1\leq i \leq n}|G^{(1)}(t|\vbeta) - G^{(1)}(\vx_i^T\vbeta^*|\vbeta_0)|\leq& 	\max_{ 1\leq i \leq n}||\vx_i||_\infty||\vbeta-\vbeta^*||_1 + (1+	c_0\max_{ 1\leq i \leq n}||\vx_i||_\infty )\sqrt{||\vbeta-\vbeta^*||_2}\\
		%\leq&\delta^2\sqrt{\log(p\vee n)} +2c_0 \delta\sqrt{\log(p\vee n)} \leq c_1 \delta\sqrt{\log(p\vee n)},
		%\end{align*} 
		
		Lemma~\ref{lem15NCL} implies that 
		%The sub-Gaussian property of $\vx_i^T(\vbeta-\vbeta^*)$ implies that 
		\begin{align*}
		P\left(\max_{1\leq i \leq n} |\vx_i^T(\vbeta-\vbeta^*)| \geq  \sigma_x \sqrt{ s\log(p\vee n)}||\vbeta-\vbeta^*||_2,\ \forall\vbeta,\vbeta^*\in\bbB\right) 
		\leq \exp[-cs\log(p\vee n)],
		\end{align*} 
		for some positive constant $c$, and all $n$ sufficiently large. %Since $\vx_j^T(\vbeta-\vbeta^*)$ are independent random variables,
		Lemma~\ref{lem15NCL} implies that  
		\begin{align*}
		&P\bigg(\max_{1\leq i \leq n} \left|(n-1)^{-1}\sum_{j=1,j\neq i}^n |\vx_j^T(\vbeta-\vbeta^*)|^2 - (\vbeta-\vbeta^*)^T\E(\vx\vx^T)(\vbeta-\vbeta^*)\right|\\
		& \qquad\qquad\qquad\qquad\qquad\qquad\qquad\geq c_0 \sigma_x^2   \sqrt{\frac{s\log(p\vee n)}{n}} ||\vbeta-\vbeta^*||_2^2, \forall\ \vbeta,\vbeta^*\in\bbB\bigg)\\
		\leq &\exp[-cs\log(p\vee n)],
		\end{align*}
		for some positive constants $c_0$, $c$, and all $n$ sufficiently large. Assumption~\ref{A2} implies that $(\vbeta-\vbeta^*)^T\E(\vx\vx^T)(\vbeta-\vbeta^*)\leq \xi_3||\vbeta-\vbeta^*||_2^2$. 
		Hence we have that 
		\begin{align*}
		\max_{1\leq i \leq n} (n-1)^{-1}\sum_{j=1,j\neq i}^n |\vx_j^T(\vbeta-\vbeta^*)|^2 \leq& \left(\xi_3+c_0\sigma_x^2\sqrt{\frac{s\log(p\vee n)}{n}}\right)||\vbeta-\vbeta^*||_2^2\\
		\leq& \sigma_x^2s\log (p\vee n)||\vbeta-\vbeta^*||_2^2,\ \forall\ \vbeta,\vbeta^*\in\bbB
		\end{align*} 
		with probability at least $1-\exp[-c\log(p\vee n)]$, for some positive constants $c_0$, $c$, and all $n$ sufficiently large.  
		
		Denote the event 
		\begin{align*}
		\mathcal{E}_2 = &\Big\{ \max_{1\leq i \leq n}    |\vx_i^T\vbeta|^2+(n-1)^{-1}\sum_{j=1,j\neq i}^n |\vx_j^T\vbeta|^2   
		\leq 2\sigma_x^2s\log(p\vee n)||\vbeta||_2^2,\ \forall\ \vbeta \in \bbB \Big\}\\
		\bigcap&\Big\{ \max_{1\leq i \leq n}   |\vx_i^T(\vbeta-\vbeta^*)|^2+(n-1)^{-1}\sum_{j=1,j\neq i}^n |\vx_j^T(\vbeta-\vbeta^*)|^2   
		\leq 2\sigma_x^2s\log(p\vee n)||\vbeta-\vbeta^*||_2^2,\ \forall\ \vbeta,\vbeta^* \in \bbB \Big\}\\
		\bigcap&\Big\{\max_{1\leq i \leq n} \big|G(\vx_i^T\vbeta|\vbeta) - G(\vx_i^T\vbeta^*|\vbeta^*) \big| \leq c_0 \sqrt{s\log (p\vee n) }||\vbeta-\vbeta^*||_2^{1/2},\ \forall\ \vbeta,\vbeta^* \in \bbB \Big\}\\
		\bigcap&\Big\{n^{-1}\sum_{i=1}^n \big|G(\vx_i^T\vbeta|\vbeta) - G(\vx_i^T\vbeta_0|\vbeta_0) \big|^2\leq  c_0||\vbeta-\vbeta_0||_2^2,\ \forall\ \vbeta\in\bbB \Big\}.
		\end{align*} 
		Hence $P(\mathcal{E}_2)\geq 1-6\exp[-c_1\log(p\vee n)]$, for some positive constant  $c_0$, $c_1$ and all $n$ sufficiently large. In the above, the third event applies (\ref{Gt}) in Lemma~\ref{lem:Gbound}, and the fourth event applies (\ref{G-2}) in Lemma~\ref{lem:Gbound}.

		Take $\delta = \frac{ h^3}{4\sqrt{s\log (p\vee n)}}$.	 
		There exist positive constants $c_0$, $c_1$, such that for all $n$ sufficiently large, 
		\begin{align*}
		|I_{n1}| \leq&  c_1 [(n-1)h^2]^{-1}\sum_{j=1,j\neq i}^n | (\vx_i-\vx_j)^T(\vbeta-\vbeta^*)|*|G(\vx_j^T\vbeta_0|\vbeta_0) - G(\vx_j^T\vbeta|\vbeta)|\\
		\leq & c_1 h^{-2} \sqrt{|\vx_i^T(\vbeta-\vbeta^*)|^2+(n-1)^{-1}\sum_{j=1,j\neq i}^n |\vx_j^T(\vbeta-\vbeta^*)|^2}  \\
		&* \sqrt{(n-1)^{-1}\sum_{j=1,j\neq i}^n |G(\vx_j^T\vbeta_0|\vbeta_0) - G(\vx_j^T\vbeta|\vbeta)|^2} \\
		\leq  &c_1h^{-2}||\vbeta-\vbeta^*||_2 *||\vbeta-\vbeta_0||_2  \sqrt{s\log(p\vee n)} \\
		\leq &c_0 h^{-2}\delta^2  \sqrt{s\log(p\vee n)}= \frac{c_0h^4}{16\sqrt{s\log (p\vee n)}}.
		\end{align*}  
		on the event $\mathcal{E}_2 $. In the above, the second last inequality applies the second and the fourth events in $\mathcal{E}_2$.

		For $I_{n2}$,  the Lipschitz condition for $K(\cdot)$ implies that there exists $\widetilde{t}_j$ between $\vx_i^T\vbeta$ and $\vx_j^T\vbeta$ such that on the event $\mathcal{E}_2 $,
		\begin{align*}
		|I_{n2}| &\leq c_1 [(n-1)h^2]^{-1}\sum_{j=1,j\neq i}^n| (\vx_i-\vx_j)^T(\vbeta-\vbeta^*)|* \big| G^{(1)}(\widetilde{t}_j|\vbeta) (\vx_i-\vx_j)^T\vbeta \big|\\
		&\leq c_1h^{-2}||\vbeta-\vbeta^*||_2 *||\vbeta||_2 s  \log(p\vee n)\\
		& \leq c_0 h^{-2}\delta^2  s \log(p\vee n)= \frac{c_0h^4}{16},
		\end{align*}  
		for all $n$ sufficiently large. In the above, the second inequality applies the first and the second events in $\mathcal{E}_2$.
		Assumption~\ref{K4}-(a) implies that $\max_{ 1\leq i \leq n}\sup_{\vbeta\in\bbB}|G^{(1)}(\vx_i^T\vbeta)|\leq b$ and  $G^{(2)}(t|\vbeta)$ is bounded for any $t\in\bbR$, $\vbeta\in\bbB$. It indicates that $\max_{  i\neq j}\sup_{\vbeta\in\bbB}|G^{(1)}(\widetilde{t}_j|\vbeta)|\leq c$ for positive constant $c$.
		
		For $|I_{n3}|$, we have
		\begin{align*}
		|I_{n3}| &\leq c_1h^{-1} |G(\vx_i^T\vbeta|\vbeta) - G(\vx_i^T\vbeta^*|\vbeta^*) |\\
		& \leq  c_1h^{-1} ||\vbeta-\vbeta^*||_2^{1/2}\sqrt{s\log (p\vee n) } \\
		&\leq c_0 h^{-1}\delta \sqrt{s\log (p\vee n)}= \frac{c_0h^2}{4},
		\end{align*} 
		on the event $\mathcal{E}_2 $,   for positive constants $c_0$, $c_1$, and all $n$ sufficiently large.  In the above, the second inequality applies the third event in $\mathcal{E}_2$.
		Combining all the previous results, we conclude that 
		\begin{align*}
		&\Big|(n-1)^{-1}\sum_{j=1,j\neq i}^n  K_h(\vx_i^T\vbeta-\vx_j^T\vbeta)\big[f_0(\vx_j^T\vbeta_0) - G(\vx_i^T\vbeta|\vbeta)\big]\\
		&-(n-1)^{-1}\sum_{j=1,j\neq i}^n K_h(\vx_i^T\vbeta^*-\vx_j^T\vbeta^*)\big[ f_0(\vx_j^T\vbeta_0) - G(\vx_i^T\vbeta^*|\vbeta^*)\big] \Big|\leq  c_0h^2/2,
		\end{align*}
		on the event $\mathcal{E}_2 $,   for all $n$ sufficiently large. Then it implies that
		\begin{align*}
		&P\left(\sup \limits_{\vbeta \in \bbB} |A_{n2}(\vx_i^T\vbeta|\vbeta)| \geq c_0h^2\right) \\
		\leq& P\left(\bigcup\limits_{\vbeta^*\in \mathcal{N}_{\delta^2}}\big|A_{n2}(\vx_i^T\vbeta^*|\vbeta^*) \big|\geq  c_0h^2/2\right)\\
		&+ P\left(\sup\limits_{ \vbeta^*\in \mathcal{N}_{\delta^2}}\sup\limits_{  ||\vbeta-\vbeta^*||_2\leq \delta^2}|A_{n2}(\vx_i^T\vbeta|\vbeta)-A_{n2}(\vx_i^T\vbeta^*|\vbeta^*)| \leq c_0h^2/2 \right) \\
		\leq& \sum_{\vbeta^*\in \mathcal{N}_{\delta^2}}P\left(\big|A_{n2}(\vx_i^T\vbeta^*|\vbeta^*) \big|\geq   c_0h^2/2\right)+5\exp[-c\log (p\vee n)] \\
		\leq &c p^{2ks}\delta^{-4ks}\exp(-cnh^5) +5\exp[-c\log (p\vee n)]\\
		\leq& \exp[-c_2\log (p\vee n)],
		\end{align*}  
		for some positive constants $c$, $c_0$, and  $c_2>1$, and all $n$ sufficiently large. 
		We thus conclude 
		\begin{align*}
		P\Big(\max_{  1\leq i\leq n}\sup \limits_{\vbeta \in \bbB} |A_{n2}(\vx_i^T\vbeta|\vbeta)| \geq  c_0h^2\Big)\leq& \sum_{i=1}^n	P\Big(\sup \limits_{\vbeta \in \bbB} |A_{n2}(\vx_i^T\vbeta|\vbeta)| \geq  c_0h^2\Big)\\
		\leq &n \exp[-c_2\log (p\vee n)]= \exp\{-[c_2\log (p\vee n)-\log n]\}\\
		\leq&  \exp[-c_1\log (p\vee n)],
		\end{align*} 
		for positive constants $c_0$ $c_1$, and all $n$ sufficiently large.
	\end{proof}

	\blem \label{lem:An3bound} 
	%Under the assumptions \ref{A1} -- \ref{A4} and \ref{K1} -- \ref{K4},
	Under the assumptions of Theorem~\ref{Lasso_error}, there exist some positive constants $c_0$, $c_1$ such that for all $n$ sufficiently large,
	\begin{align*}
	P\left(\max_{1\leq i \leq n}\sup \limits_{\vbeta \in \bbB} \left|A_{n3}(\vx_i^T\vbeta|\vbeta)-\emph{E}\left[A_{n3}(\vx_i^T\vbeta|\vbeta)\right]\right| \geq c_0h^2\right)\leq \exp(-c_1nh^5). % \label{An3bound} 
	\end{align*}  
	\elem
	\begin{proof}%f_{\vbeta}(\vx_i^T\vbeta)
		Let $A_{n3}(\vx_i^T\vbeta|\vbeta) = [(n-1)h]^{-1}\sum_{j=1,j\neq i}^n z_{ij}(\vbeta),$ where $z_{ij}(\vbeta) = K \Big(\frac{\vx_i^T\vbeta-\vx_j^T\vbeta}{h}\Big)$. Observing that $\sup \limits_{\vbeta \in \bbB} \E[z_{ij}^2(\vbeta)]\leq c_1h$.  
		Since $K(\cdot)$ is bounded on the real line, then for any fixed $\vbeta$,  by Bernstein's inequality, $\forall\ \eta>0$, there exists constant $c>0$ such that 
		\begin{align*}
		P\left(  \left|\sum_{j=1,j\neq i}^n z_{ij}(\vbeta) -  \E z_{ij}(\vbeta) \right| \geq \eta \right)  \leq 2\exp\left[\frac{-c \eta^2 }{(n-1)h+\eta}\right]. 
		\end{align*}
		Take $\eta = (n-1)h^3$, then we can conclude that $P\left(\left|A_{n3}(\vx_i^T\vbeta|\vbeta)-\E\left[A_{n3}(\vx_i^T\vbeta|\vbeta)\right]\right| \geq c_0h^{2}\right)  \leq 2 \exp(-c_1nh^5)$ for positive constants $c_0$ and $c_1$.
		
		%Use  $\delta\sqrt{\log p(p\vee n)}\times \delta-$balls to cover $\bbT\times\bbB$, with the covering number $N  \leq c p^{2ks}\delta^{-(2ks+1)}$, as shown in the proof of Lemma~\ref{lem:An1bound}.
		To obtain the uniform bound, we cover $\bbB$ with $L_2-$balls with radius $\delta$.  
		Let $\mathcal{N}_{\delta}$ be an $\delta-$cover  of $\bbB$. 	The covering number $N = |\mathcal{N}_{\delta}|$ satisfies   $N \leq c p^{2ks}\delta^{-2ks }$ for sufficiently large $c$, as shown in the proof of Lemma~\ref{lem:An1bound}. 
		For any $\vbeta$ in such a ball with center $\vbeta^*$, let us take $\delta = \frac{h^4}{4\sqrt{n}}$, then the Lipschitz condition for $K(\cdot)$ implies that
		\begin{align*}
		&\left|A_{n3}(\vx_i^T\vbeta|\vbeta)-A_{n3}(\vx_i^T\vbeta^*|\vbeta^*)\right|\\
		=&\Big|(n-1)^{-1}\sum_{j=1,j\neq i}^n [K_h(\vx_i^T\vbeta-\vx_j^T\vbeta)-K_h(\vx_i^T\vbeta^*-\vx_j^T\vbeta^*)] \Big|\\
		\leq &[c(n-1)h^2]^{-1}\sum_{j=1,j\neq i}^n \big|(\vx_i-\vx_j)^T(\vbeta-\vbeta^*)\big|\leq c_0h^2/4,
		\end{align*} 
		with probability at least $1-\exp(-cn)$, for some positive constants $c$, $c_0$, and all $n$ sufficiently large, similarly as the proof of Lemma~\ref{lem:An1bound}. It also implies that
		\begin{align*}
		&\left|\E\left[A_{n3}(\vx_i^T\vbeta|\vbeta)\right] -\E\left[A_{n3}(\vx_i^T\vbeta^*|\vbeta^*)\right]\right|\\
		=&\Big|\E_{(\vx_i^T\vbeta,\vx_i^T\vbeta^*)}\left[A_{n3}(\vx_i^T\vbeta|\vbeta)-A_{n3}(\vx_i^T\vbeta^*|\vbeta^*)\right] \Big|\\
		\leq & c_0h^2/4,
		\end{align*} 
		with probability at least $1-\exp(-cn)$, for some positive constants $c$, $c_0$, and all $n$ sufficiently large.	Then it implies that
		\begin{align*}
		&P\left(\sup \limits_{\vbeta \in \bbB}\left|A_{n3}(\vx_i^T\vbeta|\vbeta)-\E\left[A_{n3}(\vx_i^T\vbeta|\vbeta)\right]\right|  \geq c_0h^2 \right) \\
		\leq &P\left(\bigcup\limits_{\vbeta^* \in \mathcal{N}_{\delta}}\left|A_{n3}(\vx_i^T\vbeta^*|\vbeta^*)-\E\left[A_{n3}(\vx_i^T\vbeta^*|\vbeta^*)\right]\right| \geq  c_0h^2/2 \right)\\
		&+ P\left(\sup\limits_{ \vbeta^*\in \mathcal{N}_{\delta^2}}\sup\limits_{  ||\vbeta-\vbeta^*||_2\leq \delta^2}\left|A_{n3}(\vx_i^T\vbeta|\vbeta)-A_{n3}(\vx_i^T\vbeta^*|\vbeta^*)\right| \leq c_0h^2/4 \right) \\
		&+ P\left(\sup\limits_{ \vbeta^*\in \mathcal{N}_{\delta^2}}\sup\limits_{  ||\vbeta-\vbeta^*||_2\leq \delta^2}\left|\E\left[A_{n3}(\vx_i^T\vbeta|\vbeta)\right] -\E\left[A_{n3}(\vx_i^T\vbeta^*|\vbeta^*)\right]\right| \leq c_0h^2/4 \right) \\
		\leq &\sum_{\vbeta^* \in \mathcal{N}_{\delta}}P\left(\left|A_{n3}(\vx_i^T\vbeta^*|\vbeta^*)-\E\left[A_{n3}(\vx_i^T\vbeta^*|\vbeta^*)\right]\right|\geq c_0h^2/2 \right)+\exp(-cn)\\
		\leq& c p^{2ks}\delta^{- 2ks }\exp(-cnh^5)+\exp(-cn)\leq \exp(-c_1nh^5).
		\end{align*}  
		for some positive constants $c$, $c_0$, $c_1$, and all $n$ sufficiently large. We conclude 
		\begin{align*}
		&P\left(\max_{  1\leq i\leq n}\sup \limits_{\vbeta \in \bbB}\left|A_{n3}(\vx_i^T\vbeta|\vbeta)-\E\left[A_{n3}(\vx_i^T\vbeta|\vbeta)\right]\right|  \geq  c_0h^2\right)\\
		\leq& \sum_{i=1}^n	P\left(\sup \limits_{\vbeta \in \bbB}\left|A_{n3}(\vx_i^T\vbeta|\vbeta)-\E\left[A_{n3}(\vx_i^T\vbeta|\vbeta)\right]\right|  \geq  c_0h^2\right)\\
		\leq &n \exp(-c_1nh^5)= \exp [-(c_1nh^5-\log n)]\\
		\leq&  \exp(-c_2nh^5),
		\end{align*} 
		for positive constants $c_0$, $c_1$, $c_2$, and all $n$ sufficiently large.
	\end{proof}

	\blem \label{lem:Ex_eigen}  
	Under the assumptions of Theorem 1 and Lemma 2, there exist universal positive constants $c_0$, $c_1$ such that for all n sufficiently large,
	\begin{align*}
	%P\Big(\sup \limits_{t \in \bbT}  n^{-1}\sum_{i=1}^n\vv^T\Big[\vx_i-\emph{E}(\vx_i|\vx_i^T\vbeta_0=t)\Big]\Big[\vx_i-\emph{E}(\vx_i|\vx_i^T\vbeta_0=t)\Big]^T\vv \geq c_0||\vv||_2^2\Big)\leq \exp(-c_1s\log p) , \\
	P\left\{\sup \limits_{\substack{\vbeta \in \bbB\\\vv\in\bbK(p-1,2ks+\widetilde{s})}}  \left(n^{-1}\sum_{i=1}^n\vv^T\left[\vx_{i,-1}-\emph{E}(\vx_{i,-1}|\vx_i^T\vbeta)\right]\left[\vx_{i,-1}-\emph{E}(\vx_{i,-1}|\vx_i^T\vbeta)\right]^T\vv\right) \geq c_0\right\}\leq \exp(-c_1n),
	\end{align*}  
	where  $\bbK(p-1,2ks+\widetilde{s}) = \{\vv\in\bbR^{p-1}:||\vv||_2\leq 1,||\vv||_0\leq 2ks+\widetilde{s}\}$,  and $\widetilde{s}=\max_{ 2\leq j\leq p}||\vd_{0j}||_0$.
	\elem
	
	\begin{proof}
		We observe that
		\begin{align*}
		&\sup \limits_{\substack{\vbeta \in \bbB\\\vv\in\bbK(p-1,2ks+\widetilde{s})}}n^{-1}\sum_{i=1}^n\vv^T\Big[\vx_{i,-1}-\E(\vx_{i,-1}|\vx_i^T\vbeta)\Big]\Big[\vx_{i,-1}-\E(\vx_{i,-1}|\vx_i^T\vbeta)\Big]^T\vv \\
		\leq & \sup \limits_{\vv\in\bbK(p-1,2ks+\widetilde{s})} 2n^{-1}\sum_{i=1}^n\vv^T\Big[\vx_{i,-1}-\E(\vx_{i,-1}|\vx_i^T\vbeta_0)\Big]\Big[\vx_{i,-1}-\E(\vx_{i,-1}|\vx_i^T\vbeta_0)\Big]^T\vv\\
		&+\sup \limits_{\substack{\vbeta \in \bbB\\\vv\in\bbK(p-1,2ks+\widetilde{s})}}2n^{-1}\sum_{i=1}^n\vv^T\Big[\E(\vx_{i,-1}|\vx_i^T\vbeta_0)-\E(\vx_{i,-1}|\vx_i^T\vbeta)\Big]\Big[\E(\vx_{i,-1}|\vx_i^T\vbeta_0)-\E(\vx_{i,-1}|\vx_i^T\vbeta)\Big]^T\vv\\
		\triangleq&2J_{n1}+2J_{n2},
		\end{align*}
		where the definition for $J_{nk}$ is clear from the context.

		Note that $\vx_{i,-1}-\E(\vx_{i,-1}|\vx_i^T\vbeta_0)$ is sub-Gaussian with covariance matrix $\E\big[\Cov (\vx_{-1}|\vx^T\vbeta_0)\big]$. Assumption~\ref{A2} implies that
		\begin{align*}
		J_{n1} &\leq \xi_1 + \sup \limits_{\vv\in\bbK(p-1,2ks+\widetilde{s})}\big|\vv^T\bm{\Psi}_n\vv \big|,
		\end{align*}
		%Note that $\min\big\{\frac{\xi_0^2}{\sigma_x^4},1\big\}$ is a positive constant. 
		where  $ \bm{\Psi}_n  = n^{-1}\sum_{i=1}^n\big[ \vx_{i,-1}-\E(\vx_{i,-1}|\vx_i^T\vbeta_0)\big]\big[ \vx_{i,-1}-\E(\vx_{i,-1}|\vx_i^T\vbeta_0)\big]^T - \E\big[\Cov (\vx_{-1}|\vx^T\vbeta_0)\big]$.  
		Lemma~\ref{lem15NCL} implies that 
		\begin{align*}
		P\left( \sup \limits_{\vv\in\bbK(p-1,2ks+\widetilde{s})}\Big| \vv^T\bm{\Psi}_n \vv \Big| \geq  \sigma_x^2 \right) \leq \exp(-c_1n),
		\end{align*} 
		for some universal positive constant $c_1$, and all $n$ sufficiently large, since $s\log p = o(n)$, and $\widetilde{s} \log p=o(n)$. 
		Hence with probability at least $1-\exp(-c_1n)$, we have $J_{n1} \leq  \xi_1+\sigma_x^2 $.
		
		To bound $J_{n2}$, Assumption~\ref{A2}-(c) implies that
		\begin{align*}
		J_{n2} \leq &\sup_{ \vbeta \in \bbB }c^2  n^{-1}\sum_{i=1}^n\big[|\vx_i^T\vbeta_0-\vx_i^T\vbeta|  +(|\vx_i^T\vbeta|+|\vx_i^T\vbeta_0|) ||\vbeta-\vbeta_0||_2\big]^2\\
		\leq &\sup_{ \vbeta \in \bbB }\frac{2c^2}{n} \sum_{i=1}^n(\vbeta-\vbeta_0)^T\vx_i \vx_i^T(\vbeta-\vbeta_0)  + \sup_{ \vbeta \in \bbB }\frac{4c^2}{n} ||\vbeta-\vbeta_0||_2^2 \sum_{i=1}^n\big[(\vx_i^T \vbeta)^2+(\vx_i^T \vbeta_0)^2\big] ,
		\end{align*}
		for some positive constant $c$.
		Since $\vbeta\in\bbB$, we have $||\vbeta-\vbeta_0||_2\leq r$. Note that $\vx_i$ is also sub-Gaussian. Combining  Lemma~\ref{lem15NCL} and similar technique as above, we have
		\begin{align*}
		P\left\{ n^{-1} \sum_{i=1}^n [\vx_i^T (\vbeta-\vbeta_0)]^2     \leq (\xi_3+\sigma_x^2)||\vbeta-\vbeta_0||_2,\ \forall\  \vbeta\in\bbB \right\} \geq 1- \exp(-c_1n),\\
		P\left\{  n^{-1} \sum_{i=1}^n |\vx_i^T \vbeta|^2   \leq  (\xi_3+\sigma_x^2)||\vbeta||_2 ,\ \forall\ \vbeta\in\bbB\right\} \geq 1- \exp(-c_1n),
		\end{align*} 
		for some positive constant $c_1$ and all $n$ sufficiently large. We thus have that $J_{n2}\leq c_0r^2 $,  with probability at least $1- 2\exp(-c_1n)$, for some positive constants $c_0$, $c_1$, and all $n$ sufficiently large. Since $r\leq 1$, the conclusion follows.
	\end{proof}

	\blem \label{lem:subexp} 
	Under assumptions \ref{A1} and \ref{K4}, if $\log p = O(n)$, then there exist some positive constants $c_1$, $c_2$, such that for all n sufficiently large,
	\begin{align*}
	P\left( \max_{2\leq j\leq p}\Big|\Big|\frac{1}{n} \sum_{i=1}^n [G^{(1)}(\vx_i^T\vbeta_0|\vbeta_0)]^2 \vxw_{i,-j*}\vxw_{i,-1}^T\vphi_{0j} \Big|\Big|_\infty\geq c_1\sigma_x^2 \sqrt{\frac{\log p}{n}}\right)\leq \exp(-c_2\log p),%\label{subexp}
	\end{align*}
	where $\vphi_{0j} = \tau^{2}_{0j}\vtheta_j$. 
	\elem
	\begin{proof}  
		Note that for any $j$,  $G^{(1)}(\vx_i^T\vbeta_0|\vbeta_0)\vxw_{i,-1}^T\vphi_{0j} $ and $G^{(1)}(\vx_i^T\vbeta_0|\vbeta_0)\vxw_{i,-j*}$ are both sub-Gaussian with the variance proxy no larger than $2b^4\xi_1^2\xi_2^{-2}\sigma_x^2 $ and $2b^2\sigma_x^2$, respectively, by Lemma~\ref{lem:subg} and Lemma~\ref{lem:thetaj}. 
		The definition of $\vd_{0j}$ implies that $\E\big\{[G^{(1)}(\vx_i^T\vbeta_0|\vbeta_0)]^2 \vxw_{i,-j*}\vxw_{i,-1}^T\vphi_{0j}\big\} = \vnull_{p-2}$. Hence Lemma~\ref{lem14NCL} implies that 
		$$P\left(  \Big| \frac{1}{n}\sum_{i=1}^n [G^{(1)}(\vx_i^T\vbeta_0|\vbeta_0)]^2 \vxw_{i,-j*}\vxw_{i,-1}^T\vphi_{0j} \Big| \geq   2c_0b^3\xi_1\xi_2^{-1}\sigma_x^2\sqrt{\frac{\log p}{n}}\right)\leq   \exp(-c_1\log p),$$
		for some positive constants $c_0$, $c_1>1$, and all $n$ sufficiently large. Note that $2b^3\xi_1\xi_2^{-1}$ is a positive constant that does not depend on $\vx_i$.
		Then we have
		\begin{align*}
		&P\left( \max_{2\leq j\leq p}\Big|\Big|\frac{1}{n} \sum_{i=1}^n [G^{(1)}(\vx_i^T\vbeta_0|\vbeta_0)]^2 \vxw_{i,-1}^T\vphi_{0j} \vxw_{i,-j*}\Big|\Big|_\infty\geq c_1 \sigma_x^2\sqrt{\frac{\log p}{n}}\right)\\
		\leq &\sum_{j=2}^pP\left(  \Big| \frac{1}{n}\sum_{i=1}^n [G^{(1)}(\vx_i^T\vbeta_0|\vbeta_0)]^2 \vxw_{i,-1}^T\vphi_{0j} \vxw_{i,-j*}\Big| \geq   c_1 \sigma_x^2\sqrt{\frac{\log p}{n}}\right)\\
		\leq &\exp(-c_2\log p),
		\end{align*}
		for some positive constants $c_1$, $c_2$, and all $n$ sufficiently large. 
	\end{proof}
	
	\blem \label{lem:G1xx} 
	Assume the conditions of Lemma~\ref{dbound} are satisfied, then there exist universal positive constants $c_0$ and $c_1$ such that for all n sufficiently large,
	\begin{align*}
	P\left( \max_{2\leq j \leq p}\Big|\Big|\frac{1}{n} \sum_{i=1}^n \big\{[\widehat{G}^{(1)}(\vx_i^T\vbetah|\vbetah)]^2 - [G^{(1)}(\vx_i^T\vbeta_0|\vbeta_0)]^2\big\} \vxw_{i,-1} \vxw_{i,-1}^T\vphi_{0j}\Big|\Big|_\infty \geq c_0h\right)\leq \exp(-c_1\log p),\\
	P\left( \Big|\Big|\frac{1}{n} \sum_{i=1}^n \big\{[\widehat{G}^{(1)}(\vx_i^T\vbetah|\vbetah)]^2 - [G^{(1)}(\vx_i^T\vbeta_0|\vbeta_0)]^2\big\} \vxw_{i,-1}  \vxw_{i,-1}\Big|\Big|_\infty \geq c_0h\right)\leq \exp(-c_1\log p),
	\end{align*}  
	where $\vphi_{0j} = \tau^{2}_{0j}\vtheta_j$. 
	\elem
	\begin{proof}   
		We will prove the first part of the claim below. The proof of the second part is similar.
		The Cauchy-Schwartz inequality implies that
		\begin{align*}
		&\Big|\Big|\frac{1}{n} \sum_{i=1}^n \big\{[\widehat{G}^{(1)}(\vx_i^T\vbetah|\vbetah)]^2 - [G^{(1)}(\vx_i^T\vbeta_0|\vbeta_0)]^2\big\} \vxw_{i,-1} \vxw_{i,-1}^T\vphi_{0j}\Big|\Big|_\infty\\
		\leq &\max_{1\leq i\leq n}\left|  [\widehat{G}^{(1)}(\vx_i^T\vbetah|\vbetah)]^2 - [G^{(1)}(\vx_i^T\vbeta_0|\vbeta_0)]^2 \right|*\sqrt{\max_{2\leq j \leq p}n^{-1}\sum_{i=1}^n \xw_{i,j} ^2 }\sqrt{\frac{1}{n} \sum_{i=1}^n  (\vxw_{i,-1}^T\vphi_{0j})^2 }.
		\end{align*}
		
		Lemma~\ref{Gbetafunc} and  Assumption~\ref{K4}-(a) together imply that 
		\begin{align}
		&P\left(\max_{1\leq i\leq n}\big|  [\widehat{G}^{(1)}(\vx_i^T\vbetah|\vbetah)]^2 - [G^{(1)}(\vx_i^T\vbeta_0|\vbeta_0)]^2 \big|\geq ch\right)\nonumber\\
		\leq&P\left(\max_{1\leq i\leq n}\big|  \widehat{G}^{(1)}(\vx_i^T\vbetah|\vbetah) - G^{(1)}(\vx_i^T\vbeta_0|\vbeta_0) \big|\geq c_0h\right) \nonumber\\
		&+ P\left(\max_{1\leq i\leq n}\big|  \widehat{G}^{(1)}(\vx_i^T\vbetah|\vbetah) + G^{(1)}(\vx_i^T\vbeta_0|\vbeta_0) \big|\geq 2b+c_0h\right)\nonumber \\
		\leq&2\exp(-d_1\log p),\label{Ghat2}
		\end{align}
		for some universal positive constants $c$, $c_0$, $d_1$, and all $n$ sufficiently large. Note that Lemma~\ref{lem:subg} implies that $ \xw_{i,j}$ is sub-Gaussian with variance proxy at most $2\sigma_x^2$. Hence $\E (\xw_{i,j}^2)\leq 2\sigma_x^2$ uniformly in j. Lemma~\ref{lem14NCL} implies that there exist  some positive constants $d_1$, $d_2$, such that for all $n$ sufficiently large,
		\begin{align*}
		P\left(\max_{2\leq j \leq p}n^{-1}\sum_{i=1}^n \xw_{i,j} ^2\geq 3\sigma_x^2\right)
		&\leq  P\left(\max_{2\leq j \leq p}n^{-1}\sum_{i=1}^n \xw_{i,j}^2\geq \max_{ 2\leq j\leq p}\E(\xw_{i,j} ^2) +d_1\sigma_x^2\sqrt{\frac{\log p}{n}}\right) \\
		&\leq \exp(-d_2\log p).
		\end{align*}  
		Since $\vxw_{i,-1}^T\vphi_{0j}$ is sub-Gaussian with variance proxy at most $C\sigma_x^2$ for some constant $C>0$,  Lemma~\ref{lem14NCL} also implies that 
		$P\left(\frac{1}{n} \sum_{i=1}^n( \vxw_{i,-1}^T\vphi_{0j})^2 \geq 2\xi_1\xi_2^{-2} \right)\leq \exp(-d_2\log p),$
		for some positive constants $d_1$, $d_2$, and all $n$ sufficiently large.
		
		Hence we have 
		\begin{align*}
		&P\left(\max_{2\leq j\leq p} \Big|\Big|\frac{1}{n} \sum_{i=1}^n \big\{[\widehat{G}^{(1)}(\vx_i^T\vbetah|\vbetah)]^2 - [G^{(1)}(\vx_i^T\vbeta_0|\vbeta_0)]^2\big\} \vxw_{i,-1} \vxw_{i,-1}^T\vphi_{0j}\Big|\Big|_\infty \geq  ch(a^2\xi_0)^{-1}\sigma_x\sqrt{6\xi_2}\right)\\
		\leq &P\left(\max_{1\leq i\leq n}\big|  [\widehat{G}^{(1)}(\vx_i^T\vbetah|\vbetah)]^2 - [G^{(1)}(\vx_i^T\vbeta_0|\vbeta_0)]^2 \big|\geq ch\right) \\
		&+ P\left(\big|\big|\frac{1}{n} \sum_{i=1}^n\vxw_{i,-1} \vxw_{i,-1}^T\big|\big|_\infty\geq 3\sigma_x^2\right) + \sum_{j=2}^pP\left(\frac{1}{n} \sum_{i=1}^n( \vxw_{i,-1}^T\vphi_{0j})^2 \geq 2\xi_2(a^2\xi_0)^{-2} \right)\\
		\leq &\exp(-c_1\log p),
		\end{align*} 
		for some positive constants $c$, $c_1$, and all $n$ sufficiently large.
	\end{proof}

	\blem \label{lem:Ex_err} 
	Assume the conditions of Lemma~\ref{dbound} are satisfied, then there exist universal positive constants $c_0$ and $c_1$ such that for all n sufficiently large,
	\begin{align*}
	P\left( \max_{2\leq j \leq p}\Big|\Big|\frac{1}{n} \sum_{i=1}^n  [\widehat{G}^{(1)}(\vx_i^T\vbetah|\vbetah)]^2 (\vxh_{i,-1}  \vxh_{i,-1}^T-\vxw_{i,-1} \vxw_{i,-1}^T)\vphi_{0j}\Big|\Big|_\infty \geq c_0\sqrt{s}h^2\right)\leq \exp(-c_1\log p), 
	\end{align*}  
	where $\vphi_{0j} = \tau^{2}_{0j}\vtheta_j$, and $\vxh_{i,-1} = \vx_{i,-1}-\widehat{\emph{E}}(\vx_{i,-1}|\vx_i^T\vbetah)$, $\vxw_{i,-1}= \vx_{i,-1}-\emph{E}(\vx_{i,-1}|\vx_i^T\vbeta_0)$.
	\elem
	\begin{proof}  
		Observe that 
		\begin{align*}
		&\Big|\Big|\frac{1}{n} \sum_{i=1}^n  [\widehat{G}^{(1)}(\vx_i^T\vbetah|\vbetah)]^2 (\vxh_{i,-1}\vxh_{i,-1}^T-\vxw_{i,-1}\vxw_{i,-1}^T)\vphi_{0j}\Big|\Big|_\infty\\
		\leq &\Big|\Big|\frac{1}{n} \sum_{i=1}^n  \big[\widehat{G}^{(1)}(\vx_i^T\vbetah|\vbetah)\big]^2 (\vxh_{i,-1}-\vxw_{i,-1})(\vxh_{i,-1}-\vxw_{i,-1})^T\vphi_{0j}\Big|\Big|_\infty\\
		&+\Big|\Big|\frac{1}{n} \sum_{i=1}^n  [\widehat{G}^{(1)}(\vx_i^T\vbetah|\vbetah)]^2 (\vxh_{i,-1}-\vxw_{i,-1}) \vxw_{i,-1}^T\vphi_{0j}\Big|\Big|_\infty\\
		&+\Big|\Big|\frac{1}{n} \sum_{i=1}^n  [\widehat{G}^{(1)}(\vx_i^T\vbetah|\vbetah)]^2 \vxw_{i,-1}(\vxh_{i,-1}-\vxw_{i,-1})^T\vphi_{0j}\Big|\Big|_\infty.
		\end{align*}
		
		Inequality (\ref{Ghat2}) implies that 
		\begin{align*}
		P\left(\max_{1\leq i\leq n}\big[\widehat{G}^{(1)}(\vx_i^T\vbetah|\vbetah)\big]^2 \geq b^2+ch\right)  	\leq \exp(-d_1\log p),
		\end{align*}
		for some positive constants $c$, $d_1$, and all $n$ sufficiently large. 
		Note that $\vxh_{i,-1}-\vxw_{i,-1}= \widehat{\E}(\vx_{i,-1}|\vx_i^T\vbetah) - \E(\vx_{i,-1}|\vx_i^T\vbeta_0)$.
		Theorem~\ref{Lasso_error} and Lemma~\ref{Efunc} imply that 
		\begin{align*}
		P\left( \big|\big|\frac{1}{n}\sum_{i=1}^n(\vxh_{i,-1}-\vxw_{i,-1})(\vxh_{i,-1}-\vxw_{i,-1})^T\big|\big|_\infty \geq c_0 sh^4 \right)\leq  \exp(-c_1\log p),\\
		P\left(\frac{1}{n}\sum_{i=1}^n \big|(\vxh_{i,-1}-\vxw_{i,-1})^T\vphi_{0j}\big|^2 \geq c_0 sh^4 ||\vphi_{0j}||_2^2 \right)\leq  \exp(-c_1\log p),
		\end{align*}
		for some positive constants $c_0$, $c_1$, and all $n$ sufficiently large. 
		Lemma~\ref{lem:thetaj} implies that $||\vphi_{0j}||_2 = \tau_{0j}^2||\vtheta_j||_2\leq b^2\xi_1\xi_2^{-1}$.
		Hence we can conclude that 
		\begin{align*}
		&P\left(\max_{2\leq j\leq p} \Big|\Big|\frac{1}{n} \sum_{i=1}^n  [\widehat{G}^{(1)}(\vx_i^T\vbetah|\vbetah)]^2 (\vxh_{i,-1}-\vxw_{i,-1})(\vxh_{i,-1}-\vxw_{i,-1})^T\vphi_{0j}\Big|\Big|_\infty\geq c_0sh^4\right)\\
		\leq &P\left(\max_{1\leq i\leq n}[\widehat{G}^{(1)}(\vx_i^T\vbetah|\vbetah)]^2 \geq 2b^2\right)+ P\left( \Big|\Big|\frac{1}{n}\sum_{i=1}^n(\vxh_{i,-1}-\vxw_{i,-1})(\vxh_{i,-1}-\vxw_{i,-1})^T\Big|\Big|_\infty \geq c' sh^4 \right) \\
		&+ \sum_{j=2}^pP\left( \frac{1}{n}\sum_{i=1}^n \big|(\vxh_{i,-1}-\vxw_{i,-1})^T\vphi_{0j}\big|^2 \geq c' b^4\xi_1^2\xi_2^{-2} sh^4 \right)\\
		\leq &\exp(-c_1\log p).
		\end{align*} 
		for some positive constants $c_0$, $c_1$, $c'$, and all $n$ sufficiently large. 
		Similarly, we have that $P\left(\max_{2\leq j\leq p} \Big|\Big|\frac{1}{n} \sum_{i=1}^n  [\widehat{G}^{(1)}(\vx_i^T\vbetah|\vbetah)]^2 (\vxh_{i,-1}-\vxw_{i,-1})\vxw_{i,-1}^T\vphi_{0j}\Big|\Big|_\infty\geq c_0\sqrt{s}h^2\right)\leq \exp(-c_1\log p)$, and $P\left(\max_{2\leq j\leq p} \Big|\Big|\frac{1}{n} \sum_{i=1}^n  [\widehat{G}^{(1)}(\vx_i^T\vbetah|\vbetah)]^2\vxw_{i,-1} (\vxh_{i,-1}-\vxw_{i,-1})^T\vphi_{0j}\Big|\Big|_\infty\geq c_0\sqrt{s}h^2\right)\leq \exp(-c_1\log p)$, for some positive constants $c_0$, $c_1$, and all $n$ sufficiently large. This concludes the proof of Lemma~\ref{lem:Ex_err}.
	\end{proof}

	\blem \label{lem:thetax} 
	Assume the conditions of Lemma~\ref{dbound} are satisfied, then there exist some positive constant $c$ such that for all $n$ sufficiently large,
	$$\max_{2\leq j\leq p}\vthetah_j^T \left(\frac{1}{n}\sum_{i=1}\vxh_{i,-1}\vxh_{i,-1}^T\right)\vthetah_j  \leq 4\xi_1\xi_2^{-2},$$
	with probability at least $1-\exp(-c\log p)$, 
	where $\xi_1$ and $\xi_2$ are defined in Assumption~\ref{A2}-(a).%is the largest eigenvalue of $\emph{E}\big[\Cov (\vx|\vx^T\vbeta_0)\big]$, and $\inf_{\vv}\vv^T\vOmega\vv \geq \xi_2$, for any unit vector $\vv=(v_1,\cdots,v_p)^T$ with $v_1=0$ and $||\vv||_0\leq 2ks$.
	\elem
	\begin{proof}  
		Assumption~\ref{A2}-(a) implies that $\inf_{\vv\in\kV_1}\vv^T\E\big[\Cov (\vx_{-1}|\vx^T\vbeta_0)\big]\vv \geq \xi_0$, where $\kV_1 = 
		\{\vv\in\bbR^{p-1}:  ||\vv||_2=1,  	||\vv||_0\leq 2ks\}$. 
		Define $$\mathcal{E}_0=\left\{\sup_{\vv\in\mathbb{K}(p-1,2ks)}\Big|\vv^T \Big[\frac{1}{n}\sum_{i=1}^n\vxw_{i,-1}\vxw_{i,-1}^T - \E(\vxw_{i,-1}\vxw_{i,-1}^T)\Big]\vv\Big|\geq  \frac{\xi_0}{54}\right\},$$
		where $\mathbb{K}(p-1,2ks) = \{\vv\in\bbR^{p-1}: ||\vv||_0\leq 2ks, ||\vv||_2\leq 1\}$.
		By taking $s_0=\frac{n}{2\log p}$ and $t = \frac{\xi_0}{54}$, Lemma~\ref{lem15NCL} implies that 
		$P\left(\mathcal{E}_0\right)\leq 2\exp\Big(-cn\min\Big\{ \frac{\xi_0}{54}, 1\Big\}\Big), $
		for some positive constant $c$ and all $n$ sufficiently large.
		Then Lemma~13 in \citet{NCL} implies that on the event $\mathcal{E}_0$, for any $\vv\in\bbR^{p-1}$, we have that 
		$\vv^T\big(\frac{1}{n}\sum_{i=1}^n\vxw_{i,-1}\vxw_{i,-1}^T \big)\vv\leq \frac{3\xi_1}{2}||\vv||_2^2 + \frac{\xi_0\log p}{n}||\vv||_1^2 $.
		
		Lemma~\ref{lem:thetaj} and results in Lemma~\ref{dbound}-(3) of the main paper imply that 
		\begin{align*}
		\max_{2\leq j\leq p}||\vthetah_{j}||_2^2\leq 2\max_{2\leq j\leq p}(||\vtheta_{j}||_2^2+||\vtheta_{j}-\vthetah_{j}||_2^2)\leq 2(\xi_2^{-2} + c_0\eta^2\widetilde{s}),\\
		\max_{2\leq j\leq p}||\vthetah_{j}||_1^2\leq 2\max_{2\leq j\leq p}(||\vtheta_{j}||_1^2+||\vtheta_{j}-\vthetah_{j}||_1^2)\leq 2 (\widetilde{s}\xi_2^{-2}+ c_0\widetilde{s}^2\eta^2),
		\end{align*} 
		with probability at least $1-\exp(-c_1\log p)$, for some positive constants $c_0$, $c_1$, and all $n$ sufficiently large. Hence on the event $\mathcal{E}_0$,
		\begin{align*}
		&\max_{2\leq j\leq p}\vthetah_j^T \left(\frac{1}{n}\sum_{i=1}\vxh_{i,-1}\vxh_{i,-1}^T\right)\vthetah_j\\
		\leq& \max_{2\leq j\leq p}\vthetah_j^T \left(\frac{1}{n}\sum_{i=1}\vxw_{i,-1}\vxw_{i,-1}^T\right)\vthetah_j +\max_{2\leq j\leq p}\vthetah_j^T \left(\frac{1}{n}\sum_{i=1}\vxh_{i,-1}\vxh_{i,-1}^T-\vxw_{i,-1}\vxw_{i,-1}^T\right)\vthetah_j \\
		\leq& \frac{3\xi_1}{2} \max_{2\leq j\leq p}||\vthetah_j||_2^2 + \frac{\xi_0\log p}{n} \max_{2\leq j\leq p}||\vthetah_j||_1^2 +2\max_{2\leq j\leq p}||\vthetah_j||_1 \left|\left| \frac{1}{n}\sum_{i=1}(\vxh_{i,-1}-\vxw_{i,-1})\vxw_{i,-1}^T\right|\right|_\infty\\
		&+\max_{2\leq j\leq p}||\vthetah_j||_1 \left|\left| \frac{1}{n}\sum_{i=1}(\vxh_{i,-1}-\vxw_{i,-1})(\vxh_{i,-1}-\vxw_{i,-1})^T\right|\right|_\infty\\
		\leq&3\xi_1\left(\xi_2^{-2}+ c_0h^2\widetilde{s} \right)+  2\left(\frac{\xi_0\log p}{n} +c_0\sqrt{s}h^2\right) \left(\widetilde{s}\xi_2^{-2}+ c_0\widetilde{s}^2h^2\right)\\
		\leq &4\xi_1\xi_2^{-2},
		\end{align*}
		with probability at least $1-\exp(-c\log p)$, for some positive constants $c_0$, $c$, and all $n$ sufficiently large,
		since  $ n^{-1}\widetilde{s}\log p = o(h^5\widetilde{s})=o(1)$, $\widetilde{s}\sqrt{s}h^2=o(1)$ and $n^{-1}\widetilde{s}^2h^2\log p = o(h^7\widetilde{s}^2 )=o(1)$. In the above, the third inequality applies the results in Lemma~\ref{lem:Ex_err}. Hence it concludes the proof of the lemma.
	\end{proof}

	\section{Identifiability conditions for the classical low-dimensional single index model} \label{sec:id_cond}
	We assume that the underlying low-dimensional true  model for the treatment-covariates interaction term
	$f_0(\vx^T\vbeta_0)$ complies with the classical
	identification assumptions for the single-index model (i.e., our condition \ref{A1}-(c)). To be self-contained, 
	we state below a set of sufficient conditions for identifying
	$\vbeta_0$ in the low-dimension model as stated in Theorem~2.1 in \citet{horowitz2012semiparametric}.
	%	\renewcommand{\thetheorem}{2.\arabic{theorem}} 
	%	\begin{theorem}[Identification of Single-Index Models in %citet{horowitz2012semiparametric}]
	%		Suppose that $\emph{E}(Y|\vx)$ satisfied the model:
	%		$$\emph{E}(Y|\vx) = f_0(\vx^T\vbeta_0),$$
	%		and $\vx$ is a $p-$dimensional random variable. Then %$\vbeta_0$ and $f_0(\cdot)$ are identified if the following %conditions hold:
	\begin{enumerate}
		\item[(a)] $f_0(\cdot)$ is differentiable and non-constant on the support of $\vx^T\vbeta_0$.
		\item[(b)] The components of $\vx_{T_0}$ are continuously distributed random variables that have a joint probability density function, where $T_0$
		is the index set corresponding to the nonzero coefficients in $\vbeta_0$ and $\vx_{T_0}$ denotes the subvector of $\vx$ with  indices in $T_0$.
		\item[(c)] The support of $\vx_{T_0}$ is not contained in any proper linear space of $\bbR^s$, with $s=|T_0|$.
		\item[(d)] $\beta_1=1$ and $||\vbeta_0||_0\geq 2$.
	\end{enumerate}
	%	\end{theorem}
	
	%This theorem assumes that the %components of $\vx$ are all %continuous. Note that condition (a) %is equivalent to our %Assumption~\ref{A1}-(b) in the %appendix, and condition (d) is %involved  in the definition of the %candidate set for $\vbeta_0$. That %is, $\bbB_0 = %\{\vbeta=(\beta_1,\cdots,\beta_p)^T:  %||\vbeta||_0\geq2,\beta_1=1\}$.

	\noindent{\it Remark.} The literature has slightly different versions of identifiability conditions for the single index model, for example \citet{ichimura1993}. The above conditions are cited for their transparency. As discussed in \citet{horowitz2012semiparametric}, a more complex set of conditions are available to allow for some components of $\vx$ being discrete.
	In particular, the following two additional conditions are needed:
	(1) varying the values of the discrete components must not divide the support of $\vx^T\vbeta_0$ into disjoint subsets, and (2) $f_0(\cdot)$ must satisfy a non-periodicity condition.

	\section{Examples for verifying the regularity conditions} \label{sec:normal_verify}
	We  verify the key conditions on $G(\vx^T\vbeta|\vbeta)$, $G^{(1)}(\vx^T\vbeta|\vbeta)$, $\E(\vx|\vx^T\vbeta)$ and $\E(\vx\vx^T|\vx^T\vbeta)$
	when $\vx$ follows a multivariate normal distribution.
	We focus on conditions that are not much discussed in the current literature on	inference for high-dimensional linear regression.
	For notation simplicity, we assume that $\vx\sim N(\vnull, \vI_p)$. Similar results can be obtained for a multivariate normal distribution with a general covariance $\vSigma$.  
	
	Given $\vx\sim N(\vnull, \vI_p)$, then for any $\vw,\vbeta\in\bbR^p$, we have
	$$\left(\begin{array}{c}
	\vx^T\vw\\ \vx^T\vbeta 
	\end{array} \right)\sim N\left(\vnull,\left(\begin{array}{cc}
	||\vw||_2^2&\vbeta^T\vw\\ \vbeta^T\vw&||\vbeta ||_2^2
	\end{array} \right)\right).$$
	For any $\vbeta\neq \vnull_p$, the distribution of $\vx^T\vw$ conditional on $\vx^T\vbeta $ is normal with mean $\frac{\vx^T\vbeta}{||\vbeta ||_2^2}\vbeta^T\vw$, and variance $||\vw||_2^2-\frac{(\vbeta^T\vw)^2}{||\vbeta||_2^2}$. %we will validate assumptions in \ref{A2}-(b)(c) and \ref{K4}-(c)(d) successively, given Assumption~\ref{A1}. 
	We thus have
	$$\vx^T\vbeta_0|\vx^T\vbeta=t \sim N\left(\frac{\vbeta_0^T\vbeta}{||\vbeta||_2^2}t,||\vbeta_0||_2^2-\frac{(\vbeta_0^T\vbeta)^2}{||\vbeta||_2^2} \right).$$
	
	In the following subsections, we demonstrate the key assumptions in \ref{A2}-(a)(b)(c) and \ref{K4}-(b)(c) hold with high probability in the above setup.% successively, given Assumption~\ref{A1}.
	
	\subsection{Verify Assumption~\ref{A2}-(a)}
	First, we verify the eigenvalue conditions involving $\E[\Cov(\vx_{-1}|\vx^T\vbeta_0)]$  in Assumption~\ref{A2}-(a). Note that for any $\vw\in\bbR^p$, we have  
	$ \Var (\vx^T\vw|\vx^T\vbeta) =  ||\vw||_2^2 - \frac{(\vbeta^T\vw)^2}{||\vbeta||_2^2} .$ 
	Recall that $\mathcal{V}_1=
	\{\vv\in\bbR^{p-1}:  ||\vv||_2=1,  	||\vv||_0\leq 2ks\}$, 
	where $k>1$ is a positive integer. Therefore, we have
	\begin{align*}
	\inf_{\vv\in \mathcal{V}_1}\vv^T\E\big[\Cov (\vx_{-1}|\vx^T\vbeta_0)\big]\vv  
	=&\inf_{\vv\in \mathcal{V}_1}\E\left\{\Var (\vx_{-1}^T\vv|\vx^T\vbeta_0)  \right\} \\
	=&\inf_{\vv\in \mathcal{V}_1}  \left[ 1- \frac{(\vbeta_{0,-1}^T\vv)^2}{||\vbeta_0||_2^2}\right] \\
	\geq&1-\sup_{\vv\in \mathcal{V}_1}   \frac{(\vbeta_{0,-1}^T\vv)^2}{||\vbeta_0||_2^2} \\
	\geq & 1-  \frac{||\vbeta_{0,-1}||_2^2}{||\vbeta_0||_2^2} = \frac{1}{||\vbeta_0||_2^2} ,
	\end{align*}
	since $\vbeta_0=(1,\vbeta_{0,-1}^T)^T$, where the last inequality applies the Cauchy-Schwartz inequality. In the current setup, it is straightforward to show $\lambda_{\max}\left(\E (\vx\vx^T) \right) \leq 1$. Furthermore, 
	\begin{align*}
	\lambda_{\max}\left\{\E\big[\Cov (\vx_{-1}|\vx^T\vbeta_0)\big]\right\}
	=&\sup_{ \vv\in\bbR^{p-1}: ||\vv||_2=1}\E\left\{\Var (\vx_{-1}^T\vv|\vx^T\vbeta_0)  \right\} \\
	=&\sup_{\vv\in\bbR^{p-1}:  ||\vv||_2=1}\left[ 1- \frac{(\vbeta_{0,-1}^T\vv)^2}{||\vbeta_0||_2^2}\right]\leq 1.
	\end{align*} 
	
	The assumption $\lambda_{\min}(\vOmega)\geq \xi_2$ is similar to the condition
	imposed on the population Hessian matrix for high-dimensional generalized linear models. To see this is a reasonable assumption, we consider the special case that $\inf_{t}|f'_0(t)|\geq a$ for some positive constant $a$ (e.g., $f_0$ is a linear function). Then
	\begin{align*}
	\lambda_{\min}(\vOmega)
	=&\inf_{ \vv\in\bbR^{p-1}: ||\vv||_2=1}\E\left\{[G^{(1)}(\vx^T\vbeta_0|\vbeta_0)]^2(\vxw_{-1}^T\vv )^2 \right\} \\
	\geq& a^2\inf_{ \vv\in\bbR^{p-1}: ||\vv||_2=1}\E\left\{ (\vxw_{-1}^T\vv )^2 \right\}  \geq  \frac{a^2}{||\vbeta_0||_2^2},
	\end{align*}
	where the analysis is similar as above, since $\E\left( \vxw_{-1}\vxw_{-1}^T\right) = \E\big[\Cov (\vx_{-1}|\vx^T\vbeta_0)\big]$.

	Finally, we verify the eigenvalue conditions involving $\lambda_{\max}(\E(\vx_i\vx_i^T|\vx_i^T\vbeta))$ in Assumption~\ref{A2}-(a).  
	For any $\vw\in\bbR^p$, 
	$$\E [(\vx^T\vw)^2|\vx^T\vbeta] =[\E (\vx^T\vw |\vx^T\vbeta)]^2 +\Var (\vx^T\vw|\vx^T\vbeta) = \frac{(\vx^T\vbeta)^2(\vbeta^T\vw)^2}{||\vbeta ||_2^4} + ||\vw||_2^2 - \frac{(\vbeta^T\vw)^2}{||\vbeta||_2^2} .$$
	Therefore, we have
	\begin{align*}
	&\sup_{\vbeta\in\bbB}n^{-1}\sum_{i=1}^n\left[\lambda_{\max}\big(\E(\vx_i\vx_i^T|\vx_i^T\vbeta)\big)\right]^2\\
	=&\sup_{\vbeta\in\bbB}n^{-1}\sum_{i=1}^n\sup_{\vv\in\bbR^p:||\vv||_2=1}\left\{\E [(\vx_i^T\vv)^2|\vx_i^T\vbeta] \right\}^2\\
	=& \sup_{\vbeta\in\bbB}n^{-1}\sum_{i=1}^n\sup_{\vv\in\bbR^p:||\vv||_2=1}\left[\frac{(\vx_i^T\vbeta)^2(\vbeta^T\vv)^2}{||\vbeta ||_2^4} + 1- \frac{(\vbeta^T\vv)^2}{||\vbeta||_2^2}\right]^2\\
	\leq& \sup_{\vbeta\in\bbB}n^{-1}\sum_{i=1}^n\sup_{\vv\in\bbR^p:||\vv||_2=1}\left[\frac{2(\vx_i^T\vbeta)^4(\vbeta^T\vv)^4}{||\vbeta ||_2^8} +2\right]\\
	\leq& \sup_{\vbeta\in\bbB}n^{-1} \sum_{i=1}^n\left[\frac{2(\vx_i^T\vbeta)^4 }{||\vbeta ||_2^4} +2\right]\\
	\leq &  \sup_{\vbeta\in\bbB}\left\{\frac{3\E\left[(\vx_i^T\vbeta)^4\right]}{||\vbeta||_2^4}+2\right\} =11,
	\end{align*}
	with probability at least $1-\exp(-c_1\sqrt{n})$, for some positive constants $c_0$, $c_1$, and all $n$ sufficiently large. In the above, the second last inequality applies the Cauchy-Schwartz inequality,and the last inequality applies Lemma~\ref{lem:cube_rate}. Similarly, we have
	\begin{align*}
	&\max_{ 1\leq i \leq n}\sup_{\vbeta\in\bbB_1}   \lambda_{\max}\big(\E(\vx_i\vx_i^T|\vx_i^T\vbeta)\big) \\
	=&\max_{ 1\leq i \leq n}\sup_{\vbeta\in\bbB_1}\sup_{\vv\in\bbR^p:||\vv||_2=1}  \E[(\vx_i^T\vv)^2|\vx_i^T\vbeta]  \\
	=&\max_{ 1\leq i \leq n}\sup_{\vbeta\in\bbB_1}\sup_{\vv\in\bbR^p:||\vv||_2=1}  \left[\frac{(\vx_i^T\vbeta)^2(\vbeta^T\vv)^2}{||\vbeta ||_2^4} + 1- \frac{(\vbeta^T\vv)^2}{||\vbeta||_2^2} \right]\\
	\leq &\max_{ 1\leq i \leq n}\sup_{\vbeta\in\bbB_1}  \frac{(\vx_i^T\vbeta)^2 }{||\vbeta ||_2^2} +1,
	\end{align*}
	where the last inequality applies the Cauchy-Schwartz inequality.
	Since $\vx_i^T\vbeta\sim N(0,||\vbeta||_2^2)$, by the tail property of the normal distribution, we have $P\Big(\max\limits_{1\leq i\leq n}|\vx_i^T\vbeta|\geq c_0||\vbeta||_2\sqrt{\log(p\vee n)},\\ \forall\vbeta\in\bbB_1\Big)\leq \exp[-c_1\log(p\vee n)]$,  some positive constants $c_0$, $c_1$, and all $n$ sufficiently large.  Thus we have 
	\begin{align*}
	&P\left(\max_{1\leq i \leq n}\sup_{\vbeta\in\bbB_1}  \lambda_{\max}\big(\E(\vx_i\vx_i^T|\vx_i^T\vbeta)\big) \geq M \log(p\vee n)\right)\\
	\leq &P\left(\max_{ 1\leq i \leq n}  (\vx_i^T\vbeta)^2 \geq [M\log(p\vee n)-1]||\vbeta ||_2^2,\ \forall\ \vbeta\in\bbB_1\right)\\  
	\leq &\exp[-c_1\log(p\vee n)],
	\end{align*}
	for some positive constants $M$, $c_1$, and all $n$ sufficiently large.

	\subsection{Verify Assumption~\ref{A2}-(b)}
	
	Next we verify the key conditions in \ref{A2}-(b). Observe that $\E(\vx_{-1}^T\veta|\vx^T\vbeta=t) =  \frac{t}{||\vbeta||_2^2}\vbeta_{-1}^T\veta$. Hence we have $\E^{(1)}(\vx_{-1}^T\veta|\vx^T\vbeta=t) = \frac{\vbeta_{-1}^T\veta}{||\vbeta||_2^2}$, and $\E^{(2)}(\vx_{-1}^T\veta|\vx^T\vbeta=t) = 0$,   satisfying the following constraints:
	\begin{align*}
	\max_{1\leq i \leq n}\sup_{ \vbeta \in \bbB}|\E^{(1)}(\vx_{i,-1}^T\veta|\vx_i^T\vbeta) |=\sup_{ \vbeta \in \bbB} \frac{|\vbeta_{-1}^T\veta|}{||\vbeta||_2^2}\leq \sup_{ \vbeta \in \bbB} \frac{||\veta||_2}{||\vbeta||_2} \leq ||\veta||_2,\\
	\sup_{|t|\leq 2||\vbeta_0||_2\sigma_x\sqrt{\log(p\vee n)}}\sup_{ \vbeta \in \bbB}|\E^{(2)}(\vx_{-1}^T\veta|\vx^T\vbeta=t) |\leq ||\veta||_2,
	\end{align*}
	since $||\vbeta||_2\geq 1$ for any $\vbeta\in \bbB$. Recall that $\bbB_1=\{\vbeta\in\bbB: ||\vbeta-\vbeta_0||_2\leq c_0\sqrt{s}h^2,  ||\vbeta||_0\leq ks \}$, for some constants $k>1$ and $c_0>0$.
	Since $\E[(\vx_{-1}^T\veta)^2|\vx^T\vbeta=t] = \frac{t^2}{||\vbeta||_2^2}\vbeta_{-1}^T\veta + ||\veta||_2^2-\frac{(\vbeta_{-1}^T\veta)^2}{||\vbeta||_2^2}$, the Cauchy-Schwartz inequality implies that
	\begin{align*}
	\max_{1\leq i \leq n}\sup_{\vbeta\in\bbB_1}  \left\{\left|\E^{(1)}\left[(\vx_{i,-1}^T\veta)^2|\vx_i^T\vbeta\right]\right| \right\}
	=\max_{1\leq i \leq n}\sup_{\vbeta\in\bbB_1}  \frac{2|\vx_i^T\vbeta|(\vbeta_{-1}^T\veta)^2}{||\vbeta ||_2^4} 
	\leq ||\veta||_2^2\max_{1\leq i \leq n}\sup_{\vbeta\in\bbB_1}  \frac{2|\vx_i^T\vbeta|}{||\vbeta ||_2^2}.
	\end{align*}
	We have $P\left(\max_{1\leq i\leq n}|\vx_i^T\vbeta|\geq c_0||\vbeta||_2\sqrt{\log(p\vee n)},\ \forall\ \vbeta\in\bbB_1\right)\leq \exp[-c\log(p\vee n)]$, for some positive constants $c_0$, $c_1$, and all $n$ sufficiently large. Note that for any $\vbeta\in\bbB_1$, we have $\beta_1=1$.
	We thus have 
	$$1\leq ||\vbeta||_2\leq ||\vbeta_0||_2 + ||\vbeta-\vbeta_0||_2 \leq ||\vbeta_0||_2 + c_0\sqrt{s}h^2\leq c_0,$$ 
	for some constant $c_0>0$.  Then we have 
	\begin{align*}
	&P\left(\max_{1\leq i \leq n}\sup_{\vbeta\in\bbB_1}  \left\{\left|\E^{(1)}\left[(\vx_{i,-1}^T\veta)^2|\vx_i^T\vbeta\right]\right| \right\}\geq M||\veta||_2^2\sqrt{\log(p\vee n)}\right)\\
	\leq &P\left(\max_{1\leq i\leq n}\sup_{\vbeta\in\bbB_1 }\frac{2|\vx_i^T\vbeta|}{||\vbeta||_2^2}\geq M\sqrt{\log(p\vee n)}\right)\\ 
	\leq &P\left(\max_{1\leq i\leq n}\sup_{\vbeta\in\bbB_1 }|\vx_i^T\vbeta|\geq c_0\sqrt{\log(p\vee n)}\right)\\
	\leq &\exp[-c_1\log(p\vee n)],
	\end{align*}
	for some positive constants $M$, $c_0$, $c_1$, and all $n$ sufficiently large.

	\subsection{Verify the Lipschitz condition of  E{\boldmath$(x|x^T\beta)$} in Assumption~\ref{A2}-(c)}\label{sec:normal_E}
	For any $\vbeta_1,\vbeta_2\in\bbB$, we observe that
	\begin{align*}
	\E(\vx^T\vv|\vx^T\vbeta_1) -&\E (\vx^T\vv |\vx^T\vbeta_2) = \frac{\vx^T\vbeta_1}{||\vbeta_1 ||_2^2}\vbeta_1^T\vv- \frac{\vx^T\vbeta_2}{||\vbeta_2||_2^2}\vbeta_2^T\vv\\
	=& \frac{\vbeta_1^T\vv}{||\vbeta_1 ||_2^2} (\vx^T\vbeta_1-\vx^T\vbeta_2)+ \vx^T\vbeta_2\left( \frac{\vbeta_1^T\vv}{||\vbeta_1||_2^2} -\frac{\vbeta_2^T\vv}{||\vbeta_2||_2^2}\right)\\
	=& \frac{\vbeta_1^T\vv}{||\vbeta_1 ||_2^2} (\vx^T\vbeta_1-\vx^T\vbeta_2)+ (\vx^T\vbeta_2)\frac{(\vbeta_1-\vbeta_2)^T\vv}{||\vbeta_1||_2^2} + (\vx^T\vbeta_2)(\vbeta_2^T\vv) \left( \frac{1}{||\vbeta_1||_2^2} -\frac{1}{||\vbeta_2||_2^2}\right)\\
	\triangleq&A_1+A_2+A_3,
	\end{align*}
	where the definition of $A_k$, $k=1,\cdots,3$, is clear from the context. 
	Note that for the identifiability condition assumes $\beta_1=1$. We have	%$\vbeta\in\bbB=\{\vbeta\in\bbB_0: ||\vbeta-\vbeta_0||_2\leq r,  ||\vbeta||_0\leq ks \}$, with $k>1$ and constant $r\leq 1$, we have $\beta_1=1$. Hence we conclude 
	\begin{align}
	1\leq ||\vbeta||_2\leq ||\vbeta_0||_2 + ||\vbeta-\vbeta_0||_2 \leq ||\vbeta_0||_2 + r\leq c_0,\label{beta_bound}
	\end{align} 
	for some constant $c_0>0$. By the Cauchy-Schwartz Inequality and (\ref{beta_bound}), we obtain that $\sup_{\vv\in\bbK(2ks+\widetilde{s}) }|A_1|\leq \sup_{\vv\in\bbK(2ks+\widetilde{s}) } |\vx^T\vbeta_1-\vx^T\vbeta_2| *\frac{||\vbeta_1||_2||\vv||_2}{||\vbeta_1 ||_2^2} \leq |\vx^T\vbeta_1-\vx^T\vbeta_2| $, and $\sup_{\vv\in\bbK(2ks+\widetilde{s}) }|A_2|\leq  |\vx^T\vbeta_2| *||\vbeta_1-\vbeta_2||_2 $. To bound $\sup_{\vv\in\bbK(2ks+\widetilde{s}) }|A_3|$, observe that
	\begin{align*}
	\sup_{\vv\in\bbK(2ks+\widetilde{s}) }|A_3|=&\sup_{\vv\in\bbK(2ks+\widetilde{s}) } |\vx^T\vbeta_2|*|\vbeta_2^T\vv|* \Big|\frac{||\vbeta_2||_2^2-||\vbeta_1||_2^2}{||\vbeta_1||_2^2||\vbeta_2||_2^2}\Big|\\
	=& \sup_{\vv\in\bbK(2ks+\widetilde{s}) }|\vx^T\vbeta_2|*|\vbeta_2^T\vv|* \frac{\big|(\vbeta_1+\vbeta_2)^T(\vbeta_1-\vbeta_2)\big|}{||\vbeta_1||_2^2||\vbeta_2||_2^2} \\
	\leq &\sup_{\vv\in\bbK(2ks+\widetilde{s}) }|\vx^T\vbeta_2|*||\vbeta_2||_2*||\vv||_2* \frac{||\vbeta_1+\vbeta_2||_2||\vbeta_1-\vbeta_2||_2}{||\vbeta_1||_2^2||\vbeta_2||_2^2}\\
	\leq& |\vx^T\vbeta_2|* ||\vbeta_1-\vbeta_2||_2* \frac{||\vbeta_1||_2+||\vbeta_2||_2}{||\vbeta_1||_2^2||\vbeta_2||_2}\\
	\leq& 2|\vx^T\vbeta_2|* ||\vbeta_1-\vbeta_2||_2,
	\end{align*}
	where the last inequality applies (\ref{beta_bound}).
	Combining all these results, we show that
	$$\sup_{\vv\in\bbK(2ks+\widetilde{s}) }\Big|\E( \vx^T\vv|\vx^T\vbeta_1) -\E(\vx^T\vv |\vx^T\vbeta_2)\Big|\leq 3 \big(|\vx^T\vbeta_1-\vx^T\vbeta_2|  +  |\vx^T\vbeta_2|*||\vbeta_1-\vbeta_2||_2\big).$$
	%Assumption~\ref{A1} implies that $P(\sup_{\vbeta\in\bbB,1\leq i\leq n}|\vx_i^T\vbeta|\leq M )=1$.
	Similarly, we can also show that 
	$$\sup_{\vv\in\bbK(2ks+\widetilde{s}) }\Big|\E( \vx^T\vv|\vx^T\vbeta_1) -\E(\vx^T\vv |\vx^T\vbeta_2)\Big|\leq 3 \big(|\vx^T\vbeta_1-\vx^T\vbeta_2|  +  |\vx^T\vbeta_1|*||\vbeta_1-\vbeta_2||_2\big).$$
	Hence, it implies that  
	$$\sup_{\vv\in\bbK(2ks+\widetilde{s}) }\Big|\E( \vx^T\vv|\vx^T\vbeta_1) -\E(\vx^T\vv |\vx^T\vbeta_2)\Big|\leq 3\big[|\vx^T\vbeta_1-\vx^T\vbeta_2|  +\min(|\vx^T\vbeta_1|,|\vx^T\vbeta_2|) * ||\vbeta_1-\vbeta_2||_2\big].$$

	\subsection{Verify the assumptions on G{\boldmath$(x^T\beta|\beta)$} in \ref{K4}-(c) }\label{sec:normal_G}
	%As shown above, we have that $\vx^T\vbeta_0|\vx^T\vbeta=t \sim N\left(\frac{\vbeta_0^T\vbeta}{||\vbeta||_2^2}t,||\vbeta_0||_2^2-\frac{(\vbeta_0^T\vbeta)^2}{||\vbeta||_2^2} \right)$.
	Let $\sigma_{\vbeta}^2 =||\vbeta_0||_2^2-\frac{(\vbeta_0^T\vbeta)^2}{||\vbeta||_2^2} $, and $\phi(\cdot)$ be the p.d.f of $N(0,1)$. We observe that 
	\begin{align*}
	G(t|\vbeta) =\E[f_0(\vx^T\vbeta_0)|\vx^T\vbeta=t] =\int f_0(z)\sigma_{\vbeta}^{-1} \phi\left(\frac{z-\frac{\vbeta_0^T\vbeta}{||\vbeta||_2^2}t}{\sigma_{\vbeta}}\right)dz.
	\end{align*}
	Let $w =\sigma_{\vbeta}^{-1} \left(z-\frac{\vbeta_0^T\vbeta}{||\vbeta||_2^2}t\right)$, by a transformation of variable, we have 
	\begin{align}
	G(t|\vbeta) =\int f_0\left(\sigma_{\vbeta} w+\frac{\vbeta_0^T\vbeta}{||\vbeta||_2^2}t\right) \phi(w)dw.\label{Gt_norm}
	\end{align}
	%Then we validate the assumptions in \ref{K4}-(c). 
	Let $\sigma_{\vbeta_1}^2 =||\vbeta_0||_2^2-\frac{(\vbeta_0^T\vbeta_1)^2}{||\vbeta_1||_2^2} $, and $\sigma_{\vbeta_2}^2 =||\vbeta_0||_2^2-\frac{(\vbeta_0^T\vbeta_2)^2}{||\vbeta_2||_2^2} $. Then we have
	\begin{align}
	G(t|\vbeta_1) -G(t|\vbeta_2)&=\int \left[f_0\Big(\sigma_{\vbeta_1}w+\frac{\vbeta_0^T\vbeta_1}{||\vbeta_1||_2^2}t\Big)-f_0\Big(\sigma_{\vbeta_2}w+\frac{\vbeta_0^T\vbeta_2}{||\vbeta_2||_2^2}t\Big)\right] \phi(w)dw\nonumber\\
	&=\int f_0'(\widetilde{w})\left[(\sigma_{\vbeta_1}-\sigma_{\vbeta_2})w+\frac{\vbeta_0^T\vbeta_1}{||\vbeta_1||_2^2}t-\frac{\vbeta_0^T\vbeta_2}{||\vbeta_2||_2^2}t \right] \phi(w)dw,\label{Gt_diff_norm}
	\end{align}
	where $\widetilde{w}$ is between $\sigma_{\vbeta_1}w+\frac{\vbeta_0^T\vbeta_1}{||\vbeta_1||_2^2}t$ and $\sigma_{\vbeta_2}w+\frac{\vbeta_0^T\vbeta_2}{||\vbeta_2||_2^2}t$. Assumption (A1)-(b) indicates that $f_0$ is differentiable and $\max_{1\leq i \leq n}| f_0'(\vx_i^T\vbeta_0)|\leq b$. Then we can obtain
	\begin{align*}
	|G(\vx^T\vbeta_1|\vbeta_1) -G(\vx^T\vbeta_2|\vbeta_2)|&\leq b|\sigma_{\vbeta_1}-\sigma_{\vbeta_2}|*\E|w| + b|t|*\left| \frac{\vbeta_0^T\vbeta_1}{||\vbeta_1||_2^2}-\frac{\vbeta_0^T\vbeta_2}{||\vbeta_2||_2^2}\ \right|.
	\end{align*}
	As $w\sim N(0,1)$, we have $\E|w| = \sqrt{2/\pi}$. According to analysis in Section~\ref{sec:normal_E}, we have that $\Big| \frac{\vbeta_0^T\vbeta_1}{||\vbeta_1||_2^2} -\frac{\vbeta_0^T\vbeta_2}{||\vbeta_2||_2^2} \Big|\leq c_1 ||\vbeta_1-\vbeta_2||_2 $, for some positive constant $c_1$. Without loss of generality, we assume $\sigma_{\vbeta_1}\geq \sigma_{\vbeta_2}>0$. then $|\sigma_{\vbeta_1}^2-\sigma_{\vbeta_2}^2| = (\sigma_{\vbeta_1}-\sigma_{\vbeta_2})^2 +2\sigma_{\vbeta_2}(\sigma_{\vbeta_1}-\sigma_{\vbeta_2})\geq (\sigma_{\vbeta_1}-\sigma_{\vbeta_2})^2$. We thus have 
	\begin{align*}
	|\sigma_{\vbeta_1}-\sigma_{\vbeta_2}|&\leq \sqrt{|\sigma_{\vbeta_1}^2-\sigma_{\vbeta_2}^2|} = \sqrt{\frac{(\vbeta_0^T\vbeta_1)^2}{||\vbeta_1||_2^2}-\frac{(\vbeta_0^T\vbeta_2)^2}{||\vbeta_2||_2^2}}\\
	&\leq   \sqrt{\frac{(\vbeta_0^T\vbeta_1)^2-(\vbeta_0^T\vbeta_2)^2}{||\vbeta_1||_2^2} + (\vbeta_0^T\vbeta_2)^2\left(\frac{1}{||\vbeta_1||_2^2}-\frac{1}{||\vbeta_2||_2^2}\right)}\\
	&\leq   \sqrt{\frac{\vbeta_0^T(\vbeta_1+\vbeta_2)*\vbeta_0^T(\vbeta_1-\vbeta_2)}{||\vbeta_1||_2^2} }+ \frac{|\vbeta_0^T\vbeta_2|}{||\vbeta_1||_2||\vbeta_2||_2}\sqrt{ (\vbeta_1+\vbeta_2)^T(\vbeta_1-\vbeta_2)}\\
	&\leq c_1 ||\vbeta_1-\vbeta_2||_2^{1/2},
	\end{align*} 
	for some positive constant $c_1$, where the last inequality applies (\ref{beta_bound}) and Assumption~\ref{A1}-(b).
	Note that $||\vbeta_0||_2$ is bounded by Assumption~\ref{A1}-(b). Combining all these results,  we conclude that for some constant $C>0$,
	\begin{align*}
	|G(t|\vbeta_1) -G(t|\vbeta_2)|\leq C\Big(|t|*||\vbeta_1-\vbeta_2||_2 + \sqrt{||\vbeta_1-\vbeta_2||_2}\Big). 
	\end{align*} 
	%Then we can obtain that
	%\begin{align*}
	%n^{-1}\sum_{i=1}^{n}\big|G(\vx_i^T\vbeta_1|\vbeta_1) - G(\vx_i^T\vbeta_2|\vbeta_2)\big|\leq &Cn^{-1}\sum_{i=1}^{n}|\vx^T\vbeta_1-\vx^T\vbeta_2| ++ C\sqrt{||\vbeta_1-\vbeta_2||_2}  \\
	%&Cn^{-1}\sum_{i=1}^{n}|\vx^T\vbeta_1|*||\vbeta_1-\vbeta_2||_2\\
	%\leq &C\sqrt{n^{-1}\sum_{i=1}^{n}|\vx^T\vbeta_1-\vx^T\vbeta_2|^2 }+ C\sqrt{||\vbeta_1-\vbeta_2||_2} \\
	%&+C\sqrt{n^{-1}\sum_{i=1}^{n}|\vx^T\vbeta_1|^2}*||\vbeta_1-\vbeta_2||_2 \\
	%\leq &c_1\big(||\vbeta_1-\vbeta_2||_2 +\sqrt{||\vbeta_1-\vbeta_2||_2} \big),
	%\end{align*}  
	Hence we have
	\begin{align*}
	\sup_{|t|\leq c_0\sqrt{s\log(p\vee n)}}\big[G(t|\vbeta_1) -G(t|\vbeta_2)\big]^2\leq& \sup_{|t|\leq c_0\sqrt{s\log(p\vee n)}}2C^2\Big(t^2*||\vbeta_1-\vbeta_2||_2^2 +  ||\vbeta_1-\vbeta_2||_2\Big)\\
	\leq & 2C^2\Big( 2rs\log( p\vee n) ||\vbeta_1-\vbeta_2||_2 +  ||\vbeta_1-\vbeta_2||_2\Big)\\
	\leq & c_1s\log( p\vee n) ||\vbeta_1-\vbeta_2||_2,
	\end{align*}  
	for some positive constants $C$, $c_1$,   and all $n$ sufficiently large. 
	We thus have validated the assumption on $G(\cdot|\vbeta)$ in \ref{K4}-(c).

	\subsection{Verify the assumptions on G{\boldmath$^{(1)}(x^T\beta|\beta)$} in \ref{K4}-(b) and \ref{K4}-(c)} \label{sec:normal_G1}

	To validate the assumption on $G^{(1)}(\vx^T\vbeta|\vbeta)$ in \ref{K4}-(b), we first note that 
	\begin{align*}
	G^{(1)}(t|\vbeta) =\frac{\vbeta_0^T\vbeta}{||\vbeta||_2^2}\int f_0'\Big(\sigma_{\vbeta} w+\frac{\vbeta_0^T\vbeta}{||\vbeta||_2^2}t\Big) \phi(w)dw,
	\end{align*}
	by (\ref{Gt_norm}),	where $\sigma_{\vbeta}^2 =||\vbeta_0||_2^2-\frac{(\vbeta_0^T\vbeta)^2}{||\vbeta||_2^2} $.
	
	By Taylor expansion,  for some $\widetilde{w}$ between $\sigma_{\vbeta} w+\frac{\vbeta_0^T\vbeta}{||\vbeta||_2^2}\vx^T\vbeta$ and $\vx^T\vbeta_0$,
	\begin{align*}
	&G^{(1)}(\vx^T\vbeta|\vbeta) -G^{(1)}(\vx^T\vbeta_0|\vbeta_0)\\
	=&\int \left[\frac{\vbeta_0^T\vbeta}{||\vbeta||_2^2}f_0'\Big(\sigma_{\vbeta} w+\frac{\vbeta_0^T\vbeta}{||\vbeta||_2^2}\vx^T\vbeta\Big)-f_0'(\vx^T\vbeta_0)\right] \phi(w)dw\\
	=&\frac{\vbeta_0^T\vbeta}{||\vbeta||_2^2}\int \left[f_0'\Big(\sigma_{\vbeta} w+\frac{\vbeta_0^T\vbeta}{||\vbeta||_2^2}\vx^T\vbeta\Big)-f_0'(\vx^T\vbeta_0)\right] \phi(w)dw + f_0'(\vx^T\vbeta_0) \left(\frac{\vbeta_0^T\vbeta}{||\vbeta||_2^2}-1\right)\\ 
	\triangleq& D_1+D_2,
	\end{align*}
	where the definitions of $D_1$ and $D_2$ are clear from the context. By the assumptions for $f_0(\cdot)$ in Assumption~\ref{A1}-(b), we have $|D_2|\leq b*\frac{|\vbeta^T(\vbeta_0-\vbeta)|}{||\vbeta||_2^2}\leq b||\vbeta-\vbeta_0||_2 $, by (\ref{beta_bound}). 
	
	To bound $|D_1|$, by Taylor expansion and Assumption~\ref{A1}-(b), there exist some positive constants $c$, $c_1$, $c_2$, such that 
	\begin{align*}
	|D_1|\leq  &  \left|c\int  f_0''(\vx^T\vbeta_0)\left(\sigma_{\vbeta} w+\frac{\vbeta_0^T\vbeta}{||\vbeta||_2^2}\vx^T\vbeta-\vx^T\vbeta_0\right)  \phi(w)dw\right|\\
	& + c_1  \int  \left(\sigma_{\vbeta} w+\frac{\vbeta_0^T\vbeta}{||\vbeta||_2^2}\vx^T\vbeta-\vx^T\vbeta_0\right)^2  \phi(w)dw \\ 
	\leq &  c_2\left|\frac{\vbeta_0^T\vbeta}{||\vbeta||_2^2}\vx^T\vbeta-\vx^T\vbeta_0 \right| + c_2  \sigma_{\vbeta}^2 +c_2\left(\frac{\vbeta_0^T\vbeta}{||\vbeta||_2^2}\vx^T\vbeta-\vx^T\vbeta_0\right)^2.
	\end{align*}
	Similarly as in Section~\ref{sec:normal_G}, we can obtain that
	\begin{align*}
	\left|\frac{\vbeta_0^T\vbeta}{||\vbeta||_2^2}\vx^T\vbeta-\vx^T\vbeta_0 \right|&\leq c ||\vbeta_0||_2*\left(|\vx^T\vbeta-\vx^T\vbeta_0 | + |\vx^T\vbeta_0|*||\vbeta-\vbeta_0||_2\right),\\
	\mbox{and }\qquad	\sigma_{\vbeta}^2& = ||\vbeta_0||_2^2-\frac{(\vbeta_0^T\vbeta)^2}{||\vbeta||_2^2} \leq ||\vbeta_0||_2^2*||\vbeta-\vbeta_0||_2.
	\end{align*}
	Then we have for some positive constant $C$,
	\begin{align*} 
	\big|G^{(1)}(\vx^T\vbeta|\vbeta) - G^{(1)}(\vx^T\vbeta_0|\vbeta_0)\big|\leq& C \Big[|\vx^T\vbeta-\vx^T\vbeta_0 | +|\vx^T\vbeta-\vx^T\vbeta_0 | ^2 \\
	&+ (1+|\vx^T\vbeta_0|)*||\vbeta-\vbeta_0||_2 +|\vx^T\vbeta_0|^2*||\vbeta-\vbeta_0||_2^2\Big] .
	\end{align*}   
	Note that  $  ||\vbeta-\vbeta_0||_2^2\leq  2$. We have
	\begin{align*}
	&\sup_{ \vbeta \in \bbB }n^{-1}\sum_{i=1}^{n}||\vbeta-\vbeta_0||_2^{-2}\big[G^{(1)}(\vx_i^T\vbeta|\vbeta) - G^{(1)}(\vx_i^T\vbeta_0|\vbeta_0)\big]^2\\
	\leq &c_0\sup_{ \vbeta \in \bbB }n^{-1}\sum_{i=1}^{n} \Big[\frac{|\vx_i^T(\vbeta-\vbeta_0) |^2}{||\vbeta-\vbeta_0||_2^2} +\frac{|\vx_i^T(\vbeta-\vbeta_0) | ^4 }{||\vbeta-\vbeta_0||_2^2}+ 1+|\vx^T\vbeta_0|^2 +|\vx^T\vbeta_0|^4*||\vbeta-\vbeta_0||_2^2\Big]  \\
	\leq& c_1,
	\end{align*} 
	with probability at least $1-\exp[-c_2s\log (p\vee n)]$, for some positive constants $c_0$, $c_1$, $c_2$ and all $n$ sufficiently large. In the above, the first part of the inequality applies Lemma~\ref{lem15NCL},  the second part of the inequality applies Lemma~\ref{lem:cube_rate}, and the remaining part applies Lemma~\ref{lem14NCL} and Lemma~\ref{lem:cube_rate}. Then we validate Assumption \ref{K4}-(b).
	
	Next we validate assumptions for $G^{(1)}(\vx^T\vbeta|\vbeta)$ in \ref{K4}-(c). By Assumption~\ref{A1}-(b), for some positive constant $c_1$,
	\begin{align*}
	&\left|G^{(1)}(t|\vbeta_1) -G^{(1)}(t|\vbeta_2)\right|\\
	&\leq \left| \frac{\vbeta_0^T\vbeta_1 }{||\vbeta_1||_2^2}-\frac{\vbeta_0^T\vbeta_2}{||\vbeta_2||_2^2}\right|\int f_0'\left(\sigma_{\vbeta_1}w+\frac{\vbeta_0^T\vbeta_1}{||\vbeta_1||_2^2}t\right) \phi(w)dw\\
	&+\frac{|\vbeta_0^T\vbeta_2|}{||\vbeta_2||_2^2}\int \left|f_0'\left(\sigma_{\vbeta_1}w+\frac{\vbeta_0^T\vbeta_1}{||\vbeta_1||_2^2}t\right)-f_0'\left(\sigma_{\vbeta_2}w+\frac{\vbeta_0^T\vbeta_2}{||\vbeta_2||_2^2}t\right)\right| \phi(w)dw\\
	\leq &b\left| \frac{\vbeta_0^T\vbeta_1 }{||\vbeta_1||_2^2}-\frac{\vbeta_0^T\vbeta_2}{||\vbeta_2||_2^2}\right| + c_1\E|w|*\frac{|\vbeta_0^T\vbeta_2|}{||\vbeta_2||_2^2}*|\sigma_{\vbeta_1}-\sigma_{\vbeta_2}|  +  c_1\frac{|\vbeta_0^T\vbeta_2|}{||\vbeta_2||_2^2}*\left| \frac{(\vbeta_0^T\vbeta_1)}{||\vbeta_1||_2^2}t-\frac{(\vbeta_0^T\vbeta_2)}{||\vbeta_2||_2^2}t\right|\\
	\triangleq&|E_1|+|E_2|+|E_3|,
	\end{align*}
	where the definition of $E_k$, $k=1,\cdots,3$, is clear from the context. As discussed above, we have that $|E_1|\leq 3b||\vbeta_0||_2 *||\vbeta_1-\vbeta_2||_2$, $|E_2|\leq c_2*||\vbeta_0||\sqrt{||\vbeta_1-\vbeta_2||_2}$, and %$|E_3|\leq c_2*||\vbeta_0||^2\big(|\vx^T\vbeta_1-\vx^T\vbeta_2| +  ||\vbeta_1-\vbeta_2||_2\big)$ 
	$|E_3|\leq c_2|t|*||\vbeta_0||^2* ||\vbeta_1-\vbeta_2||_2 $ for some constant $c_2>0$. In conclusion, for some positive constant $C$, we have
	$\big|G^{(1)}(t|\vbeta_1) -G^{(1)}(t|\vbeta_2)\big|\leq C\Big[|t| * ||\vbeta_1-\vbeta_2||_2 + \sqrt{||\vbeta_1-\vbeta_2||_2}\Big].$
	Then applying similar techniques as those in Section~\ref{sec:normal_G}, we can validate Assumption \ref{K4}-(c).

	\section{Algorithms and additional numerical results}\label{sec:algo_numeric}
	\subsection{Pseudo Codes for the algorithms 
		in Section~\ref{sec:algo_est}}\label{sec:algo_supp} 
	In this subsection, we provide pseudo codes for the algorithms introduced in Section~\ref{sec:algo_est}. 
	Algorithm~\ref{Profiled} is the main algorithm for 
	solving the  penalized
	high-dimensional profiled estimating equation
	for the initial estimator $\vbetah$.
	It extends the proximal algorithm  
	\citep{Nesterov, agarwal2012} to estimate the profiled semiparametric estimator.
	Algorithm~\ref{project} describes the 
	details of the projection step in Algorithm~\ref{Profiled},
	using an algorithm introduced in \citet{Duchi2008}.

	\begin{algorithm} [!h] 
		\caption{An algorithm for solving the penalized profiled estimating equation.\\ Input: initial value $\vbeta^0$, $\lambda$, $\gamma_u$, data $\{\vx_i,\widetilde{Y}_i\}_{i=1}^n$} \label{Profiled} 
		\begin{algorithmic}[1]
			\State Set $t=1$, $\vbeta^t= \vbeta^{t-1}= \vbeta^0$, coef.err = $||\vbeta^0||_2+1$, model.err$^{t-1}= $model.err$^{t-2}=\Var(\widetilde{Y}_i)$. 
			\While {coef.err $> 0.01*||\vbeta^{t-1}||_2$ or model.err$^{t-1}<$ model.err$^{t-2}$} 
			%\State $t\leftarrow t+1$.
			\State $h^t\leftarrow 0.9n^{-1/6}\min\{\mbox{std}(\vx_i^T\vbeta^{t}), \mbox{IQR}(\vx_i^T\vbeta^{t})/1.34\}$. 
			\State $w_{ij}^t \leftarrow K\Big(\frac{\vx_i^T\vbeta^{t} -\vx_j^T\vbeta^{t} }{h^t}\Big)$; $w_{ij}^{'t} \leftarrow (h^{t})^{-1}K'\Big(\frac{\vx_i^T\vbeta^{t} -\vx_j^T\vbeta^{t}}{h^t}\Big)$.
			\State $\widehat{G}(\vx_i^T\vbeta^{t}|\vbeta^{t}) \leftarrow \frac{ \sum_{j\neq i}w_{ij}^t\widetilde{Y}_j}{\sum_{j\neq i}w_{ij}^t}$.
			\State $\widehat{G}^{(1)}(\vx_i^T\vbeta^{t}|\vbeta^{t}) \leftarrow \frac{ \sum_{j\neq i}w_{ij}^{'t}\widetilde{Y}_j}{\sum_{j\neq i}w_{ij}^t} - \widehat{G}(\vx_i^T\vbeta^{t}|\vbeta^{t}) *\sum_{j\neq i}w_{ij}^{'t}$.
			\State $\widehat{\E}(\vx_i|\vx_i^T\vbeta^{t})\leftarrow \frac{ \sum_{j\neq i}w_{ij}^t\vx_j}{\sum_{j\neq i}w_{ij}^t}$; $\vxh_{i}^t \leftarrow \vx_i - \widehat{\E}(\vx_i|\vx_i^T\vbeta^{t})$.
			\State model.err$^{t} \leftarrow  \frac{1}{n}\sum_{i=1}^n[\widetilde{Y}_i-\widehat{G}(\vx_i^T\vbeta^{t}|\vbeta^{t}) ]^2.$ 
			%\State $\vbeta^{t+1} \leftarrow \argmin\limits_{\vbeta\in\bbR^p: |\beta_1|=1} \frac{1}{2n}\sum_{i=1}^n[\widetilde{Y}_i - \widehat{G}(\vx_i^T\vbeta^{t}) - \widehat{G}^{(1)}(\vx_i^T\vbeta^{t}) \vxh_{i}^{t\, T}(\vbeta-\vbeta^{t}) ]^2 + \lambda||\vbeta||_1. $	
			\State $\vbeta_{-1}^{t+1} \leftarrow \argmin\limits_{\vbeta_{-1}\in\bbR^{p-1}: ||\vbeta_{-1}||_1\leq \rho}
			\frac{\gamma_u}{2}||\vbeta_{-1}-\vbeta_{-1}^t||_2^2 + [\vS_n(\vbeta^t, \widehat{G},\widehat{\E})]^T(\vbeta_{-1}-\vbeta_{-1}^t)  + \lambda||\vbeta_{-1}||_1 $,	by (\ref{update1}) and Algorithm~\ref{project}.
			\State $\vbeta^{t+1} \leftarrow (1,(\vbeta_{-1}^{t+1})^T)^T.$
			\State coef.err$ \leftarrow  ||\vbeta^{t+1}-\vbeta^{t}||_2.$
			\State $t\leftarrow t+1$, $\gamma_u\leftarrow 2*\gamma_u$.
			\EndWhile	
			\State Output $\vbeta^{t}$.	
		\end{algorithmic} 
	\end{algorithm}	
	
	\begin{algorithm} [!h] 
		\caption{An algorithm for projecting $\vbeta$ onto the $L_1$-ball: $\{\vbeta:||\vbeta||_1\leq \rho\}$.\\ Input: initial value $\vbeta$, $\rho$} \label{project} 
		\begin{algorithmic}[1]
			\If {$||\vbeta||_1\leq \rho$}
			\State Output $\vbeta$.	
			\Else  
			\State Sort $\{|\beta_j|\}_{j=1}^p$ into $b_1\geq b_2\geq \cdots\geq b_p$.
			\State Find $J = \max\{2\leq j\leq p: b_j - \frac{(\sum_{r=1}^{j}b_r) - \rho}{j}>0\}$, and $\delta = \frac{1}{J}[(\sum_{r=1}^{J}b_r) - \rho]$.
			\State Output $\vbeta^o = T_s (\vbeta, \delta)$.	
			\EndIf	 
		\end{algorithmic} 
	\end{algorithm}	 
	 
	\newpage
	\subsection{Computation of $\bf{d}_j(${\boldmath$\beta$}$, \eta)$}\label{sec:algo_dj}
	In Section~\ref{sec:inf_method}, we introduce a nodewise Dantzig estimator $\vd_j(\vbetah,\eta)$, as defined in (\ref{dj_def}), to 
	obtain the approximate inverse of $\nabla \vS_n(\vbetah,\widehat{G},\widehat{\E})$.
	This estimator can be solved via a linear programming problem as follows: 
	\begin{align}
	\min\limits_{\bm{\xi}^+,\bm{\xi}^-\in\bbR^{p-2} }||\vxi^+||_1+||\vxi^-||_1 \label{linear_prog} 
	& \mbox{ subject to } \vxi^+\geq 0, \vxi^-\geq 0, \mbox{ and }\\ 
	\frac{1}{n}\sum_{i=1}^n [\widehat{G}^{(1)}(\vx_i^T\vbetah|\vbetah)]^2\xh_{i,k} \vxh_{i,-j*}^T(\vxi^+-\vxi^-)&\geq \frac{1}{n}\sum_{i=1}^n [\widehat{G}^{(1)}(\vx_i^T\vbetah|\vbetah)]^2 \xh_{i,j}\xh_{i,k}-\eta, \mbox{ for all }  k\neq 1, k\neq j, \nonumber\\
	\frac{1}{n}\sum_{i=1}^n [\widehat{G}^{(1)}(\vx_i^T\vbetah|\vbetah)]^2\xh_{i,k} \vxh_{i,-j*}^T(\vxi^+-\vxi^-)&\leq \frac{1}{n}\sum_{i=1}^n [\widehat{G}^{(1)}(\vx_i^T\vbetah|\vbetah)]^2 \xh_{i,j}\xh_{i,k}+ \eta, \mbox{ for all }  k\neq 1, k\neq j, \nonumber 
	\end{align} 
	for any given $j\in\{2,\cdots,p\}$. Then $(\bm{\xi}^+-\bm{\xi}^-)$ is an estimator of $\vd_j$.
	%\bfblue{Connect $\bm{\xi}^+$ and $\bm{\xi}^-$ to $\vd_j$.}
	In our numerical analysis, we apply the function ``lp'' in the R package \texttt{lpSolve}  \citep{lpSolve} for linear programming.
	
	\subsection{Additional numerical results}\label{sec:add_numeric}
	\paragraph{Example 1}(Effect of tuning parameter $\eta$).  
	We compare inference performance with different choices of $\eta$
	($\eta=15h$, $20h$, $25h$, and $30h$)
	and the $\lambda$ selected by 5-fold cross-validation. 
	We consider the same model as in Section~\ref{sec:sim_res} in the main paper, with $n=300$ and $p=200$.
	% For cross-validation, we search the best %$\eta$ among $\{0.75h, h,\cdots,3h\}$, which %minimize $$\mbox{MSE}(\eta) = %n^{-1}\sum_{j\in\kG}\sum_{i=1}^n
	
	%[\widehat{G}^{(1)}(\vx_i^T\vbetah|\vbetah)]^2%[ \xh_{i,j}-\vxh_{i,-j}^T\vd_j(\vbeta,\eta) %]^2,$$
	%given the estimator $\vbetah$ and the group %index set $\kG$.
	Table~\ref{tab:infcv} summarizes the average Type I errors and powers. We observe that 
	inference performance is not very sensitive to $\eta$. Also, the choice $\eta=25h$ leads to
	performance similar as that obtained using 
	$\eta$ chosen by cross-validation.
	\begin{table}[!h]
		\centering
		\caption{Performance of the bootstrap procedure for simultaneous testing with different choices of $\eta$.}{
			\begin{tabular}{c|c|ccccc}
				\hline 
				\multirow{2}{*}{$\eta$}& Type I error & \multicolumn{5}{c}{Power} \\
				\cline{2-7}
				& $\kG_1$ &$\kG_2$ &$\kG_3$ &$\kG_4$  &$\kG_5$ &$\kG_6$  \\ 
				\hline  
				15h& 3.6\%& 79.6\% & 93.2\% & 94.4\%& 94.8\% &100\% \\  
				20h&5.0\% & 92.8\% & 95.6\% &97.0\% & 97.6\% &100\%\\ 
				25h& 5.6\%& 96.4\% & 96.2\% & 97.8\% & 98.6\% & 100\%\\  	
				30h&4.2\% & 96.8\% & 96.6\% &98.2\% & 99.2\% &100\%\\ 
				CV& 4.6\%& 97.2\% & 96.8\% & 98.4\% & 99.0\% &100\%\\  	
				\hline 
		\end{tabular}}
		\label{tab:infcv}
	\end{table}

	\paragraph{Example 2}(Comparison
	with alternative algorithms). We compare the proposed semiparametric procedure
	with the nonparametric O-learning procedure \citep{zhao2012estimating}, and the decision list based approach  \citep{list_paper}.
	We use the ``DTRlearn2'' R package 
	with the Gaussian kernel for O-learning   \citep{DTRlearn2} and the ``listdtr'' R package 
	for the decision list approach \citep{listdtr}.  As the alternative procedures do not perform inference, our comparison is focused on estimating the optimal value function. Given an estimated decision rule indexed by $\vbetah$, we can estimate the optimal value function  by
	$\widehat{V}(\vbetah) =\frac{\sum_{i=1}^n I[A_i = d(\vx_i)]Y_i}{\sum_{i=1}^n I[A_i = d(\vx_i)]}.$
	%For our model and estimation process, the %estimated treatment regime should be %$\widehat{d}(\vx,\vbetah,\widehat{G}) = %I[\widehat{G}(\vx^T\vbetah|\vbetah)>0]$. 

	We consider two different settings. The first setting (setting 1) corresponds to the index model in Section 4.2
	in the main paper, for which the optimal value is $3.423$ based an independent Monte Carlo simulation with $10^7$ replicates. In the second setting (setting 2),  $Y=1+\vx^T\veta + (A-\frac{1}{2})f_0(\vx) + \ep$, where 
	$\ep\sim N(0, 1)$, $A\sim \mbox{Bernoulli}(0.5)$, $\vx=(x_1,\cdots,x_p)^T$ has elements independently distributed as Uniform$(-1,1)$, $\veta = (2, 1, 0.5, 0,\cdots,0)^T$ and $f_0(\vx)=20(1-x_1^2-x_2^2)(x_1^2+x_2^2-0.36)$. 
	%The shape of the boundary in this model is a %ring. 
	The optimal value of setting 2 is $2.443$, based on an independent Monte Carlo simulation with $10^7$ replicates.
	
	%We estimate the optimal value function based %on the data we use to train the model.

	\begin{table}[!h]
		\centering
		\caption{Estimated bias (with standard error in the parentheses) for the optimal value and the average match ratios}{
			\begin{tabular}{cccccccc}
				\hline
				$n$&$p$& &New & O-learning& List learning\\ 
				\hline  
				\multicolumn{6}{c}{Setting 1}\\ 
				\multirow{4}{*}{300} &
				\multirow{2}{*}{200}&value & -0.034 (0.008)  & 0.195 (0.011) & -0.258 (0.010)\\ 
				&&MR& 93.76\%  & 75.89\% & 79.91\%\\  
				&\multirow{2}{*}{800}&value &-0.055 (0.008)  & 0.277 (0.010) & -0.300 (0.028)\\   
				&&MR& 92.67\%  & 63.70\% & 79.05\%\\  
				\hline 
				\multirow{4}{*}{500}
				&\multirow{2}{*}{200}&value & -0.035 (0.006) & 0.172 (0.009) &-0.289 (0.017) \\  
				&&MR& 95.51\%  & 80.82\% & 81.43\%\\  
				&\multirow{2}{*}{800}&value & -0.032 (0.006)  & 0.288 (0.008) &-0.272 (0.017) \\ 
				&&MR& 94.75\%  & 67.93\% &80.77\% \\  
				\hline 
				\hline  
				\multicolumn{6}{c}{Setting 2}\\ 
				\multirow{4}{*}{300} &
				\multirow{2}{*}{200}&value & -0.642 (0.010)  & 0.269 (0.010) &-0.676 (0.012) \\  
				&&MR& 50.60\%  & 49.98\% & 49.98\%  \\  
				&\multirow{2}{*}{800}&value & -0.649 (0.009) & 0.464 (0.009) & -0.647 (0.024)\\  
				&&MR& 50.78\%  & 49.98\% &50.07\% \\  
				\hline 
				\multirow{4}{*}{500}
				&\multirow{2}{*}{200}& value& -0.673 (0.007) & 0.141 (0.010) & -0.689 (0.020)\\  
				&&MR& 50.52\%  & 49.98\% &49.95\% \\  
				&\multirow{2}{*}{800}& value& -0.655 (0.008) & 0.434 (0.007) & -0.684 (0.017) \\ 
				&&MR& 50.68\%  & 49.96\% & 50.07\%\\  
				\hline  
		\end{tabular}}
		\label{tab:value}
	\end{table}
	
	Table~\ref{tab:value} summarizes the average bias and standard error 
	for estimating the optimal values for the two settings
	for $p=200, 800$ and $n=300, 500$. It also reports the average match ratio (MR). MR is estimated as the percentage of times the estimated optimal decision rule coincides with the true optimal decision rule, the latter of which is computed using an independent sample of size $10^4$.  Due to the computational cost, for the decision list based estimators, we run 200 simulations. For the other two estimators, the results are based on 500 simulation runs.
	
	We have the following observations. (1) In setting 1, our proposed method 
	has smaller biases for estimating the optimal value comparing with the two other approaches. 
	This is likely due to the fact the proposed method is semiparametric. In contrast,  
	the other two approaches do not make use of the model structure in estimating the optimal decision rule.
	(2) In setting 2, O-learning has smaller bias for estimating the optimal value.  
	It is noted that in this setting the 
	model does not have the index form and hence the proposed semiparametric procedure 
	is based on a misspecified model.
	(3) In both settings, the performance of 
	O-learning deteriorates as $p$ gets larger while the performance of the new method is stable.

	\paragraph{Example 3}(Correlated design with discrete covariates). 
	In this example, the covariates include three discrete components, which are independent and uniformly distributed on the set $\{-1,0,1\}$. All the other covariates follow a $(p-3)$-dimensional multivariate normal distribution with mean zero and covariance matrix $\vSigma$, with $\Sigma_{i,j}= 0.5^{|i-j|} $. The three discrete variables are the fifth and the last two of the $p$ covariates.
	The model has the same form as the example in Section 4.2 of the main paper and has $\vbeta_{0} = (1,-1,-0.8,0.6,-0.5,0,\cdots,0)$.
	
	Table~\ref{tab:est_ar5} summarizes the estimation results for $n=300,500$ and $p=200,800$, based on 500 simulations. We observe that the proposed profiled estimator has satisfactory performance in this experiment.
	
	\begin{table}[!h]
		\centering
		\caption{Performance of the estimator for the correlated design with discrete covariates}{
			\begin{tabular}{cccccc}
				\hline 
				$n$&$p$& $l_1$ error & $l_2$ error & False Negative & False Positive \\ 
				\hline 
				\multirow{2}{*}{300}
				&200& 0.81 (0.02)& 0.32 (0.00) & 0.00 (0.00) & 8.31 (0.30) \\  
				&800& 1.10 (0.02)& 0.41 (0.01) & 0.01 (0.01) & 14.82 (0.61) \\   
				\hline 
				\multirow{2}{*}{500}
				&200& 0.53 (0.01)& 0.22 (0.00) & 0.00 (0.00) &6.59 (0.25) \\  
				&800& 0.73 (0.01)& 0.28 (0.00) & 0.00 (0.00) & 12.48 (0.46) \\   
				\hline
		\end{tabular}}
		\label{tab:est_ar5}
	\end{table}
	
	Next we investigate the proposed wild bootstrap inference procedure for testing group hypotheses with $\kG_1 = \{6,7,8,9\}$, $\kG_2 = \{5,6,7,8,9\}$, $\kG_3= \{4,6,7,8,9\}$, $\kG_4 = \{4,5,6,7,8,9\}$, $\kG_5 = \{3,6,7,8,9\}$ and $\kG_6 = \{2,6,7,8,9\}$. Note that $\kG_2$ includes a discrete variable. Table~\ref{tab:inference_ar5} summarizes the results based on 1000 bootstrap samples and 500 simulation runs.  We observe that 
	the estimated type I errors and powers are reasonable for all scenarios. 
	
	\begin{table}[!h]
		\centering
		\caption{Performance of the wild bootstrap inference procedure for the correlated design with discrete covariates.}{
			\begin{tabular}{cc|c|ccccc}
				\hline 
				\multirow{2}{*}{$n$}&\multirow{2}{*}{$p$}& Type I error & \multicolumn{5}{c}{Power} \\
				\cline{3-8}
				&& $\kG_1$ &$\kG_2$ &$\kG_3$ &$\kG_4$  &$\kG_5$ &$\kG_6$  \\ 
				\hline 
				\multirow{2}{*}{300}
				&200& 5.6\%& 98.8\% & 96.8\% & 97.2\% & 100\% & 100\%\\  	
				&800& 6.6\%& 84.4\%& 86.6\%  & 88.0\%& 99.8\% &91.8\% \\   
				\hline 
				\multirow{2}{*}{500}
				&200& 5.0\%& 96.6\% &97.8\% & 98.0\% & 100\% & 100\%\\  
				&800& 7.2\%& 89.8\% & 94.8\% & 95.6\%& 100\% & 92.2\% \\   
				\hline 
		\end{tabular}}
		\label{tab:inference_ar5}
	\end{table}

	\paragraph{Example 4}(Addition results for the real-data example in Section~\ref{sec:realdata} of the main paper).  In Table~\ref{tab:realcoef}, we report the estimated coefficients for the variables in Table~\ref{tab:realdata}.
	In the table, 
	``insulin'' stands for fasting insulin, ``Cr'' 
	stands for creatinine, and ``waist'' stands for waist circumference.
	\begin{table}[!h]
		\centering
		\caption{Real data analysis: profiled estimator for variables in Table~\ref{tab:realdata} of the main paper}{
			\begin{tabular}{c|ccccccc}
				\hline 
				Variable& fasting insulin&creatinine& BMI&waist \\
				Coef& $1$&$0.0011$ &  $-0.0070$ & $-0.0047$  \\\hline
				Variable &HbA$_{1c}$& HomaS&Cr:insulin&Cr:BMI&\\
				Coef&$-0.0519$&  $0.0071$  & $0$& $0$\\\hline
				Variable& Cr:waist& Cr:HbA$_{1c}$&Cr:HomaS& insulin:BMI \\
				Coef&$0$ &$-0.0171$& $0.0110$ &$0$ \\\hline 
				Variable& insulin:waist & insulin:HbA$_{1c}$& insulin:HomaS& BMI:waist\\
				Coef&$0$ &$0$  &$0$& $-0.0046$  \\\hline
				Variable&BMI:HbA$_{1c}$&BMI:HomaS& waist:HbA$_{1c}$&waist:HomaS\\
				Coef&$-0.0349$&$0.0051$&$-0.0399$& $0.0067$ \\\hline
				Variable&HbA$_{1c}$:HomaS& LDL-C& total cholesterol& age\\
				Coef&$-0.0002$ &$0.0030$ & $0.0032$ &$0.0022$  \\\hline
				Variable& weight&&\\ 
				Coef& $0$ & & \\
				\hline 
		\end{tabular}}
		\label{tab:realcoef}
	\end{table}
	
	\newpage
	As an example of using the estimated model to interpret the covariate effect on the outcome. 
	we consider the effect of
	baseline HbA$_{1c}$ on the outcome
	of receiving the recommended treatment. 
	Figure~\ref{fig:realdata} plots the 
	$\widehat{G}(\vx^T\vbetah|\vbetah)$ versus baseline HbA$_{1c}$ while fixing all the other covariates at their respective sample averages. The plot suggests that for such an average patient, receiving pioglitazone (treatment 0)
	is likely to reduce the level of HbA$_{1c}$, and larger benefit is expected for patients with a smaller value of baseline HbA$_{1c}$.
	%that the outcome increases when the %baseline HbA$_{1c}$ increases from $6$ to %$8.430$, and then decreases if the baseline %HbA$_{1c}$ keeps increasing. 
	
	\begin{figure}[!h]
		\centering
		\includegraphics[width=3in]{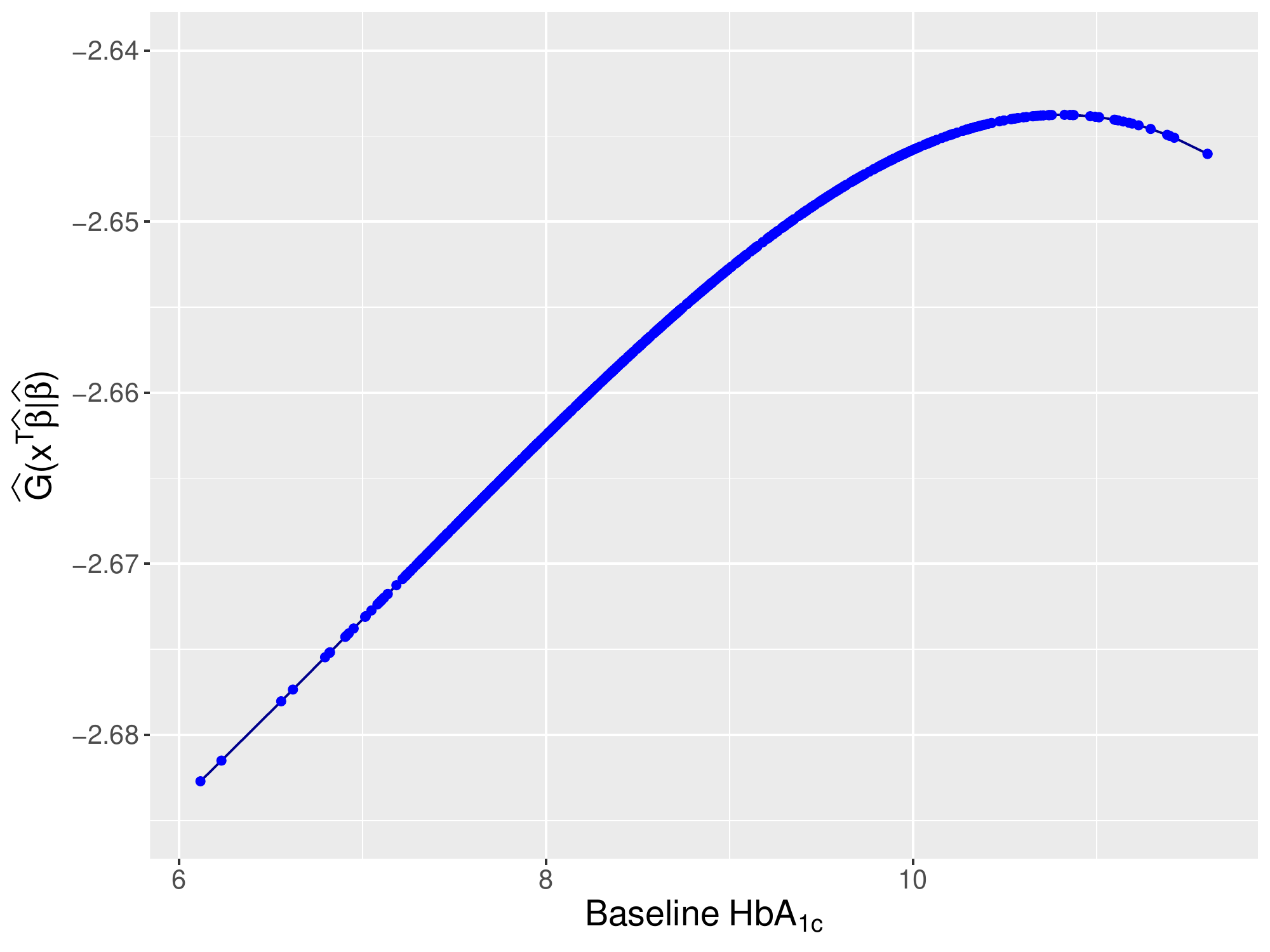}
		\caption{Plot of $\widehat{G}(\vx^T\vbetah|\vbetah)$ versus baseline HbA$_{1c}$ while fixing all the other covariates at their respective sample averages.}
		\label{fig:realdata}
	\end{figure}
	
	\paragraph{Example 5} (Extension to observational studies). 
	We investigate the procedure proposed in 
	Section~\ref{sec:discuss} of the main paper for observation studies.
	We consider the same model as in Section~\ref{sec:sim_res} of the main paper, except that $A$ is generated according to $P(A=1|\vx) = [1+\exp(-\vx^T\vxi)]^{-1}$,
	where $\vxi = (0.2,0.2,-0.4,0,\cdots,0)^T$.
	We estimate the propensity score via $L_1$-regularized logistic regression. Table~\ref{tab:obs} suggests the promising performance for the proposed estimator
	for observational studies.

	\begin{table}[!h]
		\centering
		\caption{Performance of the penalized profile least-squares estimator}{
			\begin{tabular}{cccccc}
				\hline 
				$n$&$p$& $l_1$ error & $l_2$ error & False Negative & False Positive \\ 
				\hline 
				\multirow{2}{*}{300}
				&200& 1.00 (0.01)& 0.42 (0.00) & 0.06 (0.01) & 9.13 (0.30) \\  
				&800& 1.26 (0.02)& 0.49 (0.00) & 0.10 (0.01) & 16.29 (0.62) \\  
				\hline 
				\multirow{2}{*}{500}
				&200& 0.72 (0.01)& 0.30 (0.00) & 0.00 (0.00) & 8.32 (0.26) \\  
				&800& 0.97 (0.01)& 0.39 (0.00) & 0.01 (0.00) & 15.18 (0.48) \\   
				\hline
		\end{tabular}}
		\label{tab:obs}
	\end{table}

\end{document}